\documentclass[beltcrest]{ociamthesis}  

\usepackage[latin1]{inputenc}
\usepackage{latexsym,diagbox,stackrel}
\usepackage{mathrsfs}
\usepackage{afterpage}
\usepackage{colortbl}
\usepackage{arydshln,leftidx,mathtools}
\usepackage{bbold}
\usepackage{caption}
\usepackage{titlesec}
\usepackage{pdflscape}
\usepackage{amssymb}
\usepackage{amsmath}
\usepackage{braket}
\usepackage{cite}
\usepackage{enumerate}
\usepackage{BeamerColor}
\usepackage{stmaryrd}
\usepackage{multirow}
\usepackage{setspace}
\usepackage{tikz}
\usepackage{graphicx}
\usepackage[hidelinks]{hyperref}
\usepackage{cleveref}
\usetikzlibrary{shapes.geometric, arrows}
\definecolor{light-gray}{gray}{0.95}
\definecolor{light-gray-table}{gray}{0.9}
\definecolor{light-grayI}{gray}{0.5}
\definecolor{light-grayII}{gray}{0.85}

\titleformat{\chapter}[display]
        {\normalfont\Large\bf\centering}{\chaptertitlename\ \thechapter}{10pt}{\LARGE\bf}
\crefformat{equation}{(#2#1#3)}

\usepackage
[
color=gray]
{todonotes} 

\newcommand{\half}{\frac{\scriptstyle 1}{\scriptstyle 2}}

\newcommand{\A}{\mathbb{A}}
\newcommand{\C}{\mathbb{C}}
\newcommand{\HH}{\mathbb{H}}
\newcommand{\CP}{\mathbb{CP}}

\newcommand{\R}{\mathbb{R}}
\renewcommand{\P}{\mathbb{P}}

\newcommand{\cH}{\mathcal{H}}
\newcommand{\scri}{\mathscr{I}}

\newcommand{\T}{\mathbb{T}}
\newcommand{\Z}{\mathbb{Z}}
\newcommand{\p}{\partial}
\newcommand{\dbar}{\bar\partial}
\newcommand{\e}{\mathrm{e}}
\newcommand{\g}{\mathfrak{g}}

\newcommand{\cG}{\mathcal{G}}
\newcommand{\cA}{\mathcal{A}}
\newcommand{\cC}{\mathcal{C}}
\newcommand{\cI}{\mathcal{I}}

\newcommand{\cM}{\mathcal{M}}
\newcommand{\cN}{\mathcal{N}}
\newcommand{\cO}{\mathcal{O}}
\newcommand{\cX}{\mathcal{X}}
\newcommand{\cV}{\mathcal{V}}

\renewcommand{\P}{\mathbb{P}}
\newcommand{\SL}{\mathrm{SL}}

\newcommand{\tr}{\mathrm{tr}}

\newcommand{\SU}{\, \mathrm{SU}}
\newcommand{\GL}{\mathrm{GL}}

\newcommand{\rd}{\, \mathrm{d}}
\newcommand{\Pf}{\, \mathrm{Pf}}
\newcommand{\pf}{\text{Pf} \,}
\newcommand{\pfr}{\text{Pf}^\prime \,}
\newcommand{\vol}{\mathrm{Vol}}

\newcommand{\be}{\begin{equation}\label}
\newcommand{\ee}{\end{equation}}
\newcommand{\bea}{\begin{eqnarray}\label}
\newcommand{\eea}{\end{eqnarray}}

\newcommand{\la}{\langle}
\newcommand{\ra}{\rangle}

\newtheorem{thm}{Theorem}

\newtheorem{propn}{Proposition}[section]

\newtheorem{lemma}{Lemma}[section]

\newcommand{\bd}{\bar{\delta}}
\newcommand{\wt}{\widetilde}
\newcommand{\bigslant}[2]{{\raisebox{.2em}{$#1$}\left/\raisebox{-.2em}{$#2$}\right.}}
\DeclareRobustCommand{\abinom}{\genfrac{\langle}{\rangle}{0pt}{}}
\def\d{\text{d}}

\title{{\huge \bf \sc Ambitwistor Strings:}\\[1ex]     
       {\LARGE \bf \sc Worldsheet Approaches to perturbative Quantum Field Theories}}   

\author{Yvonne Jasmin Geyer}             
\college{University College\\ Mathematical Institute}  

\degree{Doctor of Philosophy}     
\degreedate{Michaelmas Term 2015/2016}         

\begin{document}
\baselineskip=20pt plus1pt

\setcounter{secnumdepth}{3}
\setcounter{tocdepth}{2}

\hypersetup{pageanchor=false}
\pagestyle{plain}
\maketitle                  


\begin{abstract}
Tree-level scattering amplitudes in massless theories not only exhibit a simplicity entirely unexpected from Feynman diagrams, but also an underlying structure remarkably reminiscent of worldsheet theory correlators, yet essentially algebraic. These features can be explained by ambitwistor strings - two-dimensional chiral conformal field theories in an auxiliary target space, the complexified phase space of null geodesics, known as ambitwistor space. The aim of this thesis is to explore the ambitwistor string approach to understand these structures in 
amplitudes, and thereby provide a new angle  on quantum field theories.

In the first part of this thesis, the wide-ranging impact of ambitwistor strings on the study of tree-level amplitudes is highlighted in three developments: an extension of ambitwistor string worldsheet models to an extensive family of massless theories, emphasising the universality of ambitwistor strings for massless theories; a beautiful proof of the duality between asymptotic symmetries  and the low energy behaviour of a theory, relying on the contact geometry of the ambitwistor target space; and finally a twistorial representation of ambitwistor strings in four dimensions, leading to remarkably simple formulae for scattering amplitudes in Yang-Mills and gravity with any degree of supersymmetry.
 
The second part of this thesis focusses on proving a conjectured ambitwistor string formula for loop amplitudes, and extending the formalism to more general theories. Remarkably, residue theorems reduce the computationally challenging ambitwistor higher-genus expressions to simple formulae on nodal Riemann spheres. This idea is developed into a widely applicable framework for loop integrands, that 
is shown to be applicable to both supersymmetric and non-supersymmetric theories. In the case of supergravity, this provides strong evidence for the validity of the ambitwistor string at loop level, and explicit proofs are given for non-supersymmetric theories. Notably, this leads to a proposal for an all-loop integrand for gravity and Yang-Mills.

This thesis is based on \cite{Casali:2015vta,Geyer:2014lca,Geyer:2014fka,Geyer:2015bja,Geyer:2015jch} and has considerable overlap with these papers.
\end{abstract}
 

\hypersetup{pageanchor=true}
\begin{romanpages}          
\tableofcontents            
\end{romanpages}            

\chapter{Introduction}\label{chapter1}
Quantum field theory (QFT) has proven to be one of most successful mathematical frameworks in physics to date. Describing such diverse phenomena as particle physics, condensed matter and astrophysics, it also links to many areas of mathematics, most notably to topology and algebraic geometry. 
Its numerous successes have been both conceptual and relating to applications in physics; giving insight into the underlying mathematical structure as well as providing an indispensable method of probing fundamental physics. Among these advances, a few are particularly worth highlighting: Renormalisation not only provided a systematic approach to the infinities arising in quantum field theory calculations, but also revolutionised the understanding of quantum field theories via the concept of renormalisation group flow, explaining the omnipresence of renormalisable QFTs as low energy effective field theories. Moreover, Yang-Mills theory and non-abelian gauge theories bridge the gap from theory to experiment by tying into both mathematics and experimental physics via the standard model for particle physics, tested in remarkable precision calculations. Other remarkable feats include the theory of critical exponents in condensed matter physics, and the many applications of topological field theories to mathematics.\\

Quantum field theory unites the principles of quantum mechanics and relativity, and is traditionally described from a Lagrangian. In this approach, observables are calculated perturbatively around a classical vacuum solution via Feynman rules derived from the path integral.
However, there are strong hints that this traditional approach to quantum field theories is missing some fundamental underlying principles, and that another revolution in our understanding of quantum field theories is about to take place.

A strong indication in this direction is coming from exactly solvable quantum field theories: Relying on non-perturbative methods based on additional symmetries and basic properties of the theory, techniques such as the conformal bootstrap, the integrability program and localisation in supersymmetric field theories provide insights not apparent from the Feynman diagram route. 

Another hint is coming from the multitude of dualities relating quantum field theories. Stressing the aspect of a moduli space for quantum field theories, dualities demonstrate that the semi-classical description of a QFT given by its Feynman diagram expansion is not unique; and indeed two alternative `quantisations' can give rise to the same quantum field theory, as illustrated by the example of the Sine-Gordon model and the massive Thirring model. Of particular interest in this context are the dualities relating strongly coupled theories, inaccessible from Feynman diagrams, to a quantum field theory at weak coupling. Approaching quantum field theories from the point of view of the QFT moduli space also indicates that there are potentially parts of this moduli space inaccessible from a Lagrangian description, as illustrated by the six dimensional superymmetric $\cN=(0,2)$ models.

However, possibly most remarkable in this context are the AdS/CFT dualities, relating $d$ dimensional quantum field theories with conformal symmetry to quantum gravity on asymptotically Anti-de Sitter space-times in $d+1$ dimensions. This holographic relation, between quantum gravity in the full `bulk' of space-time and a quantum field theory on the boundary, can be motivated by a discussion of the asymptotic symmetry group SO$(2,d)$ of Anti-de Sitter spaces, coinciding exactly with the conformal group in $d$ dimensions. \\

None of the developments described here are obtainable from a Lagrangian point of view, and they strongly indicate that a different understanding of quantum field theories is needed.  Another piece of evidence for some underlying, as yet undiscovered fundamental principle missing in the Feynamn diagram approach is the striking simplicity and remarkably rich structure of {\it scattering amplitudes}. In any quantum field theory, scattering amplitudes are natural and fundamental observables. However, in particular for the physically interesting case of elementary particles of spin $s\geqslant 1$, calculations using the Feynman diagram approach become intractable even for simple scattering processes, having been characterised\footnote{A comment made by Zvi Bern on the complexity of scattering amplitudes in supergravity.} as containing ``more terms than there are cells in an average research student's brain''. This suggests that these most fundamental observables are highly complicated objects.\\

In the same spirit as in the examples given above however, a different route can be taken to understand scattering amplitudes from fundamental principles and symmetries instead of Feynman diagram expansions. Recall in this context that all one-particle states in a quantum field theory can be characterised by irreducible representations of the Lorentz group \cite{Weinberg-QTF-1}. In particular, they can be classified by their mass-eigenvalue $m^2=k^\mu k_\mu$ (Wigner classification) and the representation of the little group, a subgroup of the little group leaving the momentum $k_\mu$ invariant. For massless particles, the little group is isomorphic to ISO$(2)$, the abelian group of Euclidean motions in two dimensions. Two of its generator commute, and can thus be diagonalised simultaneously. Massless states of momentum $k_\mu$ are therefore characterised by their eigenvalues under the remaining generator, the angular momentum $J^3=J^{12}$, where $J^{\mu\nu}$ denote the Lorentz generators. Representations are thus classified by the helicity, which can only take integer and half-integer values due to the topology of the Lorentz group. In particular, scalars are in the trivial representation, while gluons and gravitons transform non-trivially. They are described locally (for gluons) by the polarisation data $\epsilon_\mu$ and $\epsilon_{\mu\nu}$ respectively, subject to 
\begin{equation}\label{eq1:polarisation}
 \epsilon_\mu\sim\epsilon_\mu+\alpha k_\mu\,,\qquad \epsilon_\mu\, k^\mu=0\,,
\end{equation}
where $\alpha$ is an arbitrary constant. Importantly, the scattering amplitude transforms according to its external particles under the Lorentz group, and gluon and graviton amplitudes are thus linear in the polarisation of the scattered particles, and gauge (diffeomorphism) invariant due to \cref{eq1:polarisation}.

In four dimensions, the {\it spinor helicity formalism} provides a particularly elegant and powerful tool to parametrise massless particles. This uses the isomorphism between the restricted Lorentz group on complexified space-time and the special linear group, SO$(4,\C)\cong\text{PSL}(2,\C)$, given by $p_\mu=\sigma_\mu^{\alpha\dot\alpha}p_{\alpha\dot\alpha}$. For massless particles, this can always be decomposed into the outer product of two complex two-dimensional Weyl spinors, \begin{equation}
  p_{\alpha\dot\alpha}=\lambda_\alpha\tilde{\lambda}_{\dot\alpha}\,.   
\end{equation}
The little group acts one these spinors via $(\lambda,\,\tilde{\lambda})\rightarrow(t\lambda,\,t^{-1}\tilde{\lambda})$, and scattering amplitudes have to transform according to 
\begin{equation}
 \cA_{h_1,\dots,h_n}\rightarrow \prod_it_i^{-2h_i}\cA_{h_1,\dots,h_n}\,,
\end{equation}
where $h_i$ denotes the helicity of the $i$th particle. These considerations lead to surprisingly simple formulae for tree-level MHV amplitudes\footnote{Where MHV abbreviates `maximally helicity violating', implying two particles of negative helicity and $n-2$ of positive helicity} \cite{ParkeTaylor:1986}, 
\begin{equation}
 \cA_{\text{MHV}}(1,2,\dots,n)=\frac{\delta^{4|8}\left(\sum_{i=1}^n\lambda_i\tilde{\lambda}_i\right)}{\prod_{i=1}^n \langle i,i+1\rangle}\,,
\end{equation}
where $\langle i,j\rangle=\epsilon^{\alpha\beta}\lambda_{i,\alpha}\lambda_{j,\beta}$ denotes the spinor product. Remarkably, the complicated expressions obtained from Feynman diagrams thus simplify to give beautifully compact formulae. The MHV formula provided the basis for rapid progress in the study of scattering amplitudes, fuelled particularly by the formulation of a new fundamental mathematical description underpinning scattering amplitudes in $\mathcal{N}=4$ super Yang-Mills: the {\it twistor string} \cite{Witten:2003nn,Berkovits:2004hg}, a string theory in twistor space, an auxiliary target space originally introduced as an approach to quantise gravity and field theory \cite{Penrose:1972ia,Huggett,WardWells}. Twistor string theories have proven crucial to the understanding and discovery of the simplicity and structures underlying scattering amplitudes in both $\mathcal{N}=4$ super Yang-Mills and $\cN=8$ supergravity in four dimensions. The simplifications in this framework are due to the non-locality introduced by twistor space as the target space for the worldsheet models; heuristically, a point in space-time corresponds to a line in twistor space, and a point in twistor space corresponds to a null ray\footnote{or more accurately, an $\alpha$ plane in complexified Minkowski space.} in space-time. The twistor strings for $\mathcal{N}=4$ super Yang-Mills \cite{Witten:2003nn,Berkovits:2004hg,Berkovits:2004jj,Mason:2007zv} and $\cN=8$ supergravity  \cite{Skinner:2013xp} manifest the simplicity of tree-level amplitudes \cite{ParkeTaylor:1986,Roiban:2004yf,Hodges:2012ym,Cachazo:2012kg,Cachazo:2012pz}, expressing them as integrals over the moduli space of degree $d$ maps into twistor space. These results sparked an entire new research area and led to a multitude of remarkable insights into scattering amplitudes in $d=4$, both computational and conceptual, see \cite{Elvang:2013cua} for a recent review. While twistor strings only describe maximally supersymmetric Yang-Mills theory and gravity at tree-level, they provide a tantalising paradigm both for how twistor theory could make contact with physics and how frameworks based on worldsheet methods can provide, at least perturbatively, a better understanding of quantum field theories.\\

In a spectacular recent breakthrough, Cachazo, He and Yuan \cite{Cachazo:2013gna,Cachazo:2013hca,Cachazo:2013iea} extended these results to scattering amplitudes in arbitrary space-time dimension $d$. Remarkably, the full tree-level S-matrix for massless theories localises on a set of constraints, known as the {\it scattering equations}. These equations underpin not only massless scattering in quantum field theories, but also tie in with the twistor string models described above. The expressions for tree-level amplitudes are stunningly simple -- realised as integrals over the moduli space of marked Riemann surfaces localised on the algebraic scattering equations, and describe a long list of massless theories \cite{Cachazo:2014nsa,Cachazo:2014xea}. Moreover, they provide one of the most concise representations of the `colour-kinemtaic duality' and Kawai-Lewellen-Tye relations between gauge theory and gravity amplitudes \cite{Bern:2008qj,Bern:2010ue,Kawai:1985xq}, see \cite{Cachazo:2013gna,Ellis:2015}.

The only conceivable reason for the existence of these remarkable structures in scattering amplitudes is an underlying mathematical framework similar to the worldsheet theories giving rise to string theory amplitudes, but adapted to field theory in the same manner as the twistor string. The {\it ambitwistor string theories}, proposed in \cite{Mason:2013sva} and further studied in \cite{Ohmori:2015sha}, provide this fundamental theory. These models are two-dimensional chiral conformal field theories, with an an auxiliary target space, the complexified phase space of null geodesics, known as ambitwistor space $\A$. Strongly resembling complexified worldline formulations of quantum field theory, their action encodes the contact geometry of the ambitwistor target space. This provides a very intuitive insight into the reason for their wide-reaching impact for massless scattering - the phase space of null geodesics gives a target space adapted perfectly to the problem.

Beyond providing a primal derivation of the simple representations for the full tree-level S-Matrix of Yang-Mills and gravity in arbitrary dimensions, ambitwistor strings have accomplished a number of impressive feats. Maybe most manifestly, as chiral analogues of the RNS string they provide a new angle \cite{Ohmori:2015sha,Bjerrum-Bohr:2014qwa,Reid-Edwards:2015stz,Siegel:2015axg,Berkovits:2014aia,Berkovits:2015yra} on the connection to string theory while containing only field theory degrees of freedom. Moreover, they encode the full non-linear structure of supergravity \cite{Adamo:2014wea}, albeit perturbatively, and the ambitwistor string for type II supergravity is anomaly free, and thus leads to representations of field theory loop amplitudes, again as integrals over the moduli space of marked Riemann surfaces \cite{Adamo:2013tsa}.\\


Ambitwistor strings and their applications in field theories are the main object of this thesis, aiming to understand and exploit the simplicity and structure of scattering amplitudes. Its goal is to provide a few small but significant steps towards extending the underlying theory, and to further thereby our understanding of quantum field theories. While restricted to the subject of scattering amplitudes, this provides one piece of the puzzle of the different structures unobtainable from Feynman diagrams, and a small contribution towards a modern understanding of quantum field theories.

\paragraph{Outline of this thesis.}
The main body of this thesis is divided into two parts. While the first chapters demonstrate the powerful tools different representations of ambitwistor strings provide for tree-level massless scattering, \cref{chapter6} focuses on the extension beyond the classical limit to loop amplitudes.

\Cref{chapter2} serves as an introduction to the Cachazo-He-Yuan (CHY) representation of tree-level scattering amplitudes in massless theories and the underlying ambitwistor strings for Yang-Mills theory and gravity. The contents are intended as a review, laying the base for the original work in the subsequent chapters.

In \cref{chapter3}, I extend the original ambitwistor string construction to a list of ambitwistor worldsheet models for an extensive family of massless physical theories, whose correlators give rise to the corresponding CHY formulae \cite{Cachazo:2014nsa,Cachazo:2014xea}. These models thus provide a fundamental derivation of these formulae, and effectively explain their existence from an underlying theory. This proves the universality of ambitwistor strings, with the bosonic part of the action providing the backbone for massless scattering amplitudes. The work presented here is based on \cite{Casali:2015vta}, obtained in collaboration with Eduardo Casali, Lionel Mason, Ricardo Monteiro and Kai Roehrig.

While almost algorithmically defined as the chiral complexification of a massless spinning particle, the ambitwistor string action can be constructed geometrically from the pull-back of the contact structure of its target space. Alternatively, the action can be understood as a chiral degenerate high tension limit $\alpha'\rightarrow 0$ of string theory, or a supersymmetric curved $\beta\gamma$ system. The following two chapters explore the former two perspectives in different contexts.

In \cref{chapter4}, I discuss a representation of ambitwistor strings that ties naturally into the structure of null infinity of an asymptotically flat space-time. This manifests the connections between asymptotic symmetries, known as BMS symmetries, and infrared properties of the amplitudes. More specifically, I show how Strominger's realisation of the soft theorem as a Ward identity on Ashtekar's Fock space of radiative modes \cite{Ashtekar:1981sf,ashtekar1987asymptotic} is implemented in the worldsheet conformal field theory of the ambitwistor string. These results are joint work with Arthur Lipstein and Lionel Mason, originally published in \cite{Geyer:2014lca}.

\Cref{chapter5} on the other hand focuses on four-dimensional space-times. Using the spinorial representation of ambitwistor space, I developed new ambitwistor string models, giving rise to remarkably simple formulae for Yang-Mills and gravity. They provide a particularly elegant framework realizing the relation between symmetries at null infinity and soft theorems, and highlight the connection to the original twistor strings. However, they are more flexible, allowing for any amount of supersymmetry, and contain fewer moduli integrals. Moreover, they localise completely on a four-dimensional version of the scattering equations refined by MHV degree.  The work presented in this chapter is partially based on \cite{Geyer:2014fka} and \cite{Geyer:2014lca} in collaboration with Arthur Lipstein and Lionel Mason. Full proofs are included here that were omitted in \cite{Geyer:2014fka}.\\

\Cref{chapter6} constitutes the latter part of the thesis, and explores the ambitwistor string at loop level. Its goal is two-fold, proving the conjectured ambitwistor string formula for loop amplitudes, as well as deriving a framework applicable to more general theories. More specifically, I will use a contour integral argument to reduce the computationally challenging ambitwistor higher-genus expressions to simple formulae on nodal Riemann spheres. This relies crucially on the support of the moduli integral on the one-loop scattering equations. The general idea is then developed into a framework widely applicable to loop integrands of massless scattering amplitudes, leading to simple expressions of similar complexity to tree amplitudes for both supersymmetric and non-supersymmetric Yang-Mills theory and gravity. These formulae are proven in the non-supersymmetric case from factorisation arguments, and I give strong evidence for the validity of the supersymmetric integrands. Perhaps most remarkably, this leads to a proposal for an all-loop integrand for gravity and Yang-Mills. These results are joint work with Lionel Mason, Ricardo Monteiro and Piotr Tourkine, published in \cite{Geyer:2015bja,Geyer:2015jch}.

\chapter{Review}\label{chapter2}
This chapter provides a review of the aspects of scattering amplitudes and ambitwistor strings relevant for the remainder of this thesis. As indicated in \cref{chapter1}, ambitwistor strings provide the underlying mathematical framework for an extremely simple representation of tree-level scattering amplitudes in $d$ dimensions. Therefore, we begin in \cref{sec2:SA} with a discussion of these Cachazo-He-Yuan (CHY) formulae for the full tree-level S-matrix of a family of massless theories in $d$ dimensions. The rest of the chapter focuses on ambitwistor strings, beginning with a brief introduction to ambitwistor theory in \cref{sec2:review_ambi}. However, it is not the purpose of this section to provide an in-depth review, but rather a presentation geared towards applications in scattering amplitudes of massless particles and worldsheet theories in particular. The interested reader is referred to the original papers \cite{LeBrun:1983,Baston:1987av,LeBrun:1991jh,LeBrun:1985,LeBrun:1986,Witten:1978xx,Isenberg:1978kk,Witten:1985nt} for a more detailed exposition and \cite{Mason:2013sva} for a modern review in the context of ambitwistor strings. Finally, we conclude with a review of ambitwistor strings, including a discussion of the correlator at genus zero. \\

The aim of this chapter is to review the background material for the thesis as a whole. Thus, topics relevant only for certain chapters will pedagogically be reviewed when needed. Particularly worth highlighting in this context are the introduction to BMS symmetries and soft theorems in \cref{sec4:review-scri}, the review of ambitwistor space in four dimensions in \cref{sec5:ambispace4d}, and the discussion of ambitwistor strings at genus one in \cref{sec6:review_loop}.   

\section{Scattering amplitudes}\label{sec2:SA}
To motivate the following discussion, let us briefly review the general structure of scattering amplitudes in field and string theory. 
\begin{itemize}
 \item {\bf Feynman diagrams}: In field theory, scattering amplitudes are formulated as a combinatoric problem. The S-matrix derived perturbatively from the path integral is given by an expansion in the loop order, $\cA=\sum_g \cA_g\hbar^g$, where each $\cA_g$ is computed as a sum over Feynman diagrams with $g$ loops,
 \begin{equation}
  \cA_g=\sum_{\text{graphs }\Gamma}\frac{\omega_\Gamma(k_i,\epsilon_i)}{\text{ord}(\text{Aut}(\Gamma))}\,.
 \end{equation}
 Here, $\omega_\Gamma$ denotes the weight of the Feynman diagram determined by Feynman rules derived from the path integral, and we divide this weight by the order of the symmetry group of the diagram. $\omega_\Gamma$ depends only on the momenta $k_i$ and the polarisation $\epsilon_i$ of the external particles, and at $g$ loops involves integrals over the $g$ loop momenta $\ell_a$. The graph theoretical nature of the problem gives an intuitive understanding of field theories describing point particles. However, the traditional strategy of calculating amplitudes using Feynman diagrams becomes intractable even for simple scattering processes - suggesting that scattering amplitudes are highly complicated objects.
 \item {\bf Worldsheet models}: Worldsheet theories such as string theory beautifully reformulate this into a geometric problem. Instead of an expansion in loop order, a term proportional to the Euler characteristic in the action gives different weights to different worldsheet topologies, and thus allows for an expansion in the genus $g$ of the worldsheet. 
 \begin{center}\begin{tikzpicture}[scale=1.5]
 \draw (1.1,1) circle [x radius=0.4, y radius=0.3];
  \draw [fill] (0.95,0.82) circle [radius=.3pt];
  \draw [fill] (1.1,1) circle [radius=.3pt];
  \draw [fill] (1.03,1.22) circle [radius=.3pt];
  \draw [fill] (1.37,1.05) circle [radius=.3pt];
 \draw (2.25,1) circle [x radius=0.4, y radius=0.2];
  \draw (2.05,1.045) arc [radius=0.4, start angle=240, end angle=300];
  \draw (2.4,1.015) arc [radius=0.4, start angle=70, end angle=110];
  \draw [fill] (2.15,1.1) circle [radius=.3pt];
  \draw [fill] (2,0.95) circle [radius=.3pt];
  \draw [fill] (2.55,0.98) circle [radius=.3pt];
  \draw [fill] (2.5,1.01) circle [radius=.3pt];
  \draw (3,1) to [out=90, in=180] (3.15,1.15) to [out=0, in=180] (3.35,1.07) to [out=0, in=180] (3.8,1.18) to [out=0,in=90] (4.1,0.89) to [out=270,in=0] (3.8,0.6) to [out=180,in=0] (3.3,0.9) to [out=180,in=0] (3.15,0.85) to [out=180, in=270] (3,1);
  \draw (3.09,1.01) arc [radius=0.17, start angle=240, end angle=300];
  \draw (3.23,1) arc [radius=0.17, start angle=70, end angle=110];
  \draw (3.75,0.76) arc [radius=0.17, start angle=240, end angle=300];
  \draw (3.89,0.75) arc [radius=0.17, start angle=70, end angle=110];
 \draw [fill] (3.11,0.93) circle [radius=.3pt];
 \draw [fill] (3.42,1.01) circle [radius=.3pt];
 \draw [fill] (3.75,0.87) circle [radius=.3pt];
 \draw [fill] (3.93,1.08) circle [radius=.3pt];
 \node at (0.2,1) {$\displaystyle \cA=$};
 \node at (1.68,1) {$\displaystyle +$};
 \node at (2.83,1) {$\displaystyle +$};
 \node at (4.5,1) {$\displaystyle + \;\;...$};
\end{tikzpicture}
\end{center}
This provides an excellent example of the wide-reaching impact of worldsheet theories: in contrast to the combinatoric problem posed by the expansion in Feynman diagrams, the formulation on a Riemann surface ensures the conceptual simplicity of a unique object at every order in perturbation theory, and rephrases the amplitude as a geometrical problem - an integral over the moduli space of Riemann surfaces.  However, the weight of each contribution is given by an integral over the moduli space $\mathscr{M}_{g,n}$ of the Riemann surface, which is in general difficult to perform. 

In contrast to field theory, string amplitudes depend explicitly on the string length $\sqrt{\alpha'}$, with an infinite tower of states contributing (c.f. for example the renowned Virasoro-Shaprio amplitude). Moreover, they are well-defined at high energies, whereas quantum gravity is plagued by intractable UV divergences. The two approaches make contact at the boundary of the moduli space, where the surfaces degenerate, and field theory amplitudes emerge in the high tension limit $\alpha'\rightarrow0$.
\item {\bf CHY formulae and ambitwistor strings}: Considering the seeming complexity of field theory amplitudes when represented as sums over Feynman diagrams, it is highly remarkable and surprising that a simple, compact formulation exists for massless particles in arbitrary space-time dimensions: the CHY formulae \cite{Cachazo:2013gna,Cachazo:2013hca,Cachazo:2013iea}. These represent  field theory tree-level amplitudes as the sum over a rational function evaluated at solutions $\sigma_i\in\CP^1$ to the so-called {\it scattering equations} described in \cref{sec2:review_SE},
\begin{equation}\label{eq2:CHY_sum}
 \cA_n^{(0)}=\sum_{\substack{\text{solutions}\\k_i\cdot P(\sigma_i)=0}}\frac{\cI_n(\sigma_i,k_i,\epsilon_i)}{J(\sigma_i,k_i)}\,.
\end{equation}
Here, $\cI_n$ is a theory-dependent function of the momenta $k_i$ and the polarisation data $\epsilon_i$, while $J$ is the theory-independent Jacobian obtained from solving the scattering equations $k_i\cdot P(\sigma_i)$.
\end{itemize}
In contrast to the graph combinatoric problem posed by the sum over Feynman diagrams for field theory amplitudes, and the moduli theoretic problem in worldsheet theories, the CHY formulae express tree-level scattering amplitudes in terms of solutions to an algebraic problem. From this point of view, the appeal and impact of the CHY formulae is easy to understand; algebraic problems are comparatively easy to solve. Moreover, the CHY representation for field theory amplitudes has the additional advantage of constituting structurally a middle point between field theory and string theory. To see this, note that \cref{eq2:CHY_sum} can be re-cast as an integral over the moduli space of a Riemann sphere with marked points $\sigma_i$ associated to the external particles,
\begin{equation}
 \cA_n^{(0)}=\int\frac{\prod_{i=1}^n\d\sigma_i}{\vol \,G}{\prod_i}\,\bar\delta\left(k_i\cdot P(\sigma_i)\right)\,\cI_n(\sigma_i,k_i,\epsilon_i)\,.
\end{equation}
The moduli integral over the marked sphere localises completely on the support of the scattering equations, enforced by holomorphic $\delta$-functions, and thereby avoid both the combinatoric and geometrical difficulties of the Feynman diagram expansion and string theory amplitudes. While only encoding the degrees of freedom of a field theory, the CHY formulae thus inherit the simplicity and the benefits of worldsheet amplitudes, with a unique object at every order in perturbation theory. \\

The only conceivable reason for the existence of these remarkable structures in scattering amplitudes is an underlying mathematical framework similar to the worldsheet theories giving rise to string theory amplitudes, but encoding the localisation on the scattering equations and thus describing field theories. This is indeed realised by the {\it ambitwistor string theories} proposed in \cite{Mason:2013sva}, which we will review in  \cref{sec2:review_ambistrings}. Leading to the CHY formulae, representations of field theory very reminiscent of string theory, these models thus have the potential to give a new understanding of quantum field theories. Several directions will be explored in the main body of the thesis.

\subsection{Scattering equations}\label{sec2:review_SE}
Let us begin with a more detailed discussion of the constraints on which the moduli integrals localise - the scattering equations \cite{Cachazo:2013gna,Cachazo:2013hca,Cachazo:2013iea}. Forming the backbone of the CHY formulae, they are constructed from a meromorphic map $P:\CP^1\rightarrow\C^d$ from the Riemann sphere into momentum space,
\begin{equation}\label{eq2:P0}
 P(\sigma)=\sum_{i=1}^n \frac{k_i}{\sigma-\sigma_i}\,.
\end{equation}
The scattering equations then take the form
\begin{equation}\label{eq2:SE_Ei}
  E_i\equiv k_i\cdot P(\sigma_i)=\sum_{j\neq i}\frac{k_i\cdot k_j}{\sigma_{i}-\sigma_j}=0\,.
\end{equation}
Using momentum conservation, $\sum_i k_i^\mu=0$, they respect the SL$(2,\C)$ invariance of the moduli space of Riemann spheres with $n$ marked points.\footnote{See \cref{sec6:moduli_space} for more details on the moduli space of Riemann surfaces. However, we will only need that the dimension of the moduli space $\mathscr{M}_{0,n}$ is given by $n-3$, due to the non-trivial group of automorphisms Aut$(\CP^1)=\text{PGL}(2,\C)$, acting as M\"{o}bius transformations on the coordinates.} Therefore, there are only $n-3$ linearly independent scattering equations, with $(n-3)!$ solutions. 

To see this, write the scattering equations in their polynomial form \cite{Dolan:2014ega},
\begin{equation}\label{eq2:SE_Fi}
 F_j\equiv\sum_i \sigma_i^{j+2} E_i=0\,.
\end{equation}
The powers of $\sigma_i$ in the sum are chosen such deg$F_j=j$. We then find $F_{-2}=0$ as an algebraic identity, and $F_{-1}=F_0=0$ due to momentum conservation. Since the Jacobian for the transformation from $F_j$ to $E_i$ is always non-zero for generic insertion points (so for $\sigma_{ij}\equiv\sigma_i-\sigma_j\neq0$), the set of constraints $\{E_i\}$ and $\{F_j\}$ are equivalent and have the same solutions $\{\sigma_i\}$. Since the $F_j$ are polynomials with deg$(F_j)=j$, there are $(n-3)!$ solutions to \cref{eq2:SE_Fi}, and hence to \cref{eq2:SE_Ei}. 

For a low number of marked points, explicit solutions are known, and numerical algorithms have been developed for higher numbers of external particles directly from \cref{eq2:CHY_sum}. There has also been much recent progress on evaluating the scattering amplitudes without explicitly solving \cref{eq2:SE_Ei}, see  \cite{Cachazo:2015nwa,Baadsgaard:2015voa,Baadsgaard:2015ifa,Kalousios:2013eca,Kalousios:2015fya,Cardona:2015eba,Cardona:2015ouc,Sogaard:2015dba,Huang:2015yka,Lam:2014tga,Lam:2015sqb,Dolan:2015iln}.\\

In the following, we will take a more geometric approach to the scattering equations. This point of view is very much motivated from the ambitwistor string reviewed in \cref{sec2:review_ambistrings}, and has proven extremely fruitful for extensions and generalisations of the scattering equations, for example to higher loop order. In this context, we understand $P_\mu$ as a meromorphic section of the canonical bundle $K_\Sigma$ of the worldsheet. $P_\mu\in\Omega^0(\Sigma,K_\Sigma)$ is required to associate the null momenta of the external particles $k_i$ to the marked points $\sigma_i$ on the Riemann surface via the residue Res$_{\sigma_i}P(\sigma)=k_i$, and thus satisfies the equation
\begin{equation}\label{eq2:DE_SE}
    \dbar P=2\pi i \d \sigma \sum_i k_i \,\bar \delta(\sigma-\sigma_i)\,.
\end{equation}
This is solved by the 1-form equivalent of \cref{eq2:P0},
\begin{equation}\label{eq2:P}
  P(\sigma)=\sum_{i=1}^n \frac{k_i}{\sigma-\sigma_i}\,d\sigma\,.
\end{equation}
Geometrically, the scattering equations \cref{eq2:SE_Ei} then encode the vanishing of the quadratic differential $P^2$.

To see this, first note that $P^2$ is a meromorphic quadratic differential, with vanishing residues at the double poles for on-shell null momenta. Since meromorphic quadratic differentials have at least four poles, setting $P^2=0$ is thus equivalent to requiring the residues at $n-3$ of its $n$ simple poles to vanish. In particular, these residues are given by
\begin{equation}
  \text{Res}_{\sigma_i} P^2(\sigma)=k_i\cdot P(\sigma_i)=\sum_{j\neq i}\frac{k_i\cdot k_j}{\sigma_{i}-\sigma_j}=0\,,
\end{equation}
which we identify exactly as the scattering equations $E_i$. We will see in \cref{sec2:review_ambistrings} how this prescription arises from the ambitwistor string, and in \cref{chapter6} how the construction generalises to higher genus worldsheets. \\

A crucial property of the scattering equations is that they manifest the factorisation properties of scattering amplitudes \cite{Cachazo:2013gna,Ellis:2015,Ellis:2013,Dolan:2013isa,Adamo:2013tsa}. Due to the unitarity of the S-matrix and the locality of the interactions, the amplitudes develops a pole corresponding to a propagator for a particle with momentum $k_I$ if a partial sum of the external momenta  $k_I=\sum_{i\in I}k_i$ becomes null, where $I\subset \{1,\ldots, n\}$. The residue is the product\footnote{Usually, the correct pole structure is referred to as locality, and the residue as unitarity.} of two tree amplitudes, depending on the particles in $I$ and its complement $\bar I$ respectively, as well as the propagating particle of momentum $\pm k_I$,
\begin{equation}\label{eq2:factorisation}
 \cM_n\rightarrow\sum_{\epsilon_I}\cM_{n_I+1}(\epsilon_I)\frac{1}{k_I^2}\cM_{n_{\bar I}+1}(\epsilon_I) \qquad\text{as }k_I^2\rightarrow0\,.
\end{equation}
The scattering equations relate these factorisation channels to the boundary of the moduli space of Riemann surfaces of genus $g$ with $n$ marked points, $\mathscr{M}_{g,n}$. The Deligne-Mumford compactification $\overline{\mathscr{M}}_{g,n}$  \cite{DeligneMumford,Witten:2012bh} includes these boundary configurations. They correspond to degenerations of the Riemann surface where marked points collide, or the surface develops a long, thin neck; either separating the Riemann surface into two surfaces joint at the node or shrinking a non-trivial $a$-cycle of a higher genus curve. In particular, in the Deligne-Mumford compactification, marked points are always separate $\sigma_{ij}\neq 0$, and these configurations correspond to the degeneration of the original Riemann surface into two Riemann surfaces joint at a node. As a subset $I$ of marked points come together,
\begin{equation}
 \sigma_i=\sigma_I+\varepsilon x_i +\cO(\varepsilon^2)\qquad\text{for }i\in I\,,
\end{equation}
the scattering equations imply that $k_i\cdot P_I(x_i)=\cO(\varepsilon)$, where $P_I$ is defined as above, but restricted to the $I$ component of the factorised Riemann surface. Therefore, the kinematic invariant $k_I^2$ has to behave as
\begin{equation}
 k_I^2=\frac{1}{2}\sum_{i,j\in I}k_i\cdot k_j=\sum_{\substack{i,j\in I\\i\neq j}} \frac{x_ik_i\cdot  k_j}{ x_i-x_j}= \sum_{i\in I} x_i k_i \cdot P_I(x_i)=\cO(\varepsilon)\,.
\end{equation}
This explicitly relates the boundary of the moduli space corresponding to collisions of marked points to the factorisation channels of the amplitude. The scattering equations factorise into two sets of constraints, corresponding to the components of the degenerate Riemann sphere, with an additional marked point on each corresponding to the node, with residue $\pm k_I$ ensuring momentum conservation. We will see the power of this statement in \cref{chapter6}, where we will use an extension of this argument to prove new, simple formulae for one-loop amplitudes.\\

While the discussion given here has been targeted towards scattering equations in the context of the CHY formulae and the ambitwistor string, it is highly remarkable that the scattering equations already appeared much earlier in the study of dual models \cite{Fairlie:1972zz,robertsphd,Fairlie:2008dg}, and, in the different context of the high energy limit of string scattering, in the work by Gross and Mende \cite{Gross:1987kza,Gross:1987ar} (see also \cite{Caputa:2011zk}). Interestingly, scattering equations thus make an appearance in both the high energy and low energy limit of string theory, which can be intuitively understood from the decoupling of massless modes in both limits. In a different guise, the scattering equations also play a crucial role in the four-dimensional Berkovits-Witten twistor string \cite{Witten:2003nn,Berkovits:2004hg,Witten:2004cp}, which we will explore further in \cref{chapter5}.

\subsection{CHY formulae}\label{sec2:CHY}
As motivated above, the scattering equations provide the backbone of the CHY representation for the full tree-level S-matrix of massless theories \cite{Cachazo:2013gna,Cachazo:2013hca,Cachazo:2013iea}, 
\begin{equation}\label{eq2:CHYgeneral_sum}
 \cA_n^{(0)}=\sum_{\substack{\text{solutions}\\k_i\cdot P(\sigma_i)=0}}\frac{\cI_n(\sigma_i,k_i,\epsilon_i)}{J(\sigma_i,k_i)}\,.
\end{equation}
Remarkably, these simple and compact formulae are valid in any space-time dimension $d$, and for a large class of massless theories that determine the form of the integrand $\cI_n$. The Jacobian $J$ of solving the scattering equations $k_i\cdot P(\sigma_i)=0$ for $\sigma_i$ emerges most naturally when the expression is converted into an integral over the moduli space $\mathscr{M}_{0,n}$ of an $n$-marked sphere, localised on the scattering equations,
\begin{equation}\label{eq2:CHYgeneral_int}
 \cA_n^{(0)}=\int_{(\CP^1)^n}\frac{\prod_{i=1}^n\d\sigma_i}{\vol \,G}{\prod_i}\,\bar\delta\left(k_i\cdot P(\sigma_i)\right)\,\cI_n(\sigma_i,k_i,\epsilon_i)\,.
\end{equation}
Here, the M\"{o}bius invariance of the marked Riemann sphere has been exploited to write the measure and scattering equations in a manifestly permutation invariant way, which introduces the quotient by the volume of the symmetry group $G=\SL(2,\C)^2$. Fixing this redundancy introduces the usual Faddeev-Popov Jacobian $\frac{\sigma_{ab}\sigma_{bc}\sigma_{ca}}{\d\sigma_a\d\sigma_b\d\sigma_c}$ for the measure and $\sigma_{ab}\sigma_{bc}\sigma_{ca}$ for the removed delta-functions imposing the scattering equations $E_a$, $E_b$ and $E_c$. As discussed in the previous section, there are $n-3$ linearly independent scattering equations, enforced in the CHY formula by holomorphic delta-functions that are to be understood as 
\begin{equation}
  \label{eq:deltabar-def}
\bar \delta (f(z)):= \dbar \frac{1}{2\pi i f(z)} =\delta(\Re f)\delta(\Im f) d\overline{f(z)}\, .
\end{equation}
Thus the integrand is indeed a (top,top) form. Since the dimension of the moduli space of an $n$-punctured Riemann sphere is given by dim$(\mathscr{M}_{0,n})=n-3$, the integral fully localises on the support of the scattering equations \cref{eq2:SE_Ei}. This establishes the equivalence of \cref{eq2:CHYgeneral_int} and \cref{eq2:CHYgeneral_sum}.\\

While very simple expressions were known for scattering amplitudes in maximally supersymmetric theories in four dimensions \cite{ParkeTaylor:1986,Witten:2003nn,Berkovits:2004hg,Berkovits:2004jj,Mason:2007zv,Roiban:2004yf,Hodges:2012ym,Cachazo:2012kg,Cachazo:2012pz,Cachazo:2012da,Skinner:2013xp,Adamo:2015gia}, the CHY formulae represent the first compact expressions of the full tree-level S-matrix in arbitrary dimensions for non-supersymmetric theories. In particular, these different theories are distinguished only by the form of the integrand $\cI_n$, which is a rational function of the external momenta $k_i$, the polarisation data $\epsilon_i$ for non-trivial representations of the little group, and the marked points $\sigma_i$ associated to the external particles via the residues of the meromorphic one-form $P$. To ensure SL$(2,\C)$ invariance of the integrand, $\cI_n$ has to have homogeneity $-4$ in each marked point. The scattering equations underpinning all these theories are universal for all massless theories described below.

We will focus here on the integrands for the arguably physically most interesting theories; Einstein gravity coupled to a B-field and a dilaton, and Yang-Mills theory. The amplitudes of the bi-adjoint scalar, a massless coloured cubic scalar flavoured in $U(N)\times U(\tilde{N})$, naturally tie into the discussion, as we will see below. More general theories will be discussed in \cref{sec2:masslessmodels}.

While originally formulated for the theories mentioned above, the CHY formulae have been extended to a sizeable list of  massless theories \cite{Cachazo:2014nsa,Cachazo:2014xea} that share an intrinsic structure of the integrand $\cI_n$:  it is composed of two factors, $\cI_n=\cI^L\,\cI^R$, each of homogeneity $-2$ in all marked points. We can thus characterise the theories in terms of these more fundamental building blocks.
\begin{itemize}
 \item {\it Parke-Taylor colour factors}. These are fundamental building blocks for the leading trace contribution of any gauge theory with a gauge group U$(N)$ or SU$(N)$. The integrand for the colour-ordered amplitude\footnote{An extensive review of modern techniques and approaches in the study of quantum field theories and scattering amplitudes can be found in \cite{Elvang:2013cua}.}, depending on the planar ordering $\alpha$, is given by 
 \begin{equation}\label{eq2:defPT}
  \mathcal{C}_n(\alpha)=\frac{1}{\sigma_{\alpha(i_1)\,\alpha(i_2)}\sigma_{\alpha(i_2)\,\alpha(i_3)}\dots\sigma_{\alpha(i_n)\,\alpha(i_1)}}\,.
 \end{equation}
 The full amplitude can be recovered from these partial amplitudes as a sum over all colour orderings of tree-level Yang-Mills theory, or one colour structure of the bi-adjoint scalar theory respectively,
 \begin{equation}
  \mathcal{C}_n=\sum_{\alpha\in S_n/\Z_n}\text{Tr}\left(t^{A_{\alpha(i_1)}}t^{A_{\alpha(i_2)}}\dots t^{A_{\alpha(i_n)}}\right)\, \mathcal{C}_n(\alpha)\,,
 \end{equation}
 for $t\in \g$ for some Lie algebra $\g$.
 \item {\it Pfaffians}. Integrands for particles transforming in non-trivial representations of the little group with polarisation vectors $\epsilon_i$ are build from the reduced Pfaffian\footnote{Or modifications of the Pfaffian, see \cref{sec2:masslessmodels}.} of a $2n\times 2n$ antisymmetric matrix $M$, depending on the momenta and polarisation vectors,
 \begin{equation}
  M=\begin{pmatrix} A & -C^T \\ C & B \end{pmatrix}\,,
 \end{equation}
 with entries defined by
 \begin{subequations}
 \begin{align} \label{eq2:def_M_CHY}
  &A_{ij}=\frac{k_i\cdot k_j}{\sigma_{ij}}\,,&&B_{ij}=\frac{\epsilon_i\cdot \epsilon_j}{\sigma_{ij}}\,, &&C_{ij}=\frac{\epsilon_i\cdot k_j}{\sigma_{ij}}\,,\\
  &A_{ii}=0\,, && B_{ii}=0\,, && C_{ii}=-\sum_{j\neq i}C_{ij}\,.
 \end{align}
 \end{subequations}
 The matrix $M$ has co-rank 2, its kernel being spanned by the vectors $(1,\dots,1,0,\dots,0)$ and $(\sigma_1,\dots,\sigma_n,0,\dots,0)$, and thus $\pf(M)=0$. However, an invariant quantity, the {\it reduced Pfaffian}, can be defined by removing two rows $i$ and $j$ and the corresponding columns,
 \begin{equation}
  \pf'(M)=\frac{(-1)^{i+j}}{\sigma_{ij}}\pf(M^{ij}_{ij})\,.
 \end{equation}
 This reduced Pfaffian is invariant under the choice of removed rows and columns
\end{itemize}
The scattering amplitudes for Einstein gravity, Yang-Mills and the bi-adjoint scalar theory are then given by
\begin{center}
\begingroup
\renewcommand*{\arraystretch}{1.8}
\begin{tabular}{ll}
  Einstein gravity:& $\cI_n=\pf'(M)\pf'(\wt{M})\,,$\\
  Yang-Mills (colour-ordered):& $\cI_n=\mathcal{C}(i_1,\dots,i_n)\,\pf'(M)\,,$\\
  Bi-adjoint scalar (colour-ordered): &$\cI_n=\mathcal{C}(i_1,\dots,i_n)\,\mathcal{C}(j_1,\dots,j_n)\,,$
\end{tabular}
\endgroup
\end{center}
where $\wt{M}$ is defined in analogy to $M$, but with polarisation vectors $\tilde\epsilon$. Here, we have parametrised the polarisation tensor of the graviton by $\epsilon^{\mu\nu}=\epsilon^\mu\tilde\epsilon^\nu$.\\

From the mathematical structure \cref{eq2:CHYgeneral_int} of the scattering amplitude as an integral over the moduli space localised on the scattering equations, it is not immediately obvious that the resulting expressions will always be rational, as required for tree-level scattering amplitudes. This is however ensured by taking a global residue \cite{Sogaard:2015dba}, since the scattering equations can be re-cast in polynomial form \cref{eq2:SE_Fi}, and the integrand is a rational function.

Several consistency checks, including factorisation, soft limits and sub-leading soft limits, and relations among different partial amplitudes, were checked in the original papers \cite{Cachazo:2013gna,Cachazo:2013hca,Cachazo:2013iea} and supplementary notes \cite{Ellis:2013,Ellis:2014}, and nicely summarised in \cite{Ellis:2015}. A full proof of the formulae using the BCFW recursion relation \cite{Britto:2004ap,Britto:2005fq} was given in \cite{Dolan:2013isa}.

A further remarkable property of the CHY formulae is that they provide the most concise formulation to date of the Kawai-Lewellen-Tye (KLT) relations \cite{Kawai:1985xq}, expressing gravity as the `square' of gauge theory  \cite{Cachazo:2012da,Cachazo:2013gna,Ellis:2015}. Moreover, the representation of the integrand is strongly reminiscent of colour-kinematics double copy relations between gauge theory and gravity as well \cite{Bern:2010ue}; the colour factor $\mathcal{C}_n$ in the Yang-Mills integrand gets replaced by another copy of the Pfaffian to obtain gravity. These amplitudes relations originating in string theory\cite{BjerrumBohr:2010hn,Kawai:1985xq} are manifested nicely in the CHY and ambitwistorial framework, see also \cite{Monteiro:2013rya}.

\subsection{Scattering amplitudes in massless theories}\label{sec2:masslessmodels}
Beyond the Einstein gravity, Yang-Mills and the bi-adjoint scalar, the CHY formalism can also be extended to an comprehensive family of massless theories \cite{Cachazo:2014nsa,Cachazo:2014xea}, including theories with complicated Lagrangian descriptions, like Einstein-Yang-Mills theory (with or without additional scalars), Dirac-Born-Infeld or the non-linear sigma model.\\

As indicated in the last section, all these theories share the fundamental property that their integrand $\cI_n$ factorises into $\cI_n=\cI^L\,\cI^R$. Since we have exhausted all possible combinations of the two fundamental building blocks above to construct Einstein gravity, Yang-Mills theory and the bi-adjoint scalar theory, it is evident that new building blocks will be needed:
\begin{itemize}
 \item Pfaffians of matrices $A$ independent of the polarisation data,
 \begin{equation}
 \pf'(A)\,,
\end{equation}
where $A$ is defined in \cref{eq2:def_M_CHY}. Note that these have to occur as building blocks $\cI_{L,R}=(\pf'(A))^2$ to ensure that the integrand has homogeneity $-4$ in each variable.
\item Pfaffians of matrices independent of both momentum and polarisation,
\begin{equation}
 \pf(\cX)\,,\qquad \pf(X)\,,
\end{equation}
encoding Maxwell fields. Here, $\cX$ and $X$ are given by
\begin{equation}
 \cX_{ij}=\begin{cases}\frac{\delta^{a_ia_j}}{\sigma_{ij}} &i\neq j\,\\ 0 & i=j\,\end{cases} \qquad \text{and} \qquad X_{ij}=\begin{cases}\frac{1}{\sigma_{ij}} &i\neq j\,\\ 0 & i=j\,.\end{cases}
\end{equation}
\item Generalisations of the Pfaffians of $M$ that can encode both gravity and gauge fields,
 \begin{equation}
   \prod_{a=1}^m \mathcal{C}_{\text{tr}_a}\pf'(\Pi)\,.
 \end{equation}
 These factors will appear in general multi-trace mixed graviton-gluon amplitudes. The Parke-Taylor factors $\mathcal{C}_{\mathrm{tr}_a}$ are defined in analogy to \cref{eq2:defPT} and encode the $m$ trace contributions arising from the gauge theory. For $m$ traces and $n_g$ gravitons, $\Pi$ is therefore an $2(n_g+m)\times2(n_g+m)$ matrix. For more details, see \cite{Cachazo:2014nsa,Cachazo:2014xea} and \cref{chapter3}.
\end{itemize}
The integrands for all massless theories that can be constructed  from these building blocks are summarised in the following table.
\begin{center}
\begingroup
\renewcommand*{\arraystretch}{1.8}
\begin{tabular}{l!{\color{light-gray-table}\vrule}c!{\color{light-gray-table}\vrule}c}\hline
 Theory & Integrand $\cI^L$ & Integrand $\cI^R$\\ \hline
 Einstein gravity & $\pf'(M)$ & $\pf'(\wt{M})$\\
 Yang-Mills & $\mathcal{C}_n$ & $\pf'(M)$\\
 Bi-adjoint scalar & $\mathcal{C}_n$ & $\mathcal{C}_n$ \\
 Einstein-Maxwell & $\pf(\cX)_\gamma\pf'(M)_{\hat{\gamma}}$ & $\pf'(\wt{M})$\\
 Einstein-Yang-Mills & $\prod_a \mathcal{C}_{\text{tr}_a}\pf'(\Pi)$ & $\pf'(\wt{M})$\\
 Yang-Mills-Scalar & $\mathcal{C}_n$ & $\pf(\cX)_s\pf'(M)_{\hat{s}}$\\
 generalised Yang-Mills-Scalar & $\mathcal{C}_n$ & $\prod_a \mathcal{C}_{\text{tr}_a}\pf'(\Pi)$\\
 Born-Infeld & $\pf'(M)$ & $(\pf'(A))^2$\\
 Dirac-Born-Infeld & $\pf(\cX)_s\pf'(M)_{\hat{s}}$ & $(\pf'(A))^2$\\
 extended Dirac-Born-Infeld & $\prod_a \mathcal{C}_{\text{tr}_a}\pf'(\Pi)$ & $(\pf'(A))^2$\\
 non-linear $\sigma$-model & $\mathcal{C}_n$ & $(\pf'(A))^2$\\
 special Galileon & $(\pf'(A))^2$ & $(\pf'(A))^2$ \\\hline
\end{tabular}
\endgroup
\end{center}\vspace{20pt}
In this table, lower case indices indicate that the building blocks only contain particles in a certain representation of the little group (e.g. gluons denoted by $\gamma$), while a `hat' denotes that the particles were omitted (e.g. $\hat{s}$ for scalars).\smallskip

The Lagrangian descriptions for the more exotic theories are reviewed in \cref{sec3:NewModels}, for details and extensive checks see \cite{Cachazo:2014xea}. All building blocks listed above can be obtained from the original ones via certain operations including compactifications from higher dimensions, see \cref{fig2:CHY}. However, the physical interpretation of some of these operations still remains unclear. One further comment at this stage concerns the KLT relations: similar to the gauge-gravity relations described above, generalised KLT relations express for example the special Galileon theory as the square of the non-linear $\sigma$-model \cite{Cachazo:2014xea}.\\

The existence of CHY formulae for this vast family of massless theories is a compelling evidence for the universality of the CHY formalism for massless scattering amplitudes. Beyond Einstein gravity and Yang-Mills, these formulae capture tree-level scattering amplitudes for more generic field theories in arbitrary dimensions, in a compact and simple form completely obscured by traditional approaches.

\begin{figure}[ht]
\begin{center}

\tikzstyle{block} = [rectangle, draw, fill=light-gray,
    text width=3em, text centered, rounded corners, minimum height=3em]
\tikzstyle{line} = [draw, thick, -latex']
\tikzstyle{cloud} = [draw, ellipse,fill=light-grayII, node distance=1.6cm,
    minimum height=3em]

\begin{tikzpicture}[node distance = 2cm, auto]
    \node [cloud] (gr) {Gravity: 
    };
    \node [block, below of=gr] (em) {EM
    };
     \node [block, below of=em] (eym) {EYM
     };
      \node [block, right of=eym, node distance=4.5cm] (ym) {YM
      };
       \node [block, below of=ym] (yms) {YMS
       };
     \node [block, below of=yms] (gyms) {gen.\ YMS 
     };
     \node [block, right of=gr, node distance=4.5cm] (bi) {BI};
     \node [block, below of=bi] (dbi) {DBI};
     \node[block, left of=yms, node distance=4.5cm] (phi4) {$\phi^4$};
     \node[block, right of=ym, node distance=4.6cm] (new) {NLSM};
    \path [line] (gr) --node {compactify} (em);
    \path [line, dashed] (em) --node {generalize}   (eym);
    \path [line] (ym) --node {compactify}   (yms);
     \path [line, dashed] (yms) --node {generalize}   (gyms);
    \path [line, densely dotted] (gr) --node {``compactify''}   (bi);
        \path [line, densely dotted] (em) --node {``compactify''}   (dbi);
      \path [line] (bi) --node {compactify}   (dbi);
       \path [line] (eym) --node {single trace} (ym);
        \path [line] (yms) --node {corollary}   (phi4);
        \path[line, densely dotted] (ym) --node [above] {``compactify''} (new);
 \path [line, dashed] (gr)--++ (-1.2cm, 0cm) |- node [near start, left]{squeeze} (eym);
  \path [line, dashed] (ym)--++ (2.2cm, 0cm) |- node [near start, right]{squeeze} (gyms);
\end{tikzpicture}
\caption{Theories studied by CHY and operations relating them.\label{fig2:CHY}}
\label{fig:intro}
\end{center}
\end{figure}
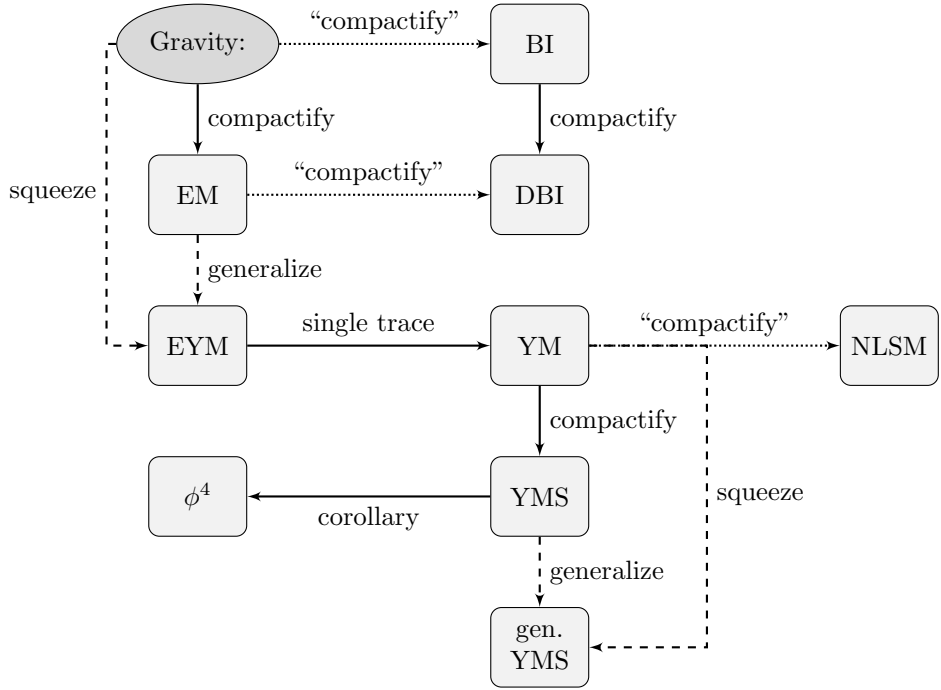

\section{The ambitwistor string} \label{sec2:review_ambistrings}
The most striking property of the CHY formulae is unquestionably that such compact and simple formulae for the full tree-level S-matrices of massless theories exist at all. This riddle is only exacerbated by the similarities of the CHY formulae to string theory amplitudes. Where are these simple formulae coming from? The only conceivable answer is that there exists a mathematical framework underlying these field theory amplitudes, similar enough in nature to conventional worldsheet theories to reproduce the similarities with string theory amplitudes, but sufficiently different to restrict to only massless degrees of freedom.\\

This task is indeed realised by the {\it ambitwistor string theories} proposed in \cite{Mason:2013sva}. Ambitwistor strings are chiral worldsheet models with an auxiliary target space, the phase space of complex null geodesics, known as (projective) ambitwistor space $P\A$. This target space gives an intuitive interpretation for the universality of the scattering equations for massless theories: Geometrically, the scattering equations ensure that the worldsheet is indeed mapped into ambitwistor space, as appropriate for propagating massless particles.

Correlators in the ambitwistor string reproduce, as claimed above, the CHY formulae, and thus the ambitwistor string provides a fundamental derivation of the results reviewed in the last section. Since it represents the mathematical theory underlying these amplitudes, it not only resolves the question about the origin of the CHY formulae, but also instigated progress in new directions. The different presentations of ambitwistor space, with different properties and advantages, and the different representations of the string provide an excellent example emphasising the flexibility of the ambitwistor string approach. Particularly worth highlighting are the RNS representation \cite{Mason:2013sva,Ohmori:2015sha} discussed in \cref{sec2:ambi_string}, the pure spinor string \cite{Berkovits:2013xba,Gomez:2013wza,Adamo:2015hoa} and the string in its four-dimensional twistorial representation \cite{Geyer:2014fka}, discussed in \cref{chapter5}. These different presentations are complemented by alternative points of view on the twistor string: formulated in the language of two-dimensional conformal field theories, its action can be understood alternatively from symplectic/contact geometry, as a complexification of the worldline action of a massless spinning particle, as a curved $\beta\gamma$ system \cite{Witten:2005px,Nekrasov:2005wg}, or as a degenerate chiral limit of string theory \cite{Mason:2013sva,Siegel:2015axg}. We will see throughout this thesis that these different aspects of the ambitwistor string lead to progress in widely different areas.\\

In the remainder of this chapter, we will give an introduction to the ambitwistor string in the RNS representation \cite{Mason:2013sva,Ohmori:2015sha}. This will first entail a brief review of ambitwistor space in \cref{sec2:review_ambi}, which provides the auxiliary target space the string is propagating in, followed by a discussion of the chiral worldsheet theory, the vertex operators and correlators.

\subsection{Ambitwistor space}\label{sec2:review_ambi}
In this section, we provide a brief review  of ambitwistor space, targeted towards the later applications in ambitwistor strings. As such, the interested reader is referred to the original papers \cite{Witten:1978xx,Isenberg:1978kk,Witten:1985nt,LeBrun:1983,LeBrun:1985,LeBrun:1986,Baston:1987av,LeBrun:1991jh} for more details and \cite{Mason:2013sva} for a modern, more extensive review. The aspects of ambitwistor space specific to four dimensions will be discussed in \cref{sec5:ambispace4d}. Moreover, rather than defining both a bosonic and a supersymmetric version of ambitwistor space, we will focus on the supersymmetric case; the bosonic analogue should always be clear from the context.\\

Ambitwistor space, the space of complex null geodesics, derives its name from its twistor representation in four-dimensional space-time, where it can be identified with a quadric in the product of twistor and dual twistor space. The name has been kept in higher dimensions, justified by the twistor-like correspondences relating it to space-time. First introduced in \cite{Witten:1978xx,Isenberg:1978kk,Witten:1985nt} in the context of gauge fields, it was extended to gravity in arbitrary dimensions in \cite{LeBrun:1983}, where it extends Penrose's non-linear graviton construction \cite{Penrose:1976jq,Penrose:1976js} to general gravitational fields. As we will discuss below in more detail, fields are encoded by deformations of the complex structure of ambitwistor space \cite{Witten:1978xx,Isenberg:1978kk,LeBrun:1983,Baston:1987av,LeBrun:1991jh}, while preserving the contact structure.\\

Ambitwistor space $\A$ is the phase space of scaled complex null geodesics, its projectivisation, $P\A$, the space of unscaled null geodesics. In $d$ dimensions, consider the complexification $(M,g)$ of space-time, with a holomorphic metric $g$. Complex null directions are then determined by cotangent vectors $p\in T_x^*M$ with 
\begin{equation}
 \mathcal{H}=\frac{1}{2}g^{-1}(p,p)=0\,. 
\end{equation}
The bundle of complex null directions is thus contained in the holomorphic cotangent bundle, $T_N^*M\subset T^*M$. For supersymmetric space-times, we consider instead\footnote{$\Pi$ is the parity reversing functor defining the fibres to be anticommuting: for a bundle $E$, $\Pi\Omega^0(E)$ denotes the space of fermion-valued sections of $E$.} $T_S^*M=(T^*\oplus \Pi T\oplus \Pi T)M$, and restrict additionally to the massless Dirac equation 
\begin{equation}
 \mathcal{G}_r=-\psi_r\cdot p=0\,,
\end{equation}
to obtain the bundle of supersymmetric null directions,
\begin{equation}
 T^*_{SN}M=\left\{(x^\mu,p_\mu,\psi_r^\mu)\in T^*_SM|g^{-1}(p,p)=0,\,p\cdot\psi_r=0\right\}\,.
\end{equation}
The cotangent bundle $T^*_SM$ is a holomorphic symplectic manifold, with symplectic potential 
\begin{equation}\label{eq2:contact-structure}
 \theta=p_\mu\d x^\mu+\frac{1}{2}\sum_{r=1}^2 g_{\mu\nu}\psi_r^\mu\d\psi_r^\nu\,,
\end{equation}
and symplectic form $\omega=\d\theta$. The symplectic form associates to the Hamiltonians $\mathcal{H}$ and $\mathcal{G}_r$ Hamiltonian vector fields\footnote{via $D_0\lrcorner\, \omega+\d(\mathcal{H})=0$ and $D_r\lrcorner\,\omega +d(\mathcal{G}_r)=0$.}
\begin{subequations}
\begin{align}
 &D_0=p^\mu\nabla_\mu=p^\mu \left(\frac{\partial}{\partial x^\mu}+\Gamma^\rho_{\mu\nu}p_\rho\frac{\partial}{\partial p^\nu}\right)\,,\\
 &D_r=\psi_r^\mu \frac{\partial}{\partial x^\mu}+p^\mu\frac{\partial}{\partial\psi_r^\mu}\,.
\end{align}
\end{subequations}
The flow along $D=\{D_0,D_r\}$ generates super null geodesics on space-time, and thus points connected by this flow should be identified in the space of null geodesics. Since $\mathcal{L}_{D_0}D_r=D_0$, the algebra closes, providing an $\mathcal{N}=2$ supersymmetry algebra along super null geodesics.

Restricting to zero energy surfaces of the Hamiltonians and quotienting by the associated Hamiltonian vector fields therefore lands us on ambitwistor space, 
\begin{equation}
 \A=\bigslant{\big\{(x^\mu,p_\mu,\psi_r^\mu)\in T^*_SM\,\big|\,g^{-1}(p,p)=0,\,p\cdot\psi_r=0\big\}}{\{D_0,D_r\}}\,.
\end{equation}
This is equivalent to taking the symplectic quotient by the Hamiltonian vector fields, and thus $\A$ is a holomorphic symplectic manifold\footnote{Since the Lie derivative of $\omega$ along a null geodesics vanishes, $\mathcal{L}_{D}\omega=0$, the symplectic form is invariant along null geodesics, and we will, in a mild abuse of notation, refer to $\omega$ as the symplectic form on $\A$. Similar statements hold below, in particular, the Euler vector field $\Upsilon$ and the symplectic potential $\theta$ descends to $\A$ as well, since $\mathcal{L}_{D}\Upsilon=\Upsilon$ and $\mathcal{L}_{D}\theta+\d(\mathcal{H})=0$.}  of dimension $2d-2|2d-2$.

To obtain projective ambitwistor space as the space of unscaled null geodesics, we quotient by the rescaling generated by the Euler vector field $\Upsilon=2p_\mu\partial_{p_\mu}+\sum_r \psi^\mu_r\frac{\partial}{\partial\psi_r^\mu}$. $P\A=\A/\Upsilon$ is thus of dimension $2d-3|2d-2$. $\A$ can then be realised as the space of the line bundle $\cO(-1)\rightarrow P\A$, where $p$ takes values in $\cO(2)$ and $\psi_r$ in $\cO(1)$. The symplectic potential descends to $P\A$, where $\theta\in\Omega^1( P\A,\cO(2))$, and thus it defines a contact structure on projective ambitwistor space.\\

The correspondence between $P\A$ and space-time is expressed as a double fibration
\begin{center}
\begin{tikzpicture} [scale=1]
 \node at (0,0) {$PT^*_{SN}M$};
 \draw[>=latex,->] (-0.4,-0.2) -- (-1.4,-1.2);
 \node at (-1.1,-0.5) {\scriptsize{$\pi_{P\A}$}};
 \node at (-1.7,-1.4) {$P\A$};
 \draw[>=latex,->] (0.4,-0.2) -- (1.4,-1.2);
 \node at (1.1,-0.5) {\scriptsize{$\pi_{M}$}};
 \node at (1.7,-1.4) {$M$};
\end{tikzpicture}
\end{center}
By construction, a point in $\P\A$ corresponds to a complex null geodesic in $M$. Conversely, a point in $M$ corresponds to a quadric $Q_x\in P\A$, which can be given the interpretation of the space of complex null rays through $x$. Similar to twistor theory, this non-locality is responsible for the simplification occurring for ambitwistor strings.

LeBrun \cite{LeBrun:1983} developed this correspondence further: $P\A$ and its contact structure are fully equivalent to the space-time $M$ with its conformal structure. In particular, small deformations of complex structure\footnote{The contact structure actually fully determines the complex structure of $P\A$ via the top-form $\theta\wedge\d\theta^{d-2}$: a vector $V$ is antiholomorphic if it obeys $V\lrcorner(\theta\wedge\d\theta^{d-2})=0$.\label{footnote:theta}} of $P\A$ preserving $\theta$ are equivalent to small deformations of the conformal structure on $M$.\\

In the context of ambitwistor strings on a flat background, we will only need the linear Penrose transform to generate amplitudes, see \cite{Baston:1987av,Mason:2013sva}. Fluctuation in the space-time metric are, by LeBrun \cite{LeBrun:1983}, determined by perturbations $\delta\theta$ of the contact structure since the complex structure is determined by the contact structure $\theta$ (see \cref{footnote:theta}). To obtain non-trivial deformations, we have to choose elements $[\delta\theta]\in H^{0,1}(P\A,\cO(2))$ of the Dolbeault cohomology class, because antiholomorphic perturbations describe diffeomorphisms along Hamiltonian vector fields. This is a specific example of the Penrose transform, relating cohomology classes of $P\A$ and fields on space-time. Whilst we are mainly interested in supersymmetric ambitwistor space, let us briefly develop the bosonic version here, the supersymmetric case is more involved, and will only be motivated below.
\begin{thm}
 The Penrose transform relates cohomology classes on projective bosonic ambitwistor space and fields on space-time: 
 \begin{equation}
  H^{0,1}(P\A,\cO(n))=\begin{cases}
                   0&n<-1\\
                   H^0(PT_{N}^*M,\cO(n+1))\,/\,D_0 (H^0(PT_{N}^*M,\cO(n))) & n\geqslant-1\,.
                  \end{cases}
 \end{equation}
 Moreover, elements of $H^0(PT_{N}^*M,\cO(m))$ are polynomials of weight $m$ in $p$, and the quotient ensures that we are considering fields modulo gauge transformations or diffeomorphisms. In particular, the case $n=1$ describes linearised trace-free metrics modulo diffeomorphisms. 
\end{thm}
Note that gravitons encoded by the Penrose transform are off-shell, the field equations arise in the context of ambitwistor strings from the quantum consistency of the worldsheet model. Similar results hold for the `heterotic' model with just one supersymmetry, used in the ambitwistor string context to describe gauge fields\footnote{However, this model contains unphysical gravity degrees of freedom, highlighted on ambitwistor space by the emergence of a two-tensor and a 3-form.} ($n=0$).  Details can be found in \cite{Baston:1987av,Mason:2013sva}, here we just give a short outline of the proof. Using the short exact sequence induced by $D_0$, 
\begin{equation}
 0\rightarrow\cO_{P\A}(n)\overset{D_0}{\rightarrow}\cO_{PT^*_NM}(n)\overset{\delta}{\rightarrow}\cO_{PT^*_NM}(n+1) \rightarrow 0\,,
\end{equation}
we can deduce the corresponding long exact sequence in cohomology. However, this long exact sequence is sparse, since the cohomology for the projective light-cone is only non-trivial for $H_N^0$, $H^{d-2}_N$. The Penrose transform then follows from the first isomorphism theorem.\footnote{In the supersymmetric case, more care is needed due to the additional two Hamiltonian vector fields $D_r$.} 
\smallskip

In \cite{Mason:2013sva}, an explicit construction was presented for gravity in super ambitwistor space. Using this, we can deduce that a variation in contact structure $\delta\theta$ corresponds to a two-tensor $H^{\mu\nu}(x)$ modulo diffeomorphisms, $\delta H^{\mu\nu}=\partial^{(\mu}v^{\nu)}$. In particular, metric fluctuations given by momentum eigenstates $H^{\mu\nu}=\epsilon^\mu_1\epsilon^\nu_2 e^{ik\cdot x}$ correspond to the deformation of the contact structure
\begin{equation}\label{eq2:ambi_momES}
 \delta\theta =\bar\delta(k\cdot p) \prod_{r=1}^2\left(\epsilon_r \cdot p +k\cdot\psi_r\epsilon_r\cdot\psi_r\right)e^{ik\cdot x}\,.
\end{equation}

\subsection{Ambitwistor strings}\label{sec2:ambi_string}
\paragraph{Action.}
Using this as a starting point, let us now construct a worldsheet action with the above described ambitwistor space as a target space. As a motivation, consider first a chiral holomorphic analogue of the worldline action for a massless spinning particle. This is indeed closely related to the integral over the pull-back of the contact structure $\theta$ of ambitwistor space \cref{eq2:contact-structure} to the worldsheet, and thus provides an ideal starting point for a theory describing maps into the space of complex null geodesics. This leads to the action of the ambitwistor string \cite{Mason:2013sva} in conformal gauge\footnote{In a general gauge, the action is written in terms of the operator $\dbar_{e}=\dbar+e \partial$ (instead of $\dbar$), parametrising the usual worldsheet diffeomorphism freedom.}, whose geometry was  studied in \cite{Ohmori:2015sha},
\begin{equation}\label{eq2:action}
 S=\frac{1}{2\pi}\int_{\Sigma}P_\mu \bar\partial X^\mu+\frac{1}{2}\sum_{r=1}^2 {\Psi_r}_\mu\bar\partial \Psi_r^\mu-\frac{1}{2} \tilde{e}P^2-\sum_{r=1}^2\chi_r \Psi_r^\mu P_\mu\,.
\end{equation}
Before we discuss this action in more detail let us make a few preliminary observations. The main feature we notice is the strong similarity to the RNS string action, which not only suggests that indeed correlation functions will bear a close resemblance to string amplitudes, but also allows us to adapt many known string techniques to the ambitwistor string. However, and perhaps most importantly, in contrast to the RNS string \cref{eq2:action} is a genuine first order action, describing after gauge fixing a free CFT, which considerably simplifies calculations. We will see below that correlation functions indeed reproduce the CHY formulae, and thus the ambitwistor string provides the underlying mathematical theory and describes gravity perturbatively around Minkowski space.\\

Let us now take a closer look at the structure of the action \cref{eq2:action}. As indicated above, the action is geometrically constructed as the chiral pull-back of the contact structure $\theta$ to the worldsheet $\Sigma$, with fields
\begin{subequations}
\begin{align}
 &X^\mu\in\Omega^0(\Sigma)\,,\\
 &P_\mu\in\Omega^0(\Sigma,K_\Sigma)\,,\\
 &\Psi_r^\mu\in\Pi\Omega^0(\Sigma,K_\Sigma^{1/2})\,.
\end{align}
\end{subequations}
As usual, we will refer to fields taking values in the line bundle $K^h\otimes\bar K^{\bar h}$ as having conformal weight $(h,\bar h)$, where $K$ and $\bar K$ denote the canonical and anti-canonical bundle of the worldsheet $\Sigma$. The field $X^\mu$ hence has zero conformal weight, and provides target space coordinates for the map of the worldsheet. Its conjugate, $P_\mu$, with conformal weight $(1,0)$, is a meromorphic 1-form on the worldsheet, and thus ties in directly to the discussion of the scattering equation in \cref{sec2:review_SE}. The Majorana fermions $\Psi_r^\mu$ are of the same chirality and weight $(1/2,0)$, and therefore the action is chiral, with all fields left-moving.\footnote{In particular, the kinetic operator is $\bar\partial$ instead of the full exterior derivative of the first order string action, 
\begin{equation}\label{eq2:Polyakov}
 S_{\text{string}}^{\text{bosonic}}=\frac{1}{2\pi}\int_{\Sigma}P_\mu \d X^\mu-\frac{1}{2}P_\mu\wedge *P^\mu\,,
\end{equation} and is thus chiral as well. Note that this action is equivalent to the bosonic part of the Polyakov string action, with  $P_\mu\in\Omega^1\cong\Omega^0(\Sigma,K_\Sigma\oplus\bar K_\Sigma)$.} 

We can thus read off that $\tilde{e}\in\Omega^{0,1}(\Sigma,T_\Sigma)$ and $\chi_r\in\Pi\Omega^{0,1}(\Sigma,T_\Sigma^{1/2})$. The field $\tilde{e}$, with conformal weight $(-1,1)$, acts as a Lagrange multiplier for the Hamiltonian constraint $\mathcal{H}=-\frac{1}{2}P^2=0$. Its weight implies in particular that it behaves like a Beltrami differential. Similarly, the fermionic fields $\chi_r$ enforce the massless Dirac equation $\mathcal{G}_r=-\Psi_r\cdot P=0$. This gives a first indication that the fields $X^\mu$, $P_\mu$ and $\Psi_r^\mu$ can be given an interpretation as a parametrisation of ambitwistor space, with the symplectic quotient implemented in the CFT via the gauge fields $\tilde{e}$ and $\chi_r$. \\

Note moreover that in contrast to the RNS superstring, the ambitwistor string action contains no dimensionful parameters, so there is no analogue of the $\alpha'$ expansion. This is consistent with the first order action implying a vanishing $XX$ OPE. The model therefore contains only the massless degrees of freedom, as appropriate for a (massless) field theory.\footnote{In particular, the bosonic action can be derived from a first-order string theory action by taking a chiral $\alpha'\rightarrow 0$ limit, yielding the interpretation of the ambitwistor string as the chiral infinite tension limit, where $T=\frac{1}{2\pi\alpha'}$. The emergence of the ambitwistor model from string theory has recently been investigated further \cite{Siegel:2015axg}; there, a degenerate (HSZ) gauge choice in conjunction with a change of boundary conditions and the zero tension limit ($\alpha'\rightarrow\infty$) gives rise to the ambitwistor string.}

\paragraph{Gauge symmetry.}
The action \cref{eq2:action} is invariant under the gauge transformations
\begin{align}\label{eq2:gauge-trafo}
 &\delta X^\mu = \alpha P^\mu +\sum_r \eta_r\Psi_r\,, &&\nonumber\\ 
 &\delta P_\mu =0\,, && \delta \tilde{e}=\bar\partial\alpha\,,\\
 &\delta \Psi_r^\mu=\eta_r P^\mu\,, && \delta \chi_r =\bar\partial\eta_r\,,\nonumber
\end{align}
where the gauge parameter $\alpha\in\Omega^0(\Sigma,T_\Sigma)$ is bosonic and $\eta_r\in\Omega^0(\Sigma, T_\Sigma^{1/2})$ are fermionic. This gauge symmetry has the geometric interpretation of generating translations along null geodesics on $T^*_S M$, the complexified cotangent bundle of Minkowski space.\footnote{Here, we have taken the alternative point of view of $(X,P,\Psi_r)$ as holomorphic coordinates on $T^*_S M$.} Then the constraints $\mathcal{H}$ and $\mathcal{G}_r$, imposed by the gauge fields, restrict the target space to the null cotangent bundle, and gauge invariance quotients by translations along null geodesics. In other words, the vanishing of $\mathcal{H}$ forces $P$ to be null, and the associated gauge freedom tells us to identify field configurations differing by a translation along a null geodesic. Therefore, the fields $X^\mu$, $P_\mu$ and $\Psi_r^\mu$ can be interpreted as a parametrisation of ambitwistor space, in direct analogy to the discussion in \cref{sec2:review_ambi}, where the symplectic quotient is implemented in the CFT via the gauge fields $\tilde{e}$ and $\chi_r$ and the gauge transformations given above. This justifies the claim stated above that \cref{eq2:action} describes a chiral two-dimensional CFT whose target space is ambitwistor space $\A$. 

\paragraph{Quantisation.}
Gauge fixing worldsheet gravity and the gravitinos via the standard BRST procedure introduces the reparametrisation $(b,c)$ ghost system, the superconformal $(\beta,\gamma)$ ghosts and a $(\tilde{b},\tilde{c})$ ghost system associated to the $\alpha$ gauge symmetry,
\begin{subequations}
\begin{align}
 & b,\tilde{b}\in \Pi\Omega^0(\Sigma,K_\Sigma^2)\,, && \beta_r\in \Omega^0(\Sigma,K_\Sigma^{3/2})\,,\\
 & c,\tilde{c}\in \Pi\Omega^0(\Sigma,T_\Sigma)\,, && \gamma_r\in \Omega^0(\Sigma,T_\Sigma^{1/2})\,.
\end{align}
\end{subequations}
Note in particular that in contrast to the RNS string, both sets of ghost fields are left-moving, in line with chirality of the model. The BRST operator is then given by
\begin{equation}
 Q=\oint c \left(T_m+\frac{1}{2}T_{\text{gh}}\right)+\tilde{c}\left(\mathcal{H}_m+\frac{1}{2}\mathcal{H}_{\text{gh}}\right)+\sum_r \gamma_r\left(\mathcal{G}_{m,r}+\frac{1}{2}\mathcal{G}_{\text{gh},r}\right)\,,
\end{equation}
where the matter currents associated to the gauge freedom \cref{eq2:gauge-trafo} are readily obtained from the action \cref{eq2:action},
\begin{align}
 &T_m=-P_\mu\partial X^\mu -\frac{1}{2}\sum_r \Psi_r \cdot\partial\Psi_r\,,\nonumber\\
 &\mathcal{H}_m=-\frac{1}{2}P^2\,,\\
 &\mathcal{G}_{m,r}=-\Psi_r\cdot P\,.\nonumber
\end{align}
and\footnote{The conventions chosen here preserve the standard action of the BRST operator on the antighost fields, e.g. $Q\cdot b=T=T_m+T_{\text{gh}}$ \cite{Adamo:2014wea}. In practical calculations, the ghost currents can be ignored when using vertex operators of canonical ghost and picture number, since the current algebra remains unchanged:
\begin{equation*}
 \mathcal{G}_r(z)\mathcal{G}_r(w)\sim\frac{-2H}{z-w}\,,\qquad \mathcal{G}_1(z)\mathcal{G}_2(w)\sim0\,.
\end{equation*}}
\begin{align}
 &T_{\text{gh}}=c\partial b-2b\partial c+\tilde{c}\partial \tilde{b}-2\tilde{b}\partial \tilde{c}-\sum_r \left(\frac{3}{2}\beta_r \partial\gamma_r+\frac{1}{2}\gamma_r\partial\beta_r\right)\,,\nonumber\\
 &\mathcal{H}_{\text{gh}}=c\partial \tilde{b}-2\tilde{b}\partial c\,,\\
 &\mathcal{G}_{\text{gh},r}=c\partial \beta_r +\frac{3}{2}\beta_r\partial c -2\tilde{b}\gamma_r\,.\nonumber
\end{align}
The central charge counting\footnote{Recall that the fields contribute to the central charge as follows: $\mathfrak{c}_{XP}=2d$ from $(X,P)$ system, $\mathfrak{c}_{\Psi_r}=2\times\frac{d}{2}$ from the $\Psi_r$ fields, $\mathfrak{c}_{\text{gh}_1}=2\times11$ from the superconformal ghosts and $\mathfrak{c}_{\text{gh}_2}=-2\times26$ from the $bc$ and $\tilde{b}\tilde{c}$ ghost systems.} proceeds exactly as for the RNS string, $\mathfrak{c}=3(d-10)$, and therefore the ambitwistor string  is critical and $Q$ nilpotent for $d=10$.\\

When gauge fixing the gauge symmetry associated to the parameter $\alpha$, an additional subtlety comes into play for correlation functions including vertex operators \cite{Adamo:2013tsa}: since the gauge parameter is required to vanish at marked points of the Riemann surface, there exists a potential obstruction to setting $\tilde{e}=0$, which we would like to do in order to obtain a free action. In fixing the gauge redundancy, we add a gauge fixing term to Lagrangian,
\begin{equation}
\{Q,F(\tilde{e})\}\,,
\end{equation}
where $F$ is a gauge fixing functional. While, as pointed out above, we aim to choose $F$ such that $\tilde{e}=0$, the gauge freedom only varies $\tilde{e}$ within a fixed Dolbeault cohomology class. Hence for genus $g$ Riemann surface, $F(\tilde{e})$ can only be chosen as
\begin{equation}
 F(\tilde{e})=\tilde{e}-\sum_{r=1}^{3g-3+n}s_r\mu_r\,,
\end{equation}
where $\mu_r\in H^{0,1}(\Sigma,T_\Sigma(-\sigma_1\dots-\sigma_n))$ is a basis\footnote{Recall that the dimension of the moduli space of an $n$-punctured Riemann surface is $3g-3+n$, see \cref{sec6:review_loop}.} of Beltrami differentials, with coefficients $s_r$. Integrating out the parameters introduced by the action of $Q$ on the gauge fixing term introduces a term \cite{Adamo:2013tsa}
\begin{equation}\label{eq2:gaugefixinginPI}
 \prod_{r=1}^{3g-3+n}\bar\delta\left(\int_\Sigma \mu_r P^2\right)\,\int_\Sigma\tilde{b}\,\mu_r\,,
\end{equation}
in the path integral. Most importantly though, the gauge fixed action is chiral and free,
\begin{equation}
 S_{\text{gauge-fixed}}=\frac{1}{2\pi} \int_{\Sigma}P_\mu \bar\partial X^\mu+\frac{1}{2}\sum_r {\Psi_r}_\mu\bar\partial \Psi_r^\mu +b\dbar c+\tilde{b}\dbar\tilde{c} +\sum_r\beta_r\dbar\gamma_r\,.
\end{equation}

\subsection{Vertex operators and tree-level amplitudes} 
Physical states in the worldsheet conformal field theory are in one-to-one correspondence with cohomology classes of the BRST operator $Q$. Vertex operators give infinitesimal deformations of the worldsheet action, and thus describe infinitesimal fluctuations of the background geometry. In particular, the constraints imposed by $\{Q, V\}\sim 0$ determine the linearised target space field equations.

The unique\footnote{The most general fixed vertex operator is given by $U=c\tilde{c}\,\delta(\gamma_1)\delta(\gamma_2)u$. Imposing a $\mathbb{Z}_2\times\mathbb{Z}_2$ symmetry on the fermions $\Psi_r$, motivated by the involutions on supersymmetric ambitwistor space and the standard GSO projection in string theory, uniquely determines $u$. Note that, while this breaks the O$(2)\times \text{O}(2)$ symmetry of the action, it does impose a further symmetry to be respected by the vertex operators.} NS-NS fixed vertex operator is given by
\begin{equation}
 U=c\tilde{c}\,\prod_r\delta(\gamma_r)\,\epsilon_r\cdot\Psi_r \,e^{ik\cdot X}\,.
\end{equation}
As above, $\epsilon^\mu_r$ are polarisation tensors, and $k_\mu$ denotes external momentum of the particle. Requiring the vertex operators to be in the cohomology of $Q$ requires $\epsilon_r\cdot k=0$ and $k^2=0$, and thus the field equations arise from the quantum consistency of the descent to ambitwistor space.\footnote{Compare this to string theory: as in the ambitwistor string, the field equations in string theory arise from requiring that the vertex operators are in the $Q$-cohomology. However, the $XX$ OPE is non-trivial in string theory, and hence the exponential carries conformal weight, allowing for massive states.  This observation corroborates the claim made above that the triviality of the $XX$ OPE precludes massive states in the ambitwistor string.}

The vertex operator $U$ clearly encodes a non-trivial representation of the little group, representing the graviton, B-field, dilaton in its symmetric, antisymmetric and traceless part respectively. From the uniqueness of the vertex operator, it is evident that these are indeed the only states in this sector. 

Via the standard descent procedure (see e.g. \cite{Friedan:1985ge,Polchinski:1998rr,Witten:2012bh}), we obtain the vertex operator  
\begin{equation}\label{eq2:VO_V}
 c\tilde{c}\,V=c\tilde{c}\prod_r\left(\epsilon_r\cdot P+k\cdot\Psi_r\,\epsilon_r\cdot \Psi_r\right)e^{ik\cdot X}\,.
\end{equation}
The integrated vertex operator can be derived by pairing $c\tilde{c}V$ with the moduli insertion \cref{eq2:gaugefixinginPI} in the path integral, obtained from the gauge fixing procedure. This yields
\begin{equation}
 \mathcal{V}_i=\bar\delta\left(\int_\Sigma \mu_i P^2\right)\,\int_\Sigma\tilde{b}\,\mu_r\,\int_\Sigma b\,\mu_i\,c \tilde{c} V\,.
\end{equation}
It is particularly convenient at this stage to choose a basis for the Beltrami differentials $\mu_i$ that extract the residues at the marked points. Then the antighost insertions $b_i$ and $\tilde{b}_i$ remove the ghost prefactor of \cref{eq2:VO_V}, and similarly the delta-function forces the residue of quadratic differential $P^2$ to vanish. The integrated vertex operator then takes the form
\begin{equation}
 \mathcal{V}_i=\int_\Sigma \bar\delta\left(\text{Res}_i P^2\right)\, V\,.
\end{equation}
To calculate the residue of $P^2$ at the marked points, note that the fields $X$ only occur in correlation functions in the exponents $e^{ik_i\cdot X}$ of the vertex operators. For a correlator including $n$ operators, there is therefore an elegant trick to avoid the complicated OPEs with the exponentials: formulate an effective action by including the factors $e^{ik\cdot X}$, and then integrate out the field $X$. The path integral over the zero modes then yields an overall momentum conserving delta function, while the non-zero modes enforce the equation 
\begin{equation}
    \dbar P=2\pi i \sum_i k_i \,\bar \delta(\sigma-\sigma_i) d\sigma\,,
\end{equation}
on which the $\mathcal{D}P$ path integral localises. This equation however is already familiar from our motivation for the scattering equations in \cref{sec2:review_SE}; it is the defining equation for $P_\mu$, used to construct the scattering equations \cref{eq2:DE_SE}. On the Riemann sphere this is solved by \cref{eq2:P},
\begin{equation}
  P(\sigma)=\sum_{i=1}^n \frac{k_i}{\sigma-\sigma_i}\,d\sigma\,.
\end{equation}
Therefore, $P_\mu$ is indeed a meromorphic section of the canonical bundle $K_\Sigma$, and holomorphic everywhere except at insertion points. We can calculate the residue explicitly to identify the argument of the delta-functions as the scattering equations \cref{eq2:SE_Ei}, and the integrated vertex operators become
\begin{equation}\label{eq2:int-VO}
 \mathcal{V}_i=\int_\Sigma \bar\delta\left(k_i\cdot P(\sigma_i)\right)\, V\,.
\end{equation}
Since at genus zero a quadratic meromorphic differential has at least four poles, the moduli integral localised on the scattering equations now forces $P^2$ to vanish everywhere; and thus the scattering equations ensure that the worldsheet is indeed mapped into ambitwistor space. Recall in this context that the localisation on the scattering equations is due to the gauge fixing of the $\alpha$ gauge redundancy. This is consistent with the argument above that the quotient by the Hamiltonian vector fields is associated to the gauge fields $\tilde{e}$ and $\chi_r$.\\

The integrated vertex operators should moreover look familiar from the review of ambitwistor space \cref{sec2:review_ambi}: they describe deformations of complex structure preserving the contact structure \cite{LeBrun:1983,Baston:1987av}. We have encountered them explicitly in \cref{eq2:ambi_momES} as deformations arising from deformations of the space-time metric by momentum eigenstates. The integrated vertex operators are hence supersymmetric extension of the Penrose ambitwistor representative of space-time plane waves, corresponding to deformations $\delta \eta^{\mu\nu}=\epsilon_1^\mu\epsilon_2^\nu e^{ik\cdot X}$. This geometric correspondence ties in beautifully with the conformal field theory perspective of vertex operators as infinitesimal deformations of the action and thus the background geometry of the target space.\\

To prove that the ambitwistor string underpins the CHY formulae \cref{eq2:CHYgeneral_int}, as claimed above, we are interested in computing $n$-point correlation functions, 
\begin{equation}\label{eq2:correlator}
 \cM^{(0)}_n=\left\langle U_1U_2c_3\tilde{c}_3V_3\prod_{i=4}^n\mathcal{V}_i\right\rangle\,.
\end{equation}
Note that BRST invariance of the vertex operators ensures that the amplitude is independent of the choice of fixed vertex operators. Moreover, we have included two fixed vertex operators and one descended vertex operator to absorb the zero modes of the ghosts, all other contributions vanish at genus zero. To see this, recall that for a chiral $\beta\gamma$ system, with $\gamma$ a holomorphic section of a vector bundle $E$, the Riemann-Roch theorem implies that the number of zero modes $n_\gamma$ and $n_\beta$ is given by
\begin{equation}\label{eq2:Riemann-Roch}
 n_{\gamma}-n_{\beta}=h^0(\Sigma,E)-h^1(\Sigma,E)=\text{deg}(E)+1-g\,.
\end{equation}
Then for $g\geqslant2$, $n_\gamma=0$, $n_\beta\neq0$, since for a generic choice of moduli the number of zero modes is minimised. By the same reasoning, we have $n_\gamma\neq0$ and $n_\beta=0$ at genus zero. However, for genus one, we have $n_\beta=n_\gamma=1$. In particular, the Riemann-Roch theorem \cref{eq2:Riemann-Roch} implies\footnote{Recall in this context that the bundle $K_{\Sigma}$ is defined as the dual of the tangent bundle $T$ with deg$(T_\Sigma)=2-2g$, and therefore deg$(K_\Sigma)=2g-2$, see e.g. \cite{Griffiths}.} for the $bc$ and $\tilde{b}\tilde{c}$ systems $n_c=n_{\tilde{c}}=3$, and for the $\beta_r\gamma_r$ systems $n_{\gamma_r}=2$.

Using the elegant trick of including the exponentials into an effective action, we can now compute the correlator \cref{eq2:correlator}. The two fermion system $\Psi_r$ give rise to the Pfaffians, with the removed rows and columns determined by the choice of fixed vertex operators. The worldsheet perspective thus manifests that the amplitude is invariant under the choice of removed rows and columns, since this corresponds directly to a different choice of basis for the Beltrami differentials $\mu_r$. The term $\epsilon_r\cdot P$ in the descended and integrated vertex operators \cref{eq2:VO_V} and \cref{eq2:int-VO} contributes the diagonal entries in the $C$ minors of the matrix $M$. The correlator \cref{eq2:correlator} thus reproduces the CHY formulae for the full tree-level S-matrix of Einstein gravity \cref{eq2:CHYgeneral_int}, and therefore not only provides a worldsheet model underpinning these expressions, but also a satisfying explanation for their seemingly miraculous existence.\\ 

It is worth pointing out that, while the discussion here focused on NS-NS vertex operators, Ramond sector vertex operators are also known \cite{Adamo:2013tsa}, however, no closed-form expressions are known for Ramond scattering amplitudes, see \cite{Adamo:2013tsa,Weinzierl:2014ava}. Moreover, the ambitwistor string has been shown to contain the full non-linear structure of classical supergravity \cite{Adamo:2014wea}, and to extend to loop amplitudes
\cite{Adamo:2013tsa,Casali:2014hfa,Adamo:2015hoa}. We review the ambitwistor string at genus one in \cref{sec6:review_loop}, and the remainder of \cref{chapter6} gives compelling evidence that the one-loop amplitudes derived from the ambitwistor string indeed correspond to type II supergravity.\\

Instead of starting with the ambitwistor action \cref{eq2:action} corresponding to the type II RNS string, we could have constructed a bosonic or heterotic model. However, the bosonic model is not invariant under space-time diffeomorphisms for non-trivial backgrounds, which is reflected by gravitational degrees of freedom not corresponding to Einstein gravity. The fermions introduced in the RNS superstring cancel this anomaly, and render the model well-defined on a curved background \cite{Adamo:2014wea}. The heterotic ambitwistor string can be obtained by replacing one set of fermions $\Psi_2$ by a worldsheet current algebra $j_a\in\Omega^0(\Sigma, K_\Sigma\otimes\mathfrak{g})$, see also \cref{sec3:freeferm} for more details. This model is critical for a current algebra central charge $c_j=41-\frac{5}{2}d$. Correlation functions in the heterotic model indeed reproduce all leading-trace Yang-Mills amplitudes in the CHY representation, but the model contains unphysical gravity states that corrupt multiple trace and loop amplitudes. Replacing the second set of fermions with another current algebra yields a model for the bi-adjoint scalar theory. This is critical if the central charges of the current algebras $c_{1,2}$ relate to the space-time dimension $d$ as $c_1+c_2=2(26-d)$. We will explore ambitwistor string models for the full family of massless theories discussed in \cref{sec2:masslessmodels} in the following chapter.

\chapter{Ambitwistor String Models for Massless Theories} \label{chapter3}
The ambitwistor string clearly resolves the riddle posed by the existence of the CHY formulae for Einstein gravity, Yang-Mills theory and the bi-adjoint scalar discussed in \cref{sec2:CHY}, giving a new perspective on scattering amplitudes in these theories. However, for the remaining massless theories like Born-Infeld and Einstein-Yang-Mills reviewed in \cref{sec2:masslessmodels}, we are still essentially in the same place as before the ambitwistor string: the origin of these formulae still remains puzzling, and the remarkable insights of the CHY formulae for a wider family of massless theories still pose the same questions regarding an underlying theory. Moreover, in the light of the ambitwistor string, the existence of these formulae would be even more miraculous if they were not based on a mathematical framework similar to the ambitwistor string.

The aim of this chapter is to provide an answer to this question by constructing ambitwistor worldsheet models giving rise to the extensive family of CHY formulae.\footnote{Indeed, Ohmori had in parallel work already provided a partial answer by constructing ambitwistor strings for Born-Infeld theory and the Galileon \cite{Ohmori:2015sha}.} In particular, while modifying the field content of the theory, these models will preserve the (bosonic) structure of ambitwistor space. While the quantum consistency of these models is in some cases not as well understood as for the type II ambitwistor string, they represent compelling evidence that both the CHY formulae {\it and } the ambitwistor strings are universal for massless scattering amplitudes.

\paragraph{Motivation.}
Recall from \cref{sec2:masslessmodels} that the tree-level amplitudes of an extensive family of massless theories can be expressed in the CHY representation as an integral over the moduli space of an $n$-marked Riemann sphere, localised on the scattering equations \cref{eq2:CHYgeneral_int},
\begin{equation}
\cM(1,\ldots,n)=\delta^d\left(\sum_i k_i\right)\int_{(\CP^1)^n}\frac{\prod_{i=1}^n\d\sigma_i}{\vol \,G}{\prod_i}\,\bar\delta\left(k_i\cdot P(\sigma_i)\right)\,\cI^L(\sigma,k,\epsilon)\,\cI^R(\sigma,k,\epsilon)\,,
\end{equation}
where the integrands were specified in \cref{sec2:masslessmodels}. Notably, the integrand naturally decomposes into factors $\cI^L$ and $\cI^R$ that depend on the null momenta $k_i$, their associated marked points $\sigma_i$, and the polarization and/or colour data of the particles scattered. These factors are the only elements of \cref{eq2:CHYgeneral_int} specifying the theory, and can be chosen from  five different choices, see \cref{chapter2} for a full list.\\

In the discussion above, we have placed the emphasis on the common structure of the CHY formulae, given by the integration over the moduli space and the localisation on the scattering equation, versus the integrand $\cI=\cI^L\,\cI^R$, which specifies the theory in question. The ambitwistor string mirrors this distinction exactly; both the type II, the heterotic model and the one describing the bi-adjoint scalar theory are built out of a bosonic model with action $S_B$, to which worldsheet matter of the form $S^L+S^R$ is added. This bosonic model is responsible for the common structures observed in the CHY amplitudes; giving rise both to the integration over the moduli space and the localisation on the scattering equations in the vertex operators. However, the {\it matter content} $S^L+S^R$ of the action determines the form of the vertex operators, and thereby the form of the integrands. The factorisation of the integrand is moreover mirrored in a factorisation of the vertex operators into two currents $v^l$ and $v^r$.

Having re-framed the ambitwistor string in this format, the task is clear: to obtain the ambitwistor models underlying the remaining CHY formulae, we have to find new matter content $S^{L,R}$ to add to the action. There will be five choices of matter corresponding to the five choices for $\cI^{L,R}$ as integrands in the CHY formulae, see table \ref{models0}. The main body of this chapter will be dedicated to constructing these models and investigating their scattering amplitudes. 

\renewcommand{\arraystretch}{2}
\begin{table}[t]{\small
\begin{tabular}{|c||l|l|l|l|l|}
  \hline
  \diagbox{$S^L$}{$S^R$}& $S_\Psi$ & $S_{\Psi_1,\Psi_2}$ & $S_{\rho,\Psi}^{(\widetilde m)}$ & $S_{YM,\Psi}^{(\widetilde N)}$ & $S_{YM}^{(\widetilde N)}$\\ \hline \hline
  $S_\Psi$ & E &  & &  & \\ \hline 
  $S_{\Psi_1,\Psi_2}$ & BI & Galileon &  &  & \\ \hline  
  $S_{\rho,\Psi}^{(m)}$ & $\stackrel[\text{U}(1)^{m}]{}{\text{EM}}$ & DBI & $\stackrel[\text{U}(1)^{m}\times \text{U}(1)^{\widetilde m}]{}{\text{EMS}}$
   &  & \\ \hline 
  $S_{YM,\Psi}^{(N)}$ & EYM & ext.\ DBI & $\stackrel[\text{SU}(N)\times \text{U}(1)^{\widetilde m}]{}{\text{EYMS}}$ & $\stackrel[\text{SU}(N)\times \text{SU}(\widetilde N)]{}{\text{EYMS}}$ & \\ \hline  
  $S_{YM}^{(N)}$ & YM &  Nonlinear $\sigma$ & $\stackrel[\text{SU}(N)\times \text{U}(1)^{\widetilde m}]{}{\text{EYMS}}$ & $\stackrel[\text{SU}(N)\times \text{SU}(\widetilde N)]{}{\text{gen YMS}}$ & $\stackrel[\text{SU}(N)\times \text{SU}(\widetilde N)]{}{\text{Biadjoint Scalar}}$\\ \hline 
 \end{tabular}}
 \caption{Theories arising from the different choices of matter models.} \label{models0}
\end{table}

\paragraph{Outline of the chapter.}
We will start by constructing all basic ingredients in \cref{sec3:matter}. As reviewed in \cref{sec2:ambi_string}, just two ingredients were used to construct $\cI^L$ and $\cI^R$ in the original models of \cite{Mason:2013sva} - a worldsheet supersymmetry $S_\Psi$, and a current algebra $S_j$.  Einstein, Yang-Mills and bi-adjoint scalar theories were obtained from  the choices $(S^L,S^R)= (S_\Psi,S_\Psi), (S_\Psi,S_j)$ and $(S_j,S_j)$ respectively.  The current algebra $S_j$ however has the defect that it also leads to multi-trace terms in its correlators that were ignored by hand.  

In this chapter, we will introduce a different worldsheet CFT in \cref{sec-CS}, the comb system\footnote{This was originally introduced \cite{Casali:2013} in the context of twistor-strings, but never published.}, $S_{CS}$.  This gives a new way to obtain colour factors and their associated Parke-Taylor factors without multi-trace terms. Furthermore, the colour factors are presented not as cyclic single trace terms, but as strings of structure constants arranged in a `comb', hence the name.  However, the number of gauge particles in this system is doubled. To remedy this issue, a reduced system $S_{YM}$ with the correct number of gauge particles is constructed in \cref{sec-YM}, but this system is always anomalous. Nevertheless, it is sufficient to produce the correct tree amplitudes and so we use this system instead of the current algebra in the table \ref{models0}. It can be replaced by $S_{CS}$ if we are seeking an anomaly-free theory, but then we must accept the doubling of gauge particles.\\  

The remaining ambitwistor models underpinning the CHY formulae for combined models are constructed from combinations of these basic ingredient in \cref{sec3:combined-matter}. This leads to a total of five different choices of worldsheet matter $S^{L,R}$, which are in a one-to-one correspondence to the choices for the integrands $\cI^{L,R}$ of the CHY formulae discussed in \cref{sec3:NewModels}, see \cref{models0}. The ambitwistor string models discussed in this chapter therefore represent the underlying worldsheet theories for the CHY formulae, and provide an explanation for their existence.

Potentially the most interesting of these models is that for Einstein-Yang-Mills, which we obtain again in two forms, using $S_{YM}$ and $S_{CS}$. Again, as in the pure case, the former reproduces the correct tree-level amplitudes, but is anomalous, while the latter is anomaly free, but contains too many gluons, see \cref{sec3:S_YM-Psi} and \cref{sec3:S_CS-Psi} respectively. In particular, the gauge theory part of the space-time action corresponding to the comb system is given by
\begin{equation}
\label{eq:tYM0}
S_{T^*\text{YM}} = \int d^D x \;\tr( a_\mu\, D_\nu F^{\mu\nu}) \, , 
\end{equation}
and we refer to it as $T^*$YM as it describes a linearised Yang-Mills field $a$ propagating on a full Yang-Mills background for the field $A$ with curvature $F$.

Of particular interest are the models that are consistent at the quantum level, and hence critical and anomaly free, since in these cases, the ambitwistor string could be used to calculate loop amplitudes along the lines of \cite{Adamo:2013tsa}.
We will explore this and other directions briefly in \cref{sec3:Discussion}.

\section{Ambitwistor worldsheet models - general aspects}
As motivated above, the general idea is remarkably simple: in analogy to the CHY formulae, the bosonic part of the ambitwistor string action forms the backbone in common to all models, giving rise to both the integration over the moduli space of $n$-marked Riemann spheres and the delta-functions responsible for the localisation on the scattering equations. The bosonic part of the action thus ensures that we preserve the geometry of bosonic ambitwistor space. Moreover, the integrands $\cI^{L,R}$ are in one-to-one correspondence with matter theories $S^L$ and $S^R$ on $\Sigma$. We thus want to consider actions of the form
\begin{equation}
 S_B+S^L+S^R\,,
\end{equation}
where the bosonic action is given by 
\begin{equation}\label{boson-str}
S_B=S_B[X,P]=\frac1{2\pi}\int_\Sigma  P\cdot \dbar X +  \tilde e P\cdot  P\, ,
\end{equation}
see \cref{sec2:review_ambistrings} for details. The worldsheet matter $S^L$ and $S^R$ determines the form of the vertex operators $V\in\Omega^0(\Sigma, K_\Sigma^2)$.  More specifically, the vertex operators all contain a factor of $\e^{ik\cdot X}$, with the remainder $V$ factorizing into two independent currents,
$$
V=v^lv^r \, , \qquad v^l, v^r \in\Omega^0(\Sigma, K_\Sigma)\,.
$$
The $v^l$, $v^r$ will be  constructed from the matter models $S^L$ and $S^R$ respectively, and are constrained by quantum consistency, BRST invariance, and possibly further discrete symmetries. Invariance under $Q_B$ for example implies $k^2=0$ because the $P^2$ term in the BRST operator brings down $k^2$ in its double contraction with $\e^{ik\cdot X}$, and thus we recover the on-shellness condition for our external fields.  We will also use the notation $u^l$, $u^r$  for such currents when they are fixed with respect to fermionic symmetries, see below.  \\

Essentially the only candidate for $v^l$ and $v^r$ in the purely bosonic model is $\epsilon\cdot P$ for some polarization vector $\epsilon^\mu$ defined up to multiples of $k^\mu$ under $Q_B$ equivalence, leading to unphysical formulae for gravity amplitudes (see also \cref{sec2:review_ambistrings} and \cite{Mason:2013sva}).  In order to obtain more interesting models,  we introduce different worldsheet matter models $S^L$ and $S^R$ that generate the currents $v^l$ and $v^r$ in the vertex operators.  In general, we will take the models $S^L$ and $S^R$ to be distinct matter theories so that the correlator factorizes into a product of correlators for the left and right currents, and we will be able to calculate them separately.  In order to ensure that the only allowed vertex operators do indeed factorize in this way, we impose discrete symmetries that are analogues of the GSO symmetries of conventional string theories, and in a slight abuse of notation, we refer to them as GSO symmetries as well.

\section{Worldsheet matter models and their correlators}\label{sec3:matter}
In \cite{Mason:2013sva}, two matter models were considered: (1) $S_\rho$, a current algebra which we will take to be generated by free fermions, and (2) $S_\Psi$, which introduces a degenerate worldsheet supersymmetry.  These led to three models with $(S^L,S^R)$ given by $(S_{\Psi_l},S_{\Psi_r})$ for type II supergravity, $(S_\Psi,S_\rho)$ for Yang-Mills amplitudes and $(S_{\rho_l},S_{\rho_r})$ for amplitudes of a bi-adjoint scalar theory.  In this section we will consider a third type of matter that we call the `comb system'  $S_{CS}$ \cite{Casali:2013} - a worldsheet conformal field theory that will be important for Yang-Mills amplitudes. As the name suggests, correlators in this model give colour invariants in the form of comb structures built out of structure constants rather than colour traces.  In the rest of this section, we describe these matter systems, and the natural currents to which they give rise as candidates for $v^l$ and $v^r$ and their correlation functions.  In the next section we see how these are altered when these systems are combined.\newpage

\subsection{Free fermions \texorpdfstring{$S_\rho$}{S-rho} and current algebras \texorpdfstring{$S_j$}{S-j}}  \label{sec3:freeferm}
 The standard action for `real' free fermions $\rho^a\in \Omega^0(\Sigma,K_\Sigma^{1/2})$, $a=1,\ldots m$, is
\begin{equation}
S_\rho=\int \rho^a\dbar \rho^a\, ,
\end{equation}
(the summation convention is assumed).  The term `real' is used to distinguish them from the complex fermion system given by 
\begin{equation}
S_{\rho,\tilde\rho}=\frac1{2\pi i} \int \tilde\rho_a \dbar\rho^a\,.
\end{equation}
The simplest currents in the real case are $j^{ab}=\rho^a\rho^b$ that form an elementary example of a current algebra for $SO(m)$ (in the complex case $j^a_b=\tilde{\rho}_b\rho^a$ generate a current algebra for $\SU(m)$).

More generally, we can consider an arbitrary current algebra $j\in \Omega^0(\Sigma,K_\Sigma\otimes \g)$ of level $k$, where $\g$ is some Lie algebra, satisfying the usual current algebra OPE,
\begin{equation}\label{current-alg}
j^a(\sigma)j^b(0)\sim \frac{k\delta^{ab}}{\sigma^2} + \frac{if^{abc}j^c(\sigma)}{\sigma}+\ldots\, ,
\end{equation}
where $f^{abc}$ are the structure coefficients, $[t^a,t^b]=f^{abc} t^c$, $\delta^{ab}$ is the Killing form ($a=1,\ldots,\dim \g$). This could be constructed from free fermions, Wess-Zumino-Witten models or some other construction and we will generally represent such matter by $S_j$.

Given choices of $t\in \g$, the current algebra can contribute 
\begin{equation}
 v=t\cdot j\,,
\end{equation}
to one or both  factors $v_l$ and $v_r$ of the vertex operators $V$. The current correlators $\la t_1\cdot  j(\sigma_1) \ldots t_n\cdot  j(\sigma_n)\ra$ lead to Parke-Taylor factors:
\begin{equation}
\mathcal{C}_n=\frac{\text{Tr}(t_1\ldots t_n)}{\sigma_{12}\sigma_{23}\ldots \sigma_{n1}} \,,
\end{equation}
where $\sigma_{ij}=\sigma_i-\sigma_j$.  However, the correlators also give rise to multi-trace terms that are ultimately problematic and unwanted.

\subsection{Worldsheet supersymmetry \texorpdfstring{$S_\Psi$}{S-Psi}} \label{sec3:WS-susy}
While discussed in \cref{chapter2}, we review this matter model here to fix the notation. Worldsheet supersymmetry is introduced by adding fermionic worldsheet spinor fields $\Psi\in \Omega^0(\Sigma, K_\Sigma^{1/2}\otimes\C^d)$, and a gauge field $\chi \in \Omega^{(0,1)}(\Sigma,T_\Sigma^{1/2})$ for the supersymmetry.   Their action is
\begin{equation}
S_\Psi= \frac{1}{2\pi i}\int \Psi\cdot \dbar \Psi + \chi P\cdot \Psi\, .
\end{equation}
The constraint leads to worldsheet gauge transformations
\begin{equation}
\delta \chi=\dbar \eta \, , \qquad \delta X^\mu=\eta \Psi^\mu\, , \qquad \delta \Psi^\mu=\eta P^\mu\, , \qquad \delta P_\mu=0\, ,
\end{equation}
where $\eta \in\Pi\Omega^0(\Sigma,T_\Sigma^{1/2})$ is  a fermionic parameter. Gauge fixing leads to bosonic ghosts $\gamma \in \Omega^0(\Sigma,T_\Sigma^{1/2})$ and corresponding antighosts $\beta\in\Omega^0(\Sigma,K_\Sigma^{3/2})$.  The BRST operator acquires an extra term
\begin{equation}
Q_\Psi=\oint \gamma \cG_{\Psi}\, , \qquad \cG_{\Psi}:= P\cdot \Psi\, .
\end{equation}
On $\CP^1$, the ghosts $\gamma$ have two zero modes by the Riemann-Roch theorem.  Thus, as far as the fermionic symmetry is concerned, we need two fixed vertex operators to fix the ghost zero modes, constructed from $\delta(\gamma) $ multiplied by a factor with values in $ K_\Sigma^{1/2}$. The `integrated' vertex operators (in the fermionic sense) arise from these via the usual descent procedure.   The relevant currents are then given by
\begin{equation}
u=\delta(\gamma)\,\epsilon\cdot \Psi\, , \qquad v=\epsilon\cdot P + k\cdot\Psi\epsilon\cdot\Psi\, .
\end{equation}

These operators are invariant under the discrete symmetry that changes the sign of $\Psi$, $\chi$ and the ghosts.  Imposing invariance under this symmetry will exclude mixing between the ingredients of these operators thought of as parts of $S^L$ and others that might be part of $S^R$.  We will refer to this as GSO symmetry.

As discussed in \cref{chapter2}, the correlators of these currents lead to the reduced Pfaffians of the CHY amplitudes,
\begin{equation}
\la u_1  u_2 v_3\ldots v_n\ra = \mathrm{Pf}'(M)=\frac{1}{\sigma_1-\sigma_2}\mathrm{Pf}(M_{12})\, ,
\end{equation}
where $M$ is the skew $2n\times 2n$ matrix with $n\times n$ block decomposition
\begin{equation}
M=\begin{pmatrix}
A&-C^T\\C&B
\end{pmatrix} \, , \qquad A_{ij}=\frac{k_i\cdot k_j}{\sigma_{ij}}\, , \qquad B_{ij}=\frac{\epsilon_i\cdot \epsilon_j}{\sigma_{ij}}\, , 
\end{equation}
and
\begin{equation}
C_{ij}=\frac{\epsilon_i\cdot k_j}{\sigma_{ij}}\, , \;i\neq j, \qquad C_{ii}= -\epsilon_i\cdot P(\sigma_i)\, ,
\end{equation}
and $M_{12}$ is $M$ with the first two rows and columns removed.

\subsection{Comb system \texorpdfstring{$S_{CS}$}{S-CS}}\label{sec-CS}
The comb system  \cite{Casali:2013} was introduced as a way of obtaining colour factors as sequences of contractions of structure constants rather than as colour ordered traces. In general, such contractions can be generated from trivalent diagrams with the structure constants $f^{abc}$ of some Lie algebra at the vertices and contractions $\delta^{ab}$ along internal edges.  It is well known that the colour factors are linearly dependent as a consequence of the Kleiss-Kuijf relations with a basis being given by `combs', with $n-2$ vertices lined up in a row \cite{KK1989,DelDuca:1999rs} and end points given by 1 and $n$:
\vspace{.4cm}
\begin{center}
\includegraphics[width=5cm]{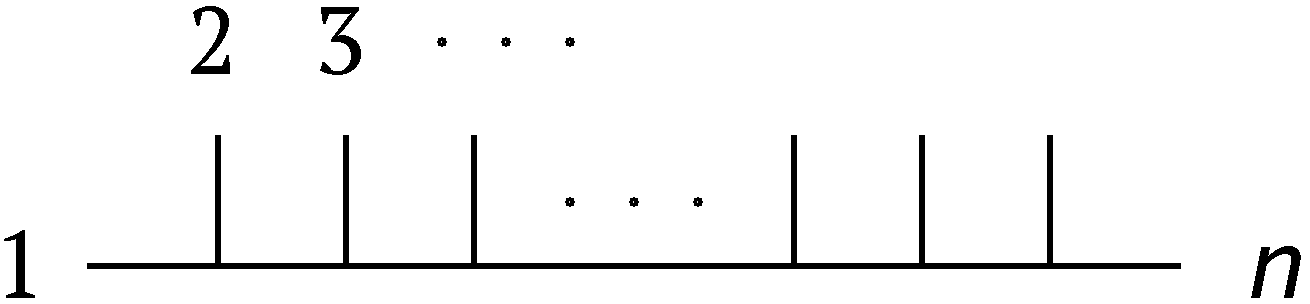} \qquad $\to \qquad f^{a_1a_2b_1}f^{b_1a_3b_2}\cdots f^{b_{n-3}a_{n-1}a_n}$.
\end{center}
\vspace{.4cm}
The comb system \cite{Casali:2013} has the remarkable property that, in conjunction with worldsheet supersymmetry, only these combs arise from correlators,  without the  multitrace terms occurring for ordinary current algebras. This system arises from an action for matter fields 
$\rho, \tilde \rho \in \Pi\Omega^0(\Sigma, K_\Sigma^{1/2}\otimes \mathfrak{g})$ and $ q, y \in \Omega^0(\Sigma, K_\Sigma^{1/2}\otimes \mathfrak{g})$, i.e. worldsheet spinors taking values in the Lie algebra $\mathfrak{g}$ of some gauge group.   The worldsheet action is
\begin{equation}
S_{CS}=\int  \tilde\rho\cdot \dbar \rho + q\cdot \dbar y + \chi \rho\cdot\left(\frac{\scriptstyle 1}{\scriptstyle 2} [\rho,\tilde\rho]+ [q,y]\right)\,,
\end{equation}
 with $\rho, \tilde \rho$ fermionic and $q,y$ bosonic and the $\cdot $ notation is used here to denote the Killing form on the Lie algebra.  As before, $\chi\in\Pi\Omega^{0,1}(\Sigma,T_\Sigma^{1/2})$ is a gauge field on the worldsheet and we are gauging 
the current\footnote{With different assignment of worldsheet spins this current would be a normal BRST current.  If we were to take $\rho, y$, scalars and $\tilde{\rho}, q$  sections of $K_\Sigma$, then $\rho$ and $\tilde\rho$ could be taken to be the ghosts associated to  gauge fixing a worldsheet gauge field $a\in\Omega^{0,1}(\Sigma, \mathfrak{g})$ with action $\int_\Sigma q\cdot \dbar y +q\cdot [a,y]$.} $\rho\cdot\left(\frac{\scriptstyle 1}{\scriptstyle 2} [\rho,\tilde\rho]+ [q,y]\right)$, which is a section of $K_\Sigma^{3/2}$.  
The gauging introduces transformations now for fermionic  $\alpha\in \Pi\Omega^0(\Sigma,T_\Sigma^{1/2})$,
\begin{equation}
\delta(\rho,\tilde\rho,q,y)=\alpha \left(\frac{1}{2}[\rho,\rho]\,, [\rho,\tilde{\rho}]\,, [\rho,q]\,,[\rho,y]\right)\, , \qquad \delta \chi=\dbar\alpha\, . 
\end{equation}
 As in the case of worldsheet supersymmetry, gauge fixing gives bosonic ghosts $\gamma\in \Omega^0(\Sigma,T_\Sigma^{1/2})$ and antighosts $\beta$ that contribute to the BRST operator as
 \begin{equation}
Q_{CS}=\oint \gamma  \cG_{CS}\, , \qquad \cG_{CS}:=\rho\cdot\left(\frac{\scriptstyle 1}{\scriptstyle 2} [\rho,\tilde\rho]+ [q,y]\right)\, .
\end{equation}

As for $S_\Psi$, at genus zero the ghosts develop two zero-modes, and so the correlator contains two fermionically fixed operators with the rest integrated.  The currents that contribute to the vertex operators in this system now depend on a Lie algebra element $t\in\g$,  with two types of fixed and integrated vertex operators respectively
\begin{equation}\label{eq3:VO-comb}
u=\delta(\gamma)t\cdot \rho\, , \quad \tilde u=\delta(\gamma)t\cdot \tilde \rho\, , \qquad  v=\half t\cdot [\rho,\rho]\, , \quad \tilde v=t\cdot\left( [\rho,\tilde\rho]+ [q,y]\right)\, .
\end{equation}
Here $v=\{Q_{CS},t\cdot\rho\}$ and $\tilde v=\{Q_{CS},\tilde t\cdot \tilde \rho\}$, and we need two fixed vertex operators in any correlator, with the remaining vertex operators unfixed.\footnote{
A  more symmetric way to understand this is to say that we choose all unintegrated vertex operators, but then we must insert $n-2$ `picture-changing operators'
\begin{equation}
\Upsilon=\delta(\beta) \rho\cdot\left(\frac{\scriptstyle 1}{\scriptstyle 2} [\rho,\tilde\rho]+ [q,y]\right)\, .
\end{equation}
These could be inserted anywhere in general.  If inserted at one of the $u, \tilde u$ insertion points, it will convert it into a corresponding $v,\tilde v$. A similar approach can be taken for correlators  associated with the $S_\Psi$ matter system. 
}
Notice that $\tilde v= t\cdot j$, where $j^a$ is a level zero current algebra, and that
\begin{equation}
\tilde v (\sigma)\; t'\cdot \tilde \rho (0) \sim - \frac{[t,t']\cdot \tilde \rho (0)}{\sigma} + \ldots \,,
\qquad \tilde v (\sigma)\; t'\cdot \rho (0) \sim - \frac{[t,t'] \cdot \rho (0)}{\sigma} + \ldots \;.
\end{equation}
The following proposal then determines the form of the correlators:
\begin{propn}[Casali-Skinner] Correlators of the currents $u,v,\tilde u,\tilde v$ are only non-vanishing when there is just one untilded current and give 
\begin{equation}\label{eq3:CS-cor}
\la u_1 \tilde v_2\ldots \tilde v_{n-1} \tilde{u}_n\ra=\la \tilde u_1 v_2 \tilde v_3\ldots \tilde v_{n-1} \tilde{u}_n\ra=
\cC(1,\ldots,n)\,,
\end{equation}
where
\begin{equation}
\cC_n=\cC(1,\ldots,n):=\frac{\tr(t_1[t_2,[\ldots,[t_{n-1} ,t_n] \ldots] ])}{\sigma_{12}\sigma_{23}\ldots \sigma_{n1}} + \mathrm{Perm}(2,\ldots,n-1)\, .
\end{equation}
\end{propn}

Instead of the colour traces arising from $S_j$, we obtain `combs', i.e. strings of structure constants $\tr(t_1[t_2,[\ldots,[t_{n-1} ,t_n] \ldots] ])$ as described in \cite{DelDuca:1999iql,DelDuca:1999rs}.  

The argument is as follows. The correlator contains exactly two fixed vertex operators $u$, $\tilde u$ due to the standard counting of $\gamma$ ghost zero modes by the Riemann-Roch theorem. Consider the $\rho, \tilde\rho$ contractions: correlators are only non-trivial for an equal number of $\rho$ and $\tilde \rho$ fields, since neither develops zero modes at genus zero. From the form of the vertex operators \cref{eq3:VO-comb}, it is easily seen that this restricts to the correlators \cref{eq3:CS-cor}. The integrated vertex operators $v$ and $\tilde{v}$ connect  along a `comb', whereas the fixed ones $u$ and $\tilde u$ form the ends, giving rise to the `colour combs' defined above.  However, we can also have contributions from a subset of $\tilde v$ vertex operators contracting in a loop.  This is where the $(q,y)$ system comes into play. These fields can only contract in loops, but, being bosonic, their loop contractions cancel the analogous loop contractions from the $(\rho, \tilde{\rho})$ system.  This can also be seen from the form of the current algebra generated by the $\tilde v$'s:  by construction, this has level zero and thus cannot generate a non-trivial trace after a sequence of OPE's.

\subsection{Other systems with comb structure \texorpdfstring{$S_{YM}$}{S-YM}}\label{sec-YM}  
A problem with the CS system above is that there are clearly two types of gluons, corresponding to the vertex operators $(\tilde u,\tilde v)$ and $(u,v)$ respectively.
We will see that this is not appropriate for pure Yang-Mills although it does give a theory that is sufficient to generate Einstein-Yang-Mills tree amplitudes correctly on certain trace sectors, the ones selected by the choice of untilded operators.\footnote{One may try to symmetrise the correlator in tilded versus untilded gluonic operators, for instance by using $u_t+\tilde u_t$ and $v_t+\tilde v_t$, but then there will be an over-counting of contributions, so that the relative factors of different terms are not correct.} The system we introduce here will give the complete Einstein Yang-Mills amplitude from a single correlator, but will be anomalous and ill-defined due to $Q^2\neq0$.

A worldsheet CFT that will generate Yang-Mills following the ideas above requires the following ingredients.  We need a fermionic worldsheet spinor $\rho^a\in \g$ for the fixed vertex operator, transforming in the adjoint representation of a current algebra $v^a\in \g$ at level zero for the integrated vertex operators; the level zero allows us to avoid multi-trace terms and loops.  Finally we need a spin 3/2 current $\cG_{YM}$ with the following OPE to give the appropriate group compatibilities and descent:
\begin{equation}\label{requirements}
\rho^a(\sigma)\rho^b(0) \sim \frac{\delta^{ab}}{\sigma}\, , \;\;  v^a(\sigma)\rho^b(0)\sim \frac{f^{abc}\rho^c(0)}{\sigma}\, , \;\; \cG(\sigma) \rho^a(0)\sim \frac{v^a(0)}{\sigma}\, , \;\; \cG(\sigma)\cG(0)\sim 0\, .
\end{equation}
It is easy to see that this can be partially realized with $\rho^a $ a `real' free fermion with action $\frac{1}{2}\int \rho^a\dbar \rho^a$, and with 
\begin{equation}
v^a= -\frac{1}{2} f^{abc}\rho^b\rho^c + j^a\, ,\quad j^a(\sigma)\rho^b(0)\sim 0\, ,
\end{equation}
we will obtain the first two of the equations above.   
$\frac{1}{2}f^{abc}\rho^b\rho^c$ is  a current algebra with level $-C$ where $f^{abc}f^{\tilde abc}=C\delta^{a\tilde a}$, so in order for $v^a$ to be a current algebra with level zero, we must take $j^a$ to be a current algebra with level $k=C$. There are many ways to do this, so let us leave this to one side for a moment.  We then need to construct $\cG$.  In order for $\cG$ to generate $v^a$ from $\rho^a$, we require
\begin{equation}
\cG= -\frac{1}{6} f^{abc}\rho^a\rho^b\rho^c + \rho^a j^a +\ldots\,,
\end{equation}       
where the $\ldots$ has non-singular OPE with $\rho^a$ and $j^a$. At this point, however, we see that an anomaly arises preventing $\{\cG,\cG\}=0$. To be specific,
\begin{equation}
\cG(\sigma)\cG(0) \sim \frac{C\,\textrm{dim}(G)}{\sigma^3} + \frac{:j^aj^a(0):}{\sigma}\,,
\end{equation}
where we recall that the energy-momentum tensor of the current algebra $j$ is given by $T(\sigma)=:j^aj^a(\sigma):/2k$. Therefore, we are able to satisfy the first three equations of \cref{requirements}, while the last equation is anomalous, and thus the BRST quantisation is inconsistent.

\subsection{Central charges}

We remark that the theories  $S_B$, $S_\rho$, $S_\Psi$ and $S_{CS}$  above respectively have central charges 
\begin{equation}
 \mathfrak{c}_B=2d-52,\qquad  \mathfrak{c}_\rho =m/2\, , \qquad  \mathfrak{c}_\Psi= d/2+11, \qquad \mathfrak{c}_{CS}=11\, ,
\end{equation}
the latter being just that of the $(\beta,\gamma)$ system as the $(\rho,\tilde\rho)$ and $(q,y)$ contributions cancel via supersymmetry.  (This can be different if the $(q,y)$ are not taken to be spin $1/2$.)  Notably, the type II supergravity model is critical in 10 dimensions as then $c_B+2c_\Psi=0$.  These considerations are less interesting for $S_{YM}$ as that theory is intrinsically anomalous, and in any case its central charge will depend on the choice of  current algebra $j^a$.


\section{Combined matter models}\label{sec3:combined-matter}

On their own, the new worldsheet matter theories $S_{CS}$ and $S_{YM}$ of the previous section do little more than give an alternative to the current algebras in the original models of \cite{Mason:2013sva} that avoids the multitrace terms that were neglected by hand.  To obtain new theories, we will consider the contributions to $S^L$ or $S^R$ of combinations of the above matter systems.
Even without $S_{CS}$ and $S_{YM}$, we will obtain a number of interesting new models. 
Here we will discuss the allowable vertex operators and the correlators of the various combinations that we can form.  These are summarized in the table \ref{matter-combinations}. 

\begin{table}[ht]  {\footnotesize
\begin{tabular}{|l||l|l|l|l|}\hline
  & Fermionic current $\mathcal{G}$ & Matter & Vertex operators & Correlator\\ \hline\hline
 $S_\Psi$ & $P\cdot \Psi$ & $\Psi$ & $u_\Psi=\delta(\gamma)\,\epsilon\cdot\Psi$ & Pf$'(M)$\\ \hline
 \multirow{2}{*}{$S_{\Psi_1,\Psi_2}$} & $P\cdot \Psi_1$ & \multirow{2}{*}{$\Psi_1,\,\Psi_2$} & $u_{\Psi_1}=\delta(\gamma_2)\,k\cdot\Psi_1$ & \multirow{2}{*}{$\big($Pf$'(A)\big)^2$}\\
 & $P\cdot \Psi_2$ & & $u_{\Psi_2}=\delta(\gamma_1)\,k\cdot\Psi_2$ & \\ \hline
 \multirow{2}{*}{$S_{\rho,\Psi}$} & \multirow{2}{*}{$P\cdot \Psi$} & \multirow{2}{*}{$\Psi$, $\rho_a \quad a=1,\dots,m$} & $u_\Psi=\delta(\gamma)\,\epsilon\cdot\Psi$ & \multirow{2}{*}{Pf$(\chi)\,$Pf$'(M|_{\text{red}})$}\\
 &&& $u_{\rho_a}=\delta(\gamma)\,\rho_a$ & \\ \hline
 \multirow{3}{*}{$S_{CS,\Psi}$} & \multirow{3}{*}{$P\cdot \Psi+$tr$\left(\rho(\frac{1}{2}[\tilde{\rho},\rho]+[q,y])\right)$} & \multirow{3}{*}{$\Psi,\,(\tilde{\rho},\,\rho),\,(q,y)$} & $u_\Psi=\delta(\gamma)\,\epsilon\cdot\Psi$ & \multirow{3}{*}{$\mathcal{C}_{(1)}\dots\mathcal{C}_{(m)}$Pf$'(\Pi)$}\\
 &  & & $\tilde{u}_{CS}=\delta(\gamma)\,$tr$(t\tilde{\rho})$ & \\
 &  & & $u_{CS}=\delta(\gamma)\,$tr$(t\rho)$ & \\ \hline
 \multirow{2}{*}{$S_{CS}$} & \multirow{2}{*}{tr$\left(\rho(\frac{1}{2}[\tilde{\rho},\rho]+[q,y])\right)$} &
 \multirow{2}{*}{$(\tilde{\rho},\,\rho),\,(q,y)$} & $\tilde{u}_{CS}=\delta(\gamma)\,$tr$(t\tilde{\rho})$ & \multirow{2}{*}{$\mathcal{C}_n$}\\
 &  & & $u_{CS}=\delta(\gamma)\,$tr$(t\rho)$ &  \\ \hline
\end{tabular}}
\caption{Table of matter models, their combinations and worldsheet correlators.}
\label{matter-combinations}
\end{table}

\subsection{\texorpdfstring{$S_{\rho,\Psi}$}{S-rho,Psi}} 
Consider the action
\begin{equation}
S^L=S_{\rho,\Psi}:=S_\rho+S_\Psi\, .
\end{equation}   
The free fermion system $S_\rho$ seems to naturally lead  to  the  $SO(m)$ current algebra $j^{ab}=j^{[ab]}=\rho^a\rho^b$,  and could therefore superficially be thought to result in the same current algebra. In the presence of worldsheet supersymmetry however, the currents $j^{ab}$ as constituents of the vertex operators are not BRST invariant, since
\begin{equation}
\{Q_\Psi, j^{ab} \e^{ik\cdot X}\}= ik\cdot\Psi j^{ab} \e^{ik\cdot X}\neq 0\, .
\end{equation}
On the other hand, allowable fixed and integrated currents are respectively given by
\begin{equation}
u^a=\delta(\gamma) \rho^a\, , \qquad v^a=k\cdot \Psi \rho^a\, , \quad a=1,\ldots , m\, .
\end{equation}
We also have the standard BRST invariant currents from $S_\Psi$, which in this context we will denote $u_\epsilon=\delta(\gamma)\epsilon\cdot \Psi$ and $v_\epsilon=\epsilon\cdot P+ k\cdot\Psi \epsilon\cdot\Psi$.

In general we will be concerned with a correlator $\la u_1 u_2 v_3 \ldots v_n\ra $ where, if $(\gamma,h)$ is a partition of $1,\ldots , n$,  for  $i\in \gamma $ the current will be one of the new photon currents, and for $i\in h$ it will be a $S_\Psi$ current depending on a polarization vector $\epsilon_{i,\mu}$.   The correlator will factorize into one for the constituent $\rho$'s and one for the $\Psi$'s.  We compute these as Pfaffians of the associated matrices of possible contractions in the correlator.  The simplest is the $\rho$ system. If we restrict it to take values in an algebra with vanishing structure constants, e.g. $\oplus^m  \mathfrak{u}(1)$, the OPEs lead to the $|\gamma| \times |\gamma|$ CHY matrix 
\begin{equation}
\cX_{ij}=\frac{\delta^{a_ia_j}}{\sigma_{ij}}\, , \quad i,j\in \gamma, \quad i\neq j, \quad  \mbox{ otherwise } \quad \cX_{ij}=0\, .
\end{equation}
The Kronecker delta fro $\tr ( t^{a_i}t^{a_j} )$ in the numerator ensures only photons of the same flavour interact.

Much as before, the $\Psi$ system leads to the matrix of possible $\Psi$ contractions
\begin{equation}
M_{\mathrm{Red}}=\begin{pmatrix}
A^{\gamma\gamma}&A^{\gamma h}& (-C^{h\gamma})^T\\
A^{h\gamma}& A^{hh}&(-C^{hh})^T\\
C^{h\gamma}&C^{hh}& B
\end{pmatrix} \,,
\end{equation}
where the matrix is expressed in a bock decomposition under $n=|\gamma|+|h|$ and
\begin{equation}
A_{ij}=\frac{k_i\cdot k_j}{\sigma_{ij}}\, , i\neq j, \quad A_{ii}=0\, ,\qquad B_{ij}=\frac{\epsilon_i\cdot\epsilon_j}{\sigma_{ij}}\, , \quad i,j\in h, i\neq j\,, 
\end{equation}
and 
\begin{equation}
C_{ij}=\frac{\epsilon_i\cdot k_j}{\sigma_{ij}} \, , \quad i\in h , i\neq j\, .
\end{equation}
Finally, the additional $\epsilon\cdot P$ term in the $S_\Psi$ vertex operator is incorporated by setting $C_{ii}=-\epsilon_i\cdot P(\sigma_i)$ as before.  We thus obtain a reduced Pfaffian associated with the two fixed vertex operators as before.  Our final correlator expression is therefore
\begin{equation}
\la u_1 u_2 v_3 \ldots v_n\ra =\Pf(\cX)\Pf'(M_{Red})\,.
\end{equation}
In comparison with \cref{chapter2}, we note that $\pf'(M_{Red})=\pf'(M_{\hat{\gamma}})$ or $\pf'(M_{Red})=\pf'(M_{\hat{s}})$, depending on the theory.

For the GSO symmetry we require all fields, $\rho, \Psi$ and the ghosts to change sign simultaneously.

\subsection{\texorpdfstring{$S_{\Psi_1,\Psi_2}$}{S-Psi1,Psi2}}  
Here we consider two worldsheet supersymmetries 
\begin{equation}
S^L=S_{\Psi_1}+S_{\Psi_2}\, .
\end{equation}
 There are two contributions to the BRST operator, $Q_{\Psi_1}+Q_{\Psi_2}$. The currents from $S_{\Psi_1}$ and  $S_{\Psi_2}$ derived in \cref{sec3:WS-susy} no longer form allowed vertex operators since they are not BRST closed,
\begin{equation}
\{Q_{\Psi_2},\,\delta(\gamma_1) \epsilon\cdot\Psi_1 \e^{i k\cdot X}\}= \delta(\gamma_1)\,\gamma_2\,  k\cdot\Psi_2 \epsilon\cdot \Psi_1 \e^{i k\cdot X} \neq 0\, .
\end{equation}
However, non-trivial  BRST invariant currents are simply given by descendants of $\delta(\gamma_1)\delta(\gamma_2)$, 
\begin{equation}
u=\delta(\gamma_1) \delta(\gamma_2)\, , \qquad v= k\cdot \Psi_1 k\cdot \Psi_2\, ,
\end{equation}
as in \cite{Ohmori:2015sha}. (Equivalently, we could have included the partial descendants $\delta(\gamma_1 ) k\cdot \Psi_2$ and $\delta(\gamma_2)k\cdot \Psi_1$.) 
Again, the correlator of $n$ such vertex operators factorizes into a product of reduced Pfaffians $\pf'(A)$, originating from all possible  $\Psi_1$ and $\Psi_2$ Wick contractions.  As before, $A$ is defined by its off-diagonal entries $k_i\cdot k_j/\sigma_{ij}$ and has co-rank two, and the reduced Pfaffian corresponds to the choice of fixed versus integrated vertex operators.  We therefore obtain
\begin{equation}
\la u_1u_2v_3v_4\ldots v_n\ra=(\Pf'(A))^2\, .
\end{equation}

One might ask whether one can carry on to combine three or more $S_\Psi$ systems into $S^L$, but this is not possible since there are no non-trivial BRST invariant currents.

Again for the GSO symmetry we require all fields, $\Psi_r$ and the ghosts, to change sign simultaneously.

\subsection{\texorpdfstring{$S_{YM,\Psi}$}{S-YM,Psi}}\label{sec3:S_YM-Psi}

In Sections \ref{sec-CS} and \ref{sec-YM}, we introduced $S_{CS}$ and $S_{YM}$ whose correlators provide the colour  comb-structure  together with Parke-Taylor factors. For the remainder of this section, we will combine each of these two systems with $S_\Psi$. The goal is to obtain the building block of Einstein-Yang-Mills amplitudes that gives the appropriate interactions  between gluons and gravitons. We start by discussing the combined theory $S_{YM,\Psi}$, which is slightly simpler than $S_{CS,\Psi}$ and possesses the main important features. Despite $S_{YM}$ not being quantum-mechanically consistent  - and this problem  extends to $S_{YM,\Psi}$ - we are able to obtain Yang-Mills tree amplitudes.  The  theory $S_{CS,\Psi}$ can be made consistent, but has two types of gluons and the corresponding amplitudes arise from an action that is not Yang-Mills (although it contains its classical solutions).

Since both worldsheet matter theories $S_\Psi$ and $S_{YM}$ involve the gauging of spin 3/2 currents $\cG_\Psi=P\cdot \Psi$ and $\cG_{YM}=\rho\cdot\left(-\frac{\scriptstyle 1}{\scriptstyle 6} [\rho,\rho]+ j\right)$, we have the option of gauging both these currents together or separately.  If we gauge them separately, we find that the resulting system is too restrictive to lead to interesting results.  Thus we gauge the sum 
\begin{equation}
\cG = P\cdot \Psi +\rho\cdot\left(-\frac{ 1}{ 6} [\rho,\rho]+ j\right),.
\end{equation}
Gauge fixing introduces a single set of ghosts $(\beta,\gamma)$, and the currents
\begin{equation}
u_t=\delta(\gamma)\rho\cdot t\, , \quad \quad u_ \epsilon=\delta(\gamma)\epsilon\cdot \Psi\,,
\end{equation}
still represent allowed fixed vertex operators. BRST descent\footnote{We highlight at this point again that the quantisation is inconsistent, and thus does not lead to a BRST cohomology. However, the operator $Q$ was defined such that the descent procedure is still valid.} then leads to the integrated vertex operators
\begin{equation}\label{eq3:YM-Psi-VO}
 v_t=k\cdot\Psi \rho\cdot t + v^0_t\,,
 \quad\quad v_\epsilon=\epsilon\cdot P+\epsilon\cdot\Psi k\cdot\Psi\,,
\end{equation}
where $v_t^0$ denotes the original $S_{YM}$ integrated vertex operator, satisfying the OPE relations \eqref{requirements} except the last.  Although the failure of the last relation means that the BRST quantisation is inconsistent, the correlator of these vertex operators does nevertheless  give the correct amplitudes.

In the previous section, we saw that the system $S_{YM}$ on its own gives the correct colour-dressed Parke-Taylor factors, in terms of a comb structure. The combination with $S_\Psi$ leads to additional insertions of $\rho\cdot t$ from the fixed vertex operators \cref{eq3:YM-Psi-VO}, and these will start additional combs. In this way we obtain multiple colour combs/traces and get the right interactions with gravity states. On the other hand, the system $S_\Psi$ on its own leads to a reduced Pfaffian. The combination with $S_{YM}$ will lead to a different but closely related Pfaffian that incorporates the multi-comb structure. We now describe the complete correlator.

\begin{thm}\label{ym-correlator}  As in \cite{Cachazo:2014xea}, let the sets $g$ index the gluons with vertex operators $u_t,v_t$, and $h$ the gravitons with vertex operators  $u_\epsilon, v_\epsilon$.  To be non-zero, a correlator must contain two fixed vertex operators $u$'s, with the remaining ones being $v$'s. The correlator  is then a sum over all partitions of the gluons into sets $T_1, T_2,\ldots, T_m$, where $\cup_{i=1}^m T_i=g$ and $|T_i|\geq 2$.  Each partition gives rise to the term
\begin{equation}
 \sum_{\substack{c_1 < d_1 \in T_1 \\ \cdots \\ c_m < d_m \in T_m }}  \mathcal{K}(c_1 ,d_1 | T_1) \cdots  \mathcal{K}(c_n ,d_n | T_n)  ~~ \mathrm{Pf} ' \left( \begin{array}{ccc|c}
A_{ab} & A_{ a c_j} &  A_{a d_j }  & (-C^T)_{ab}  \\ 
A_{c_i b} & A_{c_i c_j} &    A_{c_i d_j }  & (-C^T)_{c_i b}  \\
  A_{d_i b} &  A_{d_i c_j} & A_{d_i d_j}  &  (-C^T)_{d_i b} \\  \hline
C_{ab} & C_{a c_j} &   C_{a d_j}  & B_{ab} 
\end{array} \right) ~.
\label{eqn:KPf}
\end{equation}
Here, $a$ and $b$ label gravitons and $c_i,d_i$ label gluons in $T_i$, so that $A_{ab}$ is an $|h|\times|h|$ matrix, $A_{c_ib}$ is an $m\times|h|$ matrix, and $A_{c_ic_j}$ is an $m\times m$ matrix. Moreover, we define
\begin{equation}\label{eqn:Comb-Trace-relation}
\mathcal{K}(i,j | T) =  \sigma_{ji} \; \mathcal{C}(T) \,,
\end{equation}
where $\mathcal{C}(T)$ is $\mathcal{C}_n$ restricted to $g\in T$. The reduced Pfaffian $\pf '$ is defined in \cref{eq:redpf}.
\end{thm}

The proof is given in \cref{sec:ecomb-proof}. This correlator reproduces the main building block of the CHY formula for Einstein-Yang-Mills amplitudes in \cite{Cachazo:2014xea}.  Although not quite in the same form, the equivalence can easily be seen from equations (3.16) and (3.17) of \cite{Cachazo:2014xea} and this form is more natural from its derivation as a correlator.

\subsection{\texorpdfstring{$S_{CS,\Psi}$}{S-CS,Psi}}\label{sec3:S_CS-Psi}

While $S_{YM,\Psi}$ gives the correct  amplitude, its BRST quantisation is  inconsistent. We can obtain the same structure from $S_{CS,\Psi}$ by combining the worldsheet theories $S_\Psi$ and $S_{CS}$, which has the advantage of being anomaly free but the disadvantage of containing two types of gluons.

As for $S_{YM,\Psi}$ we gauge  the sum of spin 3/2 currents $\cG_\Psi=P\cdot \Psi$ and $\cG_{CS}=\rho\cdot\left(\frac{\scriptstyle 1}{\scriptstyle 2} [\rho,\tilde\rho]+ [q,y]\right)$, leading to the action 
\begin{equation}
S_{CS_\Psi}=\int \Psi\cdot\dbar\Psi + \tilde \rho\cdot \dbar \rho + q\cdot \dbar y +\chi\left(P\cdot \Psi+ \rho\cdot\left(\frac{\scriptstyle 1}{\scriptstyle 2} [\rho,\tilde\rho]+ [q,y]\right)\right).
\end{equation}
Now the Lie-algebra valued fermion $\rho$ is complex (i.e. not equal to $\tilde \rho$), unlike the previous case of $S_{YM,\Psi}$.  This  will change the physical content of the model. 

The gauge fixing of $\chi$ introduces just one set of ghosts $(\beta,\gamma)$, and we find the standard fixed currents for $S_{CS}$  and $S_\Psi$,
\begin{equation}
u_t=\delta(\gamma)\rho\cdot t\, , \quad \tilde u_t=\delta(\gamma)\tilde\rho \cdot t \, , \quad u_ \epsilon=\delta(\gamma)\epsilon\cdot \Psi\,.
\end{equation}
The BRST descent then leads to the following currents 
\begin{equation}
 v_t=k\cdot\Psi \rho\cdot t + v^0_t\,,
\qquad \tilde v_t=k\cdot \Psi \tilde \rho\cdot t + \tilde v^0_t\, ,\quad v_\epsilon=\epsilon\cdot P+\epsilon\cdot\Psi k\cdot\Psi,
\end{equation}
where $v^0_t$ and $\tilde v^0_t $ denote the original $S_{CS}$ integrated vertex operators. Let us highlight again that both
 $v_t$ and $\tilde v_t$ acquire a new term in $\Psi$. 

To impose GSO symmetry, we require invariance under flipping the sign of the  fields $\rho, \tilde{\rho}, q,y, \Psi, \chi $ and the corresponding ghosts.

Since we have untilded vertex operators $u_t,v_t$, and tilded ones $\tilde u_t, \tilde v_t$, the correlator will depend not only on the number of gluonic vertex operators versus gravity ones $u_\epsilon, v_\epsilon$, but also on the choice of whether the gluonic operators are of untilded or tilded type. Recall from the previous section that, for the theory $S_{CS}$ on its own, the only non-vanishing correlators were those with a single untilded operator and this led to a single comb colour structure that is equivalent to a single trace term.  This followed because of the need to have the same number of $\rho$'s and $\tilde \rho$'s in a non-trivial correlator and a single $\tilde \rho$ could only arise in one or both of the two fixed vertex operator.  Now single $\tilde\rho$'s appear in $\tilde v_t$ and this essentially represents the coupling to gravity. Thus  the coupling to gravity introduces multiple trace terms, with the interaction between each single trace structure being mediated by gravity.  It is easy to see that with the $S_{CS,\Psi}$ system we can now have as many untilded vertex operators as we like with their number corresponding precisely to the number of traces.

\begin{thm}\label{cs-correlator}  Let the set $g$ index the gluons and $h$ the gravitons. To be non-vanishing, a $S_{CS,\Psi}$ correlator must have two fixed vertex operators, with the remaining ones integrated. The correlator of such a collection of vertex operators is  a sum over all partitions of the gluons into sets $T_1, T_2,\ldots, T_m$, where $m$ is the number of untilded gluonic vertex operators, and such that there is only one such vertex operator per $T_i$, $\cup_{i=1}^m T_i=g$, $|T_i|\geq 2$. Each allowed partition gives a contribution equal to \cref{eqn:KPf}.
\end{thm}

Thus the correlator is the same as for $S_{YM,\Psi}$, except that there is a restriction on the allowed partitions of the gluons into traces.

\section{New ambitwistor string theories}\label{sec3:NewModels}
We can now assemble the full table of theories by combining the various possible choices of matter models on the left and right.  These can be identified with their  corresponding space-time theories by comparing the correlators to the formulae of CHY \cite{Cachazo:2014xea}, and this results in table \ref{models}. Hopefully the acronyms for the models are self-explanatory except perhaps that BS denotes the bi-adjoint scalar $\phi^{aa'}$, where $a$ and $a'$ are respectively indices for the Lie algebras of $\SU(N) $ and $\SU(N')$, with action
$$
S_{\mathrm{BS}}= \int d^Dx \left( -\frac{1}{2}\p_\mu\phi^{aa'}\,\p^\mu\phi^{aa'} + \frac{1}{6}\, \phi^{aa'}\phi^{bb'}\phi^{cc'}f^{abc}f^{a'b'c'}\right)\, ,  
$$
where $f^{abc}$ and $f^{a'b'c'}$ are the structure constants of $\SU(N) $ and $\SU(N')$ respectively.

\begin{table}[ht]{\small
\begin{tabular}{|c||l|l|l|l|l|}
  \hline
  \diagbox{$S^L$}{$S^R$}& $S_\Psi$ & $S_{\Psi_1,\Psi_2}$ & $S_{\rho,\Psi}^{(m')}$ & $S_{YM,\Psi}^{(N')}$ & $S_{YM}^{(N')}$\\ \hline \hline
  $S_\Psi$ & E &  & &  & \\ \hline 
  $S_{\Psi_1,\Psi_2}$ & BI & Galileon &  &  & \\ \hline  
  $S_{\rho,\Psi}^{(m)}$ & EM$\big|_{\text{U}(1)^{m}}$ & DBI & EMS$\big|_{\text{U}(1)^{m}\otimes \text{U}(1)^{m}}$ &  & \\ \hline 
  $S_{YM,\Psi}^{(N)}$ & EYM & extended DBI & EYMS$\big|_{\text{SU}(N)\otimes \text{U}(1)^{m'}}$ & EYMS$\big|_{\text{SU}(N)\otimes \text{SU}(N')}$ & \\ \hline  
  $S_{YM}^{(N)}$ & YM &  NLSM & YMS$\big|_{\text{SU}(N)\otimes \text{U}(1)^{m'}}$ & gen. YMS$\big|_{\text{SU}(N)\otimes \text{SU}(N')}$ & BS
  \\ \hline 
 \end{tabular}}
 \caption{Theories arising from the different choices of matter models.} \label{models}
\end{table}

For convenience, we list here the actions of the more exotic theories:
\begin{subequations}
\begin{align}
 S_{\mathrm{Galileon}}&=\int d^Dx \left( -\frac{1}{2}\p_\mu\phi\,\p^\mu\phi + \sum_{m=3}^\infty g_m \phi \text{det}\{\p^{\mu_i}\p_{\nu_j}\phi\}_{i,j=1}^{m-1} \right)\,,\\
 S_{\mathrm{BI}}&=\int d^Dx\; \frac{1}{\ell^{2}}\left(\sqrt{-\text{det}(\eta_{\mu\nu}-\ell^2F_{\mu\nu})}-1\right)\,,\\
 S_{\mathrm{DBI}}&=\int d^Dx \; \frac{1}{\ell^{2}} \left(\sqrt{-\text{det}\left(\eta_{\mu\nu}-\ell^2 \p_\mu\phi^a\p_\nu\phi^a -\ell F_{\mu\nu}\right)}-1\right)\, ,\\
 S_{\mathrm{NLSM}}&=\int d^Dx\left( -\frac{1}{2}\text{tr}\left((\mathbb{1}-\lambda^2\Phi)^{-1}\p_\mu\Phi(\mathbb{1}-\lambda^2\Phi)^{-1}\p^{\mu}\Phi\right)\, \right)\,.
\end{align}
\end{subequations}
For the non-linear $\sigma$-model, $\Phi=\phi^at^a$, and in the Galileon action, $g_m$ denote freely prescribable parameters. However, the CHY-amplitudes only contain one parameter. The theory that is singled out by the ambitwistor model and the CHY formulae is the one described in \cite{Cheung:2014dqa,Cachazo:2014xea,Hinterbichler:2015pqa}, which has smoother soft behaviour than the generic ones.\\

Table \ref{models}, showing how the theories are determined in terms of a pair of worldsheet systems, is a remarkable manifestation of the notion of double copy. This notion has been explored mostly in the context of gravity amplitudes, which are obtained as the double copy of gauge theory ones \cite{Kawai:1985xq,Bern:2008qj}. In the formalism of the scattering equations, this is the double copy of Pfaffian factors, and in ambitwistor string theory, this is the double copy of the worldsheet system $S_\Psi$, as in table \ref{models0}. The amplitude formulae of ref.~\cite{Cachazo:2014xea} and our results extend this notion to a range of other theories. Regarding the relation to previous work, we should mention that a double copy construction for Einstein-Yang-Mills amplitudes was first presented in \cite{Bern:1999bx} for the single trace contribution, and in \cite{Chiodaroli:2014xia} for the complete amplitude, with results extending to loop level. These double copy constructions are based on the colour-kinematics duality \cite{Bern:2008qj,Bern:2010ue}, whose relation to the scattering equations has been explored in \cite{Cachazo:2013iea,Monteiro:2013rya,Naculich:2014rta}.\\

In the table \ref{models}, we have only used $S_{YM}$.  Although this is sufficient to produce the correct tree-level amplitudes, it is an anomalous matter system and so has no hope to be extended beyond tree-level, and indeed its meaning as a string theory is unclear even at tree level.  We can obtain the same tree-amplitudes up to combinatorial factors by use of the comb system $S_{CS}$ , which is not anomalous.  However, this does lead to a doubling of the gauge degrees of freedom as described below in detail for the Einstein Yang-Mills system and bi-adjoint scalar.    

In table \ref{table-VO} we list the vertex operators in each model and the central charges. It can be seen that many models have a critical dimension, in which the central charge vanishes and for which there is some reasonable hope that loop integrands can be described via these theories, provided their one-loop correlation functions are modular invariant.\footnote{The last requirement of modular invariance might not be necessary, see \cref{chapter6}.}

\subsection{Einstein Yang-Mills and \texorpdfstring{$T^*$}{T*}YM}

The worldsheet model that we discussed in the context of Einstein Yang-Mills theory, $S_{CS,\Psi}$, has a consistent quantisation. On the other hand, it does not correspond strictly to the building block of Einstein-Yang-Mills amplitudes, 
because only trace/comb structures consistent with the choice of untilded vertex operators are allowed. Attempts to find a theory that reproduces this correlator seem to lead back to the anomalous $S_{YM,\Psi}$ system. 

Since the theory $S_{CS,\Psi}$ presents no problems, and has correlators which match part of the Einstein-Yang-Mills building block, it is natural to ask whether it is related to a known theory. This theory must contain two types of gluons, associated to tilded and untilded vertex operators, and the untilded type must give the number of allowed multiple trace terms in an amplitude. These conditions are satisfied by the following space-time action for the gauge field
\begin{equation}
\label{eq:tYM}
S_{T^*\text{YM}} = \int d^D x \;\tr( a_\mu\, D_\nu F^{\mu\nu})\,.
\end{equation}
The field $a_\mu$ is a Lagrange multiplier enforcing the Yang-Mills equations $D_\nu F^{\mu\nu}=0$, and the action can be seen as a linearisation of the Yang-Mills action, $A_\mu\to A_\mu+a_\mu$. The field $A_\mu$ corresponds to the tilded degrees of freedom, and the field $a_\mu$ corresponds to the untilded ones. Since the propagator of this action connects $a_\mu$ to $A_\mu$ and the vertices contain a single $a_\mu$, the Feynman rules and a straightforward graph-theoretic argument show that there is one and only one $a_\mu$ external field per trace, also when the system is minimally coupled to gravity. This model therefore describes a linearised Yang-Mills field $a$ propagating on a full Yang-Mills background for the field $A$ with curvature $F$, and we thus refer\footnote{Here $a$ is canonically conjugate to $F$ hence the name $T^*$YM as opposed to $T$YM.} to it as $T^*$YM.  Whilst this should give correct Yang-Mills amplitudes at one loop, it has no higher loop amplitudes in the pure gauge sector. In its critical dimension $d=10$, we would thus expect the ambitwistor string model to give a valid expression for the one-loop integrand for Yang-Mills.

\subsection{Bi-adjoint scalar}

The use of the worldsheet system $S_{CS}$, with its two types of coloured currents, ${\tilde v}$ and $v$, is the reason for the Lagrange multiplier-type action \eqref{eq:tYM}. An even simpler example is the bi-adjoint scalar theory, BS in \cref{models}. In this case, we can easily apply the procedure of \cite{Adamo:2014wea} and obtain explicitly the equations of motion. As in that paper, which was concerned with supergravity, the space-time background fields modify the worldsheet theory only via the constraints. The deformation of the constraints in the bi-adjoint scalar theory is particularly simple: the deformed ambitwistor constraint becomes
\begin{equation}
\mathcal{H} = P^2  \qquad \to \qquad \mathcal{H}_{(\phi,\Phi)}= P^2 
+ \Phi^{aa'}{\tilde v}^a{\tilde v}'^{a'}  + \phi^{aa'}{ v}^a{ v}'^{a'}\,,
\end{equation}
where we introduced currents for each of the two independent groups SU($N$) and SU($N'$). The equations of motion are obtained as anomalies obstructing the vanishing of the constraint at the quantum level,
\begin{align}
\mathcal{H}_{(\phi,\Phi)}(\sigma)\,\mathcal{H}_{(\phi,\Phi)}(0) \sim \frac{1}{\sigma^2} 
\Big( & (2\,\partial^\mu\partial_\mu\Phi^{aa'} +f^{abc}f^{a'b'c'}\Phi^{bb'}\Phi^{cc'}) \, {\tilde v}^a {\tilde v}'^{a'} \nonumber \\
&+(2\,\partial^\mu\partial_\mu\phi^{aa'} +2\,f^{abc}f^{a'b'c'}\Phi^{bb'}\phi^{cc'}) \, v^a v'^{a'} \Big)(0)  \\
+ & \;\;\text{simple} \;\; \text{pole}\,.\nonumber
 \end{align}
If the equations of motion hold, there is no double pole and in fact the OPE is finite, since there can be no simple pole in the self-OPE of a bosonic operator in the absence of higher poles.
The space-time action associated to these equations of motion takes the Lagrange-multiplier form
\begin{equation}
S_{\text{BS}} = \int d^D x \; \phi^{aa'} \left(\partial^\mu\partial_\mu\Phi^{aa'} +\frac{1}{2}\, f^{abc}f^{a'b'c'}\Phi^{bb'}\Phi^{cc'}\right)\, .
\end{equation}
It should be seen as the analogue of the gauge theory action \eqref{eq:tYM}.

\begin{table}[ht]  {\small
\begin{tabular}{|l||l|l|}\hline
 Theories & Integrated vertex operators & Central charge $c$  \\ \hline\hline
 E &  $V_h=\left(\epsilon\cdot P + k\cdot\Psi\epsilon\cdot\Psi\right)\left(\tilde\epsilon\cdot P + k\cdot\tilde\Psi\tilde\epsilon\cdot\tilde\Psi\right)$ & $3(d-10)$  \\ \hline
 \multirow{2}{*}{EM} & $V_h,\, V_\gamma$ & \multirow{2} {*}
 {$3(d-10+\frac{m}{6})$}  \\
 & $V_\gamma=\left(k\cdot \Psi\,t\cdot \rho \right) \left(\tilde\epsilon\cdot P + k\cdot\tilde\Psi \tilde\epsilon\cdot\tilde\Psi\right)$ & \\ \hline
 \multirow{2}{*}{EMS} & $V_h,\,V_\gamma,\,V_{\tilde\gamma},\,V_S$ & \multirow{2}{*}{$3(d-10+\frac{m+\tilde{m}}{6})$} \\ 
 & $V_S=\left(k\cdot \Psi \,t\cdot \rho\right)\left(k\cdot \tilde\Psi t\cdot\tilde\rho\right)$  &\\ \hline
 BI & $V_{BI}=\left(k\cdot \Psi_1 k\cdot \Psi_2\right)\left(\tilde\epsilon\cdot P + k\cdot\tilde\Psi \tilde\epsilon\cdot\tilde\Psi\right)$ & $\frac{1}{2}\left(7d-38\right)$ \\ \hline
 Galileon & $V_G=\left(k\cdot \Psi_1 k\cdot \Psi_2\right)\left(k\cdot \tilde\Psi_1 k\cdot \tilde\Psi_2\right)$ & $4d-8$ \\ \hline
 \multirow{2}{*}{DBI} & $V_{BI},\,V_{S_{BI}}$ & \multirow{2}{*}{$\frac{1}{2}(7d+m-38)$} \\
  & $V_{S_{BI}}=\left(k\cdot \Psi_1 k\cdot \Psi_2\right)\left(k\cdot \tilde\Psi\,t\cdot \tilde\rho\right)$ & \\ \hline
 \multirow{2}{*}{$T^*$YM}&$V_g=\left(\half t\cdot [\rho,\rho]\right)\left(\tilde\epsilon\cdot P + k\cdot\tilde\Psi\tilde\epsilon\cdot\tilde\Psi\right)$ & \multirow{2}{*}{$\frac{5}{2}(d-12)$} \\ 
 & $V_{\tilde{g}}=\left(t\cdot\left( [\rho,\tilde\rho]+ [q,y]\right)\right)\left(\tilde\epsilon\cdot P + k\cdot\tilde\Psi\tilde\epsilon\cdot\tilde\Psi\right)$ & \\ \hline
 \multirow{3}{*}{E$T^*$YM} & $V_h,\,V_g,\,V_{\tilde{g}}$ & \multirow{3}{*}{$3(d-10)$}  \\
 &$V_g=\left(k\cdot\Psi\,t\cdot \rho +\half t\cdot [\rho,\rho]\right)\left(\tilde\epsilon\cdot P + k\cdot\tilde\Psi\tilde\epsilon\cdot\tilde\Psi\right)$ &\\ 
 & $V_{\tilde{g}}=\left(k\cdot\Psi \,t\cdot \tilde\rho +t\cdot\left( [\rho,\tilde\rho]+ [q,y]\right)\right)\left(\tilde\epsilon\cdot P + k\cdot\tilde\Psi\tilde\epsilon\cdot\tilde\Psi\right)$ & \\ \hline
 \multirow{2}{*}{NLSM} & $V=\left(\half t\cdot [\rho,\rho]\right)\left(k\cdot \tilde\Psi_1 k\cdot \tilde\Psi_2\right)$ & \multirow{2}{*}{$3d-19$}  \\
 & $\tilde{V}=\left(t\cdot\left( [\rho,\tilde\rho]+ [q,y]\right)\right)\left(k\cdot \tilde\Psi_1 k\cdot \tilde\Psi_2\right)$ &\\\hline
\end{tabular}}
\caption{Table of the different theories and their integrated vertex operators.}
\label{table-VO}
\end{table}

\section{Discussion}\label{sec3:Discussion}
We have seen compelling evidence in this chapter for the universality of ambitwistor strings for scattering amplitudes in massless theories. In particular, we have extended the ambitwistor string to theories beyond the theories arising from string theory. This highlights the field theory nature of the ambitwistor string from a very different perspective, and strongly indicates that there exist, at least perturbatively, better ways of understanding massless field theories.\\

The work presented in this chapter opens up many possible directions for future exploration, and we will briefly mention a few of them in this discussion.

A very interesting question, concerning these models, is if they allow for an extension to loop amplitudes by taking the correlation functions on higher genus Riemann surfaces as described in \cite{Adamo:2013tsa}. In particular, this requires necessarily that the models are anomaly free and have a critical dimension where the central charge vanishes. We have included the central charges for various models that are not anomalous in \cref{table-VO}, and it can be seen that indeed many models have a critical dimension. Criticality can often be achieved by adding Maxwell fields for low enough dimension.  This suggests that a number of these models might give rise to plausible string expressions for corresponding loop integrands such as given in \cite{Adamo:2013tsa} for the type II theory in 10 dimensions.  However, an independent criterion is that the loop integrand so obtained should be modular invariant and this may well exclude many of the critical models as it does in conventional string theory. A possible way to circumvent this requirement is opened up by a different, ambitwistor string inspired approach to loop amplitudes that we will explore in \cref{chapter6}. We will see that loop amplitudes for gravity, Yang-Mills theory and the bi-adjoint scalar can be obtained alternatively from nodal Riemann spheres, and thus it would be interesting to extend this approach to the full list of (anomaly-free) massless theories.

A related question concerns the existence of further vertex operators, and therefore further sectors of these theories.  In 10 dimensions, following \cite{Adamo:2013tsa}, one can introduce a spin field $\Theta^\alpha$ associated to each $\Psi$ field and use these to introduce further vertex operators that will correspond to space-time fields with spinor indices.  For the type II Einstein theory these give rise to the Ramond sector vertex operators \cite{Adamo:2013tsa} and it can be seen that the same procedure can be applied more generally to some of the models here, particularly to the  Einstein $T^*$YM models.  Following the same procedure one then extends the Einstein NS sector to include the Ramond sectors of type II gravity theories.  However, we can see that the $T^*$YM vertex operators can only be extended in this way on the one side corresponding to the spin operator constructed from the $\Psi$ in the Yang-Mills vertex operator.  Thus one supersymmetry acts trivially on the Yang-Mills and hence is degenerate (it does not square to provide the Hamiltonian on the Yang-Mills fields). \\

By extending the worldsheet matter fields we have generated new possible couplings to space-time fields.  It would be interesting  to explore whether these couplings can be made consistent in the fully nonlinear regime as described in \cite{Adamo:2014wea,Chandia:2015sfa}.

There remain other formulae based on the scattering equations, for which  an underlying ambitwistor string theory has not yet been found. It would for example be  interesting to find ambitwistor strings that give rise to the class of formulae with massive legs \cite{Dolan:2013isa,Naculich:2014naa,Naculich:2015zha, Naculich:2015coa}, and that for ABJM theory \cite{Huang:2012vt,Cachazo:2013iaa}, in particular in the light of the four-dimensional ambitwistor string \cite{Geyer:2014fka,Lipstein:2015vxa} discussed in \cref{chapter5} and the ABJM twistor string \cite{Engelund:2014sqa}.\\

Perhaps the most irritating issue is that we have not been able to find an Einstein-Yang-Mills model that is anomaly-free without unwanted linearised modes.  Conventional string theory produces such amplitudes in open string theory and in closed string heterotic models.  However, the ambitwistor heterotic string has unphysical gravity amplitudes and so far there has been no ambitwistor analogue of open strings. Nevertheless the $T^*$YM model is likely to make sense and provide the correct amplitudes at one loop if modular, although the pure gauge sector does not have loop amplitudes beyond one loop.\\

\chapter{Ambitwistor Strings, Soft Theorems and the Geometry of Null Infinity} \label{chapter4}


Up to this point, the discussion of this thesis has focussed mainly on ambitwistor strings as the mathematical framework underlying scattering amplitudes. As briefly indicated above, a spectacular recent development indicates their importance ina  much wider context by formulating ambitwistor strings on a curved background \cite{Adamo:2014wea}. This provides an important proof that ambitwistor strings can be used to obtain non-linear information, like the fully non-linear Einstein equations, and can provide insights beyond scattering amplitudes. 

In this chapter, we will expand on another extension, relating the asymptotic symmetries of a space-time to scattering amplitudes as natural observables in the `bulk'. To motivate this, recall that diffeomorphism invariance of general relativity (and quantum gravity) implies that there are no local observables. In other words, every physical observable has to be global, and therefore be specified in terms of a theory on the boundary of space-time in some suitable sense. This is the fundamental idea behind the holographic principle, which is best known in its application to the AdS/CFT correspondence: quantum gravity in the full space-time is equivalent to a theory on the boundary. In particular, in the case of Anti-de Sitter spaces in $d+1$ dimensions, the asymptotic symmetry group SO$(2,d)$ is exactly the conformal group of the Conformal Field Theory on the boundary in $d$ dimensions. This highlights the importance of the study of asymptotic symmetries, and provides an excellent example for their relation to 'bulk' observables.

In what follows, we will apply this general idea to the case of asymptotically flat space-times, and focus on the relation between observables of massless bulk theories\footnote{The most natural diffeomorphism invariant observables in a quantum field theory, both from a theoretical and an experimental point of view, are scattering amplitudes. Note that they are also intrinsically holographic, being defined with respect to asymptotic states. } and the symmetries of null infinity $\scri$. 

Ambitwistor strings are a natural candidate theory for explaining and probing this duality due to their underlying geometric structure. In general formulated over any Cauchy hypersurface, there exists a particularly suitable representation for the study of asymptotic symmetries that identifies ambitwistor space with the cotangent bundle at null infinity. Ambitwistor strings can therefore be formulated entirely with respect to the boundary $\scri$, making them ideal candidates for the boundary theory dual to the bulk quantum gravity {\it in a regime where classical supergravity is a valid approximation}.\footnote{At this point it is worth emphasising that ambitwistor strings should not be understood as the full holographic dual of quantum gravity in an asymptotically flat space-time, but rather, as highlighted above, as a suitable effective theory in a limit where classical supergravity is a valid effective field theory in the bulk.}\\

The purpose of this chapter is to explore the duality between the asymptotic symmetries of an asymptotically flat space-time and the low-energy behaviour of a theory in the context of ambitwistor strings, rather than Ashtekar's Fock space of radiative modes used in \cite{Strominger:2013lka,Strominger:2013jfa,He:2014laa,Cachazo:2014fwa} . It also expands on the relationship between asymptotic symmetries of an asymptotically flat space-time and the vertex operators in the theory, and how a soft momentum eigenstate  becomes an extended BMS generator at leading and subleading order.  Moreover, we will see that these ideas are realised straightforwardly in the more twistorial four dimensional ambitwistor strings in \cref{sec5:scri}.\\

This chapter is structured as follows: After a brief review of the asymptotic symmetries, soft limits and the Ward identities linking them in \cref{sec4:review-scri}, we give a description of the geometry of ambitwistor space at null infinity in \cref{sec:soft_geom}, and describe the Hamiltonian lift of asymptotic symmetries of $\scri$ to ambitwistor space. This geometric picture is naturally encoded in the ambitwistor string since its action is based on the symplectic potential. The singular components of the OPEs in the ambitwistor worldsheet theory thus correspond directly to the Poisson structure on the cotangent bundle. Therefore, the Hamiltonian generating diffeomorphisms of $\scri$ induces directly the action of the symplectic diffeomorphism in the ambitwistor string model, see \cref{sec:soft_model}. We will then implement the Ward identity relating soft theorems and asymptotic symmetries in the ambitwistor worldsheet CFT in \cref{sec:soft_BMS} by expanding vertex operators in a low-energy limit and identifying the leading and subleading contribution\footnote{Recall that general vertex operators arise from Hamiltonians generating diffeomorphisms of ambitwistor space that determine the scattering from past to future null infinity.}  as the generators of symmetries at $\scri$.
This provides a beautifully geometric interpretation of the relation of asymptotic symmetries to the infrared behaviour of the bulk theory, and more generally gives an explicit perturbative correspondence between the scattering of null geodesics and that of the gravitational field via ambitwistor string theory. 


\section{Review of asymptotic symmetries and soft limits}\label{sec4:review-scri}
It has long been understood that infrared behaviour in gravity is related to supertranslation ambiguities in the choices of coordinates at null infinity \cite{Ashtekar:1981sf,ashtekar1987asymptotic}. In a recent series of papers \cite{Strominger:2013lka,Strominger:2013jfa,He:2014laa,Cachazo:2014fwa}, Strominger and coworkers have proposed a new way of understanding the Weinberg soft limit theorems as a Ward identity associated to the BMS group at null infinity and used the approach to suggest new theorems for the subleading terms in the soft limit.  

In an intriguing recent paper \cite{Adamo:2014yya}, Adamo, Casali and Skinner proposed a string model at null infinity for four dimensions to provide an explanation for these ideas.  They also suggested a link with ambitwistor strings that extends their ideas to arbitrary dimension, particularly in view of the recent proof of the subleading soft limit results in \cite{Schwab:2014xua} using the CHY formulae \cite{Cachazo:2013gna, Cachazo:2013hca,Cachazo:2013iea} as these formulae arise from ambitwistor string theory in arbitrary dimensions \cite{Mason:2013sva}. 

In this section, we will provide a general review of the background material, including a discussion of soft limits, the BMS group and their relation as proven by Strominger et al.

\paragraph{Soft limits.}
Several decades ago, Weinberg showed that photon and graviton
amplitudes behave in a universal way when one of the external
particles with momentum $s$ becomes soft \cite{Weinberg:1965nx}: 
\begin{equation}\label{eq4:softlimits}
\mathcal{A}_{n+1}\rightarrow\sum_{a=1}^{n}\frac{\epsilon_{a}\cdot
  k_{a}}{s\cdot
  k_{a}}\mathcal{A}_{n},\qquad \mathcal{M}_{n+1}\rightarrow\sum_{a=1}^{n}\frac{\epsilon_{\mu\nu}k_{a}^{\mu}k_{a}^{\nu}}{s\cdot
  k_{a}}\mathcal{M}_{n}\, , \qquad \mbox{as } s\rightarrow 0 \,.
\end{equation}
Recently, Cachazo and Strominger analysed subleading and sub-subleading terms in the soft limit of tree-level graviton amplitudes in four dimensions \cite{Cachazo:2014fwa}, finding that
\begin{equation}\label{eq4:soft-subleading}
 \cM_{n+1}=\left(S^{(0)}+S^{(1)}+S^{(2)}\right)\cM_n+\cO(s^2)\,,
\end{equation}
where 
\begin{align*}
 S^{(0)}=\sum_{a=1}^{n}\frac{(\epsilon\cdot k_a)^2}{s\cdot k_{a}}, \qquad S^{(1)}=\frac{\epsilon_{\mu\nu}k_{a}^{\mu}s_{\lambda}J_{a}^{\lambda\nu}}{s\cdot k_{a}}, \qquad
 S^{(2)}=\frac{\epsilon_{\mu\nu}(s_{\lambda}J_{a}^{\lambda\mu})(s_{\rho}J_{a}^{\rho\nu})}{s\cdot k_{a}}\,.
\end{align*}
A similar subleading factor was found by Casali for tree-level Yang-Mills amplitudes in four dimensions \cite{Casali:2014xpa}:
\begin{equation}\label{eq4:soft-subleading-YM}
 \cA_{n+1}=\left(S^{(0)}+S^{(1)}\right)\cA_n+\cO(s)\,,
\end{equation}
where $S^{(0)}$ denotes the Weinberg soft limit, and $S^{(1)}$ is the subleading contribution,
\begin{align*}
 S^{(0)}=\sum_{a\text{ adj. }s}\frac{\epsilon\cdot k_{a}}{s\cdot k_{a}}, \qquad S^{(1)}=\sum_{a\text{ adj. }s} \frac{\epsilon_{\mu}s_{\nu}J_{a}^{\mu\nu}}{s\cdot k_{a}}\,.
\end{align*}
Subleading soft limits of gauge and gravity amplitudes were previously studied in \cite{Low:1958sn,Burnett:1967km,Larkoski:2014hta} and \cite{Jackiw:1968zza,Gross:1968in,White:2011yy}, respectively. Schwab and Volovich subsequently proved these subleading soft limit formulae for tree-level Yang-Mills and gravity amplitudes in any spacetime dimension using the CHY formulae \cite{Schwab:2014xua}. Loop corrections to the subleading soft limits were subsequently studied using dimensional regularisation in \cite{Bern:2014oka,He:2014bga,Cachazo:2014dia}. 

Strominger and collaborators have shown that these soft limits are a consequence of the asymptotic symmetry group of Minkowski space \cite{Strominger:2013lka,Strominger:2013jfa,He:2014laa}, known as the BMS group (Bondi-van der Burg-Metzner-Sachs) \cite{Bondi:1962px,Sachs:1962wk}, see also \cite{Hollands:2003ie,Hollands:2003xp,Tanabe:2011es,Kapec:2015vwa} in higher dimensions.

\paragraph{BMS symmetries.}
In an asymptotically flat and simple space-time, the conformal boundary $\scri$ is a null hypersurface and decomposes into two disjoint sets $\scri=\scri^-\cup\scri^+$, corresponding to past and future null infinity. Moreover, they have the topology of a $d$-dimensional light-cone, $\scri^{\pm}\cong\mathbb{R}\times S^{d-2}$. An asymptotic flat metric can be written in a neighbourhood of $\scri$ in Bondi gauge as
\begin{equation}\label{eq4:metric}
 \d s^2 =-\d u^2 -2 \d u\d r + r^2 \gamma_{AB}\d \theta^A\d\theta^B + \cO(r^{-1})\,,
\end{equation}
where $u$ and $\theta$ are null coordinates on $\mathbb{R}$ and $S^{d-2}$ respectively, and $r$ denotes a radial coordinate. The symmetry group of $\scri^{\pm}$ in an asymptotically flat space-time is the BMS group \cite{Bondi:1962px,Sachs:1962wk,Kapec:2015vwa}, generated by diffeomorphisms of null infinity $\scri^{\pm}$ that preserve its weak and strong\footnote{preserving the null angle between two tangent directions at $\scri^{\pm}$.} conformal structure \cite{Penrose:1986ca}. More specifically, it is the group of diffeomorphisms preserving the universal structure, i.e. the boundary conditions defined by asymptotic flatness and the fall-off of the subleading terms in \cref{eq4:metric}, modulo diffeomorphisms acting trivially at $\scri^{\pm}$. The vector fields generating this group only involve the conformal Killing vectors $Y^A(\theta)$ of $S^{d-2}$ and a generating function $T(\theta)$. The BMS group therefore consists of the global conformal transformations of $S^{d-2}$ and the so-called supertranslations ST,
\begin{equation}
 \text{BMS}_{d}=\text{ST}\ltimes\text{SO}(d-1,1)\,.
\end{equation}
In particular, the supertranslations ST are angle-dependent translations along $\scri^{\pm}$ that form an infinite dimensional\footnote{In $d>4$, this symmetry enhancement from the Poincar\'{e} group has been historically ruled out: the fall-off conditions on the subleading terms in \cref{eq4:metric} restrict $T$ to the ordinary translations of the Poincar\'{e} group. However, the validity of the soft theorems in arbitrary dimension strongly suggests an extension. This was proposed in \cite{Kapec:2015vwa}, using less stringent boundary conditions at null infinity. See also \cite{Avery:2015gxa} for a discussion of these aspects.} Abelian group. They are not isometries of space-time, but rather relate different asymptotically flat solutions.\smallskip

Strominger et al. argue that the
soft limit theorem of Weinberg arises from the Ward identity following
from supertranslation invariance \cite{Strominger:2013jfa}, but taking only a diagonal subgroup BMS$^0\subset\text{BMS}^+\times\text{BMS}^-$ of the product of the groups obtained at past null infinity $\scri^-$ with that at $\scri^+$ for Christodoulou-Klainerman spaces  \cite{PhysRevLett.67.1486,Christodoulou}.  In four dimensions, this diagonal subgroup is obtained by requiring the real part of the second derivative of the shear from $\scri^-$ to be equal at space-like infinity to that from $\scri^+$.   
Furthermore, the BMS group can be
extended by so-called superrotations, which correspond to extending
the global conformal symmetry of the 2-sphere in $d=4$ to a local conformal symmetry \cite{Barnich:2009se,Barnich:2011ct,Barnich:2011mi} in 4-dimensions (and more general diffeomorphisms\footnote{The conformal group of the sphere is only enhanced in four dimensions for $S^2$, obstructing a straightforward generalisation to higher dimensions. However, it was shown in \cite{Campiglia:2014yka} that in four dimensions, supertranslations can be enhanced to Diff$(S^2)$, dropping the requirement that the vector fields are conformal Killing vectors. This can be generalised to higher dimensions, giving B$=\text{ST}\ltimes\text{Diff}(S^{d-2})$.} in higher dimensions).

\paragraph{Soft theorems and BMS symmetries.}

In \cite{Strominger:2013lka,Strominger:2013jfa,He:2014laa}, Strominger et al. have shown that the soft theorems  \cref{eq4:softlimits} are equivalent to Ward identities in Ashtekar's Fock space of radiative modes \cite{Ashtekar:1981sf,ashtekar1987asymptotic}. These Ward identities are associated with the diagonal subgroup of $\text{BMS}^0\subset\text{BMS}^+ \times \text{BMS}^-$, when it is proposed as a symmetry of the gravitational S-matrix. 
Moreover, the subleading soft theorem \cref{eq4:soft-subleading} implies a Ward identity associated with the extended BMS superrotation symmetry \cite{Kapec:2014opa}.

In particular, they proposed that this diagonal BMS generator is a spontaneously broken\footnote{Since the supertranslations in BMS$_0$ are not isometries of flat space-time, but rather connect different asymptotically flat solutions, the symmetry has to be broken spontaneously.} symmetry of the S-matrix.
In showing that Weinberg's soft theorem follows from this Ward identity, a key step in the argument of Strominger et al. is that acting with a supertranslation generator on null infinity leads to the insertion of a soft graviton. Specialising to the case of four dimensions, they showed that
\begin{subequations}
\begin{align}
T^{-}\left|in\right\rangle =F^{-}\left|in\right\rangle +\sum_{k\in in}E_{k}f\left(z_{k},\bar{z}_{k}\right)\left|in\right\rangle \,,\\
\left\langle out\right|T^{+}=\left\langle out\right|F^{+}+\sum_{j\in out}E_{j}f\left(z_{j},\bar{z}_{j}\right)\left\langle out\right|\,,
\end{align}
\end{subequations}
where $T^{\pm}$ are supertranslations acting at $\scri^{\pm}$, $F^{\pm}$ are outgoing/incoming soft graviton operators, and
$f(z,\bar{z})$ is an arbitrary function on the conformal 2-sphere of null infinity (note
that soft graviton insertions at $\scri^{\pm}$ are related by crossing
symmetry). This is a specific example of the general phenomenon that spontaneously broken symmetries modify the Ward identities, see \cite{Avery:2015gxa}. The BMS symmetry generator $B=B_{\text{soft}}+B_{\text{hard}}$ decomposes into a term $B_{\text{soft}}$ acting non-trivially only on the vacuum, and a term $B_{\text{hard}}$ annihilating the vacuum. In particular, the action of $B_{\text{soft}}$ on the vacuum is proportional to a Goldstone boson insertion, which corresponds here to a soft graviton. For the case of supertranslations described above, we then have
\begin{align}
 &B_{\text{hard}}=T_{\text{hard}}=\sum_{j\in out}E_{j}f\left(z_{j},\bar{z}_{j}\right)-\sum_{k\in in}E_{k}f\left(z_{k},\bar{z}_{k}\right)\\
 & B_{\text{soft}}=T_{\text{soft}}=F\,,
\end{align}
where $F$ denotes a soft graviton insertion.
The global Ward identity in the presence of a spontaneously broken symmetry then takes the form
\begin{equation}
\left\langle out\right|F^{+}\mathcal{S}-\mathcal{S}F^{-}\left|in\right\rangle =\left(\sum_{k\in in}E_{k}f\left(z_{k},\bar{z}_{k}\right)-\sum_{j\in out}E_{j}f\left(z_{j},\bar{z}_{j}\right)\right)\left\langle out\right|\mathcal{S}\left|in\right\rangle \,.
\end{equation}
This equation  is equivalent to Weinberg's soft theorem \cref{eq4:soft-subleading}. Using similar considerations, Strominger et al. showed that the subleading term in the soft graviton limit implies an analogous Ward identity for superrotations.

\section{Geometry and BMS symmetries}\label{sec:soft_geom}
Ambitwistor strings are a natural candidate theory for explaining the duality between BMS symmetries and soft limits due to both their ability to reproduce the most compact known expressions for amplitudes, the CHY formulae, and the underlying geometric structure of their target space. The aim of this section is to 
relate asymptotic symmetries at $\scri$ to Hamiltonians on ambitwistor space. This is achieved by identifying ambitwistor space with the cotangent bundle of null infinity in such a way that the extended BMS generators and their generalizations, indeed arbitrary symplectic diffeomorphisms of $T^*\scri$, act canonically.
\subsection{Background geometry}
Recall that ambitwistor space $\A$ is the complexification of the phase  space of
complex null geodesics with scale in a space-time, see \cref{sec2:review_ambi}.  As such, $\A$ can
be represented by the directions of the complex null geodesics and their intersection with any
Cauchy surface.  The symplectic potential $\Theta$ and symplectic form
$\rd\Theta$ on $\A$ arise from identifying $\A$ with the cotangent
bundle of the complexification of that Cauchy  hypersurface.  In an
asymptotically simple space-time,  they can therefore be represented
with respect to the complexification of null infinity, which we will
denote  $\scri$, and so $\A=T^*\scri$; and at this point $\scri$ can be
the complexification of either future or past null infinity, $\scri^+$
or $\scri^-$.

Null infinity can be represented as a light cone, although it is normal to invert the parameter up the generators to give a parameter $u$ for which the vertex is at $u=\infty$.  In order to make the symmetries manifest, we use homogeneous coordinates $p_\mu$ with $p^2=0$ for the complexified sphere of generators of $\scri$, and a coordinate $u$ of weight one also, so that $(u, p_\mu)\sim (\alpha u, \alpha p_\mu)$ for $\alpha\neq 0$.  As depicted in Figure \ref{scri}, a null geodesic through a point $x^\mu$ with null cotangent vector $P_\mu$ reaches $\scri$ at the point with coordinates 
\begin{equation}
(u,p_\mu)=w (x^\nu P_\nu,P_\mu)\, ,
\end{equation}
where $w$ encodes the scale of $P$.  The notation is intended to be
suggestive of the fact that $u$ is canonically conjugate to the frequency
here denoted by $w$.

Since  $\A=T^*\scri$, it can be described using homogeneous coordinates $(u,p_\mu, w,q^\mu)$ with $(w, q^\mu)$ of weight zero and $(u,p_\mu)$ weight one to yield the 1-form 
\begin{equation}
\Theta= w \rd u- q^\mu \rd p_\mu\, ,
\end{equation}
and this defines the symplectic potential on $\A$.
As $\Theta $ must be orthogonal to the Euler vector field $\Upsilon= u\p_u +p_\mu\p_{p_\mu}$  we have the constraint
\be{constraint}
wu-q\cdot p=0\, ,
\ee
which is the Hamiltonian for $\Upsilon$.
\begin{figure}[htbp] 
\centering
       \includegraphics[width=3.2in]{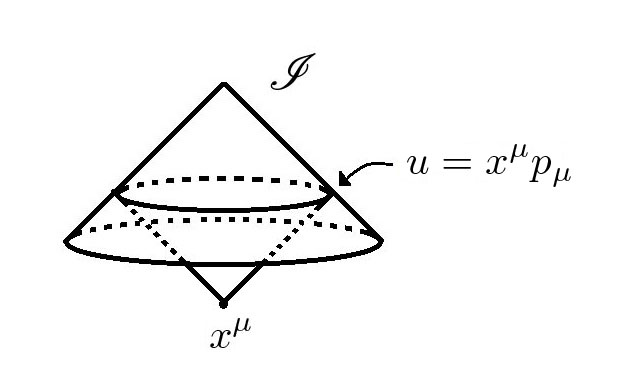}  
    \caption{Diagram of null infinity $\scri$.}
    \label{scri}
    \end{figure} 

To relate this to the original model \cite{Mason:2013sva} reviewed in \cref{chapter2}, recall that the coordinate description of $\A$ was given as the symplectic quotient of the cotangent bundle of space-time, i.e.\  as $(x^\mu, P_\mu)$ with $P^2=0$, quotiented by $P\cdot \p_{x}$.  Including the scale, $\A$ is a symplectic manifold with symplectic form $\rd \Theta $ where $\Theta=P\cdot \rd x$.
The null geodesic through $x^\mu$ with null cotangent vector $P_\mu$ has coordinates at $\scri$ determined by $u=p\cdot x=wP\cdot x$ and the symplectic potential is given by
\begin{equation}
\Theta =P_\mu \rd x^\mu =w\rd u-q^\mu \rd p_\mu \,,
\end{equation}
such that we obtain the relations
\begin{equation}
q^\mu=w x^\mu \mbox{ modulo } p^\mu\,   \quad \mbox{ and } \quad p_\mu=w P_\mu\, .
\end{equation}
On reducing by the constraint \eqref{constraint}, we can identify the scalings of $p$ with those of the momentum by scaling $w$ to 1.  For scaled null geodesics, we can therefore simply incorporate the scale of $p$.

In summary, we can express $\A$ as the symplectic quotient of $(u,p,w,q)$ space by the constraints $p^2=0$ and $uw-q\cdot p=0$. Although $(u,w)$ can be eliminated by using the constraint \eqref{constraint} and making the gauge choice $w=1$, they serve to manifest $\A$ as the cotangent bundle $T^*\scri$ at null infinity.

For the RNS ambitwistor string models, we augment the coordinates above to include either $d$ or $2d$ fermionic coordinates respectively in the heterotic case and the type II case.  These are given by coordinates $\Psi^\mu_r$, $r=1,2$ in the type II case (and $r=1$ in the heterotic case), subject to the constraint(s)  $w\,p\cdot \Psi_r=0$.  The symplectic potential is then augmented to
\begin{equation}\label{susy-pot}
\Theta=w\rd u - q\cdot \rd p + \frac{1}{2}\sum_r\Psi\cdot \rd \Psi\, .
\end{equation}

\subsection{BMS symmetries and their generalizations}
All diffeomorphisms of a manifold have a Hamiltonian lift to the cotangent bundle with Hamiltonian given by the contraction of the generating vector field with the symplectic potential  $\Theta$.

Poincar\'e motions in particular act as diffeomorphisms of $\scri$. 
Translations act by $\delta x^\mu=a^\mu$ give $\delta p=0$, $\delta u=a\cdot p$ and $\delta q^\mu=w a^\mu$ and have Hamiltonian 
\begin{equation}
H_a=wa\cdot p\, .
\end{equation}
 The more interesting supertranslations generalise these to $\delta u= f(p)$ where $f$ is now an arbitrary function of weight 1 in $p$ (i.e. a section of $\cO(1)$) but no longer necessarily linear (and generally with singularities in the complex).
These motions are all symplectic with Hamiltonian
\begin{equation}
H_f=w f(p)\, .
\end{equation}

Lorentz transformations  act by $\delta p_\mu=r_{\mu}{}^\nu p_\nu$, $\delta q^\mu=-r_\nu{}^\mu q^\nu$ ( similarly for  $x^\mu$ and $\Psi_r^\mu$) with $r_{\mu\nu}=r_{[\mu\nu]}$.  
This action has a natural lift to the total space of the line bundle $\cO(1)$ of homogeneity degree 1 functions in which $u$ takes its values.
The Hamiltonian for this action is 
\begin{equation}
H_r=  (q^{[\mu} p^{\nu]} +w\sum_r \Psi_r^\mu\Psi_r^\nu) r_{\mu\nu}\, .
\end{equation}
We can define the angular momentum to be
\begin{equation}\label{angmom}
J^{\mu\nu}=(q^{[\mu} p^{\nu]} +w\sum_r \Psi_r^\mu\Psi_r^\nu)\, .
\end{equation}
It is the sum of an orbital part and an intrinsic spin part and commutes with the constraints $p^2=0$ and $wp\cdot \Psi_r=0$.

Superrotations will be defined here by generalising $r_{\mu\nu}$ to functions that have non-trivial dependence on $p$ (but still of weight zero),
\begin{equation}\label{superr}
H_r=  J^{\mu\nu} r_{\mu\nu}(p)\, . 
\end{equation}
They also preserve the constraints $p^2=p\cdot \Psi=0$ on the constraint surface.
In general dimension, conformal motions are finite dimensional even locally. In four dimensions however, if not constrained to be global on the Riemann sphere, they become infinite dimensional and provide a non-trivial restriction on the general diffeomorphisms we have allowed above.

\section{Ambitwistor strings at null infinity}\label{sec:soft_model}
In the last section, we have constructed ambitwistor space over the Cauchy hypersurface $\scri$ by identifying it with the cotangent bundle of null infinity in such a way that arbitrary symplectic diffeomorphisms of $\scri$ act canonically, and thus lift to Hamiltonians on ambitwistor space. 
In this section, we will describe an ambitwistor string at null infinity whose action is constructed from the contact structure of ambitwistor space. This geometric interpretation of the action can be utilized to define operators from the Hamiltonians inducing the action of the symmetries in the ambitwistor string worldsheet model. While general vertex operators corresponding to graviton insertions implement diffeomorphisms of ambitwistor space, the leading and subleading terms in a soft expansion can be identified with (lifts of) BMS generators on null infinity. The analogous story for Yang-Mills is that vertex operators at null infinity correspond to certain gauge transformations at  $T^*\scri$.  Their soft expansions yield gauge transformations analogous to supertranslations at leading order and superrotations for the subleading terms.

\subsection{The string model}
As in the original ambitwistor string, the action is determined by the symplectic potential $\Theta$.  This gives the worldsheet action on a Riemann surface $\Sigma$ in the new coordinates as
\begin{equation}
S= \frac1{2\pi}\int_\Sigma w\dbar u -q^\mu \dbar p_\mu + \frac{1}{2}\sum_r\Psi_r\cdot \dbar \Psi _r + eT +  \frac{1}{2} \tilde{e} p^2 + \sum_r\chi_r w p\cdot\Psi_r + a(uw-q\cdot p)\, .
\end{equation}
Here, as in \cref{chapter2}, the fields take values in the following bundles:
\begin{subequations}
\begin{align}
 &u,\,p_\mu \in \Omega^0(\Sigma, K_\Sigma)\,, && a\in\Omega^{0,1}(\Sigma)\,,\\
 & w,\,q^\mu\in \Omega^0(\Sigma)\,, && \tilde{e}\in\Omega^{0,1}(\Sigma, T_\Sigma)\,,\\
 &\Psi_r^\mu\in\Pi\Omega^0(\Sigma, K_\Sigma^{1/2})\,, &&\chi_r\in\Pi\Omega^{0,1}(\Sigma, T_\Sigma^{1/2})\,.
\end{align}
\end{subequations}
In particular, $\tilde{e}$, $\chi_r $ and $a$ are gauge fields imposing and gauging the various constraints. The term $eT$ with $T=(w\p u -q^\mu \p p_\mu + \frac{1}{2}\sum_r\Psi_r\cdot \p \Psi _r) $  allows for an arbitrary choice of complex structure parametrised by $e$.  

These gaugings are fixed by setting $e=\tilde e=\chi_r=a=0$ but lead to
respective ghost systems $(b,c)$ and $(\tilde b,\tilde c)$ fermionic, $(\beta_r
,\gamma_r)$ bosonic and $(r,s)$ fermionic.  We are left with the
BRST operator
\begin{equation}\label{Q}
Q_{BRST}=\frac1{2\pi i}\oint c T+ \tilde c \frac {p^2}2 + \sum_r\gamma_r w\Psi_r\cdot p + s
(uw-q\cdot p)\, .
\end{equation}

This is sufficiently close to the original ambitwistor string that we
can simply adapt the Yang-Mills vertex operators in the heterotic model (with $r=1$ only) and the gravitational vertex operators in the type II case with $r=1,2$. With momentum vector $k^\mu$ and polarisation vectors
$\epsilon_{r\mu}$, this gives 
\begin{subequations}\label{Vertex}
\begin{align}
U&= \e^{ik\cdot q/w}\prod_{r=1}^2\delta(\gamma_r) \Psi_r\cdot \epsilon_r\, ,\\
\cV&=\int_\Sigma \bar\delta(k\cdot p) \, w \,\e^{ik\cdot q/w}
\prod_{r=1}^2\epsilon_{r\mu}(p^\mu +i\Psi_r^\mu \Psi_r\cdot k)\, ,
\end{align}
\end{subequations}
for gravity. For Yang-Mills we have
\begin{subequations}\label{VertexYM}
\begin{align}
U^{ym}&= \e^{ik\cdot q/w}\delta(\gamma) \Psi\cdot \epsilon j\cdot t\, , \\
\cV^{ym}&=\int_\Sigma \bar\delta(k\cdot p) w  \e^{ik\cdot q/w}
\epsilon_{\mu}(p^\mu +i\Psi^\mu \Psi\cdot k) j\cdot t\, ,
\end{align}
\end{subequations}
where $j$ is a current algebra on the worldsheet associated to the gauge group and $t$ a Lie algebra element.  As described in \cref{sec2:review_ambistrings}, $\cV$ are the integrated vertex operators, and $U$ are unintegrated
with respect to both the zero modes of $\gamma_r$ and
$\tilde c$.  At genus zero, we need two
insertions of $c\tilde c U$ to fix the two pairs of $\gamma_r$
zero-modes and a third insertion of $c\tilde c $ multiplied by an unintegrated vertex operator to fix the third
of the $c$ and $\tilde c$ zero-modes.

The new feature here is the gauge field $a$ whose ghost $s$ has a zero
mode that must also be fixed.  This can be associated also with a
$1/\mathrm{Vol \,GL}(1)$ factor from the scalings and we will treat this
as the requirement that the scale of $w$ be fixed to be 1.  This can
be done before correlators are taken because the vertex operators do
not depend on $u$.  At this
point it is easily seen that the amplitude computations directly reduce to
the original ones of \cite{Mason:2013sva} reviewed in \cref{sec2:review_ambistrings} to yield the CHY formulae.
A key feature of this derivation is that, in the evaluation of the correlation functions, the exponentials in the vertex operators are taken into the off-shell action leading to the following expression for $p$:
\be{p}
p(\sigma)=\sum_i \frac{k_i}{\sigma-\sigma_i}\, .
\ee

\subsection{Symmetries, vertex operators and diffeomorphisms} \label{symmetry}
Because the action of the worldsheet model is based on the symplectic potential, the singular parts of OPE of operators in the ambitwistor string theory precisely arise from the Poisson structure, so that for example
\begin{equation}
p_\mu(\sigma') q^\nu(\sigma) \sim \frac {\delta^\nu_\mu \rd\sigma}{\sigma-\sigma'} + \ldots \, , \qquad \Psi^{\mu}(z)\Psi_{\nu}(w)=\frac{\delta_{\nu}^{\mu}}{z-w}+...\,,\label{eq:ope}
\end{equation}
where the ellipses denote finite terms.
The Hamiltonians must all have weight one in $p$ (or weight two in $\Psi_r$) since they preserve the symplectic potential and so on the worldsheet they take values in $\Omega^{1,0}_\Sigma$.   We can therefore directly use the Hamiltonian $h$ that generates a symplectic diffeomorphism of  $\A$  to define an operator
\begin{equation}
Q_h=\frac1{2\pi i}\oint h\,,
\end{equation}
that induces the action of the symplectic diffeomorphism in the ambitwistor string model, i.e. for translations we have
\begin{equation}
Q_{p_\mu} q^{\nu}=\frac1{2\pi i}\oint  \frac {\delta^\nu_\mu}{\sigma-\sigma'} + \ldots = \delta^\nu_\mu\, .
\end{equation}
Clearly the same logic will apply to more general BMS transformations and indeed more general diffeomorphisms of $\scri$ as these all have a symplectic lift\footnote{Note that the converse is not true, not every diffeomorphism of ambitwistor space descends to a diffeomorphism of $\scri$.} to $\A=T^*\scri$.

In fact all vertex operators can be related to such motions.  This is most easily stated for gravity where we can rewrite the integrated vertex operator as 
\begin{align}
\cV&=\int_\Sigma \bar\delta(k\cdot p)\, w\,  \e^{ik\cdot q/w}
\prod_{r=1}^2\epsilon_{r\mu}(p^\mu +i\Psi_r^\mu \Psi_r\cdot k)\nonumber \\ 
&= \frac1{2\pi i}\oint \frac{\e^{ik\cdot q/w}}{k\cdot p} \, w\, 
\prod_{r=1}^2\epsilon_{r\mu}(p^\mu +i\Psi_r^\mu \Psi_r\cdot k) \, ,\label{gen-diffeo}
\end{align}
where we have used the relation
$$
\bar\delta(k\cdot p)= \frac1{2\pi i}\dbar\frac1{k\cdot p}\,,
$$ 
to reduce the integral over $\Sigma$ to a contour integral around the pole at $p\cdot k=0$. Thus we see that the vertex operator is the generator of the diffeomorphism of $\A$ with Hamiltonian given by the integrand of \eqref{gen-diffeo}.  This is to be expected in the ambitwistor construction as the data of the space-time metric is encoded in deformations of the complex structure of ambitwistor space, see \cref{sec2:review_ambi} and the original work \cite{LeBrun:1983}.  Such deformations can in turn be encoded in a Dolbeault fashion as a global variation of the $\dbar$-operator as in the first line of \eqref{gen-diffeo} or as a \v Cech deformation of the patching functions for the manifold as determined by the Hamiltonian in the second line.

The story for Yang-Mills is very similar except that now we are talking about variations of the $\dbar$-operator on a bundle in the Dolbeault description, or a non-global gauge transformation in the \v Cech description. In particular, we can rewrite the integrated vertex operator as
\begin{equation}
\mathcal{V}^{ym}=\frac{1}{2\pi i}\oint\frac{\e^{ik\cdot q/w}}{k\cdot p} w \;\epsilon_{\mu}(p^{\mu}+i\Psi^{\mu}\Psi\cdot k)\, j\,,
\label{gen-ym}
\end{equation}
where $j$ is the worldsheet current algebra.

\section{From soft limits to BMS}\label{sec:soft_BMS}

In this section, we will expand the gravitational vertex operator in \eqref{gen-diffeo} in the soft limit, which corresponds to the momentum of the graviton going to zero, and show that the leading and subleading terms in the expansion correspond to generators of supertranslations and superrotations, respectively.

Denoting the soft momentum as $s$, we can expand the vertex operator as follows:
\begin{align} 
 \cV_s&=\frac1{2\pi i}\oint  \frac{w\, \e^{is\cdot q/w}}{s\cdot p} \prod_{r=1}^2 \epsilon_{r\mu}(p^\mu+i\Psi_r^\mu\Psi_r\cdot s) \nonumber \\
&=  \cV_{s}^0+ \cV_{s}^1+\cV_{s}^2+\cV_{s}^3 +\ldots\,.
\end{align}
Simplifying to the situation where $\epsilon_1=\epsilon_2$ (which is sufficient for ordinary gravity), the first two terms in the expansion are given by
\begin{subequations}\label{expgravsoftVO}
\begin{align} 
\cV_{s}^0&= \frac1{2\pi i}\oint w\frac{ (\epsilon\cdot p)^2}{s\cdot p}  \\
\cV_{s}^1&= \frac 1{2\pi i}\oint \frac{\epsilon\cdot p}{s\cdot p} \left ({  i\epsilon\cdot p\,   s\cdot q} + iw\sum_{r=1}^2  \epsilon\cdot\Psi_r s\cdot \Psi_r \right)  \nonumber \\
&=\frac 1{2\pi }\oint \frac{\epsilon\cdot p}{s\cdot p} \epsilon^\mu s^\nu\left ({  p_{[\mu}\,   q_{\nu]}}+ w\sum_{r=1}^2  \Psi_{r\mu}  \Psi_{r\nu} \right) \nonumber \\
&=\frac 1{2\pi }\oint \frac{\epsilon\cdot p}{s\cdot p} \epsilon^\mu s^\nu J_{\mu\nu}\,,
\end{align}
\end{subequations}
where we have used the angular momentum operator defined in \eqref{angmom}
\be{J}
J_{\mu\nu}={ p_{[\mu}q_{\nu]}} +w\sum_{r=1}^2  \Psi_{r\mu}  \Psi_{r\nu}\,,
\ee 
which corresponds to a sum of orbital angular momentum and intrinsic spin. To get to the second line of \eqref{expgravsoftVO}, we note that the extra $s\cdot p$ term in the numerator cancels that in the denominator and so there is no singularity and the contour integral gives zero. 

The integrands in \eqref{expgravsoftVO} correspond precisely to the generators of the Hamiltonian lift of the supertranslations and superrotations of null infinity discussed in section \ref{symmetry}. In particular, $\cV_s^{0}$ generates the supertranslation $\delta u= \frac{(\epsilon\cdot p)^2}{s\cdot p}$, and $\cV_s^{1}$ generates the superrotation $r_{\mu\nu}= i \epsilon_{[\mu}s_{\nu]}\frac{\epsilon\cdot p}{s\cdot p}$ on $\scri$. The OPE of two $J_{\mu\nu}$'s is 
\[
J_{\mu\nu}(z)J_{\rho\lambda}(w)=\frac{2\left(\eta_{\mu\rho}\eta_{\nu\lambda}-\rho\leftrightarrow\lambda\right)}{(z-w)^{2}}+\frac{\left(\eta_{\mu\rho}J_{\nu\lambda}-\rho\leftrightarrow\lambda\right)-\mu\leftrightarrow\nu}{z-w}+...\,,
\]
where ellipsis correspond finite terms. Thus
$J_{[\mu\nu]}$ generate a Kac-Moody algebra:
\begin{equation}
J^{A}(z)J^{B}(w)=\frac{k\delta^{AB}}{2(z-w)^{2}}+\frac{i{f^{AB}}_CJ^{C}(w)}{z-w}+...\,,
\label{km}
\end{equation}
for the Lorentz group, where $k=4$ is the level and ${f^{AB}}_C$ are the structure
constants. From this, we see that the superrotation
generators in $d$ dimensions correspond to a Kac-Moody algebra associated
to $SO(d-1,1)$, which is also the conformal group
associated with the $d-2$ sphere of generators of $\scri$ parametrised
by $p$.  However, the superrotation generator $r_{\mu\nu}J^{\mu\nu}$
with $r_{\mu\nu}= i
\epsilon_{[\mu}s_{\nu]}\frac{\epsilon\cdot p}{s\cdot p}$ has
additional dependence on $p$ (although it will always be linear in $q$ and quadratic in $\Psi$). These all
respect the constraint $p^2=0$ and have the appropriate weight in $p$,
so  the algebra generated by the superrotations is the diffeomorphism
group of the sphere of generators of null infinity.

By a similar calculation to that for $\cV_{s}^1$, one finds that
\begin{align}
\cV_{s}^2 &= \frac1{2\pi i} \oint \frac{(\epsilon\cdot p)^2 (s\cdot q)^2 +2w \epsilon\cdot p s\cdot q \sum_r \epsilon\cdot \Psi s\cdot \Psi + w^2\prod_r \epsilon\cdot \Psi_r s\cdot \Psi_r}{2w\;s\cdot p}\, \nonumber\\ 
&= \frac1{2\pi i} \oint \frac{( \epsilon^\mu s^\nu J_{\mu\nu})^2}{2w\;s\cdot p}.
\end{align}  
$\cV_{s}^2$ therefore gives a `superrotation squared' on ambitwistor space. Note that $\cV_{s}^2$ does not generate a symmetry of null infinity, since the square of a symmetry generator does not in general correspond to another symmetry generator. Hence, beyond subleading order, terms in the expansion of a soft graviton vertex operator generate diffeomorphisms of ambitwistor space $\A=T^*\scri$, but not diffeomorphisms of $\scri$ itself. 

Correlators of $\mathcal{V}_s^{0}$ and $\mathcal{V}_s^{1}$ give rise to the leading and subleading terms \cref{eq4:soft-subleading} in the soft limit of graviton amplitudes:
\begin{subequations}\label{Ward}
\begin{align}
\left\langle \mathcal{V}_{1}...\mathcal{V}_{n} \cV_{s}^0\right\rangle &=\left(\sum_{a=1}^{n}\frac{(\epsilon\cdot k_a)^2}{s\cdot k_{a}}\right)\left\langle \mathcal{V}_{1}...\mathcal{V}_{n}\right\rangle \,, \label{Ward0}\\
\left\langle \mathcal{V}_{1}...\mathcal{V}_{n} \cV_{s}^1\right\rangle &=\sum_{a=1}^{n}\frac{\epsilon_{\mu\nu}k_{a}^{\mu}s_{\lambda}J_{a}^{\lambda\nu}}{s\cdot k_{a}}\left\langle \mathcal{V}_{1}...\mathcal{V}_{n}\right\rangle \,, \label{Ward1}
\end{align}
\end{subequations}
where $\epsilon^{\mu\nu}=\epsilon^{\mu}\epsilon^{\nu}$
and $J_a^{\mu\nu}=k_{a}^{[\mu}\frac{\partial}{\partial k_{a,\nu]}}+\epsilon_{a}^{[\mu} k_{a}^{\nu]}$, and we refer to \cref{gravd} for details of the calculation. These results give an alternative expression of the claims of  \cite{Strominger:2013jfa,He:2014laa,Cachazo:2014fwa} that the soft theorems are equivalent to Ward identities associated with the diagonal subgroup of $\text{BMS}^+ \times \text{BMS}^-$, when it is proposed as a symmetry of the gravitational S-matrix.  Here, however, the Ward identity is expressed in the context of the worldsheet quantum field theory of the ambitwistor string rather than the Fock space of the radiative modes of Yang-Mills or gravity.

Equations \eqref{Ward} contain all the information encoded in the Ward identities for supertranslations and superrotations. The correlation functions with insertions of $\cV_s^0$ and $\cV_s^1$, yielding the leading and subleading contributions for graviton amplitudes, imply the general Ward identities for arbitrary supertranslations and superrotations with Hamiltonians $H_f$ and $H_r$. 

There is a similar story for Yang-Mills theory. If we expand the gluon vertex operator in \eqref{gen-ym} in powers of the soft momentum $s$, we obtain the series
\begin{align} 
 \cV^{ym}_s&=\frac1{2\pi i}\oint  \frac{\e^{is\cdot q/w}}{s\cdot p}  \epsilon_{\mu}(p^\mu+i\Psi^\mu\Psi \cdot s) \, j \nonumber \\
&=  \cV^{ym,0}_{s}+ \cV^{ym,1}_{s}+\cV^{ym,2}_{s}+\cV^{ym,3}_{s} +\ldots\,,
\end{align}
where the terms in the expansion are given by
\begin{subequations}\label{ymsoftser}
\begin{align} 
\cV^{ym,0}_{s}&= \frac1{2\pi i}\oint \frac{ \epsilon\cdot p}{s\cdot p}\, j \,, \\
\cV^{ym,1}_{s}&=\frac 1{2\pi }\oint \frac{\epsilon^\mu s^\nu}{s\cdot p}  J_{\mu\nu} \, j\,.
\end{align}
\end{subequations}
Hence, the leading and subleading terms in the expansion of the gluon vertex operator generate an analogue of supertranslations and superrotations for Yang-Mills theory being respectively generators of gauge transformations that depend only on $p$ or are linear in $J_{\mu\nu}$. Unlike  gravity, $\cV^{ym,2}_{s}$ is no longer the square of $J$. 

Correlators of $\mathcal{V}^{ym,0}_{s}$ and $\mathcal{V}^{ym,1}_{s}$ give rise to the leading and subleading terms \cref{eq4:soft-subleading-YM} in the soft limit of gluon amplitudes (see \cref{ymd}): 
\begin{subequations}
\begin{align}
\left\langle \mathcal{V}_{1}...\mathcal{V}_{n} \mathcal{V}^{ym,0}_{s} \right\rangle &=\left(\frac{\epsilon\cdot k_{1}}{s\cdot k_{1}}-\frac{\epsilon\cdot k_{n}}{s\cdot k_{n}}\right)\left\langle \mathcal{V}_{1}...\mathcal{V}_{n}\right\rangle \,,\label{eq4:Ward-YM} \\
  \left\langle \mathcal{V}_{1}\dots\mathcal{V}_{n}\mathcal{V}^{ym,1}_{s}\right\rangle &= \left(\frac{\epsilon_{\mu}s_{\nu}J_{1}^{\mu\nu}}{s\cdot k_{1}}-\frac{\epsilon_{\mu}s_{\nu}J_{n}^{\mu\nu}}{s\cdot k_{n}}\right)\left\langle \mathcal{V}_{1}...\mathcal{V}_{n}\right\rangle\,. \label{eq4:Ward-YM-sub}
\end{align}
\end{subequations}
Hence, we find that the leading and subleading terms in the soft limit of gluon amplitudes arise from the action of gauge transformations that are gauge analogues of supertranslations and superrotations. 

In summary, the soft limits of tree-level graviton and gluon scattering amplitudes emerge as Ward identities for supertranslations and superrotations on $\scri$. The natural Hamiltonian lift $h$ of diffeomorphisms of $\scri$ to the cotangent bundle $T^*\scri\cong\A$ allows us to define symmetry operators $Q_h$ inducing the action of the diffeomorphism on $\scri$ in the ambitwistor string. This in turn facilitates the identification of the leading and subleading terms in the soft limit of the integrated vertex operators with the  generators of supertranslations and superrotations on $\scri$, whose insertion into correlators gives the well-known soft terms emerging from the corresponding Ward identities.

\section{Discussion}\label{sec4:Discussion}
We have seen in this chapter that ambitwistor space can be identified with the cotangent bundle of any Cauchy surface, and thus in particular in an asymptotically simple space-time with the cotangent bundle of null infinity, $T^*\scri^\pm$. The asymptotic symmetry group at $\scri$, the (extended) BMS group, acts canonically on ambitwistor space. An ambitwistor string constructed in this representation has been used to explain the results by Strominger et al. relating the diagonal subgroup of $\text{BMS}^o\subset \text{BMS}^+ \times \text{BMS}^-$ to Weinberg's soft graviton theorem, and the extended BMS superrotations to subleading terms in the soft graviton expansion. In particular, if one expands the vertex operator of a soft graviton in powers of the soft momentum, the leading and subleading terms correspond to supertranslation and superrotation generators, thus confirming the the conjectures of Cachazo and Strominger.  Furthermore, we find that higher order terms in the expansion correspond to an infinite series of new soft terms which are associated with more general diffeomorphisms of ambitwistor space, although no longer lifted from diffeomorphisms of null infinity. This realises Strominger's derivation of soft theorems as Ward identities associated to BMS symmetries, formulated on the Fock space of radiative modes, in the framework of the ambitwistor worldsheet conformal field theory.

A remarkable feature is that gravitational vertex operators in ambitwistor string theory always arise as generators of rather more general  symplectic diffeomorphisms of $\A$. That such diffeomorphisms should encode the gravitational field follows from the original ambitwistor constructions of LeBrun \cite{LeBrun:1983} in which the gravitational field is encoded in the deformed complex structure of ambitwistor space, see \cref{sec2:review_ambi}.\\  


What is therefore suggested by this picture is that we can give a description of the full nonlinear ambitwistor space  in a globally hyperbolic space-time as follows.   We glue together the flat space one  constructed from the complexification  $T^*\scri^-_\C$ to another constructed from the complexification  $T^*\scri^+_\C$ using the gluing map obtained from the diffeomorphism from the real $T^*\scri^-_\R\rightarrow T^*\scri^+_\R$ determined by the flow along the real null geodesics.  This then specifies enough of the complex structure on ambitwistor space to determine the full gravitational field and its scattering.
The scattering of null geodesics is already a complicated object and to identify those that correspond to solutions to Einstein's equations seems rather daunting in a fully nonlinear regime. 
However, within ambitwistor string theory, this is somehow achieved perturbatively, but nevertheless to all orders, as the scattering of null geodesics determined by each Fourier mode in the vertex operator determines the scattering of the gravitational field by explicit ambitwistor-string calculation.  The correlator achieves the required nonlinear superposition of the effects of each linearised Fourier mode to the required order in perturbation theory.   It would be intriguing to find a nonperturbative formulation of this correspondence.   In the ambitwistor string theory this might be expressed in the form of the structure of a curved beta-gamma system along the lines of \cite{Witten:2005px,Nekrasov:2005wg} with gluing determined by diffeomorphism from $\scri^-$ to $\scri^+$ arising from the scattering of null geodesics but pieced together from manageable ingredients as it is in the perturbative calculations. \\


The analogous story for Yang-Mills is that vertex operators at null infinity correspond to certain gauge transformations at  $T^*\scri$. The scattering here now corresponds to parallel propagation along each real null geodesic, regarding all null geodesics essentially as Wilson lines.  In its soft expansion, we obtain  gauge transformations analogous to supertranslations at leading order and superrotations for the subleading terms. This gives a realization in perturbative string theoretic terms of the ambitwistor constructions of \cite{Witten:1978xx, Isenberg:1978kk,Witten:1985nt} in which Yang-Mills fields are encoded in the complex structure of a holomorphic vector bundle over ambitwistor space with the gauge transformations playing the role of patching functions.

\chapter{Ambitwistor Strings in four dimensions} \label{chapter5} 

The previous chapter gave a powerful demonstration of the potential of different representations for the ambitwistor string: by choosing a representation of ambitwistor space as the cotangent bundle of null infinity, ambitwistor strings tie naturally into the structure of $\scri$, and are thus ideally suited to understand the interplay between asymptotic symmetries and soft theorems. We will explore this theme of choosing a representation adapted to a particular problem further in this chapter. While the focus of \cite{Mason:2013sva} reviewed in \cref{chapter2} was the RNS model in arbitrary dimension\footnote{Recall though that this model is critical in 10 dimensions.} $d$, we will focus here  on the special case of four space-time dimensions.

In four dimensions, some of the most remarkable insights and advances originated from the twistor string theories for  $\mathcal{N}=4$ super Yang-Mills \cite{Witten:2003nn,Berkovits:2004hg,Roiban:2004yf,Berkovits:2004jj,Mason:2007zv} and $\mathcal{N}=8$ supergravity \cite{Skinner:2013xp}. Correlators in both theories led to the first discovery of strikingly compact formulae whose simplicity was obscured by a Feynman diagram approach \cite{ParkeTaylor:1986,Roiban:2004yf,Hodges:2012ym,Cachazo:2012kg,Cachazo:2012pz,Cachazo:2012da,Adamo:2015gia,Adamo:2015ina}. Moreover, they also sparked a `twistor revolution', providing a tantalising paradigm for how twistor theory might eventually make contact with physics. Twistor strings led to a wide variety of results, ranging from efficient techniques for calculating scattering amplitudes (such as the MHV formalism \cite{Cachazo:2004kj,Adamo:2011cb} and the  Britto-Cachazo-Feng-Witten (BCFW) recursion relation \cite{Britto:2004ap,Britto:2005fq,ArkaniHamed:2010,Mason:2009sa,Skinner:2010cz}), to the study of conformal and dual conformal symmetry\footnote{leading to the discovery of the Yangian, an infinite dimensional symmetry algebra of scattering amplitudes in planar $\mathcal{N}=4$ SYM.} \cite{Drummond:2008vq}, the Grassmannian formalism for scattering amplitudes \cite{ArkaniHamed:2009dn,ArkaniHamed:2012nw}, and twistor actions \cite{Mason:2005zm,Mason:2007ct,Boels:2007qn,Adamo:2011cb,Adamo:2013tja,Adamo:2013cra}. However, this paradigm is still a long way from being fully realized, both due to the reliance on maximal supersymmetry and the lack of a clear route to an extension to a critical model allowing for loop calculations.\\

It is thus natural to ask whether we can choose a representation of ambitwistor strings adapted to four dimensions that makes these simplifications and advances manifest. Recall in this context that ambitwistor strings can be defined almost algorithmically by complexifying the action for a spinning massless particle, or geometrically from a chiral pull-back of the contact structure of ambitwistor space. So given either a twistorial representation of the action of a massless particle, or a twistorial representation of ambitwistor space, an ambitwistor string can be constructed.

Ambitwistor space has indeed an alternative spinorial representation in four dimensions, in which the constraints $P^2=0$ are explicitly solved. The resulting ambitwistor string models arise equivalently as the complexifications of the four-dimensional Ferber superparticle \cite{Ferber:1977qx}.  Interestingly, the original twistor-string  was similarly interpreted in \cite{Siegel:2004dj,Berkovits:2005} and the similarity with the ambitwistor approach was also remarked upon in \cite{Mason:2013sva,Bandos:2014lja}. The spinorial representation of the ambitwistor string leads to simple fomulae for any amount of supersymmetry, and the few moduli integrals are fully localised on a refined version of the scattering equations. \\

This chapter is structured as follows. After a brief review of ambitwistor space in \cref{sec5:ambispace4d}, we construct ambitwistor string models for Yang-Mills theory and gravity in four dimensions with any amount of supersymmetry, using the spinorial representation of the target space. These models yield remarkably simple new formulae for tree-level scattering amplitudes which are parity invariant, supported on the scattering equations, and dependent on very few moduli. Moreover, they are supported on a refined version of the scattering equations, adapted to the MHV degree of the amplitude. For maximal supersymmetry and $\cN=0$ Yang-Mills, we prove that these are equivalent to the Roiban-Spradlin-Volovich-Witten (RSVW) and Cachazo-Skinner (CS) formulae obtained from twistor strings. This recasts the ambitwistor string as the natural generalisation of twistor strings, allowing for arbitrary dimensions and any amount of supersymmetry in four dimensional space-time. This point of view is further corroborated by the generalisation of other features of twistor strings, including the double fibration, ambitwistor correspondence and Penrose transform expressing space-time fields in terms of geometrical data relating the auxiliary target space and space-time. 

In \cref{sec5:scri}, we revisit the Ward identities relating extended BMS symmetries and soft theorems in the context of the spinorial representation of the ambitwistor string.

\section{Ambitwistor strings in four dimensions}\label{sec5:ambi-strings}
As motivated above, the spinorial representation can be exploited to construct ambitwistor string theories adapted to the simplifications occurring in four dimensions. To this end, we briefly review ambitwistor space in four dimensional space-time, and then proceed to construct ambitwistor string theories from its contact structure, see also \cref{sec5:review} for a review of twistor space and scattering amplitudes in four dimensions.

\subsection{Ambitwistor space in four dimensions}\label{sec5:ambispace4d}
As discussed in \cref{sec2:review_ambi}, projective ambitwistor space $P\A$ is a supersymmetric extension of the space of
complex null geodesics. In four dimensions, it allows for a representation as a quadric $Z\cdot W=0$ inside the product of twistor space and dual twistor space $P\T\times P\T^*$.  Here we work with $\cN$  supersymmetries so that 
\begin{subequations}
\begin{align}
 &Z=(\lambda_\alpha,\mu^{\dot\alpha},\chi^r)\in \T=\C^{4|\cN}\,,\\
 &W=(\tilde \mu ,\tilde \lambda, \tilde \chi)\in\T^*\,,
\end{align}
\end{subequations}
where $\chi, \tilde \chi$ are fermionic, and we denote the chiral spinor indices by $\alpha=0,1$, $\dot\alpha=\dot 0,\dot 1$ and the R-symmetry indices by $r=1,\ldots \cN$.  Ambitwistor space is then represented as the quadric 
\begin{equation}\label{eq5:ambi-space}
P\A=\bigslant{\big\{(Z,W)\in \T\times \T^*| Z\cdot W=0\big\}}{\{Z\cdot\p_Z-W\cdot \p_W\}} \,,
\end{equation}  
where $Z\cdot W:=\lambda_\alpha\tilde\mu^\alpha+\mu^{\dot\alpha}\tilde\lambda_{\dot\alpha}+\chi^r\tilde\chi_r$, and we also  quotient by the relative scaling $\Upsilon-\tilde\Upsilon$, where $\Upsilon= Z\cdot \p/\p Z$ and $\tilde \Upsilon=W\cdot\p/\p W$ are the twistor and dual twistor Euler vector fields. The symplectic potential in this representation takes the form
\begin{equation}\label{eq5:contact}
\Theta=\frac i2(Z\cdot \rd W-W\cdot \rd Z)\, .
\end{equation}

Ambitwistor space inherits the twistor incidence relations
\begin{subequations}
\begin{align}
&\mu^{\dot \alpha}=i(x^{\alpha\dot\alpha} +i\theta^{r\alpha}\tilde\theta^{\dot\alpha}_r)\lambda_\alpha\, ,&&  \chi^r=\theta^{r\alpha}\lambda_\alpha\, ,\\
&\tilde \mu^{ \alpha}=-i(x^{\alpha\dot\alpha} -i\theta^{r\alpha}\tilde\theta^{\dot\alpha}_r)\tilde \lambda_{\dot \alpha}\, ,&&
\tilde \chi_r=\tilde \theta_r^{\dot \alpha}\tilde \lambda_{\dot \alpha}\, ,
\end{align}
\end{subequations}
which realize a point $(x,\theta,\tilde\theta)$ in (non-chiral) super Minkowski space as a quadric,  $\CP^1\times\CP^1$ parametrized by $(\lambda,\tilde\lambda)$.  It is easily seen that these lie inside the set $Z\cdot W=0$ and indeed, these are the only quadrics in $P\A$ of that degree. To make contact with null geodesics in vectorial representation, the momenta can be defined to be $p_{\alpha\dot\alpha}=\lambda_\alpha\tilde\lambda_{\dot\alpha}$, which now automatically satisfy the constraint $p^2=0$.  This yields the ambitwistor correspondence in its spinorial framing; by construction a point in $P\A$ now corresponds to complex null geodesic in $M$, since we have explicitly solved the constraint $p^2=0$, and the restriction to the quadric $Z\cdot W=0$ identifies points along a null geodesic. Conversely, a point in $M$ corresponds by the discussion above to a quadric $Q_x \cong \CP^1\times\CP^1\subset P\A$.

\subsection{Ambitwistor strings in twistorial representation}
Defining the ambitwistor string action via the chiral pull-back of the contact structure $\Theta$ \cref{eq5:contact}, the ambitwistor string contains the worldsheet spinors  $(Z,W)$, and a GL$(1,\C)$ gauge field $a$ acting as a Lagrange multiplier for the constraint $Z\cdot W=0$, encoding the reduction to ambitwistor space;
\begin{subequations}
\begin{align}
 &Z\in\Omega^0(\Sigma, K_\Sigma^{1/2}\otimes\T)\,,\\
 &W\in\Omega^0(\Sigma,K_\Sigma^{1/2}\otimes\T^*)\,,\\
 &a\in\Omega^{0,1}(\Sigma)\,.
\end{align}
\end{subequations}
In conformal gauge, the action is therefore given by
\begin{equation}\label{eq5:action_YM}
S=\frac1{2\pi}\int_\Sigma  W\cdot \dbar Z-Z\cdot \dbar W  +a Z\cdot W +S_j\, ,
\end{equation}
where $S_j$ is the action for a worldsheet current algebra $j^a\in\Omega^0(\Sigma, K_\Sigma\otimes\g)$ for some Lie algebra $\g$. The action is invariant under a gauge symmetry,
\begin{equation}
 Z^I\rightarrow e^\gamma Z^I\,,\quad W_I\rightarrow e^{-\gamma}W_I\,,\quad a\rightarrow a-2\dbar \gamma\,,
\end{equation}
that quotients the target space by $\Upsilon - \widetilde \Upsilon$. As in general dimensions (see \cref{chapter2}), the quotient to ambitwistor space is thus implemented in the worldsheet CFT via the gauge field $a$ and the associated gauge redundancy. While very reminiscent of the original Berkovits-Witten twistor string \cite{Witten:2003nn,Berkovits:2004hg,Mason:2007zv}, the fields $(Z,W)$ are fixed to be worldsheet spinors in the ambitwistor string, see also \cref{sec5:Discussion} for a discussion of this issue.\\


Gauge fixing worldsheet gravity\footnote{In a general gauge, the $\dbar$ operator in \cref{eq5:action_YM} is replaced by  $\bar\p_{\tilde{e}}=\bar\p+\tilde{e}\p$, parametrising the worldsheet diffeomorphism freedom.} and the gauge redundancy via the BRST procedure introduces the standard Virasoro $(b,c)$ ghost system, as well as a GL$(1)$ ghost system $(u,v)$ with 
\begin{subequations}
\begin{align}
 &c\in\Pi\Omega^0(\Sigma,T_\Sigma)\,, && v\in \Pi\Omega^0(\Sigma)\,,\\
 &b\in\Pi\Omega^0(\Sigma,K_\Sigma^{2})\,, && u\in \Pi\Omega^0(\Sigma,K_\Sigma)\,.
\end{align}
\end{subequations}
The full worldsheet action is then given by
\begin{equation}\label{eq5:action}
 S=\frac1{2\pi}\int_\Sigma  W\cdot \dbar Z-Z\cdot \dbar W  +b\dbar c+u\dbar v +S_j\,,
\end{equation}
and the BRST operator takes the form
\begin{equation}
Q=\oint c T + vZ\cdot W+Q_{\text{gh}}\,. 
\end{equation}
where $T=W\cdot \partial Z-Z\cdot \partial W +T_j$ is the world-sheet stress tensor. In general, this will be anomalous,\footnote{With central charge $\mathfrak{c}=2(4-\cN)_{ZW}-26_{bc}-2_{uv}+\mathfrak{c}_j=20-2\cN+\mathfrak{c}_j$ and a GL$(1)$ anomaly $\mathfrak{a}_{\text{GL}(1)}=4-\cN$.} although there should be choices of matter that give an anomaly free theory, see also \cref{sec5:Discussion}.  Despite the anomaly, and hence $Q^2\neq0$, vertex operators for this theory are still given by BRST-closed combinations of fields. Moreover, despite constituting a serious obstruction to defining loop amplitudes, this anomaly won't affect the tree-level calculations.\footnote{The non-vanishing central charge leads to an insertion of a factor in the path integral. In general, this is a section of a determinant line bundle over the moduli space $\cM_g$ that is not a volume-form and thus cannot be invariantly integrated. At genus zero however, the moduli space reduces to a point, and thus the central charge anomaly only leads to an overall numerical ambiguity.} Amplitudes will be obtained as correlation functions of vertex operators.  In the following we will explicitly give  integrated vertex operators, that simply differ by a factor of $c$ from their unintegrated counterpart. The ghost system will serve to give the GL$(2,\C)$ quotients that are needed in the tree-level formulae, and we will divide by the volume of $\GL(2,\C)$ in the final formula, understood in the usual Faddeev-Popov sense.  

\section{Yang-Mills} 
\subsection{Yang-Mills amplitudes from the ambitwistor string}\label{sec5:YM-amplitudes}
Physical vertex operators in a worldsheet theory arise from infinitesimal deformations of the worldsheet action in the BRST cohomology. Whilst the model is anomalous, vertex operators in the ambitwistor string will still be BRST closed. We can construct such vertex operators that couple to the current algebra and correspond to the Yang-Mills degrees of freedom. From the ambitwistor space perspective, Yang-Mills vertex operators arise from general wave functions  $a\in H^1(\P\A,\mathcal{O})$ multiplied by the currents $j\cdot t_a$ of the current-algebra $j$, to give $\cV_a=\int_\Sigma a \,j\cdot t_a$.  In general, such an $a$ corresponds to an off-shell Maxwell field on space-time, but if it extends off $P\A$ into $P\T\times P\T^*$ to 3rd order or beyond, it must be on-shell (see for example \cite{Baston:1987av}), and only when on-shell is it manifestly $Q$-closed. On shell, such wave functions are a sum of wave functions pulled back from either twistor space or dual twistor space, thus leading to two different types of vertex operators.  For momentum eigenstates, 
\begin{align}
&\cV_a'=\int\frac{\rd s_{a}}{s_{a}}\bar{\delta}^{2}(\lambda_{a}-s_{a}\lambda)\e^{is_{a}\left([\mu\,\tilde{\lambda}_{a}]+\chi^{r}\tilde{\eta}_{ar}\right)}j\cdot t_{a} \,,\label{eq5:VO-YM} \\
&\widetilde  \cV_a= \int\frac{\rd s_{a}}{s_{a}}\bar{\delta}^{2}(\tilde{\lambda}_{a}-s_{a}\tilde{\lambda})\e^{is_{a}\left(\la\tilde{\mu}\,\lambda_{a}\ra+\tilde{\chi}_{r}\eta_{a}^{r}\right)}j\cdot t_{a}\, .
\end{align}
These vertex operator obey $\{Q,\cV_a\}=\{Q,\widetilde  \cV_a\}=0$, and are thus $Q$-closed.  However,
having the supersymmetry in this form will be inconvenient in
what follows. A more suitable representation is obtained by a Fourier transform of the $\tilde
\eta$'s into $\eta$'s in \cref{eq5:VO-YM},
\begin{equation}
\cV_a=\int \frac{\rd s_a}{s_a} \bar\delta^{2|\cN}(\lambda_a-
s_a\lambda|\eta_a-s_a\chi)\e^{is_a[\mu\,
  \tilde\lambda_a]}j\cdot t_a \, ,
\end{equation}
where for a fermionic variable $\chi$, $\delta(\chi)=\chi$.  These vertex operators encode the full $\mathcal{N}=4$ super Yang-Mills degrees of freedom for $\cN=3$ (with $r=1,\ldots,\cN=4$ we would have doubled the spectrum). 

Similar to the original twistor string, this model also contains vertex operators associated to non-minimal conformal gravity degrees of freedom, $f\in \Omega^{0,1}(\Sigma,K_\Sigma\otimes \T)$ and $\tilde{f}\in \Omega^{0,1}(\Sigma,K_\Sigma\otimes \T^*)$, giving rise to vertex operators
\begin{equation}
 \cV^f=\int_\Sigma W_I f^I(Z)\,,\qquad \cV^{\tilde{f}}=\int_\Sigma Z^I \tilde{f}_I(W)\,.
\end{equation}
From the perspective of ambitwistor space, these are the variations in the symplectic potential obtained from wave functions\footnote{where $\cO(p,q)$ is the line bundle of functions of weight $p$ in $Z$ and $q$ in $W$.} in $H^1(P\A, \cO(1, 1))$, corresponding to general off-shell variations in the conformal structure. Built from $f^I\in H^1 (P\T, TP\T)$ and $\tilde{f}_I \in H^1 (P\T^*, TP\T^*)$, they correspond to deformations of twistor space and dual twistor space respectively and hence self-dual and anti-self-dual deformations of the conformal structure. $f \in H^1 (P\T, T P\T)$ can be derived from a worldsheet point of view from $\{Q,\cV^f\}=0$, implying $\partial_If^I=0$. As in the twistor string, the presence of these states implies that loop amplitudes will be corrupted by $\cN=4$ conformal gravity.\\ 

In this section, we will focus on the case where all external particles are gluons. N$^{k-2}$MHV amplitudes will then be obtained as correlation functions of the above vertex operators $\cV$ and $\wt{\cV}$, taking $k$ from dual twistor space and $n-k$ from twistor space:
\begin{equation}\label{eq5:correlator-YM}
\cA=\left\la \widetilde \cV_1 \ldots \widetilde \cV_k \cV_{k+1}\ldots  \cV_n\right\ra\, .
\end{equation}
The current algebra correlator straightforwardly gives the Parke-Taylor denominator (together with some multitrace contributions that we will ignore for the purposes of this chapter).  As in general dimension \cite{Mason:2013sva}, rather than attempt to compute the infinite number of contractions required by the exponentials,  we instead take the exponentials into the action \cref{eq5:action} to provide sources 
\[
\int_\Sigma  \sum_{i=1}^k  i s_i (\la \tilde \mu \lambda_i\ra +\tilde \chi\cdot \eta_i)\bar\delta(\sigma-\sigma_i) + \sum_{p=k+1}^nis_p[\mu\,\tilde\lambda_p]\bar\delta(\sigma-\sigma_p)\,. 
\] 
The new equations of motion for $Z$ and $W$ from this effective action are then
\begin{subequations}
\begin{align}\label{eq5:EoM}
\bar{\partial}_\sigma Z&=\bar{\partial}\left(\lambda,\mu,\chi\right)=\sum_{i=1}^{k}s_{i}\left(\lambda_{i},0,\eta_{i}\right){\bar\delta}\left(\sigma-\sigma_{i}\right)\,,  \\ \label{dbar-l}
\bar{\partial}_\sigma W&=\bar{\partial}\left(\tilde{\mu},\tilde{\lambda},\tilde{\chi}\right)=\sum_{p=k+1}^{n}s_{p}\left(0,\tilde{\lambda}_{p},0\right){\bar\delta}(\sigma-\sigma_{p}) \,,
\end{align}
\end{subequations}
and $\bar\partial \mu = \bar\partial\tilde\mu=0$. Since $(Z,W)$ are worldsheet spinors (i.e. of weight $(-1,0)$), the solutions are uniquely given by
\begin{subequations}\label{eq5:sol}
\begin{align}
Z(\sigma)&=\left(\lambda,\mu,\chi\right)=\sum_{i=1}^{k}\frac{s_{i}\left(\lambda_{i},0,\eta_{i}\right)}{\sigma-\sigma_{i}}\,,
\\
W(\sigma)&=\left(\tilde{\mu},\tilde{\lambda},\tilde{\chi}\right)=\sum_{p=k+1}^{n}\frac{s_{p} \left(0,\tilde{\lambda}_{p},0\right)}{\sigma-\sigma_{p}} \, .
\end{align}
\end{subequations}
By the Riemann-Roch theorem \cref{eq2:Riemann-Roch}, the ghosts $c$ and $u$ develop $n_c=3$ and $n_u=1$ zero modes respectively, thus leading to the quotient by GL$(2,\C)$ as claimed above. With this, the path integrals over the $(Z,W)$ system can be performed trivially, localising them on the solutions \cref{eq5:sol}, and the current correlator contributes the usual Parke-Taylor factor. The correlator \cref{eq5:correlator-YM} then becomes
\begin{equation}\label{final-form_ampl}
\begin{split}
\cA=\int  \frac{1}{\vol \,\GL(2,\C)} 
  &\prod_{a=1}^n\frac{\rd s_a\rd\sigma_a}{s_a (\sigma_a-\sigma_{ a+1})}  
  \prod_{i=1}^k  \bar\delta^2(\tilde \lambda_i -s_i\tilde\lambda) \\
  &\prod_{p=k+1}^n\bar\delta^{2|\cN} 
  (\lambda_p-s_p\lambda(\sigma_p),\eta_p-s_p\chi(\sigma_p)) \, .
  \end{split}
\end{equation}
We can write this in terms of homogeneous coordinates on the Riemann sphere $\sigma_\alpha=\frac1s (1,\sigma)$ using the notation\footnote{With indices raised and lowered by the usual skew symmetric $\epsilon_{\alpha\beta}$.} $(i \,j)=\sigma_{i\alpha}\sigma_j^\alpha$ as follows
\begin{equation}\label{dbar-sol-hgs}
Z(\sigma)= \sum_{i=1}^{k}\frac{\left(\lambda_{i},0,\eta_{i}\right)}{(\sigma\, \sigma_{i})}
\, , \qquad 
W(\sigma)= \sum_{p=k+1}^{n}\frac{(0,\tilde{\lambda}_{p},0) }{(\sigma \,\sigma_{p})}\, ,
\end{equation}
where we have rescaled $W$ and $Z$ by a factor of $1/s$. The final formula for the amplitude then takes the form
\begin{equation}\label{final-form-hgs}
 \begin{split}
  \cA=\int  \frac1{\vol \,\GL(2,\C)} &\prod_{a=1}^n\frac {\rd^2 \sigma_a}{(a\, a+1)}  \,
   \prod_{i=1}^k  \bar\delta^2(\tilde \lambda_i - \tilde\lambda(\sigma_i))  \\ 
   &\prod_{p=k+1}^n\bar\delta^{2|\cN} (\lambda_p-\lambda(\sigma_p),\eta_p-\chi(\sigma_p))\, .
 \end{split}
\end{equation}
For notational simplicity we have taken the colour order to be $(1,\ldots,n )$; of course any other choice will just lead to the obvious re-ordering of the Parke-Taylor denominator.\\

Since we did not specify $\cN$ in this derivation, the expression \cref{final-form-hgs} for Yang-Mills amplitudes holds for any amount of supersymmetry. Moreover, the delta-functions localise the moduli integral completely: there are $2n$ bosonic delta functions and  $2n-4$ moduli integrals, the minus four coming from the $\vol (\GL(2,\C))$ quotient. The remaining delta functions encode momentum conservation, which can be seen explicitly from
\begin{equation}
\sum_{p=k+1}^{n}\lambda_{p}\tilde{\lambda}_{p}=\sum_{p=k+1}^{n}\tilde{\lambda}_{p}\sum_{j=1}^{k}\frac{\lambda_{j}}{(p \, j)}=-\sum_{j=1}^{k}\lambda_{j}\tilde{\lambda}_{j}.
\end{equation}
Here, we used the first (second) set of delta functions in \eqref{final-form-hgs} to get the first (second) equality. Similarly $\sum_{a=1}^{n}\tilde{\lambda}_{a}\eta_{a}=0$.

Defining $P_{\alpha\dot\alpha}(\sigma)=\lambda_\alpha(\sigma)\tilde\lambda_{\dot\alpha}(\sigma)$, the scattering equations \begin{equation}   
 k_a\cdot P(\sigma_a)=\lambda_a^\alpha\tilde\lambda_a^{\dot\alpha} P_{\alpha\dot\alpha}(\sigma_a)=0\,,
\end{equation}
follow on the support of the delta functions. The amplitude therefore localises, as in general dimensions, on the scattering equations. However, the scattering equations are here {\it refined} to give just those appropriate to the N$^k$MHV degree of the amplitude,
\begin{subequations}\label{eq5:ref-SE}
\begin{align}
&[\tilde \lambda_i\, \tilde\lambda(\sigma_i)]=0\, , && i=1\ldots k\,,\label{eq5:SE_1}\\
&\la \lambda_p \, \lambda(\sigma_p)\ra=0 , && p=k+1 \ldots n\,.\label{eq5:SE_2}
\end{align}
\end{subequations}
At N$^k$MHV, these equations have $\abinom{n-3}{k-2}$ solutions, where we denote by $\abinom{p}{q}$ the $(p,q)$ Eulerian\footnote{The Eulerian number $A(p,q)$ is the number of permutations of 1 to $p$ where $q$ elements are larger than their preceding element. They are defined recursively by $A(p,q)=(p-q)A(p-1,q-1)+(q+1)A(p-1,q)$.} number. In particular, these obey $\sum_{k}\abinom{n-3}{k-2}=(n-3)!$, and thus we recover the expected number of solutions to the scattering equation in vectorial representation when summing over all MHV degrees. Yang-Mills amplitudes obtained from the spinorial representation of ambitwistor strings are therefore valid for any amount of supersymmetry, and localise on the {\it refined scattering equations} \cref{eq5:ref-SE}.\\

The formula \eqref{final-form-hgs} can be verified at $\cN=0$ by comparison with  Witten's parity invariant representation of Yang-Mills amplitudes \cite{Witten:2004cp} (or analogous formulae in \cite{Cachazo:2013gna}), as discussed in \cref{N=0}, and extended to arbitrary $\cN$ by superconformal invariance. 
In particular, at $\cN=3$, it is equivalent to the RSVW formula \cref{eq3a:RSVW}, \cite{Roiban:2004yf}, as we will prove in \cref{sec5:RSVW}.


\subsection{Comparison to Witten's parity invariant formula for Yang-Mills} \label{N=0}

We now verify that \eqref{final-form-hgs} agrees with the parity invariant representation of the amplitude at $\cN=0$ due to Witten \cite{Witten:2004cp}. Using $u_\alpha$ as the homogeneous coordinates on the Riemann sphere,  the ambidextrous expression takes the form
\begin{multline}
\mathcal{M}=\int \rd^{2k}P_{\alpha}\, \rd^{2(n-k)}T_{\dot\alpha}\int\prod_{a=1}^{n}\frac{d^{2}u_{a}}{( u_a\, u_{a+1})}\prod_{b\neq a}\frac{1}{(u_b\, u_a) ^{2}}
\\
\prod_{l\in
  L}\delta^{2}\left(\lambda_{l}-\frac{P\left(u_{l}\right)}{\prod_{p\neq
      l} ( u_{p}u_{l}) }\right)
\prod_{r\in R}\delta^{2}\left(\lambda_{r}-P\left(u_{r}\right)\right)
\\
\prod_{r\in R}\delta^{2}\left(\tilde{\lambda}_{r}-\frac{T\left(u_{r}\right)}{\Pi_{t\neq r}( u_{t} u_{r}) }\right)
\prod_{l\in L}\delta^{2}\left(\tilde{\lambda}_{l}-T\left(u_{l}\right)\right)
.\label{eq:witten1}
\end{multline}
For convenience we have chosen $L=\{1,\ldots ,k\}$ and $R=\{k+1,\ldots,n\}$ (although in general there no correlation is implied between the colour ordering in the Park-Taylor denominator and the choice of $L$ and $R$). 

In order to compare with \eqref{final-form-hgs}, we need to integrate
out the moduli for the $P$'s and $T$'s against the delta functions for
$\lambda_l$, $l\in L$, and $\tilde \lambda_r$ for $r\in R$. This
can be done most simply by choosing a basis of the homogeneity
degree $k-1$ functions that form $P$ and degree $n-k-1$
functions that form $T$ respectively
\begin{equation}\label{basis-fns}
C_{l}(u)=\prod_{j\in L, j\neq l} \frac{(u_j\,
u)}{(u_j\, u_l)}\, , \qquad \widetilde C_{r}(u)=\prod_{q\in R, q\neq
r} \frac{(u_q\, u)}{(u_q\, u_r)}\, .
\end{equation}
These functions satisfy $C_l(u_j)=\delta_{lj}$ for $l,j\in L$, and
so integrating the moduli against the first set of delta functions in the second
and third lines of \eqref{eq:witten1} gives
\begin{equation}\label{P}
P(u)=\sum_{l\in L} \lambda_l C_l(u) \prod_{i\neq
  l}(u_i\,u_l)=
\sum_{l\in L} \frac{\lambda_l}{(u \, u_l)} \left(\prod_{r\in R} (u_r\, u_l)\right)  \prod_{j\in L}(u_j \, u)
 \, .
\end{equation}
If we now consider the remaining delta functions, they involve $P(u_r)$ for $r\in R$, that is
$$
P(u_r)=
\sum_{l\in L} \frac{\lambda_l}{(u_r \, u_l)} \left(\prod_{r'\in R} (u_r'\, u_l)\right) \prod_{j\in L}(u_j \, u_r)\, .
$$
Thus,  if for $l\in L$, $r\in R$ we make the substitutions
\begin{equation}\label{u-to-sigma}
\sigma_l=\frac{u_l}{\prod_{r\in R} (u_r\, u_l)}\, , \qquad \sigma_r=\frac{u_r}{\prod_{l\in L}(u_l\,  u_r)}\, ,
\end{equation}
we find
$$
P(u_r)=\sum_{l\in L} \frac{\lambda_r}{(\sigma_l\, \sigma_r)}=\lambda(\sigma_r)\, ,
$$
as required for agreement with the ambitwistor string formula.  Since the substitution \eqref{u-to-sigma} is left-right symmetric, we will similarly have, following the analagous procedure for $T(u)$, 
\begin{equation}\label{T}
T(u)=\sum_{r\in R} \tilde \lambda_r \widetilde C_r(u)\prod_{a\neq r} (u_a\, u_r)=\sum_{r\in R} \frac{\tilde \lambda_r }{(u\, u_r)} \left(\prod_{l\in L} (u_l\, u_r)\right)\prod_{q\in R} (u_q\, u)\, ,
 \end{equation}
and so similarly
$$
T(u_l)=\sum_{r\in R}\frac{\tilde \lambda_r}{(\sigma_l\, \sigma_r)}=\tilde\lambda(\sigma_l)\, .
$$

To see that we now obtain \eqref{final-form-hgs} with $\cN=0$ we must check the determinants following the various changes of variables.  Firstly we observe that 
$$
\prod_{a=1}^{n}\frac{d^{2}u_{a}}{( u_a\, u_{a+1})}=\prod_{a=1}^{n}\frac{d^{2}\sigma_{a}}{( \sigma_a\, \sigma_{a+1})}\, .
$$
Furthermore, when we integrate the moduli against the first set of delta functions in the second and third lines of \eqref{eq:witten1}, this gives the Jacobians $\Pi_{l\in L,p\neq l}\left(u_{p}u_{l}\right)^{2}$ and $\Pi_{r\in R,t\neq r}\left(u_{t}u_{r}\right)^{2}$, respectively. These Jacobians cancel the factor of $\Pi_{b\neq a}\frac{1}{\left(u_{b}u_{a}\right)^{2}}$ in \eqref{eq:witten1}, yielding \eqref{final-form-hgs} for $\cN=0$.

\subsection{Comparison to the RSVW formula}\label{sec5:RSVW}
Let us also briefly discuss the equivalence of \cref{final-form-hgs} for $\mathcal{N}=3$ to the formula \cref{eq3a:RSVW} for $\cN=4$ super Yang-Mills amplitudes \cite{Witten:2003nn,Berkovits:2004hg,Roiban:2004yf} derived from the twistor string. This comparison is facilitated by the similarity of the underlying worldsheet theories, as we will see below.\\

Recall that for maximal supersymmetry, the twistor string \cite{Witten:2003nn,Berkovits:2004hg} leads to the RSVW formula \cite{Roiban:2004yf},
\begin{equation}
 \cA=\int \frac{\prod_{r=0}^d d^{4|4}Z_r} {\vol \,\GL(2,\C)} \prod_{a=1}^n \frac{\rd \sigma_a}{ (a\,a+1)}\prod_{a=1}^n A_a(Z)\,.
\end{equation}
To make contact with the scattering amplitudes derived from the ambitwistor string, we will need to intergrate out the map moduli $Z_r(\sigma)=(\lambda,\mu,\chi)$. On the support of momentum eigenstates,
\begin{equation*}
A_{a}(Z)=\int\frac{\rd t_{a}}{t_{a}}\bd^{2|4}\left(\lambda_{a}-t_{a}\lambda(\sigma_{a})|\eta_a-t_a\chi(\sigma_a)\right)e^{it_{a}[\mu(\sigma_{a}),\tilde{\lambda}_{a}]}\,,
\end{equation*}
with the map moduli
\[
\left(\lambda,\mu,\chi\right)(\sigma)=\sum_{r=0}^{d}\left(\rho_{r},\mu_{r},\chi_{r}\right)C_{r}(\sigma)\,,
\]
the integration over the moduli $\mu_r$ is straightforward, and we obtain the delta-functions
\begin{equation*}
\prod_{r=0}^{d}\delta^{2}\left(\sum_{a=1}^{n}t_{a}\tilde{\lambda}_{a}C_{r}(\sigma_{a})\right)\,.
\end{equation*}
The remaining moduli $(\rho_r,\chi_r)$ can be integrated from the $k$ of the delta-functions associated with the momentum eigenstates by choosing a convenient basis for $C_r(\sigma)$, where $C_{i-1}(\sigma_j)=\delta_{ij}$ for $i,j\in\{1,\dots,k\}$, similar to the strategy in \cref{N=0};
\begin{equation}
 C_{i-1}(\sigma)=\prod_{l\neq i, l=1}^k \frac{(\sigma\,\sigma_l)}{(i\,l)}, \qquad\qquad\text{for } i=1,\dots,k\,.
\end{equation}
This trivially determines $\rho_{r=i-1}=\lambda_i/t_i$ and $\chi_{r=i-1}=\eta_i/t_i$, and integrating out the moduli leaves us with
\begin{equation*}
\begin{split}
 \cA=\int &\frac{\prod_{a=1}^n \rd \sigma_a} {\vol \,\GL(2,\C)} \prod_{a=1}^n \frac{\rd t_a}{t_a} \frac1{(a\,a+1)}\\
 &\qquad \prod_{i=1}^{k}\delta^{2}\left(\tilde{\lambda}_{i}+\frac{1}{t_{i}}\sum_{p=k+1}^{n}t_{p}\tilde{\lambda}_{p}\prod_{l=1,l\neq i}^{k}\frac{(p\, l)}{(i\, l)}\right)\\
 &\qquad\qquad \prod_{p=k+1}^{n}\delta^{2|4}\left(\lambda_{p}-t_{p}\sum_{i=1}^{k}\frac{\lambda_{i}}{t_{i}}\prod_{l\neq i,l=1}^{k}\frac{(p\, l)}{(i\, l)}\right)\,.
 \end{split}
\end{equation*}
As a final step, we can now perform a change of variables;
\begin{align*}
 &s_i=\frac1{\prod_{l=1,l\neq i}^k (i\,l)\,t_i} &&i=1,\dots k\,,\\
 &s_p=\prod_{l=1}^k (p\, l)\, t_p && p=k+1,\dots,n\,,
\end{align*}
which maps the RSVW representation of super Yang-Mills amplitudes to
\begin{equation}
 \begin{split}
\cA=\int&\frac{1}{\vol\,\GL(2,\C)}\prod_{a=1}^{n}\frac{\rd^{2}\sigma_{a}}{(a\, a+1)}\\
&\qquad\prod_{i=1}^{k}\bar{\delta}^{2}(\tilde{\lambda}_{i}-\tilde{\lambda}(\sigma_{i}))\prod_{p=k+1}^{n|4}\bar{\delta}^{2}\left(\lambda_{p}-\lambda(\sigma_{p})|\eta_{p}-\chi(\sigma_{p})\right)\,,
  \end{split}
\end{equation}
where the fields $\lambda(\sigma)$, $\tilde{\lambda}(\sigma)$ and $\tilde\chi(\sigma)$ are given by
\begin{align}
(\lambda(\sigma),\chi(\sigma))=\sum_{i=1}^{k}\frac{(\lambda_{i},\eta_i)}{(\sigma\,\sigma_{i})}\,,
\qquad\qquad
\tilde{\lambda}(\sigma)=\sum_{p=k+1}^{n}\frac{\tilde{\lambda}_{p}}{(\sigma\,\sigma_{p})} \, .
\label{ne4}
\end{align}

To make contact with the $\cN=3$ representation of the ambitwistor string, note that the on-shell superfield of $\mathcal{N}=4$ super Yang-Mills can be encoded in two $\mathcal{N}=3$ on-shell superfields \cite{Elvang:2011fx}. In particular, the $\mathcal{N}=3$ superfields can be obtained by either integrating out $\eta_4$ or setting $\eta_4=0$ in the $\mathcal{N}=4$ superfield, and both descriptions encode the same field content. Hence, $\mathcal{N}=4$ SYM amplitudes are equivalent to $\mathcal{N}=3$ super Yang-Mills amplitudes at tree-level, and we can obtain an $\mathcal{N}=3$ representation of the N$^{\text{k-2}}$MHV amplitude \eqref{ne4} by integrating out $\eta_i$ for $i=1,...,k$ and setting $\eta_j=0$ for $j=k+1,...,n$. We are then left the amplitude in the ambitwistor string representation \eqref{final-form-hgs} for $\mathcal{N}=3$, which concludes the proof for $\cN=4$ Yang-Mills amplitudes.

\section{Einstein gravity}\label{sec5:gravity}
\subsection{Einstein gravity in the ambitwistor string}
While the action of the ambitwistor string was straightforwardly constructed from the pull-back of the contact structure, the worldsheet model for Einstein gravity needs additional structures. We thus construct it as an ambitwistor analogue of the $\cN=8$ twistor string \cite{Skinner:2013xp} for flat space-times. Note however that it can be constructed more generically for non-vanishing cosmological constant $\Lambda$, and has beautiful derivation from a split $(1|2)$ worldsheet supermanifold, see \cite{Skinner:2013xp,Adamo:2013tca}.\\

The basic ingredients of the model are again the worldsheet spinor fields $(Z,W)$, encoding the map into $\T\otimes\T^*$. In addition, we introduce (dual) twistor- and dual twistor-valued fermionic worlsheet spinors $(\rho,\tilde\rho)$, and thus the field content of the theory is given by
\begin{subequations}
\begin{align}
 &Z\in\Omega^0(\Sigma, K_\Sigma^{1/2}\otimes\T)\,, && \rho\in\Pi\Omega^{0}(\Sigma,  K^{1/2}\otimes\T)\,,\\
 &W\in\Omega^0(\Sigma, K_\Sigma^{1/2}\otimes\T^*)\,, && \tilde\rho\in\Pi\Omega^{0}(\Sigma,K^{1/2}\otimes\T^*)\,,\\
 &a\in\Omega^{0,1}(\Sigma)\,.
\end{align}
\end{subequations}
In order to describe gravity, we introduce infinity twistors $I_{IJ}$ and $I^{IJ}$ satisfying $I^{IJ}I_{JK}=\Lambda\delta^{I}_{K}$. These break conformal invariance by determining a metric on space-time, see \cref{sec5:rev-twistor} for details. In general the infinity twistor can encode a cosmological constant $\Lambda$ and a gauging of $R$-symmetry \cite{Wolf:2007tx}, but they are rank $2$ tensors in the simplest zero cosmological constant ungauged case that we will work with here, setting $I_{IJ}Z_1^IZ_2^J=\la \lambda_1\, \lambda_2\ra=:\la Z_1, Z_2\ra$ and $I^{IJ}W_{1I}W_{2J}=[\tilde\lambda_1\, \tilde\lambda_2]=:[W_1, W_2]$.   
In the ambitwistor string for gravity, conformal invariance is broken by the infinity twistors in the following gauged currents:
\begin{equation}\label{eq5:K_a}
K_a=\left(Z\cdot W, \rho\cdot \tilde \rho,  W\cdot \rho, [W, \tilde \rho], Z\cdot \tilde \rho, \la Z, \rho\ra, \la \rho, \rho\ra, [\tilde \rho, \tilde \rho]\right)\,.
\end{equation}
Note in particular that, as in the Yang-Mills model, the constraint $Z\cdot W=0$ and the associated gauge symmetry ensure that the worldsheet is mapped into ambitwistor space. Gauging these currents and gauge fixing all the gauge fields to zero leads to the introduction of the corresponding ghost systems $(\beta_a,\gamma^a)$, 
together with the by now familiar fermionic $(b,c)$ worldsheet reparametrisation ghosts\footnote{The original model \cite{Skinner:2013xp} was introduced without a $(b,c)$ ghost system. In its absence, there exists no natural set of coordinates for the worldsheet moduli, and the automorphism group has to be fixed `by hand'. We follow the route proposed in  \cite{Adamo:2013tca} by simply adding the $(b,c)$ system here.} \cite{Adamo:2013tca}.  The BRST $Q$-operator for this theory is then given by
\begin{equation}
Q=\int c T + \gamma^aK_a -\frac i2 \beta_a \gamma^b\gamma^c  C^a_{bc}\, ,
\end{equation}
where $C^a_{bc}$ are the structure constants of the current algebra $K_a$, and the ghosts take values in 
\begin{subequations}
\begin{align}
 &c\in\Pi\Omega^0(\Sigma,T_\Sigma)\,, && \gamma_r\in \Pi\Omega^0(\Sigma)\,,&& \gamma_s\in \Omega^0(\Sigma)\,,\\
 &b\in\Pi\Omega^0(\Sigma,K_\Sigma^{2})\,, && \beta_r\in \Pi\Omega^0(\Sigma,K_\Sigma)\,, && \beta_s\in \Omega^0(\Sigma,K_\Sigma)\,,
\end{align}
\end{subequations}
for $r=1,2,7,8$, $s=3,4,5,6$. While $Q$ is again anomalous, vertex operators will still be BRST-closed. The gauge fixed ambitwistor string action then becomes
\begin{equation}\label{eq5:S-gravity}
 S=\frac{1}{2\pi}\int_{\Sigma} W\cdot\dbar Z-Z\cdot\dbar W+\tilde{\rho}\dbar\rho -\rho\dbar\tilde{\rho}+b\,\dbar c+\sum_a \beta_a\dbar\gamma^a\,.
\end{equation}
\smallskip

In contrast to the Yang-Mills ambitwistor string, vertex operators in this model must also account for the zero-modes\footnote{In the original set-up on a split $(1|2)$ worldsheet supermanifold, these correspond to fermionic automorphisms, see \cite{Skinner:2013xp}.} of the ghosts $\gamma^a$. The model therefore contains fixed and integrated\footnote{where integration is with respect to the ghost zero modes.} vertex operators, with respect to both $\nu=(\gamma^3,\gamma^4)$ and $\tilde\nu=(\gamma^5,\gamma^6)$,
\begin{subequations}\label{eq5:VO-gravity}
\begin{align}
&\widetilde{V}=\int_{\Sigma}\delta^{2}(\tilde{\nu})\,\tilde{h}(W)\,, &&\widetilde{\cV}=\int_{\Sigma}\left\la Z,\frac{\partial \tilde{h}}{\partial W}\right\ra +\left\la\rho, \frac{\partial}{\partial W}\right\ra \tilde{\rho}\cdot\frac{\partial\tilde{h}}{\partial W}\,,\\
&V=\int_{\Sigma}\delta^{2}(\nu)\,h(Z)\,,&& \cV=\int_{\Sigma}\left[W,\frac{\partial h}{\partial Z}\right]+\left[\tilde{\rho},\frac{\partial}{\partial Z}\right]\rho\cdot\frac{\partial h}{\partial Z}\,.
\end{align}
\end{subequations}
$Q$-invariance $\{Q,V\}=\{Q,\widetilde{V}\}=0$ implies that the vertex operators are built from $\dbar$-closed $(0,1)$-forms $h$ ($\tilde{h}$) of weight two on (dual) twistor space, $h\in H^{1}(P\T,\cO(2))$ and $\tilde{h}\in H^{1}(P\T^{*},\cO(2))$, in accordance with the Penrose transform for gravitons. 

The integrated picture is obtained from fixed vertex operators via the usual superconformal descent procedure (cf. \cite{Friedan:1985ge, Witten:2012bh}).  In the worldsheet correlation function, this can be accomplished by the insertion of picture changing operators (PCOs) for the corresponding ghost systems. Following \cite{Verlinde:1987sd}, a picture changing operator for a generic bosonic $\beta\gamma$-system takes the form
\begin{equation*}
 \Upsilon_{\beta}=\left\{Q,\Theta(\beta)\right\}=\delta(\beta)\,\{Q,\beta\}\,.
\end{equation*}
This implies in the ambitwistor string
\begin{equation}
\Upsilon=\delta^{2}(\mu)\,W\cdot\rho\,[W,\tilde{\rho}]\,,\qquad \widetilde{\Upsilon}=\delta^{2}(\tilde{\mu})\,\la\rho,Z\ra\,\tilde{\rho}\cdot Z\,,
\end{equation}
where, $\mu=(\beta^3,\beta^4)$, and $\tilde\mu=(\beta^5,\beta^6)$. The integrated vertex operators are then obtained from the OPE of the picture changing operator $\Upsilon$ and the fixed vertex operator.

To make contact with the twistor string calculations, a convenient choice for the wave functions will be momentum eigenstates, with $h$ and $\tilde{h}$ are given by
\begin{subequations}
\begin{align} \label{grav_mom_ES}
 h_a&=\int \frac{\rd s_a}{s_a^3}\bar\delta^{2|\cN}(\lambda_a-
 s_a\lambda|\eta_a-s_a\chi)\e^{is_a[\mu\,
  \tilde\lambda_a]}\,,\\
 \tilde{h}_a&=\int\frac{\rd s_{a}}{s_{a}^{3}}\bar{\delta}^{2}(\tilde{\lambda}_{a}-s_{a}\tilde{\lambda})\e^{is_{a}\left(\la\tilde{\mu}\,\lambda_{a}\ra+\tilde{\chi}_{r}\eta_{a}^{r}\right)}\,.
\end{align}
\end{subequations}
To describe full $\cN=8$ supergravity  we can again use the above vertex operators for $\cN=7$. This suggests an interesting connection with Hodges' $\cN=7$ formalism, \cite{Hodges:2011wm}.

To obtain tree-level amplitudes, we are interested in calculating worldsheet correlation functions on $\Sigma\cong\P^1$ for degenerate infinity twistors (i.e. $\Lambda= 0$). According to the Riemann-Roch theorem, the ghosts $\nu=(\gamma^3,\gamma^4)$, $\tilde\nu=(\gamma^5,\gamma^6)$ each develop one zero mode, that is fixed by the insertion of one each of the unintegrated vertex operators $V$ and $\tilde{V}$. Amplitudes are therefore given by the worldsheet correlation function saturating all zero-modes,
\begin{equation}
 \cM=\left\langle \widetilde V_{\tilde h_1}\prod_{i=2}^{k} \widetilde{\mathcal{V}}_{\tilde h_i} \prod_{p=k+1}^{n-1} \mathcal{V}_{{h}_p} V_{{h}_n}\right\rangle\,.
\end{equation}
For vanishing cosmological constant, the number of Wick contractions contributing to this correlator is greatly reduced.  Using standard matrix-tree theorem arguments for the fermionic $(\rho,\tilde \rho)$ system \cite{Feng:2012sy, Skinner:2013xp}, the gravity amplitude is then given by 
\begin{equation} \label{Mgrav_4d}
  \cM=\int \frac{\prod_{a=1}^n \rd^2 \sigma_a}{\vol \,\GL(2,\C)}\: \text{det}'(\cH)\, \prod_{i=1}^k  \bar\delta^2(\tilde \lambda_i - \tilde\lambda(\sigma_i))\prod_{p=k+1}^n\bar\delta^{2|\cN} (\lambda_p-\lambda(\sigma_p)\,\eta_p-\chi(\sigma_p))\, .
\end{equation}
Let us expand on this in a bit more detail. The matrix $\cH$ is built from the correlation function of the $(\rho,\tilde\rho)$ system. Since it is a free CFT, its two-point function is given by
\begin{subequations}
\begin{align}
 &\langle \rho_\alpha(\sigma_i)\tilde\rho^\beta(\sigma_j)\rangle =\delta^\beta_\alpha \frac{(\sigma_i\,\d\sigma_i)^{1/2}(\sigma_j\,\d\sigma_j)^{1/2}}{(i\,j)}\,,\\
 &\langle \rho^{\dot\alpha}(\sigma_i)\tilde\rho_{\dot\beta}(\sigma_j)\rangle =\delta_{\dot\beta}^{\dot\alpha} \frac{(\sigma_i\,\d\sigma_i)^{1/2}(\sigma_j\,\d\sigma_j)^{1/2}}{(i\,j)}\,.
\end{align}
\end{subequations}
This leads to 
\begin{equation}
\cH= \begin{pmatrix}{ \mathbb{H}}& 0\\ 0&\widetilde{\mathbb{H}}\end{pmatrix},
\end{equation}
where, for $i,j\in\{1,...,k\}$ and $p,q\in\{k+1,...,n\}$, with $i\neq j$ and $ p\neq q$,
\begin{subequations}
\begin{align}
 &{\mathbb{H}}_{ij}=\frac{\braket{i\, j}}{(i\, j)},  && {\mathbb{H}}_{ii}=-\sum_{j=1,j \neq i}^{k} {\mathbb{H}}_{ij}\\
 &\widetilde{\mathbb{H}}_{pq}=\frac{[p\, q]}{(p\, q)}, &&  \widetilde{\mathbb{H}}_{pp}=-\sum_{q=k+1,q \neq p}^n \widetilde{\mathbb{H}}_{pq}.
\end{align}
\end{subequations}
The off-diagonal element $\cH_{ij}$ corresponds to the contraction of the $\rho$-term in the $i$th vertex operator with the $\tilde \rho$-term in the $j$th, and the diagonal elements of $\cH$ come from the remaining terms in the integrated vertex operators (see \cite{Skinner:2013xp}). Since the system $(\rho,\tilde\rho)$ does not develop zero modes, all $\rho$ and $\tilde{\rho}$ insertions must be absorbed by Wick contractions. However, due to the degenerate infinity tensor, $\cV$ only contains $\dot\alpha$ components of $\tilde\rho_{\dot\alpha}$, while $V$ only involves $\rho^\alpha$. Therefore, there are no contractions between the two types of vertex operators, and the matrix $\cH$ becomes block-diagonal. Moreover, the fixed vertex operators imply that instead of obtaining the full determinant, the correlator is given by $\det' \cH$, where we omitting a row and column from each of $\widetilde{\mathbb H} $ and $\mathbb{H}$ associated to the unintegrated vertex operators for 1 and $n$. The amplitude is independent of this choice because each matrix has co-rank one, with the kernel spanned by the vector $(1,\ldots,1)$. As in the Yang-Mills case, the path integral over the $(Z, W)$ system enforces the equations of motion \cref{eq5:EoM}, and thus the amplitude is given by \cref{Mgrav_4d}.\\

In analogy to the Yang-Mills case, the graviton scattering amplitude localises on the scattering equations {\it refined by MHV degree}, and is valid for any amount of supersymmetry.\\

The equivalence to the Cachazo-Skinner formula \cref{eq:CS} \cite{Cachazo:2012kg}, proven in \cref{sec5:compgrav}, is seen by following the Yang-Mills strategy to integrate out the moduli and making judicious identifications of reference spinors with the given $\sigma_a$.

\subsection{Comparison to the Cachazo-Skinner formula}\label{sec5:compgrav}
In order to prove the new gravitational formula \eqref{Mgrav_4d}, we will make contact with
the Cachazo-Skinner formula for maximal supergravity \cite{Cachazo:2012pz, Cachazo:2012kg}; note that the similarities in the construction for the theories will facilitate this comparison. Recall the Cachazo-Skinner formula
\begin{equation} \label{M_CS}
 \cM_{CS}=\int \frac{\prod_{r=0}^d \rd^{4|8}Z_r}{\vol \,\GL(2,\C)}\prod_{a=1}^n (\sigma_a\, \rd \sigma_a) |\Phi|'|\widetilde{\Phi}|'\prod_{a=1}^{n}h_a(Z),
\end{equation}
where, for momentum eigenstates, the matrices $\Phi$ and $\wt\Phi$ are defined by
\begin{subequations}
\begin{align}
 \widetilde{\Phi}_{ab}&=\frac{[a\,b]}{(a\,b)}t_at_b && \widetilde{\Phi}_{aa}=-\sum_{b\neq a}\frac{[a\,b]}{(a\,b)}t_at_b\prod_{r=0}^{d}\frac{(b\,w_r)}{(a\,w_r)},\\
 \Phi_{ab}&=\frac{\la a\,b \ra}{(a\, b)}\frac1{t_at_b} && \Phi_{aa}=-\sum_{b\neq a}\frac{\la a\,b \ra}{(a\, b)} \frac1{t_at_b}\prod_{r=0}^{\tilde{d}}\frac{(b\,u_r)}{(a\,u_r)}\prod_{c\neq a,b}\frac{(a\, c)}{(b\, c)}.
\end{align}
\end{subequations}
The reduced determinant $|\Phi|'$ and $|\widetilde{\Phi}|'$ are obtained by removing $d+2$ rows and columns of the matrix $\widetilde{\Phi}$ and $n-d$ rows and columns of $\Phi$ respectively, and including the appropriate Vandermonde factors, see \cite{Cachazo:2012pz, Cachazo:2012kg} and \cref{sec5:review} for more details. \\

The Cachazo-Skinner formula \eqref{M_CS} can be mapped onto the newly derived formula \eqref{Mgrav_4d} for gravity scattering amplitudes by integrating out the moduli $Z_r=(\rho_r, \mu_r,\chi_r)$ and a convenient choice of reference spinors $u_r$, $w_r$. As a general outline for the proof of the equivalence of these formulae, we will proceed as follows:
\begin{enumerate}[i.]
 \item Eliminate convenient rows and columns from the matrices $\Phi$ and $\wt{\Phi}$,
 \item integrate out the moduli $Z_r$, using the delta-functions appearing in the first $k$ momentum eigenstates,
 \item change variables from $t_a$ to $s_i$, $s_p$, where $i=1,\dots,k$, $p=k+1,\dots,n$,
 \item simplify the new matrices $\Phi'$ and $\wt{\Phi'}$ by extracting common factors occuring in all rows and columns, and lastly
 \item make a convenient choice for the reference spinors $u_r$, $w_r$.
\end{enumerate}
In what follows, we will discuss these steps in more detail.\\

As already manifest in the structure of the matrices $\Phi$ and $\wt{\Phi}$, these will correspond to $\mathbb{H}$ and $\wt{\mathbb{H}}$ respectively. To make this correspondence more explicit, we can choose to remove the $k=d+1$ rows and columns associates to the particles $i=1,\dots,k$ from $\wt\Phi$, and equivalently the $n-k=n-d-1$ rows and columns $p=k+1,\dots,n$ from $\Phi$. Note that with this choice, both reduced matrices $\Phi_{(k\times k)}$ and $\wt\Phi_{(n-k\times n-k)}$ have now co-rank 1, and we obtain
\begin{subequations}
\begin{align}\label{detPhi}
  |\wt\Phi|'&=\frac{\text{det}(\wt\Phi_{\text{red}})}{\prod_{\substack{i,j=1\\ i<j}}^k (i\,j)^2\prod_{j=1}^k (j\, q_{\text{rem}})^2}\\
  |\Phi|'&=\frac{\text{det}(\Phi_{\text{red}}) \prod_{l=1,l\neq j_{\text{rem}}}^k (l\,j_{\text{rem}})^2}{\prod_{i<j,i,j=1}^k (i\,j)^2}\,, 
\end{align}
\end{subequations}
where we have denoted the particles removed in addition to those chosen above by $j_{\text{rem}}$ and $q_{\text{rem}}$ for $\Phi$ and $\wt\Phi$ respectively. \\
 
With the choice of momentum eigenstates \eqref{grav_mom_ES}, it is straightforward to explicitly integrate out the map moduli $\mu_r$ and $\chi_r$ from the Cachazo-Skinner formula following the same procedure as in Yang-Mills, leading to the delta-functions
 \begin{equation*}
  \prod_{r=0}^d \bd^{2|8}\Big(\sum_{a=1}^n t_a \tilde{\lambda}_a C_r(\sigma_a)\Big).
 \end{equation*}
 Recall that, in addition, momentum eigenstates supply $2n$ delta-functions which localize $\lambda(\sigma_a)$ on $\lambda_a$;
 \begin{equation*}
  \prod_{a=1}^n\bd^2\left(\lambda_a-t_a\lambda(\sigma_a)\right)=\prod_{a=1}^n\bd^2\Big(\lambda_a-t_a\sum_{r=0}^d \rho_r C_r(\sigma_a)\Big)\,.
 \end{equation*}
$d+1=k$ of these delta-functions determine the map coefficients $\rho_r$, and thus localize the integral over the bosonic moduli. In this context, we can use the same convenient choice of basis for $C_r(\sigma)$ as for Yang-Mills, where $C_{i-1}(\sigma_j)=\delta_{ij}$ for $i,j\in\{1,\dots,k\}$;
 \begin{equation}
  C_{i-1}(\sigma)=\prod_{l=1,l\neq i}^k\frac{(\sigma\,\sigma_l)}{(i\,l)}.
 \end{equation}
To compare the resulting formula directly to the gravitational formula derived from the ambitwistor string model, we will perform a change of variables;
\begin{subequations}
 \begin{align} \label{change-t-s-i}
  &s_i=\frac1{\prod_{j=1,j\neq i}^k (ij)\, t_i} && i\in\{1,\dots,k\}\\
  &s_p=\prod_{j=1}^k (pj)\, t_p && p\in\{k+1,\dots,n\}. \label{change-t-s-p}
 \end{align}
 \end{subequations}
 The delta-functions now take the form
 \begin{align*}
  \prod_{i=1}^k\delta^{2|8} \Big(\tilde\lambda_i -s_i\sum_{p=k+1}^n \frac{s_p \tilde\lambda_p}{(i\,p)}\Big) \,\prod_{p=k+1}^n \delta^2 \Big(\lambda_p-s_p\sum_{i=1}^k \frac{s_i\lambda_i}{(p\,i)}\Big),
 \end{align*}
 where the remaining factors of $(i\,p)$ in the denominator stem from the different range of the indices in the products of \eqref{change-t-s-i} and \eqref{change-t-s-p}. After the change of variables, the entries of the matrices $\wt\Phi$ and $\Phi$ can be simplified immensely by extracting the common factors $(\prod_{l=1}^k (p\,l)\prod_{l=1}^k (q\,l))^{-1}$ occuring in the $p^{\text{th}}$ row and $q^{\text{th}}$ column of $\wt\Phi'$ (and $\prod_{l=1,l\neq i}^k (i\,l)\prod_{l=1,l\neq j}^k (j\,l)$ from the $i^{\text{th}}$ row and $j^{\text{th}}$ column of $\Phi'$). Denoting the matrices obtained this way by $\Phi^{H}$ and $\tilde\Phi^{H}$, we find
 \begin{align*}
  &\wt\Phi^H_{pq}= \frac{[p\,q]}{(p\,q)}s_ps_q =\wt\HH_{pq}\,,\\
  &\wt\Phi^H_{pp}= -\sum_{q=k+1,q\neq p} \frac{[p\,q]}{(p\,q)}s_ps_q \frac{\prod_{l=1}^k (p\,l)}{\prod_{l=1}^k (q\,l)} \prod_{r=0}^d \frac{(q\,w_r)}{(p\,w_r)}\\
  &\quad\qquad - \sum_{j=1}^k \frac{[p\,j]}{(p\,j)}\frac{s_p}{s_j}\frac{\prod_{l=1}^k(p\,l)}{\prod_{l=1,l\neq j}^k(j\,l)} \prod_{r=0}^d \frac{(j\,w_r)}{(p\,w_r)}\,,
 \end{align*}
 and 
 \begin{align*}
  &\Phi^H_{ij}= \frac{\la i\,j\ra}{(i\,j)}s_is_j=\HH_{ij}\,, \\
  &\Phi^H_{ii}= -\sum_{j=1,j\neq i}^k \frac{\la i\,j\ra}{(i\,j)}s_is_j \frac{\prod_{r=k+1}^n (i\,r)}{\prod_{r=k+1}^n (j\,r)}\prod_{r=0}^{\tilde{d}}\frac{(j\,u_r)}{(i\,u_r)}\\
  &\quad\qquad - \sum_{q=k+1}^n  \frac{\la i\,q\ra}{(i\,q)}\frac{s_i}{s_q} \frac{\prod_{r=k+1}^n (i\,r)}{\prod_{r=k+1,r\neq q}^n (q\,r)}\prod_{r=0}^{\tilde{d}}\frac{(q\,u_r)}{(i\,u_r)}\,.
 \end{align*}
In writing down the full amplitude, we also have to include Jacobians from the measure after the change of variables, the choice of basis for $C_r(\sigma)$, the Vandermonde factors associated to the generalized determinants $|\Phi|'$ and $|\wt\Phi|'$ \eqref{detPhi}, and the rescaling the matrices $\Phi'$ and $\wt\Phi'$ to $\Phi^H$ and $\wt\Phi^H$ discussed in the previous step. Combining these contributions leaves us with 
 \begin{equation} \label{CS2}
 \begin{split}
 \cM=\int &\frac{\prod_{a=1}^n (\sigma_a\, \rd \sigma_a)}{\vol \,\GL(2,\C)} \prod_{a=1}^n\frac{\rd s_a}{s_a^3}   \,|\Phi^{H}|_{j_\text{rem}}^{j_\text{rem}}|\tilde\Phi^{H}|_{q_\text{rem}}^{q_\text{rem}}\\
  &\prod_{i=1}^k\delta^{2|8} \big(\tilde\lambda_i -s_i\tilde\lambda(\sigma_i)\big) \,\prod_{p=k+1}^n \delta^2 \big(\lambda_p-s_p\lambda(\sigma_p)\big)\,,
 \end{split}
 \end{equation}
 where we have used the definitions of $\lambda(\sigma)$ and $\tilde\lambda(\sigma)$ familiar from above,
 \begin{equation}
  \lambda(\sigma)=\sum_{i=1}^k \frac{s_i\lambda_i}{(\sigma\,\sigma_i)}\,, \qquad\qquad \tilde\lambda(\sigma)= \sum_{p=k+1}^n \frac{s_p \tilde\lambda_p}{(\sigma\,\sigma_p)}\,.
 \end{equation}
 
As a last step, we will make a convenient choice for the reference spinors: Let $\{\sigma_{w_r}|r=0,\dots d\}=\{\sigma_l|l=1,\dots,k\}$, and furthermore $\{\sigma_{u_r}|r=0,\dots \tilde{d}\}=\{\sigma_q|q=k+1,\dots,n\}$. Note in particular that, with this choice of reference spinors, we do not encounter any singularities, as $w_r=l\neq p$ for all $r$ (and $u_r=q\neq i$). The diagonal entries of the matrices thus become
 \begin{align*}
  \wt\HH_{pp}\equiv\wt\Phi^H_{pp}&=-\sum_{q=k+1,q\neq p} \frac{[p\,q]}{(p\,q)}s_ps_q\,,\\
  \HH_{ii}\equiv\Phi^H_{ii}&=-\sum_{j=1,j\neq i}^k \frac{\la i\,j\ra}{(i\,j)}s_is_j\,.
 \end{align*}
With these choices for the reference spinors, the matrices $\Phi^H$ and $\wt\Phi^H$ coincide with $\HH$ and $\wt\HH$. Again, let us remark briefly on the representation of the $\cN=8$ supersymmetry. Whereas the representation chosen here makes the full $\cN=8$ supersymmetry of the amplitude manifest, the scattering amplitude derived from the ambitwistor string model exhibits only manifest $\cN=7$ supersymmetry. Note however that the two formulations are equivalent; as in Yang-Mills, the $\cN=8$ supermultiplet field content can be encoded in two $\cN=7$ superfields obtained by either integrating out or setting to zero $\eta_a^8$. In particular, to obtain the N$^{\text{k-2}}$MHV amplitude, we will integrate out $\tilde\eta_i^8$ and set $\tilde\eta_p^8=0$ for $i\in\{1,\dots,k\}$, $p\in\{k+1,\dots,n\}$.\\

This concludes the proof of \eqref{Mgrav_4d} for gravitational scattering amplitudes; we have shown explicitly that for $\cN=8$, the amplitude derived as the worldsheet correlation function of an ambitwistor string yields the same result as the Cachazo-Skinner formula. However, \eqref{Mgrav_4d} constitutes a greatly simplified form, where the moduli have been integrated out. Moreover, the formula is manifestly permutation invariant and localizes manifestly on the support of the scattering equations, thereby highlighting the connection between the earlier Cachazo-Skinner formula and the recently proposed CHY formulae.

\paragraph{Grassmannian and link representations.}
Before concluding the discussion of gravity amplitudes from the ambitwistor string, let us briefly comment on the relation of the representation for gravity amplitudes obtained here to the Grassmannian or link representations \cite{Cachazo:2012pz,He:2012er}. These express the graviton N$^{k-2}$MHV amplitude as a multidimensional contour integral over the Grassmannian G$(k,n)$ in terms of link variables \cite{ArkaniHamed:2009si}, reminiscent of the connected prescription for $\cN=4$ super Yang-Mills \cite{ArkaniHamed:2009dg,Bourjaily:2010kw,Bullimore:2009cb,Dolan:2009wf,Dolan:2010xv,Dolan:2011za,Nandan:2009cc,Spradlin:2009qr}. 

In the ambitwistor string, the Grassmannian and link representations are obtained from choosing the external vertex operators to be constructed from elemental states. These are pull-backs of twistor or dual twistor eigenstates,  supported at points in twistor space and are cohomology valued, $f\in H^1(P\T,\cO(2h-2))$ and $\tilde{f}\in H^1(P\T^*,\cO(2h-2))$, and thus correspond via the Penrose transform to fields on space-time.  The prototypes for such wave functions on twistor space are
\begin{subequations}
\begin{align}
&f_{Z_i}(Z)=\int \frac{\rd s}{s^{2h-1}} \delta^{4|\cN}(Z_i-sZ)\, ,&& f_{W_i}(Z)= \int \frac{\rd s}{s^{2h-1}} \exp( s W_i\cdot Z)\, ,\\
&\tilde{f}_{W_i}(W)=\int \frac{\rd s}{s^{2h-1}} \delta^{4|\cN}(W_i-sW)\, , &&\tilde{f}_{Z_i}(W)= \int \frac{\rd s}{s^{2h-1}} \exp( s Z_i\cdot W)\, .
\end{align}
\end{subequations}
In particular, this gives rise directly\footnote{With the same definition of the link variables as in \cite{Cachazo:2012pz}.} to the Grassmannian-like representation (5.9) of \cite{Cachazo:2012pz} when all external wave functions are obtained from twistors, with $k$ vertex operators using $f_{Z_i}(Z)$, and $n-k$ obtained from $f_{W_i}(Z)$. Similarly, we obtain the link representation of \cite{He:2012er} from choosing all wave functions as exponential elemental states. These representations thus emerge naturally from the ambitwistor string.

\section{Four dimensional ambitwistor strings at null infinity} \label{sec5:scri}
An interesting question to ask in the context of \cref{chapter4} is whether there exists as well a spinorial representation of ambitwistor strings in four dimensions tying into the asymptotic structure of space-time. Since the twistorial representation of ambitwistor strings action is constructed as the pull-back of the contact structure to the worldsheet, we can use the same idea of identifying the cotangent bundle at null infinity with projective ambitwistor space, and thus formulate the model with target space $T^*\scri$. As in \cref{chapter5}, the ambitwistor string gives a particularly beautiful and transparent geometric proof of the relation between extended BMS symmetries and soft limits of the amplitudes from a worldsheet CFT perspective. \\

In four dimensions, adapting ambitwistor strings of \cref{sec5:ambi-strings} and \cref{sec5:gravity} to null infinity is perhaps more elegant, requiring no new coordinates.  This model uses again the twistorial representation of ambitwistor space.  As above, we use the coordinates 
$$
Z=(\mu^{\dot\alpha},\lambda_\alpha,\chi^a)\in \T\qquad \mbox{ and }\qquad W=(\tilde\lambda_{\dot\alpha},\tilde\mu^\alpha,\tilde \chi_a)\in\T^*\,,
$$ 
where $\alpha=0,1$ and $\dot\alpha=\dot 0,\dot 1$ and $a=1\ldots \cN$ to represent ambitwistor space spinorially as
\begin{equation}
P\A=\bigslant{\big\{(Z,W)\in \T\times \T^*| Z\cdot W=0\big\}}{\{Z\cdot\p_Z-W\cdot \p_W\}} \,.
\end{equation}    

With homogeneous coordinates $(u,P_{\alpha\dot\alpha})$ on $\scri$ as in \cref{chapter4} (using the spinorial decomposition of the vector index on $P$) we have that the projection from this representation of ambitwistor space to null infinity follows by setting \cite{Eastwood:1982} 
\begin{equation}
u=-i\la \lambda\, \tilde\mu\ra \, , \qquad \tilde u = i[\tilde\lambda\,\mu] \, , \qquad wp_{\alpha\dot\alpha}=P_{\alpha\dot\alpha}=\lambda_\alpha\tilde\lambda_{\dot\alpha}\, ,
\end{equation}
where we have introduced the usual spinor helicity bracket notation to denote spinor contractions, $\la \lambda \,\tilde\mu\ra:=\lambda_\alpha \tilde\mu^\alpha$ and $[\tilde\lambda\,\mu]:=\tilde\lambda_{\dot\alpha}\mu^{\dot\alpha}$.  The spinorial representation here explicitly solves the constraint $P^2=0$, and, working without supersymmetry,  we see that on ambitwistor space $u=\tilde u$ due to $Z\cdot W=0$. 

\subsection{Symmetries and Hamiltonians} \label{4dhamiltonian}
Poincar\'e generators and supertranslations can easily be adapted to act on this representation of ambitwistor space.  The Hamiltonian for the supertranslations $\delta u=f(\lambda,\tilde\lambda)$ with $f$ of weight $(1,1)$ in this model is simply $f$ itself as it induces the transformation
$$
\delta \tilde \mu^\alpha=i\frac {\p f}{\p\lambda_\alpha}\, , \quad \mbox{so} \quad \delta u =\lambda_\alpha \frac {\p f}{\p\lambda_\alpha}=f\,,
  $$
with the latter equality following by homogeneity.  Superrotations can similarly be taken to be those transformations generated by Hamiltonians $H_r$ of weight $(1,1)$ that are linear in $(\mu,\chi)$ and in $(\tilde \mu,\tilde\chi)$ but have more complicated dependence in $(\lambda,\tilde\lambda)$, which will then of necessity include poles. These Poisson commute with $Z\cdot W$ on $Z\cdot W=0$ as they have weight $(1,1)$.  Thus we obtain
\begin{equation}
H_r=[\mu, \tilde{r}]+\la\tilde\mu, r\ra\, , 
\end{equation}
for $ r_\alpha $  and $\tilde{r}_{\dot\alpha}$ respectively weight $(1,0)$ and $(0,1)$ functions of $(\lambda,\tilde\lambda)$.  These are linear functions respectively of $\tilde \lambda$ or $\lambda$ for ordinary rotations or dilations but for superrotations will be allowed to have poles and more general functional dependence on $(\lambda,\tilde\lambda)$. Below we will make the further requirement that 
$$
\frac{\p r_\alpha}{\p\lambda_\alpha}+ \frac{\p \tilde r_{\dot\alpha}}{\p \tilde\lambda_{\dot\alpha}} =\frac{\p^2 H_r}{\p Z^I\p W_I}=0\,,
$$
which will ensure that we are working with $\SL(4)$ rather than $\GL(4)$.  \smallskip

In order to incorporate Einstein gravity in the worldsheet model, we introduce further coordinates $(\rho,\tilde\rho)\in \T\times \T^*$  of opposite statistics to $(Z,W)$ as in \cref{sec5:gravity} and perform the symplectic quotient by the following  further constraints
\begin{equation}\label{constraints}
Z\cdot\tilde \rho=W\cdot \rho=\rho\cdot \tilde\rho= \la Z\, \rho \ra = [W\, \tilde\rho]=0\, ,
\end{equation}
where $\la Z_1 Z_2\ra=\la \lambda_1 \lambda_2\ra$ and $[W_1\, W_2]=[\tilde\lambda _1\, \tilde \lambda_2]$.   In this model, the symplectic potential is 
$$
\Theta = \frac i2 \left (  Z\cdot \rd W -W\cdot \rd Z + \rho\cdot\rd \tilde\rho-\tilde \rho \rd \rho\right)\, .
$$ 

In order to extend the supertranslations and superrotations to this space, we need to extend the above Hamiltonians to commute with these further constraints on the constraint submanifold.  It can be checked that this can be done automatically by taking the Hamiltonians above and acting on them with $1+\rho \cdot \p_Z \tilde \rho\cdot \p_W$.  Thus for supertranslations we obtain the extensions
$$
\left(1+\rho \cdot \p_Z \tilde \rho\cdot \p_W\right)H_f=f + \rho^I \tilde\rho_J\frac{\p^2 f}{\p Z^I\p W_J}\, ,
$$
and for superrotations we get
$$
\left( 1+\rho \cdot \p_Z \tilde \rho\cdot \p_W\right)H_r=[\mu, r]+\la\tilde\mu,\tilde r\ra +  \rho^I \tilde\rho_J\frac{\p^2 ([\mu, r]+\la\tilde\mu,\tilde r\ra)}{\p Z^I\p W_J}\, .
$$

\subsection{The string model}
As before we base the action on the symplectic potential so that the
Poisson brackets will be reflected in the conformal field theory OPEs.  
Since no additional coordinates were required in four dimensions, the action is identical to the original ambitwistor string \cref{eq5:action_YM},
\begin{equation}
S=\frac1{2\pi}\int_\Sigma  W\cdot \dbar Z-Z\cdot \dbar W  +a Z\cdot W +S_j\, ,
\end{equation}
with $Z\in\Omega^0(\Sigma, K_\Sigma^{1/2}\otimes\T)$, $W\in\Omega^0(\Sigma,K_\Sigma^{1/2}\otimes\T^*)$ and the gauge field $a$ enforcing the constraint $Z\cdot W=0$. Gauge fixing the action proceeds as discussed above, and in particular the vertex operators are given by \cref{eq5:VO-YM}. The gravity sector of this model, as discussed in \cref{sec5:YM-amplitudes} is the Berkovits-Witten non-minimal version of conformal supergravity \cite{Berkovits:2004jj}.  

\smallskip

For Einstein gravity we also incorporate the $(\rho,\tilde\rho)$ system described above to give the gauge-fixed action \cref{eq5:S-gravity},
\begin{equation}
 S=\frac{1}{2\pi}\int_{\Sigma} W\cdot\dbar Z-Z\cdot\dbar W+\tilde{\rho}\dbar\rho -\rho\dbar\tilde{\rho}+b\,\dbar c+\sum_a \beta_a\dbar\gamma^a\,.
\end{equation}
Recall in particular that here, we gauged all the currents \cref{eq5:K_a}
$$
K_a=\left(Z\cdot W, \rho\cdot \tilde \rho, Z\cdot \tilde \rho, W\cdot \rho, \la Z \rho\ra,  [W\, \tilde \rho], \la \rho\, \rho\ra, [\tilde \rho\, \tilde \rho]\right)\,,
$$  
and gauge fixed all the gauge fields to zero, leading to the BRST operator  
$$
Q_{BRST}=\int cT+ \gamma^aK_a -\frac i2 \beta_a \gamma^b\gamma^c  C^a_{bc}\, .
$$

The pull-back from twistor space and dual twistor space leads to vertex operators for self-dual and anti self-dual fields \cref{eq5:VO-gravity}. For future convenience, we will re-cast the integrated vertex operators as
\begin{subequations}\label{VOgravity4d}
\begin{align}
\cV_p&=\int_\Sigma \left(1+\rho \cdot \p_Z \tilde \rho\cdot \p_W\right)\frac{\rd t_p}{t_p^3} \bar\delta^2(\lambda_p-t_p\lambda
(\sigma_p)) \, [\tilde\lambda(\sigma_p) \, \tilde \lambda_p]   \, \e^{it_p
  [\mu(\sigma_p) \tilde\lambda_p]} \,,  \\
\wt{\cV_i}&=\int_\Sigma \left(1+\rho \cdot \p_Z \tilde \rho\cdot \p_W\right)\frac{\rd t_i}{t_i^3} \bar\delta^2(\tilde\lambda_i-t_i\tilde\lambda
(\sigma_i)) \, \la\lambda(\sigma_i) \,  \lambda_i\ra   \, \e^{it_i
  \la\tilde\mu(\sigma_i) \lambda_i\ra}   \,,
\end{align}
\end{subequations}
which is easily seen to agree with the original definition in \cref{sec5:gravity}.\\

The amplitude calculations for both Yang-Mills and gravity then reduce trivially to those of the original four-dimensional ambitwistor string \cite{Geyer:2014fka}, thus yielding the expected scattering amplitudes. Following the same strategy as in higher dimensions, the correlation function will be evaluated by incorporating the exponentials of the vertex operators into the off-shell action. For $k$ insertions of $\wt\cV$ and $n-k$ insertions of $\cV$, the equations of motion determine $\lambda(\sigma)$ and $\tilde{\lambda}(\sigma)$ to be
\begin{equation}
 \lambda(\sigma)=\sum_{i=1}^k \frac{t_i\lambda_i}{\sigma-\sigma_i}, \qquad \tilde\lambda(\sigma)=\sum_{p=k+1}^n \frac{t_p\lambda_p}{\sigma-\sigma_p}\,.
\end{equation}

As in the higher dimensional case covered in \cref{chapter4}, the ambitwistor string theory naturally incorporates the geometry encoded in the Poisson structure via the singular part of the OPE, due to the construction based on the symplectic potential. The discussion of section \ref{symmetry} is therefore directly applicable in the four-dimensional case; any Hamiltonian $h$ generating a symplectic diffeomorphism on $\A$ preserves the symplectic potential. Therefore, it has the correct weights in the fields to define a corresponding symmetry operator $Q_h$. In particular, the Hamiltonians for the extended BMS transformations discussed above will lead to operators inducing the action of the diffeomorphism of $\scri$ in the ambitwistor string.

\subsection{Soft limits}
BMS symmetries and soft limits were discussed in more detail in \cref{sec4:review-scri}. At this point, we therefore just review briefly the spinorial representation of soft limits \cref{eq4:soft-subleading} and \cref{eq4:soft-subleading-YM} in four dimensional space-time. 

As discovered in \cite{Cachazo:2014fwa}, the leading and subleading terms in the soft limit of tree-level gravity amplitudes take the form 
\begin{equation}
 \cM_{n+1}=\left(S^{(0)}+S^{(1)}+S^{(2)}\right)\cM_n+\cO(s^2)\,,
\end{equation}
where $S^{(0)}$ denotes the Weinberg soft limit, and $S^{(1)}$, $S^{(2)}$ are the subleading contributions,
\begin{subequations}
\begin{align}
 &S^{(0)}=\sum_{a=1}^n\frac{[as]\la \xi\,a\ra^2}{\la a\,s\ra\la \xi\, s\ra^2}\,, \\ &S^{(1)}=\sum_{a=1}^n\frac{[a\,s]\la\xi\,a\ra}{\la a\,s\ra\la\xi\,s\ra}\tilde\lambda_s\cdot\frac{\p}{\p \tilde\lambda_a}\,, \\
 &S^{(2)}=\frac1{2}\sum_{a=1}^n \frac{[a\,s]}{\la a\,s\ra}\tilde{\lambda}_s^{\dot\alpha}\tilde{\lambda}_s^{\dot\beta}\frac{\p^2}{\p\tilde\lambda_a^{\dot\alpha}\p\tilde\lambda_a^{\dot\beta}}\,.
\end{align}
\end{subequations}
Moreover, the analogous expansion for tree-level Yang-Mills amplitudes is given by \cite{Casali:2014xpa}
\begin{equation}
 \cA_{n+1}=\left(S^{(0)}+S^{(1)}\right)\cA_n+\cO(s)\,,
\end{equation}
where again $S^{(0)}$ denotes the Weinberg soft limit, and $S^{(1)}$ is the subleading contribution,
\begin{subequations}
\begin{align}
 &S^{(0)}=\frac{\langle 1\,n\rangle}{\langle s\,1\rangle\langle s\,n\rangle}, \\
 &S^{(1)}=\left(\frac{1}{\left\langle s1\right\rangle }\tilde{\lambda}_{s}\cdot\frac{\partial}{\partial\tilde{\lambda}_{1}}+\frac{1}{\left\langle ns\right\rangle }\tilde{\lambda}_{s}\cdot\frac{\partial}{\partial\tilde{\lambda}_{n}}\right)\,.
\end{align}
\end{subequations}
Without loss of generality, we have chosen particles $1$ and $n$ to be adjacent to the soft particle $s$.

In the following section, we will see how these soft theorems emerge from the ambitwistor string, and tie into the asymptotic symmetries at null infinity.

\subsubsection*{Yang-Mills in 4d}
Following the same strategy as in higher dimensions, we will again expand the integrated vertex operators \eqref{eq5:VO-YM} in the soft gluon limit to show that the leading and subleading terms correspond to generators of supertranslations and superrotations.\\

The scaling integral for the soft momentum occurring in the Yang-Mills vertex operators can be performed explicitly against one of the delta
functions with a choice of reference spinors $\xi_\alpha$ or
$\tilde \xi_{\dot \alpha}$
 to give $s_s= \la \xi \, \lambda_s\ra/\la \xi \,
\lambda(\sigma_s)\ra$ (or its tilded version respectively).  For $\cV$ this leads to 
\begin{align}
\cV^{ym}_s&=\int_\Sigma  \frac{\la \xi \, \lambda(\sigma_s)\ra}{\la \xi \,
  \lambda_s\ra }\bar\delta(\la \lambda_s\, \lambda (\sigma_s)\ra)\;
\exp \left(i
 \frac{\la \xi \, \lambda_s\ra [\mu(\sigma_s) \tilde\lambda_i  ] s}{\la \xi \,
 \lambda(\sigma_s)\ra }  \right)\;
j\cdot t_s \nonumber \\
&=\oint
\frac{\la \xi \, \lambda(\sigma)\ra}{\la \xi \, \lambda_s\ra 
\la \lambda_s\, \lambda (\sigma_s)\ra }
\exp \left(i 
 \frac{\la \xi \, \lambda_s\ra [\mu(\sigma_s) \tilde\lambda_s]}{\la \xi \,
 \lambda(\sigma_s)\ra }  \right)\; 
j\cdot t_s \nonumber\\
&= \cV^{ym,0}_s+\cV^{ym,1}_s +\cV^{ym,2}_s+\ldots\,,
\end{align}
where, as before, in the second line we have used the definition of $\bd$,
\be{deltabar}
\bar\delta(\la \lambda_s\, \lambda (\sigma_s)\ra )=\dbar \frac1{2\pi i\la
\lambda_s\, \lambda (\sigma_s)\ra }\,,
\ee
to reduce the integral to a
contour integral around $\la \lambda_s\, \lambda (\sigma_s)\ra =0$.  In
the last line we are expanding the exponential in the soft gluon limit
$\lambda_s\tilde\lambda_s\rightarrow 0$.  We obtain
\begin{subequations}
\begin{align}
\cV_s^{ym,0}&=\oint \rd \sigma_s
\frac{\la \xi \, \lambda(\sigma_s)\ra}{\la \xi \, \lambda_s\ra 
\la \lambda_s\, \lambda (\sigma_s)\ra }j\cdot t_s\,, 
 \\
V_s^{ym,1}&=\oint 
\frac{i[\mu(\sigma_s) \tilde\lambda_s]}{
\la \lambda_s\, \lambda (\sigma)\ra } \; j\cdot t_s\,,
 \\
V_s^{ym,2}&=\oint 
\frac{-\la \xi\, \lambda_s\ra [\mu(\sigma_s) \tilde\lambda_s]^2}{ \la \xi
  \, \lambda(\sigma_s)\ra 
\la \lambda_s\, \lambda (\sigma_s)\ra } \; j\cdot t_s\,.
\end{align}
\end{subequations}
These can be thought of as singular gauge transformations, the gauge analogues of the supertranslations and superrotations in the gravitational case below.

As we show in appendix \ref{ym4}, a single insertion of the charges generating those singular gauge transformations directly give the leading and subleading terms of the soft gluon limit:
\begin{subequations}
\begin{align}
\left\langle \widetilde{\mathcal{V}}_{1}...\widetilde{\mathcal{V}}_{k}\mathcal{V}_{k+1}...\mathcal{V}_{n}\cV_s^{ym,0}\right\rangle& =\frac{\left\langle 1n\right\rangle }{\left\langle s1\right\rangle \left\langle sn\right\rangle }\left\langle \widetilde{\mathcal{V}}_{1}...\widetilde{\mathcal{V}}_{k}\mathcal{V}_{k+1}...\mathcal{V}_{n}\right\rangle \,,\\
\left\langle \widetilde{\mathcal{V}}_{1}...\widetilde{\mathcal{V}}_{k}\mathcal{V}_{k+1}...\mathcal{V}_{n}\cV_s^{ym,1}\right\rangle& =\left(\frac{1}{\left\langle s1\right\rangle }\tilde{\lambda}_{s}\cdot\frac{\partial}{\partial\tilde{\lambda}_{1}}+\frac{1}{\left\langle ns\right\rangle }\tilde{\lambda}_{s}\cdot\frac{\partial}{\partial\tilde{\lambda}_{n}}\right)\times \nonumber\\
&\qquad\qquad\times
\left\langle \widetilde{\mathcal{V}}_{1}...\widetilde{\mathcal{V}}_{k}\mathcal{V}_{k+1}...\mathcal{V}_{n}\right\rangle \,.
\end{align}
\end{subequations}

\subsubsection*{Einstein gravity in 4d}
In analogy to the discussion in Yang-Mills, we can identify the leading and subleading terms in the soft expansion of the integrated gravity vertex operators as generators of supertranslations and superrotations on $\scri$, with the corresponding Ward identities yielding the soft graviton contributions found by Cachazo and Strominger \cite{Cachazo:2014fwa}.

Following through the same steps as before, we get
\begin{align}
\cV_s&=\int_\Sigma \left(1+\rho \cdot \p_Z \tilde \rho\cdot \p_W\right)\frac{\rd t}{t^3} \bar\delta^2(\lambda_s-t\lambda
(\sigma_s)) \, [\tilde\lambda(\sigma_s) \, \tilde \lambda_s]   \, \e^{it
  [\mu(\sigma_s) \tilde\lambda_s]}   \nonumber \\
&= \int_\Sigma \left(1+\rho \cdot \p_Z \tilde \rho\cdot \p_W\right) \bar\delta(\la \lambda_s\, \lambda (\sigma_s)\ra) \frac{\la \xi \, \lambda(\sigma)\ra^2 [\tilde\lambda(\sigma) \, \tilde \lambda_s]}{\la \xi \,  \lambda_s\ra^2 }  \, 
\e^{ \left(i 
 \frac{\la \xi \, \lambda_s\ra [\mu(\sigma) \tilde\lambda_s]}{\la \xi \,
 \lambda(\sigma_s)\ra }  \right)}\nonumber \\
&=\oint    \left(1+\rho \cdot \p_Z \tilde \rho\cdot \p_W\right)  \frac{\la \xi \, \lambda(\sigma_s)\ra^2 [\tilde\lambda(\sigma_s) \, \tilde \lambda_s]}{\la \xi \,  \lambda_s\ra^2 \la \lambda_s\, \lambda (\sigma_s)\ra} \,
\e^{ \left(i 
 \frac{\la \xi \, \lambda_s\ra [\mu(\sigma_s) \tilde\lambda_s]}{\la \xi \,
 \lambda(\sigma_s)\ra }  \right)}
 \nonumber\\
&= \cV^0_s+\cV^1_s +\cV^2_s+\ldots\,,
\end{align}
where, as above,  to get to the second line we have performed the $s$-integrals  against one of the delta
functions with a choice of reference spinor $\xi_\alpha$ to find $s= \la \xi\, \lambda_s\ra/ \la \xi \, \lambda(\sigma_s)\ra$.  To get to the third line we  have again used $\bar\delta(\la \lambda_s\, \lambda(\sigma_s)\ra)=\dbar (1/\la \lambda_s\,\lambda(\sigma_s)\ra)$.  In the last line we are simply expanding out the exponential as before to find
\begin{subequations}\label{4dgravV}
\begin{align}
\cV^0_s&=\oint    \left(1+\rho \cdot \p_Z \tilde \rho\cdot \p_W\right)  \frac{\la \xi \, \lambda(\sigma_s)\ra^2 [\tilde\lambda(\sigma_s) \, \tilde \lambda_s]}{\la \xi \,  \lambda_s\ra^2 \la \lambda_s\, \lambda (\sigma_s)\ra} \,, 
\\
\cV^1_s&= \oint  \left(1+\rho \cdot \p_Z \tilde \rho\cdot \p_W\right)  \frac{i\la \xi \, \lambda(\sigma_s)\ra [\tilde\lambda(\sigma_s) \, \tilde \lambda_s]  [\mu(\sigma_s) \tilde\lambda_s]}{\la \xi \,  \lambda_s\ra \la \lambda_s\, \lambda (\sigma_s)\ra} \,, \\
\cV^2_s&= \oint  \left(1+\rho \cdot \p_Z \tilde \rho\cdot \p_W\right)  \frac{ [\tilde\lambda(\sigma_s) \, \tilde \lambda_s]  [\mu(\sigma_s) \tilde\lambda_s]^2}{ \la \lambda_s\, \lambda (\sigma_s)\ra} \, .
\end{align}
\end{subequations}
From the discussion in section \ref{4dhamiltonian}, we can identify $\cV^0_s$ as a supertranslation generator, and of $\cV^1_s$ as a superrotation. $\cV^2_s$ corresponds, as in higher dimensions, to the `square' of a superrotation. As terms in the soft expansion of the vertex operators, all of these contributions generate diffeomorphisms of $\A$, but only $\cV^0_s$ and $\cV^1_s$ can be seen to arise from Hamiltonian lifts of diffeomorphisms of $\scri$.

In the four-dimensional case, the transformations generated by superrotations are actually singular in the worldsheet coordinates. To see this, recall that in four dimensions, a general superrotation corresponds to the Hamiltonian
\[
H_r=[\mu, \tilde{r}]+\la\tilde\mu, r\ra\, , 
\]
for $ r_\alpha $  and $\tilde{r}_{\dot\alpha}$ respectively weight $(1,0)$ and $(0,1)$ functions of $(\lambda,\tilde\lambda)$. 
Comparing to the generators in \eqref{4dgravV}, we see that
\[
[\mu, \tilde{r}]=\frac{\left\langle \xi\lambda\right\rangle }{\left\langle \xi s\right\rangle \left\langle s\lambda\right\rangle }\left[\tilde{\lambda}s\right]\left[\mu s\right],\,\,\, \left\langle \tilde{\mu},r\right\rangle =\frac{\left[\xi\tilde{\lambda}\right]}{\left[\xi s\right]\left[s\tilde{\lambda}\right]}\left\langle \lambda s\right\rangle \left\langle \tilde{\mu} s\right\rangle \,.
\] 
Hence, $ r_\alpha $  and $\tilde{r}_{\dot\alpha}$ have poles, which arise from the poles of $(\lambda,\tilde\lambda)$, and this will continue to be true for general functions which are weight $(1,0)$ and $(0,1)$ functions of $(\lambda,\tilde\lambda)$, respectively. Furthermore, the superrotation generators have the following OPE:
\[
[\mu, \tilde{r}](z)\left\langle \tilde{\mu},r\right\rangle(w)=\frac{1}{\left(z-w\right)^{2}}+\frac{\left\langle \tilde{\mu}(w)s\right\rangle \left\langle \xi\lambda(w)\right\rangle /\left\langle \xi s\right\rangle +\left[\mu(w)s\right]\left[\xi\tilde{\lambda}(w)\right]/\left[\xi s\right]}{z-w}+...\label{eq:4dope}\,.
\]
As in general dimensions, this has the form of a Kac-Moody algebra, where the analogue of the structure constants in \eqref{km} are now functions of the worldsheet coordinates.

With the vertex operators defined above for momentum eigenstates, an insertion of a soft graviton leads to the following Ward identities or the generators of supertranslations and superrotations respectively,
\begin{subequations}
\begin{align}
 &\left\langle\wt{\mathcal{V}}_{1}...\wt{\mathcal{V}}_{k}\mathcal{V}_{k+1}...\mathcal{V}_{n}\;\cV^{0}_s\right\rangle=\sum_{a=1}^n\frac{[as]\la \xi\,a\ra^2}{\la a\,s\ra\la \xi\, s\ra^2}\left\langle\wt{\mathcal{V}}_{1}...\wt{\mathcal{V}}_{k}\mathcal{V}_{k+1}...\mathcal{V}_{n}\right\rangle\,,\\
 &\left\langle \wt{\mathcal{V}}_{1}...\wt{\mathcal{V}}_{k}\mathcal{V}_{k+1}...\mathcal{V}_{n}\;\cV^{1}_s\right\rangle=\sum_{a=1}^n\frac{[a\,s]\la\xi\,a\ra}{\la a\,s\ra\la\xi\,s\ra}\tilde\lambda_s\cdot\frac{\p}{\p \tilde\lambda_a}\left\langle \wt{\mathcal{V}}_{1}...\wt{\mathcal{V}}_{k}\mathcal{V}_{k+1}...\mathcal{V}_{n}\right\rangle\,,
\end{align}
\end{subequations}
and we refer to appendix \ref{grav4} for more details on the calculation. These Ward identities can immediately be seen to be equivalent to the leading and subleading terms in the soft graviton limit. 

In the sub-subleading case, $\cV^2_s$ generates a diffeomorphism of ambitwistor space corresponding to the `square' of a rotation, but does not descend to $\scri$ itself. Inserting it into correlators yields at tree-level the sub-subleading soft graviton contribution found by Cachazo and Strominger,
\begin{equation}
 \left\langle \wt{\mathcal{V}}_{1}...\wt{\mathcal{V}}_{k}\mathcal{V}_{k+1}...\mathcal{V}_{n}\;\cV^{2}_s\right\rangle
 =\frac1{2}\sum_{a=1}^n \frac{[a\,s]}{\la a\,s\ra}\tilde{\lambda}_s^{\dot\alpha}\tilde{\lambda}_s^{\dot\beta}\frac{\p^2}{\p\tilde\lambda_a^{\dot\alpha}\p\tilde\lambda_a^{\dot\beta}}\left\langle \wt{\mathcal{V}}_{1}...\wt{\mathcal{V}}_{k}\mathcal{V}_{k+1}...\mathcal{V}_{n}\right\rangle\,.
\end{equation}

\subsection{Brief comparison to the  model of Adamo et al.\ }

This four dimensional twistorial ambitwistor model is closely connected to the 2d CFT recently proposed by Adamo, Casali, and Skinner \cite{Adamo:2014yya}, and extended to gauge theory in \cite{Adamo:2015fwa}. Indeed,  ambitwistor strings provide a very flexible framework and one can take different coordinate realizations of the space of null geodesics, adding further variables and corresponding constraints  to bring out different structures or features.  This can lead to quite different realizations with different properties (as witnessed by the distinction between the 4d RNS model versus the twistoral one which manifests both the underlying conformal invariance and its breaking).  The general strategy of identifying the worldsheet action as the chiral pull-back of a  symplectic potential guarantees that Hamiltonians will give rise to operators that can be realized in the worldsheet theory.

 Their  model also lives on a supersymmetric extension of the cotangent bundle of the complexification of null infinity.  This model uses a different presentation of the supersymmetry, but the main coordinates can be identified from the symplectic potential $\Theta$ on $T^*\scri$.  Thus we see for example that in their coordinates the symplectic potential is 
$$
\Theta= w \rd u + \chi\rd \xi+ \nu^A\rd \lambda_A+ \tilde \nu^{\dot A} \rd \tilde \lambda_{\dot A} + \bar \psi^A \rd \psi_A + \bar{\tilde \psi}^{\dot A}\rd \tilde\psi_{\dot A}\, , \quad A=(\alpha, a)\, ,
$$
and here $a=1,\ldots ,4$ is an R-symmetry index corresponding to a representation of $\cN=8$ supersymmetry.  We can clearly identify 
$$
Z=(\lambda_A,i\tilde \nu^{\dot A})\, , \quad  W= (-i\nu^{\dot A},\tilde \lambda_{\dot A})\, , \quad \rho=(\psi_A,\bar{\tilde \psi}^{\dot A})\, ,\quad  \tilde \rho=(\bar{ \psi}^{ A},\tilde\psi_{\dot A})\, ,
$$ 
because the bosonic parts of $\lambda$ and $\tilde\lambda$ are geometrically identical to that of the original twistorial ambitwistor model, 
although the representation of supersymmetry is somewhat different to the usual one on twistor space.  There are additional variables $(w,u)$ playing an identical role to the $(w,u)$ in the $d$-dimensional ambitwistor string model at $\scri$. These are again associated with an additional constraint that can be used to eliminate them.  They also introduce a further pair of fields $(\chi,\xi)$.     One can readily identify  $Z\cdot W=\la Z \rho \ra=[W\tilde \rho]=0$ constraints of the 4d ambitwistor string model amongst the gaugings in the ACS model.   

There are nevertheless important distinctions.  Firstly, the vertex operators are quite distinct from ours, and secondly their formulae work with the worldsheet fields taking values in line bundles of more general degree (ours are taken to be spinors on the worldsheet) leading to a larger integral over moduli in the evaluation of scattering amplitudes. The distinction between the vertex operators seems to be quite substantial and would appear to correspond to a realization of linearized fields in $H^2$ rather than in $H^1$ as in the models presented in this chapter. 

Conceptually, the ACS model \cite{Adamo:2014yya} is formulated in terms of radiative data at $\scri$, including the asymptotic shear $\sigma^0$ manifestly in its field content, whereas the model presented here is more targeted towards scattering amplitudes, manifesting the interpretation of soft gravitons giving rise to the Ward identities.

\section{Discussion}\label{sec5:Discussion}
Ambitwistor strings provide a chiral infinite tension limit of conventional strings.  Here we have formulated them in four space-time dimensions in terms of twistors and dual twistors, leading to remarkably simple new formulae for tree amplitudes for (super) Yang-Mills and (super) gravity. These are nontrivially related to previous twistor string formulae, as we have seen in \cref{sec5:RSVW} and \cref{sec5:compgrav}. Our gravitational formula is similar to the link representation of \cite{He:2012er}, and so one can regard ambitwistor strings as providing the theory underlying such representations. Moreover, the twistorial representation of ambitwistor strings gives a particularly elegant framework for implementing the Ward identity relating asymptotic symmetries of an asymptotically flat four-dimensional space-time to low energy theorems. 

\paragraph{Comparison to the RNS model.}
There are many directions for future exploration. One important question regards the representation of loop amplitudes. Although our model is sufficient for computing tree-level amplitudes, in general it is noncritical and anomalous (the gauge anomalies require $\cN=4$ for the first model and $\cN=8$ for the Einstein gravity model, which suggest a doubling of the spectrum in our context). On the other hand, it is likely that a critical, anomaly-free theory can be obtained by coupling to appropriate matter as for example obtained by reduction from an anomaly-free theory in 10 dimensions \cite{Mason:2013sva,Adamo:2013tsa, Berkovits:2013xba}. Note in this context that the bosonic part of the action in spinorial representation emerges naturally from the RNS ambitwistor string via $P_{\alpha\dot\alpha}= \lambda_\alpha\tilde\lambda_{\dot\alpha}$ and the incidence relations,
\begin{equation}
\begin{aligned}
 S\supset \frac{1}{2\pi}\int_\Sigma P\cdot \dbar X &= \frac{1}{2\pi}\int_\Sigma \lambda_\alpha\tilde\lambda_{\dot\alpha} \,\dbar X^{\alpha\dot\alpha}\\
 &=\frac{1}{2\pi}\int_\Sigma \tilde\lambda_{\dot\alpha}\,\dbar (X^{\alpha\dot\alpha}\lambda_\alpha) -\tilde\lambda_{\dot\alpha} X^{\alpha\dot\alpha}\,\dbar\lambda_\alpha\\
 &=\frac{1}{2\pi}\int_\Sigma \tilde\lambda_{\dot\alpha}\,\dbar \mu^{\dot\alpha} -\tilde\mu^{\alpha}\,\dbar\lambda_\alpha\equiv \frac{1}{2\pi}\int_\Sigma W\cdot\dbar Z\,.
 \end{aligned}
\end{equation}
Treating $Z$ and $W$ on an equal footing then gives rise to the bosonic part of the spinorial ambitwistor action. However, completing the rest of the correspondence is still an open problem. Nevertheless, the relation to the RNS ambitwistor string might make it possible to represent loop amplitudes as integrals over higher genus moduli spaces. 

\paragraph{Degree of the line bundle.}
An issue raised by the gauging associated with the gauge field $a$ is the validity of imposing a choice of degree for the line bundles on $\Sigma$  in which the worldsheet spinor fields $(Z,W)$ take their values. 

The obvious alternative is to not gauge the current $Z\cdot W$, and instead restrict to ambitwistor space globally. This of course changes the interpretation of the model, which is now a string theory with target space $P\T\times P\T^*$. The final factor of GL$(1,\C)$ in the amplitude does not arise from a local gauge redundancy in this case, but rather from a global symmetry. While we have chosen to gauge this current in the discussion above to make the derivation of the factor of GL$(2,\C)$ transparent, this is a valid approach to take. However, note that in the gravitational case, more care is needed since the current algebra only closes after imposing that $Z\cdot W=0$ globally.\\

If we choose to gauge the current $Z\cdot W$ on the other hand, we also have to sum over the degree of line bundle $L$ associated to the worldsheet gauge field $a$, with
\begin{subequations}
\begin{align}
 &Z\in\Omega^0(\Sigma,  L\otimes\T)\,,\\
 &W\in\Omega^0(\Sigma,\tilde{L}\otimes\T^*)\,,\\
 &a\in\Omega^{0,1}(\Sigma)\,,
\end{align}
\end{subequations}
and where 
\begin{equation}\label{eq5:deg-L}
 L\otimes\tilde{L}\cong K_\Sigma\,.
\end{equation}
This ties into the relation to RNS string mentioned above: since $P_{\alpha\dot\alpha}\in\Omega^0(\Sigma,K_\Sigma)$ is a section of the worldsheet canonical bundle, there is no obvious split for $\lambda_\alpha$ and $\tilde\lambda_{\dot\alpha}$, and thus we take $\lambda_\alpha\in \Omega^0(\Sigma,L)$ and $\tilde\lambda_{\dot\alpha}\in \Omega^0(\Sigma,\tilde{L})$. $P_{\alpha\dot\alpha}=\lambda_\alpha\tilde\lambda_{\dot\alpha}\in\Omega^0(\Sigma,K_\Sigma)$ then implies \cref{eq5:deg-L}. 

Let us thus consider the case where $L$ is a line bundle of degree $d$. In this case, $Z$ is of degree $d$, and $W$ of degree $-d-2$. This affects the solutions to the equations of motion from the effective action \cref{eq5:EoM} and \cref{dbar-l}. In particular,
\begin{equation}
 \dbar \tilde\lambda_{\dot\alpha}=\sum_{i=1}^k \tilde\lambda_{i,\dot\alpha}\,\bar\delta_{(-d-2,d)}((\sigma\,\sigma_i))\,,
\end{equation}
cannot be solved generically.\footnote{Here, we have introduced a generalised delta-function \cite{Witten:2004cp} of arbitrary homogeneity by
\begin{equation}
 \bar\delta_{(m-1,-m-1)}((\sigma\,\sigma_i))=\left(\frac{\sigma}{\sigma_i}\right)^m\bar\delta_{(-1,-1)}((\sigma\,\sigma_i))\,.
\end{equation}
This definition is well-defined since on the support of the delta-function $\sigma$ is a non-zero multiple of $\sigma_i$.} However, integrating out the zero modes associated to the field $\mu$ leads to additional delta function insertions in the path integral,
\begin{equation}
 \prod_{m=0}^d \bar\delta^2\left(\sum_{i=1}^k \tilde{\lambda}_i\,C_m(\sigma_i)\right)\,,
\end{equation}
where $C_m(\sigma)$ are a basis of $d+1$ linearly independent degree $d$ polynomials. This implies that the obstruction for solving for $\tilde{\lambda}$ vanishes.\footnote{Note that the obstruction is given by the argument of the delta-functions due to 
\begin{equation*}
 0=\int\dbar\left(\tilde\lambda \,C_m(\sigma)\right)=\int\sum_{i=1}^k\tilde{\lambda}_i\,\bar\delta_{(-d-2,d)}((\sigma\,\sigma_i))\, C_m(\sigma)=\sum_{i=1}^k \tilde{\lambda}_iC_m(\sigma_i)\,.
\end{equation*}}
Moreover, these delta-functions restrict the degree to $d+1\leqslant k$, since for $d+1> k$, the amplitude has distributional support.

This strongly restricts the degrees of the line bundle $L$ giving a non-trivial contribution, since the terms with $d+1> k$ vanish for generic momenta. A full investigation of the remaining degrees $d\leqslant k-1$ however is still needed. Since the final expressions obtained from correlators of the ambitwistor string as discussed above have been confirmed against known formulae, this leaves two possibilities: one is that the contributions from all degrees contribute equally, leading to an overall numerical factor. However, this seems unsatisfactory, and one would like to argue that only one degree of the line bundle is contributing. This could be fixed by the choice of external states, with the twistor string emerging from taking all external states as twistor representatives, and the ambitwistor string from an ambidextrous choice. Further research in this area would not only shed some light on these issues, but could also further clarify the relation between the twistor and ambitwistor string at the level of conformal field theories.

\paragraph{Non-zero cosmological constant and extension to other massless theories.}
Another direction is the generalization of our formulae to nonzero
cosmological constant. The choice of infinity twistor in our model already allows for this, 
leading to modifications in the computation of the correlator. This could provide an efficient method for computing tree-level correlation functions in $AdS_{4}$  and $dS_{4}$ that can be compared to the formulae of \cite{Adamo:2012nn,Adamo:2013tja,Adamo:2015ina} and may in turn have applications to the $AdS_{4}/CFT_{3}$  correspondence and cosmology.\\

Moreover, it would be interesting to extend the twistorial representations in four dimensions to a wider family of massless theories similar to \cref{chapter3}. An especially prominent candidate theory is minimal conformal supergravity, since it can be used to derive graviton scattering amplitudes on asymptotically de Sitter spaces \cite{Maldacena:2011mk}. An ambitwistor model for minimal conformal gravity would thus allow us to access space-times with non-zero cosmological constant from a different direction. This also ties in with recent results \cite{Adamo:2015ina} providing new twistorial representations for gravity scattering amplitudes in space-times with non-zero cosmological constant. Moreover, twistorial representations of Einstein-Yang-Mills tree amplitudes have been discovered recently \cite{Adamo:2015gia}. It would be highly interesting to derive these new formulae from ambitwistor string models.

\paragraph{Ambitwistor strings at null infinity.}
In the context of the representation of ambitwistor strings at null infinity, a natural direction to explore further  is the relation to \cite{Adamo:2014yya,Adamo:2015fwa}. Both models manifest different properties, and a more direct mapping between the two would give insights into how these properties arise in the opposite models. Furthermore, recent progress has been made on understanding the algebra of soft limits derived here from braiding of the soft vertex operators \cite{Lipstein:2015rxa}, and in obtaining 1-loop correction to the subleading soft graviton theorem due to infrared divergences. This is of particular interest in the context of recent developments in the study of loop amplitudes, developed in \cref{chapter6}.

Moreover, the spinorial representation of the ambitwistor string in four dimensions provides the easiest setting in which the a full non-linear ambitwistor string could be explored, obtained from gluing the ambitwistor space $T^*\scri^-_\C$ to $T^*\scri^+_\C$ via the flow along real null geodesics, see \cref{sec4:Discussion}.
\chapter{Loop Integrands from the Riemann Sphere}\label{chapter6}

The wide-ranging impact of ambitwistor strings on the study of amplitudes at tree-level discussed in \cref{chapter3,chapter4,chapter5} naturally raises the question whether, and how far, this progress can be extended to loop amplitudes. Adamo, Casali and Skinner (ACS) \cite{Adamo:2013tsa} explored this further, extending the powerful framework of the scattering equations to loop level using the ambitwistor string, and leading to a conjecture for the one-loop integrand of type II supergravity from genus one ambitwistor string correlators. But while the mathematical framework describing ambitwistor strings on higher genus Riemann surfaces \cite{Adamo:2013tsa,Casali:2014hfa,Adamo:2015hoa} exhibits a conceptual simplicity and provides important evidence for the validity of the resulting formulae, it is computationally challenging and obscures the relative simplicity of the expected amplitudes. 

Recall in this context the role the ambitwistorial representation of scattering amplitudes play in relation to traditional approaches: they act as a third alternative to the combinatorial approach using Feynman diagrams, and the geometric problem posed by the integration over the moduli space of marked Riemann surfaces in worldsheet models. In particular, the localisation of the moduli space integrals on the scattering equations reduced the calculation of scattering amplitudes to a purely algebraic problem. This demonstrates very clearly the origin of the complexity of the higher genus amplitudes: on higher genus worldsheets, the algebraic scattering equations involve Jacobi theta functions that obscure the purely rational representation of loop integrands.
Despite the ACS one-loop proposal \cite{Adamo:2013tsa}, the description of loop amplitudes thus effectively remained a burning open question.\\

In this chapter, we will take a different approach, inspired by the localisation of the amplitudes on the scattering equations, as well as the Feynman tree theorem \cite{Feynman:1963ax} at one loop: starting from the ambitwistor correlators on higher genus Riemann surfaces, the loop integrands can be reduced to expressions on nodal Riemann spheres, that exhibit a similar complexity to tree-level amplitudes involving two particles for every loop momentum. This is achieved by a multidimensional residue theorem relying crucially on the localisation of the amplitude on the scattering equations. This residue theorem effectively moves the support off ambitwistor space to allow for off-shell loop momenta. 

We develop this idea into a widely applicable framework for loop integrands in quantum field theories that localises the expressions on new off-shell scattering equations. Moreover, we demonstrate that this framework not only gives the correct supergravity amplitudes, thereby providing strong evidence for the ACS conjecture, but also extend it to include super Yang-Mills integrands at leading trace, as well as non-supersymmetric theories. For these non-supersymmetric one-loop formulae, we provide a systematic proof relying on the factorisation properties of the nodal worldsheet underlying these formulae. These ideas have a natural extension to higher genus, and we conclude on a proposal for an all-loop integrand that would realise the primary motivation behind the Feynman tree-theorem - expressing a higher loop amplitude as an object of a complexity comparable to a tree-level amplitude.

\paragraph{Summary.}
More specifically, we will start with a brief review of the progress in formulating the ambitwistor string on the elliptic curve \cite{Adamo:2013tsa,Casali:2014hfa,Adamo:2015hoa} in \cref{sec6:review_loop}. Building on this, \cref{sec6:general} will develop the new mathematical framework for integrands from the nodal Riemann sphere. To allow for the identification of a well-defined loop-integrand, we will briefly review the scattering equations on the elliptic curve with a different representation of the solution to the differential equation for the meromorphic differential $P$. The localisation of the formulae on the scattering equations then facilitates the application of a contour integral argument for the complex structure $\tau$ of the elliptic curve, leading to new formulae on a nodal Riemann sphere. In particular, these are supported on new off-shell scattering equations determining both the marked points associated to insertions of external momenta and the location of the nodes of the Riemann sphere that correspond to insertions of the off-shell loop momentum. In this representation, the one-loop integrands of the scattering amplitudes are manifestly rational functions of the external momenta, and achieve the goal of having a similar complexity to tree-amplitudes.\\

We demonstrate both the computational simplicity and the mathematical beauty of this framework on nodal Riemann spheres in \cref{sec:supersymm-theor-1} and \cref{sec:non-supersymm-theor} by exploring these amplitudes in more detail for both supersymmetric and non-supersymmetric theories. This provides very strong evidence for the validity of the ACS one-loop fomula. For supersymmetric theories in \cref{sec:supersymm-theor-1}, we derive supergravity amplitudes from the ACS one-loop conjecture, with an integrand constructed from one-loop extensions of the CHY Pfaffians.  If one of the one-loop Pfaffians is replaced by a one-loop extension of the Parke-Taylor factor, super-Yang-Mills amplitudes are obtained. If both one-loop Pfaffians are replaced by one-loop Parke-Taylor factors, it was shown by \cite{He:2015yua} (see also \cite{Baadsgaard:2015hia}) that certain subtleties arise as additional degenerate solutions of the scattering equations contribute, and diagrams with bubbles on the external legs need to be considered.

Non-supersymmetric theories are presented in \cref{sec:non-supersymm-theor}, where we provide a detailed study of the individual contributions of the Neveu-Schwarz and Ramond sectors to the one-loop amplitudes. We express the NS contribution as a reduced Pfaffian of CHY type for a larger $(n+2)\times (n+2)$ matrix of co-rank two. The resulting formulae have been subjected to various checks at low low particle number, and are proved systematically in \cref{sec:factorization}.

A subtlety that arises here follows from an analysis of \cite{He:2015yua} in which it is argued that a degenerate class of solutions to the scattering equations might contribute non-trivially for non-supersymmetric theories, and that of \cite{Baadsgaard:2015hia} who point out that on these degenerate solutions there is a risk of divergence, and some regularization might be required.  For our proposed integrands we show (in the subsequent section) that no regularization is required at the divergent solutions.  Nevertheless, we propose that these degenerate solutions should not be included as we see in our proof in the subsequent section that they do not contribute to the Q-cuts \cite{Baadsgaard:2015twa}, and so are not needed in the final formula.  It seems most likely that they correspond to degenerate contributions that will vanish under dimensional regularization and are discarded in the Q-cut formalism.

In \cref{sec:factorization}, we provide a full proof at one loop for the $n$-gon conjecture, and for the non-supersymmetric gauge, gravity and bi-adjoint scalar amplitudes.  The basic strategy is to study factorisation of the Riemann sphere.   The one-loop scattering equations link factorisation channels of the integrands, apart from the explicit\footnote{where $\ell$ is the loop momentum} $1/\ell^2$, to degenerations of the Riemann sphere. We can use this to identify all the poles involving the loop momentum and the corresponding residues. Moreover, similar factorisation ideas allow us to identify the fall-off as $\ell\rightarrow \infty$.  This immediately gives  the poles and residues in the case of the $n$-gon conjecture.  For  gauge and gravity amplitudes, we also need to study the the Parke-Taylor factors and Pfaffians that arise. The poles and residues that we find give perfect agreement with the Q-cut representation of the loop amplitude, as obtained recently in \cite{Baadsgaard:2015twa}. This completes a proof of our formulae; the Q-cut procedure applied to our formula will yield the correct Q-cut representation.  We are restricted to a proof for the non-supersymmetric theories since there are no closed-form formulae for tree amplitudes with two Ramond sector particles.\\

To conclude, we give a brief discussion of the extension of these ideas to an all-loop conjecture for the one-loop amplitudes of type-II supergravity and super Yang-Mills theory in \cref{sec:Loops_all-loop}. In particular, at $g$ loops, these integrands have the same complexity as $n+2g$-particle tree-amplitudes. This realises and extends the primary motivation behind the Feynman tree theorem \cite{Feynman:1963ax}, and initiates a program with the potential to resolve long-standing questions regarding the UV behaviour of maximal supergravity.

\section{Review of the ambitwistor string at genus one}\label{sec6:review_loop}
In this section, we will review the ambitwistor string at genus one \cite{Adamo:2013tsa}, including in particular the  proposals for loop amplitude expressions following strategies developed in string theory.  While at tree-level the gain might not yet be immediately obvious, the loop extension beautifully highlights the benefit of an underlying mathematical structure and the strength of worldsheet theory, leading to new insights and formulae. We will focus here on the RNS ambitwistor string \cite{Adamo:2013tsa,Casali:2014hfa}, however, progress has also been made in the pure spinor formalism \cite{Adamo:2015hoa}.

Adamo, Casali and Skinner (ACS) \cite{Adamo:2013tsa} used the ambitwistor string to extended the CHY formulae for type II supergravity in 10 dimensions to one-loop in terms of scattering equations on an elliptic curve (and, in principle, to $g$-loops on curves of genus $g$). While supergravity loop amplitudes are divergent without a suitable renormalisation scheme, the origin of these divergences is well-understood and can be factored out, leaving finite integrands obtained from the anomaly free ambitwistor worldsheet conformal field theory. We will review the ACS one-loop proposal in this section, focussing on the one-loop scattering equations in \cref{sec6:SE_g1} and on the full integrands, localised again on the scattering equations, in \cref{sec6:amplitudes_g1} (see also \cref{sec6:moduli_space} for a general review of moduli space of curves, especially focussing on genus one, and \cite{Verlinde:1987sd,D'Hoker:1988ta,D'Hoker:2002gw,Witten:2012bh} for reviews of multiloop amplitudes in string theory).

\subsection{Scattering equations at genus one}\label{sec6:SE_g1}
To study the ambitwistor string at genus one, consider an elliptic curve $\Sigma_q =\C/\{\Z\oplus \Z \tau\}$ with complex coordinate $z$ and modulus $\tau$, where we have introduced for future convenience the variable $q=\e^{2\pi i \tau}$. We will discuss first the form of the scattering equations at genus one, derived again from solutions to the defining equation \cref{eq2:DE_SE} for the meromorphic one-form $P_\mu\in\Omega^0(\Sigma,K_\Sigma)$.

Recall from the discussion in \cref{sec2:ambi_string} that gauge fixing the $\alpha$ gauge symmetry leads to insertions \cref{eq2:gaugefixinginPI} in path integral. The gauge fixing functional vanishes at the insertion points $z_i$,  giving an obstruction to setting $e=0$ parametrised by 
\begin{equation}
 \sum_{r=1}^{h^{0,1}(\Sigma,T_\Sigma(-\sigma_1\dots-\sigma_n))}s_r\mu_r\,,
\end{equation}
where $\mu_r$ is a basis of Beltrami differentials. In the case of genus one, the dimension of the moduli space is $h^{0,1}(\Sigma,T_\Sigma(-\sigma_1\dots-\sigma_n))=n$, and for the unmarked elliptic curve $h^{0,1}(\Sigma,T_\Sigma)=1$. We can thus choose $n-1$ of the Beltrami differentials in the basis to extract residues at marked points, similar to the choice at tree level. The remaining Beltrami differential $\mu$ is then associated to (variations in) the complex structure parameter $\tau$ of the elliptic curve:
\begin{equation}\label{eq6:SE_tau}
 \bar\delta\left(\int_\Sigma \mu P^2\right)\,\int_\Sigma\tilde{b}\mu = \tilde{b}_0\,\bar\delta\left(P^2 (\tau,z_0)\right)\,,
\end{equation}
where $z_0$ is an arbitrarily chosen reference point.
By the Riemann-Roch theorem \cref{eq2:Riemann-Roch} at genus one, the antighost $\tilde{b}$ has exactly one zero mode, which is absorbed\footnote{The $c$ and $\tilde{c}$ zero modes at genus one are absorbed by the insertion of one fixed vertex operator in the correlator.} by the insertion of $\tilde{b}$. The delta-function in \cref{eq6:SE_tau} is understood to be part of the scattering equations at genus one, and fixes the integral over the modulus $\tau$. Geometrically, the scattering equations encode the same content as at genus zero, ensuring that the target space of the ambitwistor string is indeed the space of complex null geodesics $P\A$. To see this, note that the meromorphic quadratic differential $P^2$ has again at most simple poles, since the coefficients of the double poles vanish on-shell. The first $n-1$ scattering equations obtained from the Beltrami differentials $\mu_i$ imply that the residues of $P^2$ at these simple poles vanish, and thus $P^2$ must be globally holomorphic on $\Sigma$. On the elliptic curve, there exists a unique holomorphic quadratic differential $u\, \d z$, which is constant. The final scattering equation $P^2 (\tau,z_0)$ is thus independent of the insertion point $z_0$ and ensures that the constant $u$ vanishes. The scattering equations therefore encode the geometric information that $P^2=0$ everywhere, and the localisation on the support of the scattering equations ensures that the target space of the chiral worldsheet model remains ambitwistor space at higher genus;
\begin{equation}\label{eq6:SE_genus1}
 \text{Res}_iP^2=0\, \quad i=2,\dots,n\qquad \text{and}\quad P^2(z_0)=0\,.
\end{equation}
Indeed, this was to be expected, since it emerges as a direct consequence of the gauge fields $\tilde{e}$ and $\chi_r$ implementing the quotient by the Hamiltonian vector fields $D$ (and since the gauge redundancy is not anomalous). 

To find the explicit form of these scattering equations, we can again perform the $(X,P)$ path integral for a correlator with $n$ vertex operator insertions. As at genus zero, the integral over the zero modes of $X$ yields an overall momentum conserving delta function $\delta^{10}(\sum_ik_i)$, while the non-zero modes impose the constraint \cref{eq2:DE_SE},
\begin{equation}
    \dbar P=2\pi i \d z \sum_i k_i \,\bar \delta(z-z_i) \,.
\end{equation}
On the elliptic curve, this equation has the homogeneous solution $\ell_\mu \d z$ where $\d z$ is the unique holomorphic Abelian differential on the torus $\Sigma_q=\C/\{\mathbb{Z}\oplus\mathbb{Z}\tau\}$. The general solution to \cref{eq2:DE_SE} is therefore given by the sum over the zero modes and inhomogeneous solutions,
\begin{equation}\label{P-def_genus1}
 P_\mu(z)=2\pi i \ell_\mu \d z +\sum_{i=1}^n k_{i,\mu}\,\wt{S}_1(z,z_i|\tau)\,.
\end{equation}
Here, we have chosen a representation where $\wt{S}_1$ is the $PX$ propagator on the genus one Riemann surface, \cite{Adamo:2013tsa},
\begin{equation}\label{eq6:SE_ACS}
 \wt{S}_1(z,z_i|\tau)=\left(\frac{\theta_1'(z-z_i;\tau)}{\theta_1(z-z_i;\tau)}+4\pi\frac{\text{Im}(z-z_i)}{\text{Im}(\tau)}\right)\d z\,,
\end{equation}
where $\theta_1(z;\tau)$ are the Jacobi theta functions associated with the spin structure $\pmb{\alpha}=1$, see \cref{eq:theta-def}. Note that $\wt{S}$ is indeed meromorphic in $z$ since the term proportional to Im$(z)$ vanishes by momentum conservation, and is invariant under modular transformations $\wt{S}_1(z,z_i|\tau)=\wt{S}_1(z,z_i|\tau+1)=\wt{S}_1(\frac{1}{z},\frac{1}{z_i}|-\frac{1}{\tau})$ (see \cref{sec6:moduli_space} for a brief review of the modular group). An alternative proposal \cite{Casali:2014hfa} uses instead the purely holomorphic Szeg\H{o} kernel for the spin structure $\pmb{\alpha}=1$,
\begin{equation}\label{eq6:SE_CT}
 \wt{S}_1(z,z_i|\tau)=S_1(z,z_i|\tau)=\frac{\theta_1'(z-z_i;\tau)}{\theta_1(z-z_i;\tau)}\d z\,.
\end{equation}
As argued in \cite{Casali:2014hfa}, this choice is necessary to reproduce the correct pole structure of the field theory integrand. Moreover, \cite{Casali:2014hfa} clarifies the relation of \cref{eq6:SE_genus1} to the scattering equations to the genus one scattering equations found by Gross and Mende. 

The scattering equations completely localise the integral over the $n$-dimensional moduli space $\mathscr{M}_{1,n}$. Therefore, the introduction of $\ell^\mu$ as zero modes of $P$ is of crucial importance, since they are explicitly integrated over after the correlator is calculated, and will acquire the interpretation of a {\it loop momentum}. This provides an important distinction from usual string theory, which is UV finite.\footnote{The Deligne Mumford compactification of moduli space of marked curves includes the singularities on boundary, but these correspond to IR divergences, see \cref{sec6:moduli_space} for a very brief review and e.g. \cite{Witten:2012bh} for more details. (Also, \cite{Witten:2015mec} provides a nice heuristic point of view.)} Since the moduli integral in the ambitwistor string includes these zero modes, it is non-compact, and can thus potentially give rise to the UV divergences of field theory.

\subsection{One-loop scattering amplitudes} \label{sec6:amplitudes_g1}
Before discussing the one-loop correlation function, recall that on the elliptic curve, we get two contributions from different spin structures for the fermions.\footnote{A general manifold has to fulfil a topological condition to allow for spinors. In particular, the transition functions of the tangent frame bundle must have a lift from SO$(d)$ to its double cover $Spin(d)$. If this holds, the set of transition functions is known as the spin structure. Riemann surfaces are spin, and in general admit more than one spin structure. We can classify the spin structures in terms of periodic and antiperiodic boundary conditions for the fermions around the generators of the Riemann surface. This corresponds to $2^{2g}$ choices at genus $g$, and thus there are $2^{2g}$ spin structures. Note that both periodic (Ramond) and antiperiodic (NS) boundary conditions transform irreducibly under modular transformations.} In particular, there are $2^2=4$ spin structure, three even ones with $\pmb{\alpha}=2,3,4$ and one odd one, $\pmb{\alpha}=1$, that we will treat separately.\\

The odd spin structure, characterised by the fermions being periodic around each non-trivial cycle, only contributes for $n\geqslant5$ due to the zero modes of the fermion fields $\Psi_r$, and won't be relevant for the remainder of this thesis. We therefore refer to the original work \cite{Adamo:2013tsa} for a complete discussion.\\

For the even spin structures, the fermionic fields do not acquire zero modes, but both the ghosts and antighosts of the $(b,c)$ and $(\tilde{b},\tilde{c})$ systems develop a constant zero mode. As discussed above, these are absorbed for the ghosts by the insertion of one fixed vertex operator in the correlation function, while the insertion in path integral originating from the gauge fixing absorb the antighost zero modes. The correlator is therefore given by
\begin{equation}\label{eq6:ampl-genus-1}
 \cM^{(1),\text{even}}_n=\left\langle b_0\tilde{b}_0\,\bar\delta(P^2(z_0))\,c\tilde{c} V_1(z_1)\prod_{i=2}^n\mathcal{V}_i \right\rangle\,.
\end{equation}
Due to the absence of fixed vertex operators, the correlators of the fields $\Psi_r$ give rise to full Pfaffians, whilst the non-trivial partition function at genus one contributes $Z_{\pmb{\alpha};\pmb{\beta}}$. Imposing the standard type IIB GSO projection,\footnote{Hence, in absence of vertex operators, the partition function takes the form
\begin{equation}
 Z_{\text{IIB}}=\left(Z_1+\sum_{\pmb{\alpha}=2,3,4}(-1)^{\pmb{\alpha}}Z_{\pmb{\alpha}}\right) \left(\wt{Z}_1+\sum_{\pmb{\beta}=2,3,4}(-1)^{\pmb{\beta}}\wt{Z}_{\pmb{\beta}}\right)\,.
\end{equation}
Therefore, the one loop zero-point correlation function vanishes due to the Jacobi vanishing identity, and thus there is no one-loop contribution to the space-time cosmological constant.}
the amplitude, stripped off the overall momentum conserving delta function, becomes
\begin{equation}\label{elliptic-amp}
 \cM^{(1),\text{even}}_n=\int \d^{10}\ell\wedge \d\tau\,\bar\delta(P^2(z_0))\,\prod_{i=2}^n \bar\delta\left(k_i\cdot P(z_i)\right)\,\cI_q\,,
\end{equation}
where the integrand is given by
\begin{equation}
 \cI_q=\sum_{\pmb{\alpha},\pmb{\beta}=2,3,4}(-1)^{\pmb{\alpha}+\pmb{\beta}}Z_{\pmb{\alpha};\pmb{\beta}}(\tau)\pf(M_{\pmb{\alpha}})\pf(\wt{M}_{\pmb{\beta}})\,.
\end{equation}
The integrand, obtained from the CFT correlator of the vertex operators, is thus given as a sum over products of Pfaffians for even spin structures, weighted by partition functions of the same spin structure. It depends, just as at tree level, only on the kinematics, polarisation and insertion points of the external particles. Let us highlight again that due to the trivial $XX$ OPE, the theory does not contain massive states, hence there are only massless states propagating in the loop. 
As discussed above, the scattering equation completely localise the integral over the $n$-dimensional moduli space $\cM_{1,n}$, hence the only integration to be performed is the one over the zero modes $\ell_\mu$ of $P_\mu$.

The Pfaffians in the integrand $\cI_q$ are given by 
\begin{equation}\label{eq:Malpha}
 M_{\pmb{\alpha}}=\begin{pmatrix} A & -C^T\\ C &B \end{pmatrix}\,,
\end{equation}
with entries, for $i\neq j$,
\begin{align}
 &A_{ij}=k_i\cdot k_j S_{\pmb{\alpha}}(z_{ij}|\tau)\,, && B_{ij}=\epsilon_i\cdot \epsilon_j S_{\pmb{\alpha}}(z_{ij}|\tau)\,, && C_{ij}=\epsilon_i\cdot k_j S_{\pmb{\alpha}}(z_{ij}|\tau)\,,\\
 & A_{ii}=0\,, && B_{ii}=0\,, && C_{ii}=\epsilon_i\cdot \ell\d z_i+\sum_{j\neq i}\epsilon_i\cdot k_j \wt{S}_1(z_{ij}|\tau)\,. \nonumber
\end{align}
Here, the Szeg\H{o} kernels $S_{\pmb{\alpha}}$ are the torus free fermion propagators for the respective spin structures, interchanged under modular transformations.
\begin{equation}
 S_{\pmb{\alpha}}(z_{ij}|\tau)=\frac{\theta_1'(0;\tau)}{\theta_1(z_{ij};\tau)}\frac{\theta_{\pmb{\alpha}}(z_{ij};\tau)}{\theta_{\pmb{\alpha}}(0;\tau)}\sqrt{\d z_i \d z_j}\,.
\end{equation}
The tilded matrix $\wt{M}_{\pmb{\alpha}}$ is defined as $M_{\pmb{\alpha}}$, but with different polarisation vectors $\tilde\epsilon$, such that
the polarisation tensors
$\epsilon_i^{\mu\nu}=\epsilon_i^\mu \tilde\epsilon_i^\nu$
correspond to the NS-NS states of supergravity,
graviton, the dilaton and the B-field. In terms of $\epsilon^{\mu\nu}$, the dilaton corresponds to the trace part, the $B$-field to the skew part, and the graviton to the traceless symmetric part.\\

The partition function for the $\pmb{\alpha};\pmb{\beta}$ spin structure receives a contribution of $1/\eta(\tau)^{16}$ from the $(P,X)$ system, and $\theta_{\pmb{\alpha}}(0|\tau)^4/\eta(\tau)^{4}$ from each of the $\Psi_r$ fermion systems, where the powers of the Dedekind eta function are given by twice and one-half the transverse directions of $10$ dimensional Minkowski space. The full $Z_{\pmb{\alpha};\pmb{\beta}}$ then take the form
\begin{equation}\label{eq6:pt-funs}
 Z_{\pmb{\alpha};\pmb{\beta}}(\tau)=\frac{1}{\eta(\tau)^{16}}\frac{\theta_{\pmb{\alpha}}(0|\tau)^4}{\eta(\tau)^{4}} \frac{\theta_{\pmb{\beta}}(0|\tau)^4}{\eta(\tau)^{4}}\,,
\end{equation}
where, for $q=e^{2i\pi\tau}$, the Dedekind eta function and the Jacobi theta functions are given by
\begin{align}\label{eq:theta-def}
 \theta_{\pmb{\alpha}}(z|\tau)&=\sum_{n\in\mathbb{Z}}q^{\frac{1}{2}(n-a/2)^2}e^{2i\pi(z-b/2)(n-a/2)}\,,\\
 \eta(\tau)&=q^{1/24}\prod_{n=1}^\infty (1-q^n)\,.
\end{align}
Here, $\pmb{\alpha}=3,4,2:=(a,b)=(0,0),(0,1),(1,0)$ denote the even spin structures and $\pmb{\alpha}=1:=(1,1)$ is the odd characteristic, corresponding to periodic $(1)$ and antiperiodic $(0)$ boundary conditions around the non-trivial cycles of the elliptic curve.\\

An important characteristic of the amplitude \cref{elliptic-amp} is its modular invariance: if under transformations $T:(z|\tau)\rightarrow(z|\tau+1)$ and $S:(z|\tau)\rightarrow(\frac{z}{\tau}|-\frac{1}{\tau})$, we require $\ell_\mu\rightarrow \tau\ell_\mu$ such that the zero modes $\ell_\mu\d z$ of the meromorphic one-form $P_\mu$ do not transform,\footnote{Note however that care is needed with the definition of the integration cycle for the loop momentum.} the amplitude is  modular invariant. Therefore, the integration region for the modular parameter $\tau$ of the elliptic curve is given by the fundamental domain $\mathcal{F}$, see \cref{fig6:fund-dom0}. Note that the only boundary of $\mathcal{F}$ is at $\Im(\tau)=\infty$, or $q=0$, corresponding to the non-separating degeneration pinching a generating cycle of the elliptic curve.\\

Strong evidence for the validity of this formula was given in \cite{Adamo:2013tsa} by studying its factorisation properties. Moreover, in \cite{Casali:2014hfa}, it was shown that for $n=4$, as in conventional string theory, the integrand $\cI_q$ is independent of $z_i$ and $q$, and factors out of the integral. Moreover, the formulae were shown to reproduce the known
integrands of four-points supergravity amplitudes at a triple cut.
The non-trivial remaining integral for four particles is the $n=4$ version of the more general integral
\begin{equation}
\cM^{(1)}_{n\mathrm{-gon}}=\int \,d^{d} \ell \,  d\tau \, \bar\delta
(P^2(z_0'))\prod_{i=2}^n \bar\delta(k_i\cdot P(z_i)) d z_i^2
\,,\label{eq:ngon}
\end{equation}
where the integral can be checked to be modular invariant in the dimension $d=2n+2$.  In \cite{Casali:2014hfa}, this was conjectured to be equivalent to a sum over permutations of $n$-gons and, if true at $n=4$, this would confirm the 4-particle supergravity conjecture at one loop. \\

This chapter not only provides the missing proof that the ambitwistor string at genus one \cref{elliptic-amp} reproduces supergravity amplitudes, but also develops a completely new approach to loop scattering amplitudes in massless theories, reformulating them not as arising from a higher genus worldsheet, but rather from nodal Riemann spheres. This result relies crucially on the localisation on the scattering equations, and significantly reduces the (manifest) complexity of loop amplitudes.

\section{Loop integrands from the nodal Riemann sphere}\label{sec6:general} 
As discussed above, the scattering equation completely localise the integral over the $n$-dimensional moduli space $\mathscr{M}_{1,n}$ and thus restrict the integral to a sum over a discrete set of solutions. While Galois theory guaranteed a rational result at genus zero, this is obscured at one loop: how do rational integrands emerge from sum over solution to scattering equations, that contain Jacobi theta functions? We will see that the answer takes a surprising form: on the support of the scattering equations, we can use a residue theorem to localise the amplitude on the boundary of the moduli space $\mathcal{F}$ for the modulus $\tau$, corresponding to a non-separating degeneration to a nodal Riemann sphere. Loop amplitudes thus have the same complexity as tree-amplitudes, albeit with effectively more external particles.

To develop this new framework, we first need to reformulate the scattering equations in a manifestly doubly periodic and holomorphic form in \cref{sec:scatt-equat-torus}. The localisation on these holomorphic doubly periodic scattering equations provides the basis for the application of the residue theorem that reformulates the amplitude on a nodal Riemann sphere, see \cref{sec6:torus-to-RS}. As a first, non-trivial check, we will use this to prove the $n$-gon conjecture \cref{eq:ngon} of \cite{Casali:2014hfa}.

\subsection{The scattering equations on a torus}\label{sec:scatt-equat-torus}

In this section, we define holomorphic and single-valued scattering equations on a torus.  These are motivated by the definitions \cref{eq6:SE_ACS} and \cref{eq6:SE_CT} given in \cite{Adamo:2013tsa, Casali:2014hfa}, but the definition has been changed to ensure that they are holomorphic and single-valued on the torus with a well defined loop momentum.   

On the elliptic curve $\Sigma_q =\C/\{\Z\oplus \Z \tau\}$ with $q=\e^{2\pi i \tau}$, we obtain the scattering equations by solving \cref{eq2:DE_SE} for the meromorphic one-form $P(z,z_i|q)$ on $\Sigma_q$,
\begin{equation}
\dbar P=2\pi i \sum_i k_i \,\bar \delta(z-z_i) dz\, , 
\label{eq:dbarP}
\end{equation}
where the delta function is defined as in \cref{eq:deltabar-def}. Again parametrising the zero modes of $P$ by $\ell\in\R^{1,d-1}$, we can use momentum conservation to express the solution in a manifestly holomorphic and doubly-periodic manner,
\begin{equation}
\label{P-def}
P_\mu(z,z_i)=2\pi i\,\ell_\mu \d z +\sum_{i=1}^n k_{i,\mu}\,\wt{S}_1(z,z_i|\tau)\,,
\end{equation}
where
\begin{equation}
 \wt{S}_1(z,z_i|\tau)=\left(\frac{\theta'_1 (z-z_i)}{\theta_1 (z-z_i)} + \frac{\theta'_1
  (z_i-z_0)}{\theta_1 (z_i-z_0)} + \frac{\theta'_1 (z_0-z)}{\theta_1
  (z_0-z)}\right)\d z\,.
\end{equation}
Here the prime denotes $\partial/\partial z$,  $z_0$ is a choice of reference point, and the Jacobi theta functions $\theta_{\pmb\alpha}$ were defined in \cref{eq:theta-def}. Since $\theta_1(z)\underset{z\to0}{\sim} z$, the meromorphic one-form $P_\mu$ has poles at $z=z_i$ for $i=1,\ldots , n$. However, momentum conservation implies that the coefficient of $\theta'_1(z_0-z)/\theta_1(z_0-z)$ is in fact zero, so $P$ is holomorphic at $z_0$.  We include the last term to make the double periodicity manifest.    
Theta functions are trivially periodic under $z\rightarrow z+1$, but
under $z\rightarrow z+\tau$ we have 
\begin{equation}
\frac{ \theta'(z+\tau)}{\theta(z+\tau)}
=\frac{\theta'(z)}{\theta(z)}-2\pi i\, .\label{eq:theta-periodicity}
\end{equation}
It is easy to see that our  expression for $P$ is doubly periodic in $z$ as a consequence of momentum conservation, but it is also doubly periodic in the $z_i$ as a consequence of the extra last term involving the reference point in \eqref{P-def}.

Using this, we define the scattering equations to be
\begin{subequations} \label{eq:SE}
 \begin{align}
  &\mathrm{Res}_{z_i} P^2(z)=2k_i\cdot P(z_i)=0\, , \qquad\qquad i=2,\dots,n \label{eq6:SE_res}\\
  &P^2(z_0')=0\,.\label{eq6:SE_P^2}
 \end{align}
\end{subequations}
where $z'_0$ is another choice of reference point. Because the sum of residues of $P^2$ vanishes, the first scattering equation Res$_{z_1}P^2=0$ follows from \cref{eq6:SE_res}. Translation invariance moreover implies that we can fix the location of $z_1$ arbitrarily. Next note that on the support of the other scattering equations \cref{eq6:SE_res}, $P^2(z_0')$ is global and holomorphic in $z'_0$ and hence independent of the choice of reference point.  Therefore, the last scattering equation \cref{eq6:SE_P^2} serves to determine the complex structure $\tau$ of the torus.\smallskip

Some remarks are in order here. Since our $P$ is meromorphic and doubly periodic both in $z$ \emph{and} the $z_i$, so are the scattering equations.\footnote{An earlier version presented in \cite{Geyer:2015bja} is holomorphic and doubly periodic, but concerns were raised about factorisation by Adamo, Casali \& Skinner. We would like to thank them for suggesting this approach.}  It differs from the previous versions \cref{eq6:SE_ACS} and \cref{eq6:SE_CT} in the choice of an additive `constant' term in $\ell$ that depends on the $z_i$ and $k_i$.  The ACS version \cref{eq6:SE_ACS} is not holomorphic in the $z_i$; this leads to non-holomorphic scattering equations and it was argued in \cite{Casali:2014hfa} that they do not give the correct $1/\ell^{2}$ pole structure. A holomorphic version \cref{eq6:SE_CT} was proposed there
for which factorisation was checked, which is also the version used in \cite{Ohmori:2015sha}.  However, that version is not doubly periodic so the scattering equations are not well defined on the elliptic curve
for fixed constant loop momentum $\ell$; there are different numbers of solutions on
the different fundamental domains of the lattice as well as those
related by SL$(2,\Z)$ as observed numerically
in~\cite{Casali:2014hfa}.\footnote{This fact leads to a well-known
  apparent ambiguity in the definition of the loop momentum in all
  first quantized theories (worldline, strings \cite{D'Hoker:1988ta}).
  This ambiguity drops out of the physical observables after
  integration of the loop momentum and does not alter the modular
  properties of the string amplitudes.
However, the case of the first quantized ambitwistor string is undoubtedly more
subtle because of the presence of the scattering equations and the fact that we must integrate only over a real contour in the loop momentum variable. Therefore
we must proceed by making two assumptions.  Firstly, we must cure the ambiguity in the loop momentum in
the integrand by defining  $P$ by 
\eqref{P-def}. Secondly, we want to define the integration cycle of
the theory (in the sense of \cite{Ohmori:2015sha,Witten:2012bh}) as
including only the solutions to the scattering equations in the
fundamental domain, as described below.}

With this version of the scattering equations, the ACS proposal for the one-loop integrand of type-II supergravity amplitudes \cref{elliptic-amp} is not only modular invariant, but also doubly periodic in the marked points $z_i$.


\subsection{From a torus to a Riemann sphere}\label{sec6:torus-to-RS}
Here we use a residue theorem (or integration by parts in our
notation) to reduce the formula \cref{eq6:ampl-genus-1} on the elliptic curve to one on the
nodal Riemann sphere at $q=0$. 
The argument relies on the intuitive fact that the scattering equation imposed
by $\bar\delta (P^2(z_0))$ has a separate status from the others,
serving to fix $\tau$, and can be analysed on the $\tau-$plane alone. We can use a residue theorem to convert it into an equation
enforcing $q=0$.  Such `global residue theorems' have already been
applied to tree-level CHY formulae by \cite{Dolan:2013isa} to relate
the scalar CHY formulae to their Feynman diagrams.  
We apply the same strategy here, and we reduce the ACS conjecture \cref{eq6:ampl-genus-1} to a formula for the one-loop integrand based on off-shell scattering equations on a nodal Riemann sphere.\footnote{In fact, they are a forward limit of those of \cite{Naculich:2014naa}.}  These one-loop off-shell scattering equations strongly resemble the tree-level ones, but explicitly depend on off-shell momenta associated to $\ell$.

\begin{figure}[ht]
\begin{center}
 \begin{tikzpicture} [scale=2.5]
  \filldraw [fill=light-gray, draw=white] (0.5,0.866) arc [radius=1, start angle=60, end angle= 120] -- (-0.5,2.7) -- (0.5,2.7) -- (0.5,0.866);
  \draw [black] (1.1,0) -- (-1.1,0);
  \draw [black] (0,0) -- (0,2);
  \draw [black,dashed] (0,2) -- (0,2.5);
  \node at (0.5,-0.15) {$\frac{1}{2}$};
  \node at (-0.5,-0.15) {-$\frac{1}{2}$};
  \node at (1.195,2.63) {$\tau$};
  \draw (1.1,2.7) -- (1.1,2.55) -- (1.25,2.55);
  \draw [gray] (0,0) arc [radius=1, start angle=0, end angle= 90];
  \draw [black] (1,0) arc [radius=1, start angle=0, end angle= 180];
  \draw [gray] (1,1) arc [radius=1, start angle=90, end angle= 180];
  \draw [gray] (0.5,0) -- (0.5,0.866);
  \draw [gray] (-0.5,0) -- (-0.5,0.866);
  \draw (0.5,0.866) -- (0.5,2);
  \draw (-0.5,0.866) -- (-0.5,2);
  \draw [dashed] (0.5,2) -- (0.5,2.5);
  \draw [dashed] (-0.5,2) -- (-0.5,2.5);
  \draw (0.5,2.5) -- (0.5,2.7);
  \draw (-0.5,2.5) -- (-0.5,2.7);
  \draw [red3,thick,->] (0.48,2.7) -- (0,2.7);
  \draw [red3,thick] (0,2.7) -- (-0.48,2.7);
  \draw [red3,thick] (0.48,2.5) -- (0.48,2.7);
  \draw [red3,thick,<-] (-0.48,2.5) -- (-0.48,2.7);
  \draw [red3,thick,dashed] (0.48,2) -- (0.48,2.5);
  \draw [red3,thick,dashed] (-0.48,2) -- (-0.48,2.5);
  \draw [red3,thick,->] (0.48,1.2) -- (0.48,2);
  \draw [red3,thick,->] (0.48,0.9) -- (0.48,1.2);
  \draw [red3,thick,<-] (-0.48,1.5) -- (-0.48,2);
  \draw [red3,thick] (-0.48,0.9) -- (-0.48,1.5);
  \draw [red3,thick] (0.48,0.9) arc [radius=0.97, start angle=60, end angle= 119];
  \draw (-0.025,1.125) -- (0.025,1.075);
  \draw (-0.025,1.075) -- (0.025,1.125);
  \draw (0.225,1.225) -- (0.175,1.175);
  \draw (0.225,1.175) -- (0.175,1.225);
  \draw (-0.225,1.225) -- (-0.175,1.175);
  \draw (-0.225,1.175) -- (-0.175,1.225);
  \draw (-0.025,1.455) -- (0.025,1.405);
  \draw (-0.025,1.405) -- (0.025,1.455);
  \draw (-0.385,1.735) -- (-0.435,1.685);
  \draw (-0.385,1.685) -- (-0.435,1.735);
  \draw (0.385,1.735) -- (0.435,1.685);
  \draw (0.385,1.685) -- (0.435,1.735);
  \draw (0.325,1.805) -- (0.375,1.755);
  \draw (0.325,1.755) -- (0.375,1.805);
  \draw (-0.325,1.805) -- (-0.375,1.755);
  \draw (-0.325,1.755) -- (-0.375,1.805);
  \draw (0.205,2.425) -- (0.255,2.375);
  \draw (0.255,2.425) -- (0.205,2.375);
  \draw (-0.205,2.425) -- (-0.255,2.375);
  \draw (-0.205,2.375) -- (-0.255,2.425);
  \draw [gray] (0,0) arc [radius=0.333, start angle=0, end angle= 180];
  \draw [gray] (-0.334,0) arc [radius=0.333, start angle=0, end angle= 180];
  \draw [gray] (0.666,0) arc [radius=0.333, start angle=0, end angle= 180];
  \draw [gray] (1,0) arc [radius=0.333, start angle=0, end angle= 180];
  \draw [gray] (0,0) arc [radius=0.196, start angle=0, end angle= 180];
  \draw [gray] (-0.608,0) arc [radius=0.196, start angle=0, end angle= 180];
  \draw [gray] (0.392,0) arc [radius=0.196, start angle=0, end angle= 180];
  \draw [gray] (1,0) arc [radius=0.196, start angle=0, end angle= 180];
  \draw [gray] (0,0) arc [radius=0.142, start angle=0, end angle= 180];
  \draw [gray] (-0.716,0) arc [radius=0.142, start angle=0, end angle= 180];
  \draw [gray] (0.284,0) arc [radius=0.142, start angle=0, end angle= 180];
  \draw [gray] (1,0) arc [radius=0.142, start angle=0, end angle= 180];
  \draw [gray] (0,0) arc [radius=0.108, start angle=0, end angle= 180];
  \draw [gray] (-0.784,0) arc [radius=0.108, start angle=0, end angle= 180];
  \draw [gray] (0.216,0) arc [radius=0.108, start angle=0, end angle= 180];
  \draw [gray] (1,0) arc [radius=0.108, start angle=0, end angle= 180];
  \draw [gray] (0.5,0) arc [radius=0.122, start angle=0, end angle= 180];
  \draw [gray] (-0.5,0) arc [radius=0.122, start angle=0, end angle= 180];
  \draw [gray] (0.744,0) arc [radius=0.122, start angle=0, end angle= 180];
  \draw [gray] (-0.256,0) arc [radius=0.122, start angle=0, end angle= 180];
  \draw [gray] (0.5,0) arc [radius=0.060, start angle=0, end angle= 180];
  \draw [gray] (-0.5,0) arc [radius=0.060, start angle=0, end angle= 180];
  \draw [gray] (0.620,0) arc [radius=0.060, start angle=0, end angle= 180];
  \draw [gray] (-0.380,0) arc [radius=0.060, start angle=0, end angle= 180];
  \draw [gray] (0.666,0) arc [radius=0.039, start angle=0, end angle= 180];
  \draw [gray] (0.412,0) arc [radius=0.039, start angle=0, end angle= 180];
  \draw [gray] (-0.588,0) arc [radius=0.039, start angle=0, end angle= 180];
  \draw [gray] (-0.334,0) arc [radius=0.039, start angle=0, end angle= 180];
  \draw [gray] (0.334,0) arc [radius=0.062, start angle=0, end angle= 180];
  \draw [gray] (0.790,0) arc [radius=0.062, start angle=0, end angle= 180];
  \draw [gray] (-0.666,0) arc [radius=0.062, start angle=0, end angle= 180];
  \draw [gray] (-0.210,0) arc [radius=0.062, start angle=0, end angle= 180];
  \draw (3.05,1) circle [x radius=0.4, y radius=0.2];
  \draw (2.85,1.045) arc [radius=0.4, start angle=240, end angle=300];
  \draw (3.2,1.015) arc [radius=0.4, start angle=70, end angle=110];
  \draw [fill] (2.95,1.1) circle [radius=.3pt];
  \draw [fill] (2.8,0.95) circle [radius=.3pt];
  \draw [fill] (3.35,0.98) circle [radius=.3pt];
  \draw [fill] (3.3,1.01) circle [radius=.3pt];
  \draw (3.05,1.75) circle [x radius=0.4, y radius=0.2];
  \draw (3.05,1.92) to [out=180, in=60] (2.9,1.85) to [out=250, in=180] (3.05,1.72) to [out=0, in=290] (3.2,1.85) to [out=120, in=0] (3.05,1.92);
  \draw [fill] (2.78,1.82) circle [radius=.3pt];
  \draw [fill] (2.9,1.63) circle [radius=.3pt];
  \draw [fill] (3.35,1.72) circle [radius=.3pt];
  \draw [fill] (3.3,1.76) circle [radius=.3pt];
  \draw (2.5,2.5) circle [x radius=0.4, y radius=0.2];
  \draw (2.5,2.7) to [out=240, in=60] (2.35,2.6) to [out=250, in=180] (2.5,2.47) to [out=0, in=290] (2.65,2.6) to [out=120, in=300] (2.5,2.7);
  \draw [fill] (2.23,2.57) circle [radius=.3pt];
  \draw [fill] (2.35,2.38) circle [radius=.3pt];
  \draw [fill] (2.8,2.48) circle [radius=.3pt];
  \draw [fill] (2.75,2.51) circle [radius=.3pt];
  \draw (3.6,2.5) circle [x radius=0.4, y radius=0.2];
  \draw [dashed,red3] (3.53,2.65) to [out=140, in=180] (3.6,2.9) to [out=0, in=40] (3.67,2.65);
  \draw [fill] (3.33,2.57) circle [radius=.3pt];
  \draw [fill] (3.45,2.41) circle [radius=.3pt];
  \draw [fill] (3.8,2.43) circle [radius=.3pt];
  \draw [fill] (3.85,2.55) circle [radius=.3pt];
  \draw [fill,red3] (3.53,2.65) circle [radius=.3pt];
  \draw [fill,red3] (3.67,2.65) circle [radius=.3pt];
  \node at (3.05,2.5) {$\leftrightarrow$};
 \end{tikzpicture}
 
\caption{Residue theorem in the fundamental domain.}
\label{fund-dom}
\end{center}
\end{figure}
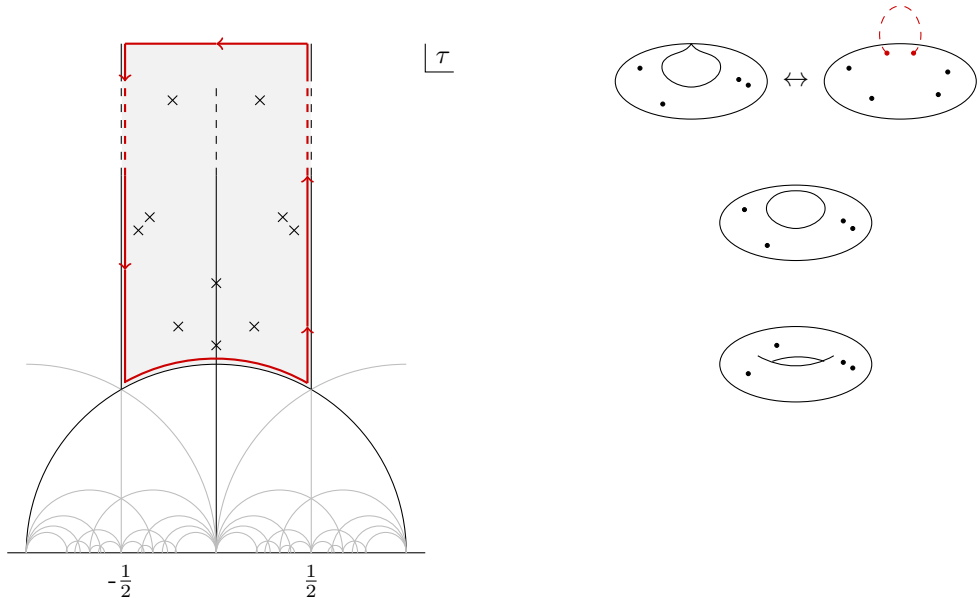  

In order to obtain a formula for the amplitude on the Riemann sphere,
we need the integrand $\cI_q:=\cI(\ldots|q)$ to be holomorphic as a function of
$q$ on the fundamental domain
$\mathcal{F}_\tau=\{|\tau|\geq 1, \Re (\tau) \in[-1/2,1/2]\}$ for the modular
group.  In the case of the $n$-gon conjecture investigated below, 
$\cI_q=1$, and holomorphicity is obvious.
For supergravity, however, $\cI_q$ is a product of two one-loop analogues of CHY Pfaffians that in particular have many contributions of the form $1/ \theta_1(z_i-z_j)$, which provide potential poles when $z_i\rightarrow z_j$, and it is conceivable that as $q$ varies, these might lead to poles in $q$. However,  such poles are suppressed by the scattering equations for generic choices of the momenta.   As $z_i\to z_j$ for $i,j \in I$ and $I$ some subset of $1, \ldots , n$, $P_\mu$ is well approximated by its counterpart on the Riemann sphere near the concentration point, and it is easily seen that such factorisation of the $z_i$ can only occur if the corresponding partial sum of the momenta for $i\in I$ becomes null.  See \cref{sec:factorization} for a detailed discussion of the argument.  Thus, for generic momenta $k_i$, we cannot have $z_i\to z_j$ on the support of $k_i\cdot P(z_i)=0$, and so the integrand $\cI_q$ has no poles.  

 It was shown in
\cite{Adamo:2013tsa} that the holomorphicity of the supergravity
integrand at $q=0$ is a consequence of the GSO projection. For other
values of $q$, the possible poles in the integrand can only occur when
$z_i\rightarrow z_j$. However, the standard factorisation argument
\cite{Dolan:2013isa} implies that this can only
happen when the momenta are factorising, and hence non-generic. \\

The main argument is then the following global residue theorem in $q$:
\begin{align}
\cM^{(1)}_{SG}&=\int \cI_q\,d^d \ell \,  \frac{d q}{ q} \, \dbar \left(\frac{1}{2\pi i P^2(z_0')}\right)\prod_{i=2}^n \bar\delta(k_i\cdot P(z_i)) d z_i \nonumber \\
&= -\int \cI_q\,d^d \ell \,  \dbar\left(\frac{d q}{2\pi i q} \right)\,  \frac{1}{ P^2(z_0')}\prod_{i=2}^n \bar\delta(k_i\cdot P(z_i)) d z_i \nonumber \\
&= -\int \cI_0\,d^d \ell \,    \frac{1}{  P^2(z_0')}\prod_{i=2}^n \bar\delta(k_i\cdot P(z_i)) d z_i \,
\Big|_{q=0}\, .\label{eq6:pre-1-loop}
\end{align}
In the first line, we put $d\tau=dq/(2\pi i q)$ and inserted the definition of $\bar\delta (P^2(z_0'))$.  In the second line, we integrated by parts in the fundamental domain $\mathcal{F}$, yielding a delta function supported at $q=0$ that is then integrated out.  The boundary terms cancel because of the modular invariance.  This is equivalent to a contour integral argument in the fundamental domain $\mathcal{F}$, as in figure \cref{fund-dom}. The sum of the residues at the poles of $1/P^2(z_0' |q)$ simply gives the contribution from the residue at the top, $q=0$, since the contributions from the sides and the unit circle cancel by modular invariance.  

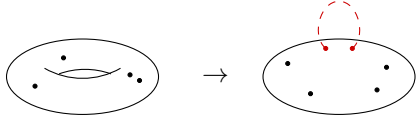
\begin{figure}[ht]
\begin{center}\vspace{-10pt} 
\begin{tikzpicture}[scale=2.5]
 \draw (3.6,1) circle [x radius=0.4, y radius=0.2];
  \draw [dashed,red3] (3.53,1.15) to [out=140, in=180] (3.6,1.4) to [out=0, in=40] (3.67,1.15);
  \draw [fill] (3.33,1.07) circle [radius=.3pt];
  \draw [fill] (3.45,0.91) circle [radius=.3pt];
  \draw [fill] (3.8,0.93) circle [radius=.3pt];
  \draw [fill] (3.85,1.05) circle [radius=.3pt];
  \draw [fill,red3] (3.53,1.15) circle [radius=.3pt];
  \draw [fill,red3] (3.67,1.15) circle [radius=.3pt];
  \node at (2.95,1) {$\rightarrow$};
  \draw (2.25,1) circle [x radius=0.4, y radius=0.2];
  \draw (2.05,1.045) arc [radius=0.4, start angle=240, end angle=300];
  \draw (2.4,1.015) arc [radius=0.4, start angle=70, end angle=110];
  \draw [fill] (2.15,1.1) circle [radius=.3pt];
  \draw [fill] (2,0.95) circle [radius=.3pt];
  \draw [fill] (2.55,0.98) circle [radius=.3pt];
  \draw [fill] (2.5,1.01) circle [radius=.3pt];
\end{tikzpicture}
\end{center}
\caption{The residue theorem maps from an elliptic curve to a nodal Riemann sphere.\label{fig6:degen-1-loop}}
\end{figure}

The fundamental domain for $z$ maps,
\begin{equation}
\sigma= \e^{2\pi i (z-\tau/2)}\,,
\end{equation}
 to $\{\e^{-\pi\Im \tau} \leq|\sigma|\leq\e^{\pi\Im \tau}\}$, with the
 identification $\sigma\sim q \sigma$. As  $q\rightarrow0$, we obtain
 $\sigma\in\CP^1$ with $\sigma=0,\infty$  identified, giving a double point
 corresponding to the pinching of $\Sigma_q$ at a non-separating
 degeneration as illustrated in figure \ref{fund-dom} and \cref{fig6:degen-1-loop}. We have
 $dz=\frac{d\sigma}{2\pi i\sigma}$ and, at $q=0$,
\begin{equation}
 \frac{\theta'_1 (z-z_i)}{\theta_1 (z-z_i)}dz=\frac{\pi}{\tan \pi
   (z-z_i)} \,dz =-\frac{d\sigma}{2\,\sigma} +
 \frac{d\sigma}{\sigma-\sigma_i}\,.
 \label{eq:thetaq0}
\end{equation}
Using momentum conservation we obtain 
\begin{equation}
\label{Psphere}
P(z)=P(\sigma)= \ell \,\frac{d\sigma}{\sigma} +\sum_{i=1}^n \frac{k_i\, d\sigma}{\sigma-\sigma_i}\, ,
\end{equation}
where here we have translated $\ell$ by\footnote{This was the extra term included in  $\ell$ to make $P$ doubly periodic at constant $\ell$.  In this limit on the Riemann sphere, it no longer plays a useful role and corresponds to taking $\sigma_0=e^{2i\pi z_0'}\rightarrow\infty$. } $\sum_{i} k_i \cot \pi
z_{i0 }$. 

If we now consider the meromorphic quadratic differential $P^2(\sigma)$, we find that it has double poles at $0$, $\infty$ (along with the usual simple poles at $\sigma_i$).  Defining 
\begin{equation}
S=P^2-\ell^2 \frac{\rd \sigma^2}{\sigma^2}\,,
\end{equation}
$S$ now has only simple poles. The vanishing of the residues of $S$ gives the {\em off-shell scattering equations} 
\begin{equation}\label{SE2n}
0=\mathrm{Res}_{\sigma_i}S =k_i\cdot P(\sigma_i) = \frac{k_i\cdot \ell}{\sigma_i} + \sum_{j\neq i}\frac{k_i\cdot k_j}{\sigma_i-\sigma_j}\, . 
\end{equation}
Again, the sum of the residues of $\sum\sigma_\alpha\sigma_\beta S$ must vanish with $\sigma_\alpha=(1,\sigma)$ in affine coordinates, and thus the equations for $i=2,\ldots , n$ imply the vanishing of the residues of $S$ at $\sigma_1$, $0$ and $\infty$. Therefore, any $n-1$ of the scattering equations imply all $n+2$. Since $S$ is holomorphic and of negative weight, it vanishes, and hence $P^2=\ell^2\rd \sigma^2/\sigma^2$. Geometrically, the contour integral argument thus takes the correlator off ambitwistor space. This is reflected in the interpretation of the residues of the new poles in $P$ as off-shell loop momenta. In particular, $P_\mu$ strongly resembles the tree-level solution, though with two off-shell legs with equal/opposite momentum associated to the node of the Riemann sphere.\smallskip

With these off-shell scattering equations, the one-loop formula \cref{eq6:pre-1-loop} becomes
\begin{equation}\label{1-loop}
\cM^{(1)}_{SG}= -\int \cI_0\,d^d \ell \, \frac{1}{\ell^2}\prod_{i=2}^n \bar\delta(k_i\cdot P(\sigma_i))\frac{d \sigma_i}{\sigma_i^2}\, ,
\end{equation}
where we have used the identity $\bar\delta(\lambda f)=\lambda^{-1}\bar\delta(f)$ to give $\bar{\delta}(k_i\cdot P(z_i)) dz_i=
\bar{\delta}(k_i\cdot P(\sigma_i))d\sigma_i/\sigma_i^2$. The formula \eqref{1-loop} is our new proposal for the supergravity loop integrand, with $\cI_0$ the $q=0$ limit of the ACS correlator.

The simpler `$n$-gon' conjecture presented in \cite{Casali:2014hfa} is given by $\cI_0=1$. For both this and supergravity, modular invariance is no longer an issue on the Riemann sphere, and the new formulae make sense in any dimension.  However, the link to a formula on the elliptic curve will only be valid in the critical dimension.

Since, at this stage, the formula does not require modular invariance any more, it is strongly suggestive of a framework widely applicable in massless scattering. Similar to tree-level, the one-loop off-shell scattering equations and the measure form the universal backbone, while the integrand $\cI_0$ specifies the theory. Moreover, note that \cref{1-loop} has manifestly the same complexity as a tree-level amplitude with $n+2$ external particles! This signifies a huge simplification in comparison to \cref{eq6:ampl-genus-1} that allows us to check the conjecture directly. Below, we will discuss several different theories, starting from the simplest $n$-gons and then proceeding with supersymmetric and non-supersymmetric gravity and Yang-Mills theory amplitudes.\\

Let us include a few more comments at this point clarifying the relation to string theory.
The integration by parts localises the amplitude on $q=0$. This is also the regime in
full string theory where one extracts the field theory or
$\alpha'\to0$ limit of the amplitudes. The difference here is that
this limit is obtained by application of the residue theorem, so we
are not discarding any terms, whereas in string theory we
would be projecting out the contribution of massive states running in the
loop by doing so.
At the moment it is unclear if the similarity between the method we
use here and string theory is just a consequence of the fact that both
strings are physical and hence factorise properly at the boundary of
the moduli space, or if this goes deeper.
In any case, the similarity between the $\alpha'\to0$ limit and our
integration by parts will allow us to reuse some standard technology from string
theory.\footnote{Several restrictions apply; there is no fully fledged
  well defined heterotic ambitwistor string (see
  \cref{sec:motiv-form-ambitw}), there are no winding modes
  which can become massless at self-dual radii in compactifications to
  enhance abelian to non-abelian gauge groups, see
  \cref{sec:dimens-reduct}, and it is yet unknown how to
  include the contribution of non-perturbative states of
  supergravity.}

\subsection{A first check: the \texorpdfstring{$n$}{n}-gon conjecture, partial fractions and shifts}\label{sec:n-gon}

The question arises as to how the $\ell$ appearing in \eqref{1-loop} relates to the loop momentum flowing in any given propagator.  We will see that the answer requires a new way of expressing one-loop amplitudes. The expression \eqref{1-loop} is a representation of the one-loop contribution to the scattering amplitude of a theory specified by $\cI_0$.  In this subsection, we consider the integrand $\cI_0=1$, which was conjectured in \cite{Casali:2014hfa} to give rise to a permutation sum of polygons.  
When $n=4$, the $n$-gon conjecture implies the supergravity conjecture \cite{Casali:2014hfa}.

For $n=4$, the off-shell scattering equations can be solved exactly with two solutions\footnote{This problem is identical to that arising in factorisation as studied in \cite{Casali:2014hfa} except that now $\ell$ is off-shell.  It was  conjectured there to have $(n-1)!-2(n-2)!$ solutions giving 2 at $n=4$.} given explicitly in \cref{4pt-soln}.   
After substituting into \eqref{1-loop}, this yields
\begin{equation}\label{4pt-result}
\hat{\cM}^{(1)}_4 = \frac{1}{\ell^2} \sum_{\rho \in S_4} 
\frac{1}{2\ell\cdot k_{\rho(1)}(2\ell\cdot (k_{\rho(1)}+ k_{\rho(2)})+2k_{\rho(1)}\cdot k_{\rho(2)})(-2\ell\cdot k_{\rho(4)})}\,,
\end{equation}
where we defined the loop integrand as $\hat{\cM}^{(1)}$,
\begin{equation}\label{4pt-result-int}
{\cM}^{(1)} = \int d^D\ell \;\hat{\cM}^{(1)}.
\end{equation}
This result is not obviously equivalent to the permutation sum of the boxes
\begin{equation}
  \label{eq:Ibox}
I^{1234} = \frac{1}{\ell^2(\ell+k_1)^2(\ell+k_1+k_2)^2(\ell-k_4)^2}\,,
\end{equation}
as the only manifest propagator in $\cM^{(1)}_4$ is the pre-factor ${1}/{\ell^2}$, and all the other denominator factors are linear in $\ell$. 
However,  the partial fraction identity\footnote{This
identity is easily proven by induction. Equivalently, it can also be derived from a residue theorem argument in the following integral \cite{Baadsgaard:2015twa}:
 \begin{equation}
   \label{eq:Di-partfracz}
   \frac{1}{2\pi i}\oint_{|z|=\epsilon}\frac1{z\,\prod_{i=1}^{n}(D_{i}-z)}\,.
 \end{equation}}
\begin{equation}
  \label{eq:Di-partfrac}
  \frac1{\prod_{i=1}^{n}D_{i}} =
  \sum_{i=1}^{n}\frac{1}{D_{i}\prod_{j\neq i}(D_{j}-D_{i})}\,,
\end{equation}
can be applied to a contribution such as \eqref{eq:Ibox}.   The right-hand-side of this identity is a sum of terms with a single factor of the type $D_i=(\ell+K)^2$, and several factors of the type $D_{j}-D_{i}=2\ell\cdot K +\mathcal{O}(\ell^0)$, where $K$ denotes a partial sum of the momenta. We then perform a shift in the loop momentum for each term such that the respective $D_i$ is simply $\ell^2$. Applying this procedure to the permutation sum, we precisely obtain $\hat{\cM}_4^{(1)}$.

We are now in a position to address the $n$-gon conjecture of \cite{Casali:2014hfa}. It states that $\cI$=1 corresponds to a permutation-symmetric sum of $n$-gons, which can be written as
\begin{align}
\hat{\cM}^{(1)}_{n\text{-gon}}= \frac1{\ell^2} \sum_{\rho\in S_n}
 \frac{1}{{\prod_{i=1}^{n-1} }\bigg(2\ell\cdot\sum_{j=1}^i k_{\rho_i} +\left(\sum_{j=1}^i
k_{\rho_i}\right)^2\bigg) } .
\label{formulangon}
\end{align}
This can be verified analytically at four points, using the explicit solutions to the scattering equations in \cref{4pt-soln}, and numerically at five points. We will see later in \cref{factorize} that this equality can be proved directly by factorisation arguments.

The $n=2$ and $3$ examples are also instructive.  The bubble (2-gon) example gives\footnote{Henceforth, we use a capital symbol $K$ to distinguish a possibly massive momentum.}
\begin{align}
  \frac1{\ell^2(\ell+K)^2} & =
  \frac{1}{\ell^2(2\ell\cdot K+K^2)} + \frac{1}{(\ell+K)^2(-2\ell\cdot K-K^2)} \nonumber \\
& \stackrel{\textrm{shift}}{\longrightarrow} \frac1{\ell^2}
\left( \frac1{2\ell\cdot K+K^2} + \frac1{-2\ell\cdot K+K^2} \right),\label{bubble}
\end{align}
where a shift $\ell \to \ell -K$ was applied to the second term. If $K$ is null, the bubble vanishes, which is also the result of dimensional regularisation. The triangle (3-gon) with massless corners, $k_1^2=k_2^2=k_3^2=0$, also vanishes:
\begin{align}
  \frac1{\ell^2(\ell+k_1)^2(\ell-k_3)^2} \stackrel{\textrm{shift}}{\longrightarrow} -\frac{ \ell\cdot(k_1+k_2+k_3)}{4\,\ell^2(\ell\cdot k_1)(\ell\cdot k_2)(\ell\cdot k_3)} =0.\label{triangle}
\end{align}

In general, the loop momentum dependence of typical integrands is not restricted to propagator factors in the denominator, and numerators of theories like gauge theory or gravity depend on $\ell$. The loop momentum in the numerators should thus also be shifted.
 For more general amplitudes, this can be achieved with a
 shift in the loop momentum together with a contour integral argument, and this has been explored and considerably
 generalised in  \cite{Baadsgaard:2015twa} and reviewed in \cref{sec:factorization}.

\section{Supersymmetric theories}
\label{sec:supersymm-theor-1}

Supergravity and Yang-Mills one-loop amplitudes can be expressed on the Riemann sphere using different choices for $\cI_0$ in \eqref{1-loop}.  While the former are readily derived from the type II RNS ambitwistor string, the Yang-Mills one cannot be derived from a full ambitwistor string calculation due to the corrupted gravity in the heterotic model (see some motivational comments in \cref{sec:motiv-form-ambitw}). We show that these
integrands pass several non-trivial consistency checks, and demonstrate that they factorise on the correct poles in \cref{sec:factorization}.

\subsection{Supergravity}
\label{sec:supergravity}
Let us start by recalling the form of genus-one graviton amplitudes in
ambitwistor string, as derived by ACS in \cite{Adamo:2013tsa} and reviewed in \cref{sec6:review_loop}.
As in the usual RNS string, the worldsheet correlator incorporates a
GSO projection to remove the unwanted states. The integrand $\cI_q$, derived from the worldsheet correlator of the vertex operators, is a sum over spin structures on the torus.  The odd-odd spin structure gives a fermionic 10-dimensional
zero-mode integral that leads to a $10$-dimensional Levi-Civita $\epsilon$ symbol.  This will vanish if all the polarisation data and external kinematics are restricted to $7$ dimensions or less.\footnote{In string theory, the $\epsilon_{10}\epsilon_{10}$ term is actually an $\epsilon_{8}\epsilon_{8}$ one, see for instance \cite{Peeters:2000qj}. It is not clear how this situation transposes to the ambitwistor string.} For simplicity we will assume this in the following\footnote{In doing this, we miss the term that leads to the Green-Schwarz anomaly \cite{Mafra:2012kh}.} and focus only on the even spin structures labelled by $\pmb{\alpha}=2,3,4$.  With this, the ACS proposal for the amplitude explicitly reads as \eqref{elliptic-amp} with 
\begin{equation}
\cI_q:=\frac14\sum_{\pmb{\alpha};\pmb{\beta}=2,3,4}
	    (-1)^{\pmb{ \alpha}+ \pmb{\beta}} Z_{ \pmb{\alpha}; \pmb{\beta}}(\tau)
	\ {\rm Pf}(M_{ \pmb{\alpha}})\,{\rm Pf}(\widetilde{M}_{ \pmb{\beta}})\,,
\label{e:graveven}
\end{equation}
where the numerical factor of $1/4$ comes from the two GSO projections. The vertex operators are naturally a product of two factors, and since these two parts essentially decouple, the full correlator decomposes also as a product as follows:
\begin{equation}
  \mathcal{I}_q = \mathcal{I}^L_q\,\mathcal{I}^R_q, \qquad \text{with} \qquad \mathcal{I}^L_q=\frac12\left(Z_2\,\pf(M_2) -Z_3\,\pf(M_3) +Z_4\,\pf(M_4)\right) .
\label{defIsugra}
\end{equation}  
with an analogous definition for $\mathcal{I}^{R}$.
The matrices $M_{\pmb\alpha}$ and $\wt{M}_{\pmb\alpha}$ are generalisations of the CHY matrix, and
arise from a straightforward application of Wick's theorem to the one factor of the vertex operators in the spin structure $\pmb{\alpha}$. In particular, they depend on the kinematic data, the insertion points $z_i$ and the polarisation vectors (see also \cref{sec6:amplitudes_g1}), and thus the amplitude encodes the NS-NS states of supergravity via the polarisation tensors $\epsilon^{\mu\nu}_i=\epsilon_i^\mu\tilde{\epsilon}_i^\nu$, with the dilaton corresponding to the trace, the $B$-field to the skew and the graviton to the traceless symmetric part. Explicit expressions for the matrices and the partition functions were given in \cref{sec6:amplitudes_g1}, we will use these for the expansions detailed below.\\

Applying our contour integral argument to go from the torus to the nodal
Riemann sphere localises the amplitude on the limit $q\to0$. The partition
functions possess $1/\sqrt{q}$ poles which extract higher order
terms in the Szeg\H{o} kernels. Hence we need the following
$q$-expansions:
\begin{equation}
Z_2(\tau) = 16+O(q^{2}),\quad
Z_3(\tau)= \frac1{\sqrt{q}}+8+O(q),\quad
Z_4(\tau) = \frac1{\sqrt{q}}-8+O(q).
\label{eq:Z-q-exp}
\end{equation}
and
\begin{equation}
\begin{aligned}
S_1(z_{ij}|\tau) &\to \frac{1}{2}\,\frac{1}{\sigma_ i-\sigma_ j} \left(\sqrt{\frac{\sigma_ i}{\sigma_ j}}+ \sqrt{\frac{\sigma_ j}{\sigma_ i}}\right)  \sqrt{d\sigma_ i} \sqrt{d\sigma_ j}  ,
  \\ 
S_2(z_{ij}|\tau) &\to \frac{1}{2}\,\frac{1}{\sigma_ i-\sigma_ j} \left(\sqrt{\frac{\sigma_ i}{\sigma_ j}}+ \sqrt{\frac{\sigma_ j}{\sigma_ i}}\right)  \sqrt{d\sigma_ i} \sqrt{d\sigma_ j} ,
  \\ 
S_3(z_{ij}|\tau) &\to \left(\frac{1}{\sigma_ i-\sigma_ j}+\sqrt{q} \,\frac{\sigma_ i-\sigma_ j}{{\sigma_ i}{\sigma_ j}}\right) \sqrt{d\sigma_ i} \sqrt{d\sigma_ j} ,
  \\ 
S_4(z_{ij}|\tau) &\to \left(\frac{1}{\sigma_ i-\sigma_ j}-\sqrt{q} \,\frac{\sigma_ i-\sigma_ j}{{\sigma_ i}{\sigma_ j}}\right) \sqrt{d\sigma_ i} \sqrt{d\sigma_ j} .
\end{aligned}\label{eq:szego-limit}
\end{equation}
in terms of the coordinates $\sigma=e^{2\pi i (z-\tau/2)}$.
The limit of $P(z_i)$ required for the components $C_{ii}$ was already
given in \eqref{Psphere}. 
The $q=0$ residue of \eqref{defIsugra} is then given by 
\begin{equation}
  \begin{aligned}
  \label{eq:sugra-sphere}
   \mathcal{I}^{L}=\frac{1}{2\sqrt{q}} \left( \pf(M_{3})\big|_{q^{0}} -
    \pf(M_{4})\big|_{q^{0}} \right)+
\frac12\left(\pf(M_{3}) \big|_{\sqrt{q}}+\pf(M_{4}) \big|_{\sqrt{q}}\right)+\\
  4\left(\pf(M_{3}) \big|_{q^0}+\pf(M_{4}) \big|_{q^0}-2\pf(M_{2} )\big|_{q^0}\right)+O(\sqrt{q})\,,
  \end{aligned}
\end{equation}
where the symbol $(\cdots)|_{q^r}$ with $r=0,1/2$ denotes
the coefficient of $q^r$ in the Taylor expansion around
$q=0$.\footnote{In the original ACS paper, the $O(\sqrt{q})$ were not
    included in the analysis of the factorisation channel.}

Some simplifications occur at this stage. Firstly, it is easy to see from \cref{eq:szego-limit} that 
\begin{equation}
  \pf(M_{3})\big|_{q^{0}} = \pf(M_{4})\big|_{q^{0}}\,,
  \label{eq:m3m4q0}    
  \end{equation}
  which reflects the projection of the ambitwistor string ``NS--tachyon''
  (we come back on this later). Furthermore, we also have that
\begin{equation}
  \pf(M_{3})\big|_{\sqrt{q}} = -\pf(M_{4})\big|_{\sqrt{q}}\,.
  \label{eq:m3m4q12}    
  \end{equation}
Using the two previous identities finally leads to the following expression for the full {\it one-loop supergravity integrand} for any number of external particles:
\begin{equation}
  \label{eq:IL0-sugra}
  \mathcal{I}_{0}^{L} = \pf(M_{3}) \big|_{\sqrt{q}}+
  8\left(\pf(M_{3}) \big|_{q^0}-\pf(M_{2} )\big|_{q^0}\right)\,.
\end{equation}

Let us comment briefly on the complexity of this object. As observed above, due to the simplicity of the scattering equations on the nodal sphere, the amplitude has a similar complexity to an $n=2$ point tree-level amplitude. Moreover, while the structure of $\mathcal{I}_{0}$ may initially appear to be quite complicated compared to the extreme simplicity of one-loop maximal supergravity
integrands, it simplifies considerably due to standard stringy theta function identities \cite{Casali:2014hfa}, see also \cite{Tsuchiya:1988va,Broedel:2014vla} for a discussion in string theory.
The simplest identities involve products of up to three Szeg\H{o} kernels,
\begin{equation}
\sum_{\alpha=2,3,4} (-1)^\alpha Z_\alpha \prod_{r=1}^m S_\alpha (w_{(r)}|\tau) = 0, \quad 
\text{for} \quad m=0,1,2,3,
\end{equation}
where the $w_{(r)}$ can be arbitrary. At $n=0$, this is the well known
Jacobi's identity
$\theta_{2}(0,\tau)^{4}-\theta_{3}(0,\tau)^{4}+\theta_{4}(0,\tau)^{4}=0$. For
$m>3$, the analogous identities are valid only for
\begin{equation}
\label{sumztheta}
w_{(1)}+\ldots+w_{(m)}=0.
\end{equation}
Let us consider the case $m=4$. In our application, the condition \cref{sumztheta} on
the $w_{(r)}$ is naturally achieved by the set
$(z_{ij},z_{jk},z_{kl},z_{li})$, and the corresponding identity is
\begin{equation}
  \sum_{\alpha=2,3,4} (-1)^\alpha Z_\alpha \prod_{r=1}^4 S_\alpha
  (w_{(r)}|\tau)  = (2\pi)^4 \,,
\end{equation}
where we have omitted the global form degree $dz_i dz_j dz_k dz_l$.
Applied to \eqref{defIsugra}, these identities tell us that
$\mathcal{I}^L$ is a constant for four-point scattering
\cite{Casali:2014hfa}. This follows from the structure of the
Pfaffians, or equivalently from the structure of the vertex
operators: as in string theory, only the terms with 8 $\psi$'s or more
contribute. At $n$ points, each term in $\Pf(M_\alpha)$ is a product
of $m$ Szeg\H{o} kernels of type $\alpha$ and $m-n$ factors
$C_{ii}$. The Szeg\H{o} kernels of type $\alpha$ appear with arguments
which precisely satisfy the condition \eqref{sumztheta}. At four
points, the sum over spin structures ensures that no $C_{ii}$
contributes, as $m<4$ for those terms, whereas the $m=4$ identity
implies that $\mathcal{I}^L$ is a constant. For $n>4$, the sum over
spin structures ensures that there are no terms with more than $n-4$
factors of the type $C_{ii}$. See \cite{Tsuchiya:1988va} for
identities up to $n=7$.
Since the loop momentum enters explicitly in $\mathcal{I}^L$ only
through $C_{ii}$, this means that $\mathcal{I}^L$ is a polynomial of
order $n-4$ in the loop momentum, which is always contracted with a
polarisation vector. This discussion holds for any value of $\tau$. In
the limit $q\to 0$ ($\tau\to i\infty$) of interest here, the Riemann identities become
algebraic identities, and can be easily checked at low multiplicity.\\

With few external particles, the following consistency checks have been performed on \cref{1-loop} with \cref{eq:IL0-sugra}: for $n=4$, $\mathcal{I}$ is a constant for any $q$, giving the expected $t_8t_8R^4$ kinematic tensor \cite{Casali:2014hfa}, and the $n$-gon results above suffice to give the correct answer. For $n=5$,  the integrand $\mathcal{I}_0$ depends on the $\sigma_ i$ and the loop momentum. The amplitude can be written in terms of pentagon and box integrals, and we can apply the shift procedure above to connect to our results, yielding 
\begin{align*}
\label{eq:5ptint}
& \hat{\cM}_5^{(1)}=\frac{1}{32\,\ell^2} \sum_{\sigma\in S_5}
\frac{1}{\prod_{i=1}^4 \big(\ell\cdot\sum_{j=1}^i k_{\sigma_i} +\frac{1}{2} (\sum_{j=1}^i k_{\sigma_i})^2\big)} 
\nonumber \\
& \times 
\left( N^5_{\sigma,\ell} \,+ 
\frac{1}{2} \sum_{i=1}^4 N^{\text{box}}_{\sigma_i\sigma_{i+1}} \,\frac{\ell\cdot\sum_{j=1}^i k_{\sigma_i} +\frac{1}{2} (\sum_{j=1}^i k_{\sigma_i})^2}{k_{\sigma_i}\cdot k_{\sigma_{i+1}}} \,   \right).
\end{align*}
The supergravity numerators $N^5$ and $N^{\text{box}}$ are the square of the gauge theory numerators given in \cite{Carrasco:2011mn} or \cite{Mafra:2014gja}, which satisfy the colour-kinematics duality \cite{Bern:2008qj,Bern:2010ue}.  This formula precisely matches that from the off-shell scattering equations at 5 points numerically, see \cref{sec:checks} for details.

\subsection{Super-Yang-Mills theory}
\label{sec:super-yang-mills}
The supergravity amplitude was derived in \cref{sec:supergravity} from the genus-one ambitwistor string expression of \cite{Adamo:2013tsa}. However, a Yang-Mills analogue of the latter on the torus is not known, despite the progress in formulating an ambitwistor string version of gauge theory at tree level \cite{Mason:2013sva,Geyer:2014fka,Casali:2015vta,Adamo:2015gia}, see \cref{chapter3}. Nevertheless, a proposal for super Yang-Mills amplitudes can be given, using the tree-level integrand as a motivation and relying on the relation between gauge theory and gravity.

At tree level, CHY \cite{Cachazo:2013hca} found that the expression for the gauge theory amplitude is obtained from the supergravity one by substituting one Pfaffian by a Parke-Taylor factor. The fact that gauge theory has only one Pfaffian, depending on a set of polarisation vectors ($\epsilon_i^\mu$), while gravity has two Pfaffians, each depending on a different set of polarisation vectors ($\epsilon_i^\mu$ and $\tilde\epsilon_i^\mu$), is a clear manifestation of gravity as a `square' of gauge theory, in agreement with the Kawai-Lewellen-Tye relations \cite{Kawai:1985xq} and with the Bern-Carrasco-Johansson (BCJ) double copy \cite{Bern:2008qj,Bern:2010ue}. At loop-level, the BCJ double copy is known to hold at one-loop in a variety of cases, including certain classes of amplitudes at any multiplicity \cite{Mafra:2012kh,Boels:2013bi,Bjerrum-Bohr:2013iza,Mafra:2014gja,He:2015wgf}, so it is natural to propose that one-loop formulae based on the scattering equations will also exhibit this property. We therefore propose that the {\it full  one-loop super Yang-Mills integrand} is given by\footnote{This gives the planar (single-trace) contribution to the amplitude. At one loop, the non-planar (double-trace) contribution is determined by the planar part for any gauge theory involving only particles in the adjoint representation of $SU(N_c)$ \cite{Bern:1994zx}.}
\begin{equation}
\mathcal{I}^{SYM}= \mathcal{I}^L_0\,\mathcal{I}^{PT}\,,
\label{defIsym}
\end{equation}
where $\mathcal{I}^L_0$ is defined in \eqref{eq:IL0-sugra}. The one-loop analogue of the Parke-Taylor factor is conjectured to be
\begin{equation}
\label{eq:PTloop-def}
\mathcal{I}^{PT}=\sum_{i=1}^n \frac{\sigma_{\ell^+\,\ell^-}}{\sigma_{\ell^+\,i}\sigma_{i\, i+1}\sigma_{i+1\, i+2}\ldots \sigma_{i+n\,\,\ell^-}}\, ,
\end{equation}
where $\sigma_{\ell^+}$ and $\sigma_{\ell^-}$ represent the pair of insertion points of the loop momentum, and where we identify the labels $i\sim i+n$. In our choice of coordinates used above, we have fixed $\sigma_{\ell^+}=0$ and $\sigma_{\ell^-}=\infty$, so that
\begin{equation}
\mathcal{I}^{PT}= - \sum_{i=1}^n \frac1{\sigma_{i}\sigma_{i\, i+1}\sigma_{i+1\, i+2}\ldots \sigma_{i+n-1\,i+n}}\, .
\end{equation}
At four points, $\cI^L_0$ is constant as mentioned above and factors out of the moduli integral.
Moreover, this proposal \cref{defIsym} has been checked numerically at both four and five points, see \cref{sec:checks} for details.

In \cref{sec:motiv-form-ambitw}, we present a motivation for our conjecture based on the
heterotic ambitwistor models.

\section{Non-supersymmetric theories}
\label{sec:non-supersymm-theor}

In this section, we describe new formulae for Yang-Mills theory and
gravity amplitudes in the absence of supersymmetry. The main tool in
arriving at these formulae is the detailed analysis of the sum over
spin structures (or GSO sum), which was part of the formulae for supergravity and
super Yang-Mills theory presented in \cref{sec:supersymm-theor-1}.

On the torus, these GSO sectors correspond to the various states
propagating in the loop. Once taken down to the sphere, we will see how
they provide amplitudes with $n$ external on-shell gravitons (or
gluons) and additional NS-NS, R-NS, NS-R or R-R additional
off-shell states (resp. NS or R), running in the loop.  
In particular, we are able to see that the $M_2$ contribution in $\eqref{eq:IL0-sugra}$ corresponds to the Ramond sector.  Furthermore the $M_3$ contributions naturally combine as a reduced Pfaffian of an $(n+2)\times (n+2)$ matrix in which the number of NS states running in the loop can be chosen at will.  

Taken individually, these one-loop amplitudes are non-supersymmetric. Using
these building blocks, one can engineer various types of
amplitudes.
For gravity, we discuss both NS-NS gravity (graviton, dilaton,
B-field) and pure Einstein gravity (graviton only).
We later show that our formulae match the known 4-point one-loop amplitudes in Yang-Mills theory and gravity, in a certain helicity sector.

A subtlety that arises however is that a class of degenerate solutions to the scattering equation becomes non-trivial (and in fact potentially divergent) for these non-supersymmetric amplitudes, as described by \cite{Baadsgaard:2015hia,
He:2015yua} for the bi-adjoint scalar theories.  So we first rephrase the scattering equations in a more SL$(2,\C)$ invariant manner to give a less degenerate formulation of these solutions.  In the next section, we will see that the contribution of these degenerate solutions is finite for our proposed formulae, and can furthermore be discarded without changing the integrated amplitude.  

In \cref{sec:factorization}, we present a full systematic proof for the non-supersymmetric one-loop amplitudes proposed in this section, relying on the recently proposed Q-cut formalism \cite{Baadsgaard:2015twa}.

\subsection{General form of the one-loop scattering equations}
\label{sec:general-form-one}
Before proceeding, we rewrite our previous expressions in
order to use their different building blocks for non-supersymmetric
theories. The reason for this, as pointed out in \cite{He:2015yua}, is
that the one-loop scattering equations on the sphere possess, in their
general form, more solutions than manifest from
\eqref{SE2n}.  We used part of the SL$(2,\mathbb{C})$
freedom on the Riemann sphere to fix the positions of the
loop-momentum insertions at $\sigma_{\ell^+}=0$ and
$\sigma_{\ell^-}=\infty$ as was natural from the
degeneration of the torus into a nodal Riemann sphere. However there are  extra solutions to the scattering
equations for which $\sigma_{\ell^+}=\sigma_{\ell^-}$ with the remaining
$\sigma_i$ satisfying the tree-level scattering equations. These solutions do arise in the previous gauge fixing with $\sigma_1=1$ as all the marked points collide, $\sigma_i\rightarrow\sigma_1$, but this gauge is much less convenient for these solutions.  We will see in \cref{sec:factorization} that these
extra solutions do not contribute to the formulae for maximal
supergravity and super Yang-Mills theory given in \cref{sec:supersymm-theor-1}
and reviewed above, but do contribute for generic theories, e.g. the
bi-adjoint scalar theory. As discussed in \cite{He:2015yua}, the total
number of solutions contributing is $(n-1)!-(n-2)!$, of which
$(n-1)!-2(n-2)!$ are the `regular' solutions considered in
\eqref{SE2n}, and $(n-2)!$ are the `singular' solutions for which
$\sigma_{\ell^+}=\sigma_{\ell^-}$.

Hereafter, we will write the one-loop formulae based on the general scattering equations as
\begin{equation}\label{1-loopgeneral}
\cM^{(1)}= -\int d^d \ell \, \frac{1}{\ell^2} \int 
\frac{d\sigma_{\ell^+} d\sigma_{\ell^-} d^n\sigma}{\text{vol}\,G} \;\; \hat{\cI}
\;\bar\delta(\text{Res}_{\sigma_{\ell^+}} S) \bar\delta(\text{Res}_{\sigma_{\ell^-}} S) \prod_i{} \;\bar\delta(\mathrm{Res}_{\sigma_i}S)\,, 
\end{equation}
where we should not fix the positions of both $\sigma_{\ell^+}$ and $\sigma_{\ell^-}$ in choosing the $G=\SL(2,\mathbb{C})^2$ gauge, to avoid losing the `singular' solutions. Since
\begin{equation}
 P=   \left( \frac{\ell}{\sigma-\sigma_{\ell^+}} - \frac{\ell}{\sigma-\sigma_{\ell^-}} + \sum_{i=1}^n 
\frac{k_i}{\sigma-\sigma_i} \right) d\sigma\,,
\end{equation}
and
\begin{equation}
S= P^2-  \left( \frac{\ell}{\sigma-\sigma_{\ell^+}} - \frac{\ell}{\sigma-\sigma_{\ell^-}}\right)^2 d\sigma^2,
\end{equation}
the scattering equations take the form
\begin{subequations}
\begin{align}
\mathrm{Res}_{\sigma_i}S = k_i\cdot P(\sigma_i) = \frac{k_i\cdot
  \ell}{\sigma_i-\sigma_{\ell^+}} -\frac{k_i\cdot
  \ell}{\sigma_i-\sigma_{\ell^-}} + \sum_{j\neq i}\frac{k_i\cdot
  k_j}{\sigma_i-\sigma_j} =0\, ,  \\
\mathrm{Res}_{\sigma_{\ell^-}}S = -\sum_{i}\frac{\ell\cdot k_j}{\sigma_{\ell^-}-\sigma_i} =0\, ,\\
\mathrm{Res}_{\sigma_{\ell^+}}S = \sum_{i}\frac{\ell\cdot k_j}{\sigma_{\ell^+}-\sigma_i} =0\, .
\end{align}
\end{subequations}
In the formula \eqref{1-loopgeneral}, the quotient by the volume of SL$(2,\C)$ denotes that only $n-3$ of these equations should be enforced (with those at $\sigma_{\ell^\pm}$ now on an equal footing with the external particle insertions $\sigma_i$).  The three remaining scattering equations hold by the same reasoning as in \cref{sec6:torus-to-RS}.

\subsection{General form of the one-loop integrand}
The interesting part of formula \eqref{1-loopgeneral} is the quantity $\cI$ specifying the theory. We introduced
\begin{equation}
\hat{\cI} = \frac1{(\sigma_{\ell^+\,\ell^-})^4} \, \cI\,,
\label{ihatdef}
\end{equation}
so that $\hat{\cI}$ has the same $\SL(2,\mathbb{C})$ homogeneity in $\{\sigma_{\ell^+},\sigma_{\ell^-},\sigma_i\}$, as required by the integration, whereas $\cI$ has zero weight in $\{\sigma_{\ell^+},\sigma_{\ell^-}\}$. The $n$-gon formula now corresponds to
\begin{equation}
\cI^{n-\text{gon}} = \prod_{i=1}^n \left( \frac{\sigma_{\ell^+\,\ell^-}}{\sigma_{i\,\ell^+}\,\sigma_{i\,\ell^-}} \right)^2\,.
\end{equation}
The relation to the $n$-gon representation in \cite{He:2015yua} follows from the identity
\begin{equation}
\sum_{\alpha\in S_n} \, \frac{\sigma_{\ell^+\,\ell^-}}{ \sigma_{\ell^+\,\alpha(1)}\,
\sigma_{\alpha(1)\,\alpha(2)} \ldots \sigma_{\alpha(n-1)\,\alpha(n)}\, \sigma_{\alpha(n)\,\ell^-}} =
\prod_{i=1}^n  \frac{\sigma_{\ell^+\,\ell^-}}{\sigma_{\ell^+\,i}\,\sigma_{i\,\ell^-}}\,,
\end{equation}
derived by  induction and  partial fractions.

For supergravity and for super Yang-Mills theory, we have
\begin{equation}
\cI^{SG} = \cI^L_0\, \cI^R_0 \qquad \text{and} \qquad \cI^{SYM} =  \cI^L_0 \,\cI^{PT} \,,
\end{equation}
where $\cI^{PT}$ was defined in \eqref{eq:PTloop-def}.
The quantities $\cI^L_0$ and $\cI^R_0$ are defined as in \eqref{eq:IL0-sugra}, but the Szeg\H{o} kernels in the matrices $M_\alpha$ are now expressed as
\begin{equation}
\begin{aligned}
S_2(z_{ij}|\tau) &\to \frac{1}{2}\,\frac{1}{\sigma_{i\,j}} \left(\sqrt{\frac{\sigma_{i\,\ell^+}\,\sigma_{j\,\ell^-}}{\sigma_{j\,\ell^+}\,\sigma_{i\,\ell^-}}}+ \sqrt{\frac{\sigma_{j\,\ell^+}\,\sigma_{i\,\ell^-}}{\sigma_{i\,\ell^+}\,\sigma_{j\,\ell^-}}} \right)  \sqrt{d\sigma_ i} \sqrt{d\sigma_ j}\, ,
  \\ 
S_3(z_{ij}|\tau) &\to \frac{1}{\sigma_{i\,j}} \left(1 +\sqrt{q} \;\frac{(\sigma_{i\,j}\,\sigma_{\ell^+\,\ell^-})^2}{\sigma_{i\,\ell^+}\,\sigma_{i\,\ell^-}\,\sigma_{j\,\ell^+}\,\sigma_{j\,\ell^-}}\right) \sqrt{d\sigma_ i} \sqrt{d\sigma_ j} \,,
\end{aligned}\label{eq:szegolimitsl2c}
\end{equation}
in the limit $q\to 0$.

Regarding the `singular' solutions to the scattering equations, it is
clear that they do not contribute in the $n$-gon case, since
$\hat{\cI}^{n-\text{gon}}\to 0$ for $\sigma_{\ell^+}\to
\sigma_{\ell^-}$. However, they do contribute in the case of the
non-supersymmetric Yang-Mills and gravity formulae to be presented
below, and some care is needed in their evaluation, due to the factor
$(\sigma_{\ell^+\,\ell^-})^{-4}$ in \eqref{ihatdef}. It is easy to see
that $\cI^{PT}\sim
\mathcal{O}\left((\sigma_{\ell^+\,\ell^-})^2\right)$. We will show in \cref{sec:powercounting} that
\begin{align}
\pf(M_2)|_{q^0}, \, \pf(M_3)|_{q^0} & \sim \mathcal{O}\left((\sigma_{\ell^+\,\ell^-})^2\right), \nonumber\\
\left(\pf(M_2)-\pf(M_3)\right)|_{q^0}, \, \pf(M_3)|_{\sqrt{q}} & \sim \mathcal{O}\left((\sigma_{\ell^+\,\ell^-})^3\right).
\end{align}
This is irrespective of the context there of taking the limit of large $\ell$, or of considering the `singular' solutions.
The contributions from these solutions to our formulae are therefore finite, as expected, and they vanish in the case of $\hat{\cI}^{SG}$ and $\hat{\cI}^{SYM}$.
Furthermore, we will see that the degenerate solutions do not contribute to the integrated amplitudes, and can thus be discarded.
It would, however, be useful to have an explicit formula for the limit.

\subsection{Contributions of GSO sectors and the NS Pfaffian}
\label{sec:analys-indiv-gso}
We now turn to the individual contributions of each GSO sector to the
supergravity amplitudes. This analysis is based on standard
string theory, the reader is referred to standard string textbooks such as that by Polchinski, or \cite{Adamo:2013tsa} for further details.

We work in dimension $d$ for $ d\leq10$ by
dimensional reduction from $d=10$. Since there are no winding modes, taking the radii of
compactification to zero is enough to decouple the Kaluza-Klein
modes, see \cref{sec:dimens-reduct} for further comments.
We consider first the ``left'' and ``right'' sectors
independently.\footnote{In string theory, this is justified by the
  chiral splitting of the worldsheet correlator whose dramatic consequences include  the KLT
  relations~\cite{Kawai:1985xq,BjerrumBohr:2010hn}. In the ambitwistor string this follows from KLT orthogonality
  \cite{Cachazo:2013gna}.} These
consists of $\mathcal{N}=1$ super Yang-Mills multiplets in $d=10$, and their
dimensional reduction is well known~\cite{Brink:1976bc}. The 10
dimensional vector $A^{(10)}_{\mu}$ splits into a $d$-dimensional
vector and $(10-d)$ scalars -- the fermionic case is discussed below.

The important point for the present analysis is that the partition
functions $Z_{a,b}$ as defined in \cref{eq6:pt-funs} arise from the respective sectors of the theory. More specifically, $a=0$ and $a=1$
correspond to the NS and R sectors, while $b=0,1$ correspond to the
periodicity of the boundary conditions. Thus
\begin{center}
  \begin{tabular}[h]{rcl}
  $Z_{3},Z_{4}$& $\longleftrightarrow$ & NS sector,\\
  $Z_{1},Z_{2}$ &  $\longleftrightarrow$ &R sector.
  \end{tabular}
\end{center}
As noted above, we will ignore the odd spin structure $Z_1$ here as it only contributes when the kinematics are in $d\geqslant 8$.  In analogy to the partition functions, the sum over spin structures in the correlators corresponds to particular sectors of the theory. So we define
\begin{align}
  \mathcal{I}_{NS} &=\pf(M_{3})\big|_{\sqrt{q}}+8\,
                                    \pf(M_{3})\big|_{q^{0}}\label{eq:chiral-NS}\,,\\ 
\mathcal{I}_{R\phantom{S}}&=8\pf(M_{2})\big|_{q^{0}}\label{eq:chiral-R}\,.
\end{align}

In $10$ dimensions, these correspond to chiral integrands for one vector
and one Majorana-Weyl fermion.
When we reduce  to $d<10$ dimensions, the problem that one
faces is how to decide which parts of the integrand
(\ref{eq:chiral-NS}) correspond to the $10-d$ scalars and which part
corresponds to the vector.
Following in particular the string theory analysis of
\cite{Tourkine:2012vx}, it is easy to identify first the scalar
contribution by reading off the (vanishing) coefficient
$\frac12(\pf(M_{3})\big|_{q^{0}}-\pf(M_{4})\big|_{q^{0}})$ of the
$1/\sqrt{q}$ pole in \cref{eq:sugra-sphere}.  Ignoring
the minus sign of the GSO projection, it corresponds to the (vanishing)
propagation of the unphysical scalar state
$\delta(\gamma_{1}) \delta(\gamma_{2})c\tilde c \exp(i k\cdot X)$.
With this we identify the scalar integrand as
$\pf(M_{3})\big|_{q^{0}}$ (recalling \cref{eq:m3m4q0}) and we can deduce  
\begin{subequations}
  \begin{align}
    \mathcal{I}^{L}_{\mathrm{scal}} &=  \pf(M_{3})\big|_{q^{0}}\label{eq:chiral-scalar}\\
    \mathcal{I}^{L}_{\mathrm{vect}} &= \pf(M_{3})\big|_{\sqrt{q}}+ (d-2)\,
                                      \pf(M_{3})\big|_{q^{0}}\label{eq:chiral-vector}\\
    \mathcal{I}^{L}_{\mathrm{ferm}} &=
                                      -c_{d}\,\pf(M_{2})\big|_{q^{0}}\, .\label{eq:chiral-fermion}
  \end{align}
\end{subequations}
The fermion integrand \cref{eq:chiral-fermion} comes
with a constant $c_{d}$ that follows from dimensional reduction of the 10d
Majorana-Weyl spinor, which produces an 8d Weyl spinor, four 6d
symplectic-Weyl spinors, and four 4d Majorana spinors. From
\cref{eq:sugra-sphere} we read off $c_{10}=8$, therefore we have
$c_{8}=8$, $c_{6}=2$, $c_{4}=2$. 

We can therefore obtain the reduced gravitational states in the loop by taking the tensor
product of the two sectors
\begin{subequations}
\label{eq:chiral-int-grav}
  \begin{align}
    \label{eq:grav-scalar}
    \mathcal{I}_{\mathrm{NS-NS-grav}}&=( \pf(M_{3})\big|_{\sqrt{q}}+ (d-2)\,\pf(M_{3})\big|_{q^{0}})^{2}
    \,,\\
    \label{eq:scalar}
    \mathcal{I}_\mathrm{scalar}&=(\pf(M_{3})\big|_{q^{0}})^{2}\,,\\
    \label{eq:vect}
    \mathcal{I}_\mathrm{vector}&=(\pf(M_{3})\big|_{q^{0}})( \pf(M_{3})\big|_{\sqrt{q}}+ (d-2)\,\pf(M_{3})\big|_{q^{0}})\,,
  \end{align}
\end{subequations}
from the NS-NS sector. Here, by NS-NS gravity 
in \cref{eq:grav-scalar} we mean Einstein gravity coupled to a B-field and a dilaton.
By convention, the squares are to be understood as incorporating a replacement of
the $\epsilon$'s by $\tilde \epsilon$'s in the second factor.

In the R-NS and NS-R sectors, we have
\begin{subequations}
  \begin{align}
    \label{eq:fermion}
    \mathcal{I}_\mathrm{fermion}&=-c_{d}\,\pf(M_{2})\big|_{q^{0}}\pf(M_{3})\big|_{q^{0}}\,,\\
    \label{eq:gravitino}
    \mathcal{I}_\mathrm{gravitino}&=-c_{d}\,\pf(M_{2})\big|_{q^{0}}\left( \pf(M_{3})\big|_{\sqrt{q}}+ (d-2)\,
    \pf(M_{3})\big|_{q^{0}}\right)\,.
  \end{align}
\end{subequations}
The R--R states in $d=10$ simply involve the square
\begin{equation}
  \label{eq:RR}
  \left(c_d \pf(M_{2})\big|_{q^{0}}\right)^{2}=\mathcal{I}_{\mathrm{RR}}\,,
\end{equation}
and it would be interesting to investigate this sector further.

With these interpretations of how different fields in the loops correspond to different  ingredients of the one-loop correlator, we can make the following proposals. 

\subsection{Pure YM and gravity amplitudes} 
Firstly, by adjusting the building
blocks in \eqref{eq:chiral-int-grav} in an appropriate way, we conjecture that a four-dimensional
one-loop pure gravity amplitude can be written as follows;
\begin{equation}
  \label{eq:pure-gravity}
  \mathcal{I}_{\mathrm{pure-grav}}^{(d=4)}=( \pf(M_{3})\big|_{\sqrt{q}}+ 2\,\pf(M_{3})\big|_{q^{0}})^{2}-2\,(\pf(M_{3})\big|_{q^{0}})^{2}\,,
\end{equation}
where the subtraction removes the two scalar degrees of freedom of the
dilaton and B-field.\footnote{The single degree of freedom of the
  B-field in four dimensions is the axion.} This subtraction is analogous to the
prescription of~\cite{Johansson:2014zca}, where scalars with fermionic
statistics were introduced to implement the BCJ double copy in
loop-level amplitudes of pure gravity theories.

Using the prescription reviewed in \cref{sec:super-yang-mills}, we
can also build four-dimensional pure YM amplitudes, by simply
multiplying the vector integrand of \cref{eq:chiral-vector}
with the Parke-Taylor factor \cref{eq:PTloop-def},
\begin{equation}
  \label{eq:pure-YM}
  \mathcal{I}_{\mathrm{pure-YM}}^{(d=4)}=
  ( \pf(M_{3})\big|_{\sqrt{q}}+ 2\,\pf(M_{3})\big|_{q^{0}})\,\mathcal{I}^{cPT}.
\end{equation}
These formulae \cref{eq:pure-gravity} and \cref{eq:pure-YM} generalise straightforwardly to $d$ dimensions, and explicit formulae for {\it pure Yang-Mills and gravity integrands} in arbitrary dimensions are given by
\begin{align}
 \mathcal{I}_{\mathrm{pure-YM}}&=
  ( \pf(M_{3})\big|_{\sqrt{q}}+ (d-2)\,\pf(M_{3})\big|_{q^{0}})\,\mathcal{I}^{cPT}\,, \label{eq:pure-YM-d}\\
  \mathcal{I}_{\mathrm{pure-grav}}&=( \pf(M_{3})\big|_{\sqrt{q}}+ (d-2)\,\pf(M_{3})\big|_{q^{0}})^{2}-\alpha\,(\pf(M_{3})\big|_{q^{0}})^{2}\,, \label{eq:pure-gravity-d}
\end{align}
where $\alpha=\frac{1}{2}(d-2)(d-3)+1$ counts the degrees of freedom of the B-field and the dilaton, see also \cref{sec:non-supersymm-theor-1}.\\

We will perform checks on these amplitudes in \cref{sec:gso-projection} and give a general proof in the next section. 
Note that although the standard string ideas discussed here are suggestive of the above proposals, they do not constitute a proof, so it is important to produce an independent proof\footnote{Amongst other issues, a point that is missing is
  that the abelian gauge groups do not get enhanced at self-dual radii
  of compactification as there are no winding modes that could become massless.} in \cref{sec:factorization}.

\paragraph{Pfaffian structure of the new amplitudes.}

A feature of the previous formulae is that they
provide information on the structure of tree-level amplitudes.
The finite residue that we extract at $q=0$ coincides with the residue
at the factorisation channel $q\simeq \ell^{2}\to0$. The only
difference between our expression and a ``single cut'' is the presence
of $1/\ell^{2}$ and the full $d$-dimensional integral
$\int d^{d}\ell$.
Therefore, we have  a variety of tree-level amplitudes with
$n+2$ (on-shell) particles, in a forward limit configuration where
$k_{n+1}=-k_{n+2}=\ell$ are off-shell, but traced over their polarization states.

One may therefore expect that  the integrands of the pure gravity and
Yang-Mills amplitudes \eqref{eq:pure-gravity} and (\ref{eq:pure-YM})
can be reformulated to look more like CHY Pfaffians. 

For Yang-Mills, this can be done as follows: the full supergravity integrands
$\hat{\cI}^{L,R}_0=\frac{1}{\sigma_{\ell^+\,\ell^-}}\cI^{L,R}_0$ can be expressed more compactly in terms of a single NS sector matrix $M_{NS}$,  defined explicitly below, as
\begin{equation}\label{eq:defInt_MNS}
 \hat{\cI}^{L,R}_0=\sum_r\pf'(M_{\text{NS}}^r)-\frac{c_d}{\sigma_{\ell^+\,\ell^-}^2}\pf(M_2)\,,
\end{equation}
where
$\pf'(M_{\text{NS}}^r)\equiv\frac{-1}{\sigma_{\ell^+\,\ell^-}}\pf({M_{\text{NS}}^r}_{(\ell^+\,\ell^-)})$, and the brackets $(\ell^+\,\ell^-)$ indicate that the rows and columns associated to $\ell^+$ and $\ell^-$ have been removed. In particular, this implies that
\begin{equation}
  \sum_r \pf'(M_{\text{NS}})=\pf(M_3)\,\big|_{\sqrt{q}}+(d-2)\pf(M_3)\,\big|_{q^0}\,.
\label{eq:PfNs-def}
\end{equation}
The matrix $M_{\text{NS}}^r$ is defined by
\begin{equation}\label{eq:defMNS}
 M_{\text{NS}}^r=\begin{pmatrix}A & -C^T\\ C & B \end{pmatrix}\,,
\end{equation}
and more specifically
\begin{equation}  \label{eq:defMNS_details}
\begingroup
\renewcommand*{\arraystretch}{1.4}
 M_{\text{NS}}^r=\left( \begin{array}{cccc:cccc}
    0 & 0 & \frac{\ell\cdot k_i}{\sigma_{\ell^+i}} & \frac{\ell\cdot k_j}{\sigma_{\ell^+j}} & -\epsilon^r\cdot P(\sigma_{\ell^+}) & \frac{\ell\cdot\epsilon^r}{\sigma_{\ell^+\ell^-}} & \frac{\ell\cdot\epsilon_i}{\sigma_{\ell^+i}} & \frac{\ell\cdot\epsilon_j}{\sigma_{\ell^+j}}\\
    & 0 & -\frac{\ell\cdot k_i}{\sigma_{\ell^-i}} & -\frac{\ell\cdot k_j}{\sigma_{\ell^-j}} & -\frac{\ell\cdot\epsilon^r}{\sigma_{\ell^-\ell^+}} & -\epsilon^r\cdot P(\sigma_{\ell^-}) & -\frac{\ell\cdot\epsilon_i}{\sigma_{\ell^-i}} & -\frac{\ell\cdot\epsilon_j}{\sigma_{\ell^-j}}\\
    &   & 0 & \frac{k_i\cdot k_j}{\sigma_{ij}} & \frac{k_i\cdot\epsilon^r}{\sigma_{\ell^+i}} & \frac{k_i\cdot\epsilon^r}{\sigma_{\ell^-i}} & -\epsilon_i\cdot P(\sigma_{i}) & \frac{k_i\cdot\epsilon_j}{\sigma_{ij}}\\
    &   &   & 0 & \frac{\epsilon^r\cdot k_j}{\sigma_{\ell^+j}} & \frac{k_j\cdot\epsilon^r}{\sigma_{\ell^-j}} & \frac{k_j\cdot\epsilon_i}{\sigma_{ji}} & -\epsilon_j\cdot P(\sigma_j)\\
    \hdashline
    & & & & 0 & \frac{d-2}{\sigma_{\ell^+\ell^-}} & \frac{ \epsilon^r\cdot\epsilon_i}{\sigma_{\ell^+i}} & \frac{ \epsilon^r\cdot\epsilon_j}{\sigma_{\ell^+j}}\\
    & & & &   & 0 & -\frac{\epsilon^r\cdot\epsilon_i}{\sigma_{\ell^-i}} & \frac{\epsilon^r\cdot\epsilon_j}{\sigma_{\ell^-j}}\\
    & & & &   &   & 0 & \frac{\epsilon_i\cdot \epsilon_j}{\sigma_{ij}}\\
    & & & &   &   &   & 0\\
  \end{array} \right)\,.
  \endgroup
\end{equation}
The sum runs over a basis of polarisation vectors $\epsilon^r$, and
$d$ denotes the space-time dimension. Note in particular that the reduced
Pfaffian is well-defined since this matrix has indeed co-rank
two: similar to the structure at tree-level, the vectors $(1,\dots,1,0,\dots,0)$
and $(\sigma_{\ell^+},\sigma_{\ell^-},\sigma_1,\dots,\sigma_n,0,\dots,0)$ span the
kernel of $M_{\text{NS}}^r$ on the support of the scattering
equations.

The matrix depends explicitly on the number of NS polarisation states running in the loop and we can adjust this to give different theories in different dimensions (we comment on dimensional reduction in \cref{sec:dimens-reduct}).\smallskip

The proof of \cref{eq:PfNs-def}, relies on standard properties of
Pfaffians, and the interested reader is referred to \cref{sec:MNS}. In
this form, the NS contribution to the integrand is very suggestive of
a worldsheet CFT correlator, and indeed it is not hard to see that
this Pfaffian arises from an off-shell correlator on the Riemann sphere, with two marked 
points associated to the loop momentum. The corresponding polarizations should be replaced by a photon propagator
in a physical gauge.

The gravity case uses also $M_{NS}$, and is treated in more
details later in \cref{sec:non-supersymm-theor-1}, when we discuss
the factorisation properties of these pure Yang-Mills and gravity
amplitudes. Basically, we simply decompose the difference of squares
in \cref{eq:pure-gravity} as a product.

To conclude this discussion, we note that the fermion integrand of
\cref{eq:fermion} for a two-fermion-$n$-graviton integrand seems to
arise naturally as a factorised product of Pfaffians.
Although amplitudes with fermions have been computed in
\cite{Adamo:2013tsa}, no closed form for higher-points amplitude is
known, partly because of the non-polynomial nature of the spin-field
OPE's in the RNS framework. While these difficulties can be circumvented and there exist tree-level expressions with external fermions in the literature obtained from the pure spinor formalism \cite{Mafra:2011nv,Gomez:2013wza},\footnote{See also \cite{Mafra:2015vca} for new efficient methods to integrate out the variables of the pure spinor superspace.} these do not lead to closed-form  formulae, and the connection to the Pfaffians found here is highly non-trivial.
It is possible that \cref{eq:fermion} is different from the
generic -- i.e. non-forward -- amplitude due to terms vanishing with
$\ell^{2}$. Nevertheless, this hints at some unexpected simplicity.

\section{Proof for non-supersymmetric amplitudes at one-loop}
\label{sec:factorization}

We now give a full proof of the formulae for one-loop amplitudes derived above for non-supersymmetric theories, i.e. the $n$-gons, bi-adjoint scalar theory, Yang-Mills and gravity.\footnote{In particular, the proof holds for both the NS sector (including the B-field and the dilaton for gravity) and the pure theories.}
There are two main ingredients in our proof.  The first is to identify the poles in our formulae arising from factorisation or bubbling of the Riemann sphere, which allows us to determine the location of the poles and their residues. Since the $1/\ell^2$ is already apparent, this analysis of factorisation will lead to the identification of the residue at two poles. The second is the theory of `Q-cuts' introduced in \cite{Baadsgaard:2015twa} that expresses  a general one-loop amplitude in terms of tree amplitudes that is perfectly adapted to the factorisation of our formulae (this is perhaps not completely surprising as their construction was motivated by our formulae). \smallskip 

As reviewed briefly in \cref{sec2:review_SE}, an amplitude must factorise in the sense that if a partial sum of the external momenta  $k_I=\sum_{i\in I}k_i$, where $I\subset \{1,\ldots, n\}$, becomes null, then there will be a pole corresponding to a propagator with momentum $k_I$ flowing through it.  Furthermore, the residue is the product of two tree amplitudes for the theory in question with external legs consisting of $\pm k_I$ and the elements of $I$ or its complement $\bar I$. We have seen that the scattering equations \cite{Cachazo:2013hca} relate the factorisation channels of the momenta precisely to factorisation channels of the Riemann surface, i.e. when a subset of the marked points collide. This concentration point  can then be blown up to give a bubbled-off Riemann sphere connected to the original at the concentration point, see  below (or \cite{Dolan:2013isa}).  \smallskip

Our scattering equations at one-loop will give worldsheet factorisation channels that lead to poles associated to loop momenta, but these are not immediately recognizable as loop propagators; they instead correspond to poles of the form of those in the sum of \eqref{formulangon}.  These however can be understood as naturally arising in the `Q-cuts' of \cite{Baadsgaard:2015twa}.  These are a systematic extension of the contour integral argument that leads to the partial fractions expansion of \eqref{eq:Di-partfrac} applicable to any one-loop integrand.  They follow from a two-step process.  The first follows the contour integral argument of \eqref{eq:Di-partfracz}.  Consider a one-loop integrand 
\begin{equation}
\cM(\ell,k_i,\epsilon_i)=\frac{N(\ell,\ldots)}{D_{I_1}\ldots D_{I_m}}\,,
\end{equation}
for the theory under consideration, where $N$ is a polynomial numerator, and $D_I=(\ell+k_I)^2$ a propagator.  We shift the loop momentum $\ell\rightarrow \ell + \eta$ where $\eta$ is in some higher dimension than the physical momenta and polarization vectors, so that the only shift that occurs in the invariants is $\ell^2\rightarrow \ell^2+z$, with all other inner products remaining unchanged.  One then runs the contour integral argument that expresses the amplitude as the residue at $z=0$ of $\cM(\ell+\eta,k_i,\epsilon_i)/z$ in terms of  the sum of the other residues of this expression.  Such residues arise at shifted propagators $1/(D_I+z)$ with poles at $z=-D_I$. One then shifts $\ell\rightarrow \ell-k_I$ in each of these new residues so that $z$ becomes $\ell^2$.  This gives a representation of a one-loop amplitude as a sum of terms of the form
\begin{equation}
\frac{1}{\ell^2}\left[
\frac{\tilde 
N(\ell)}{(2\ell\cdot k_{I_1} +k_{I_1}^2)
\ldots (2\ell\cdot k_{I_m}+ k_{I_m}^2)
}
\right]\,,
\end{equation}
giving a generalization of the partial fraction formulae of \eqref{eq:Di-partfrac}.  

In order  to interpret constituents of this expression as tree amplitudes, \cite{Baadsgaard:2015twa} considers a further contour integral argument with integrand
\begin{equation}
\frac{\cM(\alpha \ell)}{\alpha-1}\,,
\end{equation}
where $\cM(\ell)$ is now the expression with shifted $\ell$s obtained above.  The residue at $\alpha=1$ returns the original $\cM(\ell)$.  The residues  at zero and infinity can be discarded as they vanish in dimensional regularization.
It can then be argued \cite{Baadsgaard:2015twa} that the finite residues finally yield the `Q-cut' decomposition
\begin{equation}
 \cM\Big|_{\text{Q-cut}}\equiv\sum_{I}\cM_I^{(0)}(\dots,\tilde \ell_I, \tilde \ell_I+k_I)\frac{1}{\ell^2 \,(2\ell\cdot k_I+k_I^2)}\cM_{\bar I}^{(0)}(-\tilde \ell_I,-\tilde \ell_I-k_I,\dots)\,,
\end{equation}
where $\tilde \ell_I=\alpha(\ell +\eta)$, with $\alpha=-k_I^2/2\ell\cdot k_I$, $\eta^2=-\ell^2$, $\eta\cdot\ell=\eta\cdot k_i=0$, and $\cM_I^0$ and $\cM_{\bar I}^0$ are now tree amplitudes.

\begin{figure}[ht]
\begin{center}
 \begin{tikzpicture}[scale=3]
 \draw (1.75,1.2) to [out=30, in=180] (2.125,1.35) to [out=0, in=150] (2.5,1.2);
 \draw (1.75,0.8) to [out=330, in=180] (2.125,0.65) to [out=0, in=210] (2.5,0.8);
 \draw (1.75,1.2) -- (1.375,1.4);
 \draw (1.75,0.8) -- (1.375,0.6);
 \draw (2.5,1.2) -- (2.875,1.4);
 \draw (2.5,0.8) -- (2.875,0.6);
 \draw (1.75,1.07) -- (1.3,1.14);
 \draw (1.75,0.93) -- (1.3,0.86);
 \draw (2.5,1.07) -- (2.95,1.14);
 \draw (2.5,0.93) -- (2.95,0.86);
 \draw [fill, light-grayII] (1.75,1) circle [x radius=0.2, y radius=0.4];
 \draw (1.75,1) circle [x radius=0.2, y radius=0.4];
 \draw [fill, light-grayII] (2.5,1) circle [x radius=0.2, y radius=0.4];
 \draw (2.5,1) circle [x radius=0.2, y radius=0.4];
 \node at (0.72,1) {$\displaystyle \mathcal{M}^{(1)}\Big|_{\text{Q-cut}}= \,\sum_I$};
 \node at (2.125,1.5) {$\scriptstyle\ell$};
 \node at (2.125,0.5) {$\scriptstyle \ell+K_I$};
 \node at (1.75,1) {$\scriptstyle I$};
 \node at (2.5,1) {$\scriptstyle \overline{I}$};
 \node at (3.7,1) {$\displaystyle =\sum_I\frac{\cM_I^{(0)} \cM^{(0)}_{\bar I}}{\ell^2(2\ell\cdot K_I +K_I^2)}$};
\end{tikzpicture}
\end{center}
\caption{Representation of the amplitude as a sum over Q-cuts.}
\label{fig:q-cut}
\end{figure}  

We will see in this section that factorisation of the worldsheet in our formulae gives precisely these poles and residues for $\alpha=1$ and therefore precisely yield this decomposition, including the contributions from the residues at $\alpha=0,\infty$.\footnote{While this decomposition should strictly speaking be referred to as a `pre-Q-cut decomposition', we will for simplicity denote them as Q-cuts as well in the discussion below. Moreover, following the work presented here, the relation of the formulae to Q-cuts for the bi-adjoint scalar and gauge theory was explored further in \cite{Cachazo:2015aol}, which also gives a derivation of the formulae presented here from higher-dimensional tree-amplitudes.}
Moreover, similar factorisation arguments allow us to determine the UV behaviour of the amplitudes. These two results combine to prove our formulae by a multidimensional analogue of Liouville's theorem.

\begin{thm}{}\label{thm:fact}
Consider a subset $I\subset\{1,\ldots, n, \ell^+,\ell^-\}$ with one fixed point in $I$ and two in its complement $\bar I$. Suppose furthermore that we have solutions to the scattering equations behaving as $\sigma_i= \sigma_I + \varepsilon x_i+\cO(\varepsilon^2)$ for $i\in I$ and $\varepsilon \rightarrow 0$, with $x_i=O(1)$, $\sigma_{ij}=O(1)$ for $j\in\bar I$. 
Then we must also have  $\tilde{s}_I =O(\varepsilon)$ where 
 \begin{equation}
  \tilde{s}_I=\begin{cases}
               k_I^2 & \sigma_{\ell^+},\sigma_{\ell^-}\in\bar I\,,\\
               2\ell\cdot k_I + k_I^2 & \sigma_{\ell^+}\in I, \sigma_{\ell^-}\in \bar I\,.
              \end{cases}
 \end{equation}
 In this factorisation channel, our one-loop formulae $\cM^{(1)}$ on the Riemann sphere for n-gons, bi-adjoint scalar theory, Yang-Mills and gravity have poles at $\tilde{s}_I= 0$. For $\sigma_{\ell^+},\sigma_{\ell^-}\in\bar I$ (case I of figure \ref{fig:degen}), the residue of this pole is given by the separating degeneration of the nodal Riemann sphere 
 \begin{equation}
  \cM^{(1)}(\dots,-\ell,\ell)=\cM_I^{(0)}(\dots)\frac{1}{s_I}\cM_{\bar I}^{(1)}(\dots,-\ell,\ell) \,,
 \end{equation}
 and for $\sigma_{\ell^+}\in I, \sigma_{\ell^-}\in \bar I$ (case II of figure \ref{fig:degen}) by the Q-cut degeneration
 \begin{equation}\label{eq:Q-cut fact}
  \cM^{(1)}(\dots,-\ell,\ell)=\cM_I^{(0)}(\dots,\ell_I,\ell_I+k_I)\frac{1}{\ell^2}\frac{1}{\tilde{s}_I}\cM_{\bar I}^{(0)}(-\ell_I,-\ell_I-k_I,\dots)\,,
 \end{equation}
 where $\ell_I=\ell+\eta$, with $\eta^2=-\ell^2$, $\eta\cdot\ell=\eta\cdot k_i=0$.
\end{thm}

\begin{proof}
 Here, we outline the idea of the proof, all details will be developed in \cref{sec:fact_SE} and \cref{sec:fact_int}. The central observation is that poles in \cref{1-loop} occur only if a subset $I$ of the marked points approach the same marked point $\sigma_I$; so that $\sigma_i \rightarrow \sigma_I+\varepsilon x_i +\cO(\varepsilon^2)$ for $i\in I$.
This is conformally equivalent to a degeneration of the Riemann sphere into two components, connected by a node. All such poles receive contributions from both the measure and scattering equations. In particular, the question whether a pole occurs for a given integrand reduces to a simple scaling argument in the degeneration parameter $\varepsilon$, and we can straightforwardly identify the residues.

To be more explicit, for some $m\in I$ fix $\sigma_m= \varepsilon$ so that $x_m=1$ is the new fixed point on the $I$ component of the degenerate Riemann surface, then the measure and the scattering equations factorise as
 \begin{equation}
  d \mu\equiv\prod_{i=2}^n \bar{\delta}(k_i\cdot P(\sigma_i))d\sigma_i=\varepsilon^{2(|\cI|-1)}
  \frac{d\varepsilon}{\varepsilon}  \,\bar{\delta}(s_I+\varepsilon\mathcal{F})\,d\mu_I\, d\mu_{\bar I}\,.
\end{equation}
Moreover, \cref{sec:fact_int} provides details on how the integrands for $n$-gons, Yang-Mills theory and gravity factorise as well,
\begin{equation}
 \cI^{(1)}=\varepsilon^{-2(|\cI|-1)}\cI_I^{(0)}\cI_{\bar I}^{(0)}\,,
\end{equation}
where $\cI_{I,\bar I}^{(0)}$ depend only on on-shell momenta $k_i$ and the loop momentum in the on-shell combination $\ell_I=\ell+\eta$, $\eta^2=-\ell^2$. The full amplitude therefore factorises on the expected poles, and the residue gives the Q-cut factorisation described above;
\begin{equation}
 \cM^{(1)}=\int \cI^{(1)} d\mu=\int  \frac{d\varepsilon}{\varepsilon} \,\bar{\delta}(\tilde{s}_I+\varepsilon\mathcal{F})\,\cI_I^{(0)}\cI_{\bar I}\,d\mu_I^{(0)}\, d\mu_{\bar I}=\cM_I^{(0)}\frac{1}{\tilde{s}_I}\cM_{\bar I}^{(0)}\,.
\end{equation}
\end{proof}

\begin{thm}{}\label{lemma:powercounting}
 The amplitudes $\cM^{(1)}$ scale as $\ell^{-N}$ for $\ell\rightarrow\infty$, where
 \begin{center}
\begin{tabular}{l|c}
 theory & scaling $\ell^{-N}$\\ \hline 
 $n$-gon & $N=2n$\\
 supergravity & $N=8$\\
 super Yang-Mills & $N=6$\\
 pure gravity & $N=4$\\
 pure Yang-Mills & $N=4$\\
 bi-adjoint scalar & $N=4$
\end{tabular}
\end{center}
\end{thm}

\begin{proof} 
  This  follows from the fact that as $\ell\rightarrow \infty$,  the insertions of $\sigma_{\ell^+}$ and $\sigma_{\ell^-}$ must approach each other. This is conformally equivalent to a degeneration of the worldsheet into a nodal Riemann sphere with no further insertions and another Riemann sphere carrying all the external particles, see case III in figure \ref{fig:degen}. Moreover, this is also the configuration that corresponds to the singular/degenerate solutions described in the previous sections and so our analysis of fall-off in $\ell$ will also give information about the finiteness of the contributions from these degenerate solutions. 
 We  give the full details in \cref{sec:powercounting}.\end{proof}

\Cref{thm:fact} and \cref{lemma:powercounting} will now allow us to prove that the representation of one-loop amplitudes from the nodal Riemann sphere is equivalent to the Q-cut representation reviewed above.

\begin{thm}{}
 $\cM^{(1)}_{\mathrm{BS}}$, $\cM^{(1)}_{\mathrm{YM}}$ and $\cM^{(1)}_{\mathrm{gravity}}$ with the degenerate solutions omitted are representations of the one-loop amplitudes for the bi-adjoint scalar theory, Yang-Mills and gravity respectively.
\end{thm}
\begin{proof}
We use the fact that our formula must be rational in the external data and $\ell$, and that the only poles in our formulae arise at the boundary of the moduli space when the $\sigma_i$ come together.  This theorem  is then an immediate consequence of the correct factorisation on Q-cuts and the scaling behaviour in $\ell$: 
 \begin{equation}
    \Delta=\cM^{(1)}-\cM\Big|_{\text{Q-cut}}\,,
 \end{equation}
cannot have any further poles in $\ell$ by \cref{thm:fact}, so by a multidimensional analogue of Liouville's theorem, $\Delta$ has to be a constant. However, the non-trivial fall-off of $\cM^{(1)}$ and $\cM\Big|_{\text{Q-cut}}$ in $\ell$ implies that this constant has to vanish, and thus
 \begin{equation}
    \cM^{(1)}=\cM\Big|_{\text{Q-cut}}\,.
 \end{equation}
 In particular, the degenerate solutions to the scattering equations do not contribute to the  singularities that give rise to the Q-cuts, since these all arise from case II of figure \ref{fig:degen}, whereas the degenerate solutions are all case III of figure \ref{fig:degen}.   
\end{proof}

\subsection{Factorisation I - scattering equations and measure}\label{sec:fact_SE}
As discussed above,\footnote{See also the tree-level discussion in \cref{sec2:review_SE}.} poles of $\cM^{(1)}$ only occur  when a subset of the marked points
(possibly including $\sigma_{\ell^+}$ or $\sigma_{\ell^-}$) approach the same point, giving rise to a degeneration of the Riemann sphere into a pair of Riemann spheres connected at a double point. The scattering equations then imply that this pole is associated with a partial sum of the momenta becoming null. This is an extension of the discussion at tree-level in \cref{sec2:review_SE}, relating the boundary of the moduli space of marked Riemann surfaces to factorisation channels of the amplitude.\smallskip

Let $I$ be a subset of $\{\ell^+,\ell^-, 1,\ldots ,n\}$ that contains just one of the fixed points that we shall denote $\sigma_I$.   Moreover, let $k_I= \sum_{i\in I}k_i$, $k_{\bar I} =\sum_{i\in \bar I} k_i$ with $k_0=\ell_I$, $k_{n+1}=-\ell_I$.   Consider now a solution to the scattering equations implying a degeneration of the Riemann surface\footnote{Two of the three original fixed points must be in $\bar I$ and just one in $I$ to obtain a stable degeneration  as we cannot make two of the fixed points  approach each other but $\bar I$ cannot contain three as, after factorisation, it will also have the fixed point $\sigma_I$ which would be too many.} so that for $i\in I$ 
\begin{equation}
\sigma_i \rightarrow \sigma_I+\varepsilon x_i +\cO(\varepsilon^2)\,,
\end{equation} 
for some small parameter $\varepsilon$, with $x_I=0$, $x_m=1$ for some $m\in I$ and $x_i =O(1)$ for all other $i\in I$.

 We first wish to see that with these assumptions, the scattering equations imply that $\tilde s_I =O(\epsilon)$.  Firstly we have
\begin{subequations}
\begin{align}
 &P(\sigma)=\frac{1}{\varepsilon}P_I(x)+\tilde{P}_{\bar I}(\sigma_I)+\cO(\varepsilon) && \sigma=\sigma_I+\varepsilon x+\cO(\varepsilon^2)\,,\\
 &P(\sigma)=P_{\bar I}(\sigma)+\cO(\varepsilon) && \sigma\neq\sigma_I+\varepsilon x+\cO(\varepsilon^2)\,,
\end{align}
\end{subequations}
where 
\begin{equation}
P_I(x)=\sum_{i\in I}\frac{k_i}{x-x_i} 
\, , \qquad \tilde{P}_{\bar I}(\sigma)=\sum_{i\in \bar I}\frac{k_i}{\sigma-\sigma_i}\,,\qquad P_{\bar I}(\sigma)=\sum_{p\in \bar I}\frac{k_p}{\sigma-\sigma_p}-\frac{k_I}{\sigma-\sigma_I}\,.
\end{equation}
Thus the scattering equations give\footnote{Setting $k_0=\ell_I$, $k_{n+1}=-\ell_I$, the scattering equations for the marked points $\sigma_{\ell^\pm}$ are
\begin{equation*}
 0=\mathrm{Res}_{\sigma_{\ell^\pm}}=\pm\sum_i \frac{\ell_I\cdot k_i}{\sigma_{\ell^\pm i}}=\pm\ell_I\cdot P(\sigma_{\ell^\pm})\,,
\end{equation*}
so the conclusions hold for $\sigma_{\ell^\pm}$ also.}
\begin{subequations}\label{eq:SE-degen}
    \begin{align}
      & 0=k_i\cdot P(\sigma_i)=\frac{1}{\varepsilon} k_i\cdot P_I( x_i) + k_i\cdot \tilde{P}_{\bar I}(\sigma_I) + O(\varepsilon)\, , && i\in I\,,\\
      & 0=k_p\cdot P(\sigma_p)=k_p\cdot P_{\bar I}(\sigma_p)\,, &&
      p\in \bar I\,.
    \end{align}
\end{subequations}
In particular, for $i\in I$, this implies
$k_i\cdot P_I(x_i)=O(\varepsilon)$ as $\varepsilon\rightarrow 0$ since
\begin{equation}
 \label{scale}
 k_i\cdot P_I( x_i) =-\varepsilon k_i\cdot \tilde{P}_{\bar I}(\sigma_I) + O(\varepsilon^2)\,.
\end{equation}  
By summing we obtain as an algebraic identity
\begin{equation}
\label{factorize}
\tilde{s}_I:= \frac12 \sum_{i\neq j\in I} k_i\cdot k_j=  \sum_{i,j\in I} \frac{x_ik_i\cdot  k_j}{ x_i-x_j}= \sum_{i\in I} x_i k_i \cdot P_I(x_i)=\cO(\varepsilon)
 \, ,
\end{equation}
so $\tilde{s}_I$ vanishes to order $\varepsilon$, and any (potential) pole is
associated with the vanishing of $\tilde s_I$ where 
\begin{equation*}
  \tilde{s}_I=\begin{cases}
               k_I^2 & \sigma_{\ell^+},\sigma_{\ell^-}\in\bar I\,,\\
               2\ell\cdot k_I + k_I^2 & \sigma_{\ell^+}\in I, \sigma_{\ell^-}\in \bar I\,.
              \end{cases}
 \end{equation*}.

We now focus on the measure of the amplitude expression with a generic integrand
\begin{equation}
\cM=\int \cI \prod_{i=2}^n \bar{\delta}(k_i\cdot P(\sigma_i))d\sigma_i\, .
\end{equation}
First, let us determine the weight of the measure in $\varepsilon$ as $\varepsilon\rightarrow 0$ (the integrand $\cI$ will have some weight also which we discuss later).
For each $i\in I$, the scattering equations contribute
\begin{subequations}
  \begin{align}
    &\bar{\delta}(k_i\cdot P(\sigma_i))\,d \sigma_i=\varepsilon^2\,\bar{\delta}(k_i\cdot P_I(x_i))\,d x_i && i\in I\,,\\
    &\bar{\delta}(k_p\cdot P(\sigma_p))\,d \sigma_p=\bar{\delta}(k_p\cdot P_{\bar I}(\sigma_p))\,d \sigma_p && p\in \bar I\,.
  \end{align}
\end{subequations}

Thus we obtain  scattering equations on the factorised Riemann surface, multiplied by a factor of $\varepsilon^2$ for each $i\in I$. Note however that there is a subtlety; we expect three fixed marked points on each Riemann surface. On the Riemann surface $\Sigma_{\bar I}$, this is trivially true since there are two fixed points and the degeneration point $\sigma_I$. On $\Sigma_I$, the fixed points are given by the degeneration point $x_{\bar I}=\infty$ and $x_I=0$, and our choice of parametrisation for the degeneration $x_m=1$.
This gives the required independent $n_I-3$ scattering equations, but we still have to consider the integration over $\sigma_m=\sigma_I+\varepsilon$ and its associated delta function imposing its scattering equation associated to $\sigma_m$.  Using 
\begin{equation}
 s_I=\sum_{i\neq m \in I}x_i k_i\cdot P_I(x_i) + k_m \cdot P_I(x_m)\,,
\end{equation}
and the support of the remaining scattering equations \cref{eq:SE-degen}, we find
\begin{align}
 k_m\cdot P(\sigma_m)&=\frac{1}{\varepsilon}\left(k_m\cdot P_I(x_m)+\varepsilon\sum_{p\in \bar I}k_m\cdot \tilde{P}_{\bar I}(\sigma_I)+\cO(\varepsilon^2)\right)\nonumber\\
 &=\frac{1}{\varepsilon}\left(\tilde{s}_I+\varepsilon\sum_{i\in I,\, p\in \bar I}x_i k_i\cdot \tilde{P}_{\bar I}(\sigma_I)+\cO(\varepsilon^2)\right)\nonumber\\
 &\equiv\frac{1}{\varepsilon}\left(\tilde{s}_I+\varepsilon\mathcal{F}+\cO(\varepsilon^2)\right)\,.
\end{align}
Thus the last equation becomes
\begin{equation}
 \bar{\delta}(k_m\cdot P(\sigma_m))\,d \sigma_m=\varepsilon\,\bar{\delta}(\tilde{s}_I+\varepsilon\mathcal{F})\,d \varepsilon  \,,
\end{equation}
and  the measure factorises as
 \begin{equation}
  d \mu\equiv\prod_{i=2}^n \bar{\delta}(k_i\cdot P(\sigma_i))d\sigma_i=\varepsilon^{2(|\cI|-1)}\,\frac{d\varepsilon}{\varepsilon} \,\bar{\delta}(\tilde{s}_I+\varepsilon\mathcal{F})\,d\mu_I\, d\mu_{\bar I}\,.
\end{equation}

We  now distinguish three cases according to whether $\sigma_{\ell^\pm}$ are  in $I$ as in  \cref{fig:degen}.
\begin{enumerate}[\text{Case} I]
 \item \label{Case1} If $\sigma_{\ell^+}$ and $\sigma_{\ell^-}$ are not in $I$, and $I$ is a strict subset of $1,\ldots ,n$,  this is the standard factorisation channel with  $s_I=k_I^2\rightarrow 0$. The relevant boundary of the moduli space describes a Riemann sphere connected to a nodal sphere, corresponding to a tree-level amplitude factorising from a one-loop amplitude. The measure is
 \begin{equation}
  d \mu^{(1)}=\varepsilon^{2(|\cI|-1)}\,\frac{d\varepsilon}{\varepsilon} \,\bar{\delta}(s_I+\varepsilon\mathcal{F})\,d\mu_I^{(0)}\, d\mu_{\bar I}^{(1)}\,.
\end{equation}
 \item \label{Case2} If without loss of generality $\sigma_{\ell^+} \in I$ but $\sigma_{\ell^-}\notin I$ the scattering equations imply 
\begin{equation}
\tilde{s}_I=k_I\cdot\ell + \frac12 k_I^2=\cO(\varepsilon)\, .
\end{equation}
as $\varepsilon\rightarrow 0$. This non-separating degeneration describes two Riemann spheres, connected at {\it two} double points, see \cref{fig:degen}. The corresponding measure is given by
\begin{equation}
  d \mu^{(1)}=\varepsilon^{2(|\cI|-1)}\,\frac{d\varepsilon}{\varepsilon} \,\bar{\delta}(\tilde{s}_I+\varepsilon\mathcal{F})\,d\mu_I^{(0)}\, d\mu_{\bar I}^{(0)}\,,
\end{equation}
leading to the expected poles from the Q-cut factorisation channels.
 \item \label{Case3} The case $I=\{\sigma_{\ell^+},\sigma_{\ell^-}\}$ is of particular interest since this configuration arises for large $\ell$, $\ell =\cO(\lambda^{-1})$ (see \eqref{SE2n}), and for the singular solutions in non-supersymmetric theories. It is discussed in \cref{lemma:powercounting} and \cref{sec:powercounting}, and determines the UV behaviour of our one-loop amplitudes.
\end{enumerate}

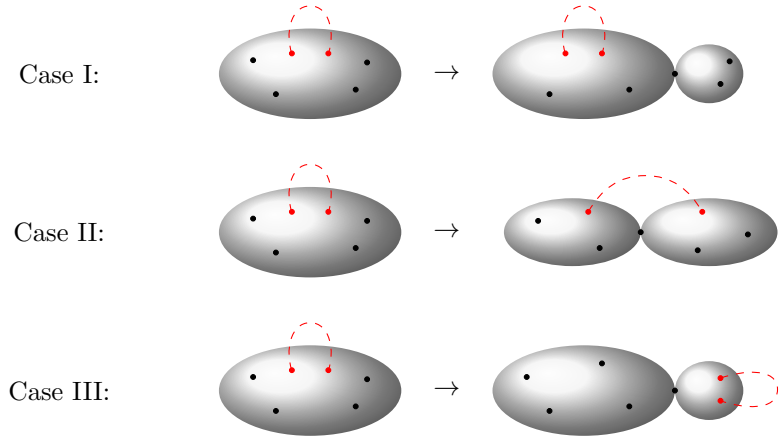
\begin{figure}[ht]
\begin{center} \vspace{-5pt}\begin{tikzpicture}[scale=3]
 \shade[shading=ball, ball color=light-gray] (3.6,2.4) circle [x radius=0.4, y radius=0.2];
  \draw [dashed,red] (3.52,2.49) to [out=110, in=180] (3.6,2.7) to [out=0, in=70] (3.68,2.49);
  \draw [fill] (3.35,2.46) circle [radius=.3pt];
  \draw [fill] (3.45,2.31) circle [radius=.3pt];
  \draw [fill] (3.8,2.33) circle [radius=.3pt];
  \draw [fill] (3.85,2.45) circle [radius=.3pt];
  \draw [fill,red] (3.52,2.49) circle [radius=.3pt];
  \draw [fill,red] (3.68,2.49) circle [radius=.3pt];
  \node at (4.2,2.4) {$\rightarrow$};
  \node at (2.5,2.4) {Case I:};
 \shade[shading=ball, ball color=light-gray] (4.8,2.4) circle [x radius=0.4, y radius=0.2];
 \shade[shading=ball, ball color=light-gray] (5.35,2.4) circle [x radius=0.15, y radius=0.13];
  \draw [dashed,red] (4.72,2.49) to [out=110, in=180] (4.8,2.7) to [out=0, in=70] (4.88,2.49);
  \draw [fill] (4.65,2.31) circle [radius=.3pt];
  \draw [fill] (5,2.33) circle [radius=.3pt];
  \draw [fill] (5.44,2.455) circle [radius=.3pt];
  \draw [fill] (5.4,2.355) circle [radius=.3pt];
  \draw [fill,red] (4.72,2.49) circle [radius=.3pt];
  \draw [fill,red] (4.88,2.49) circle [radius=.3pt];
  \draw [fill] (5.2,2.4) circle [radius=.3pt];
  \shade[shading=ball, ball color=light-gray] (3.6,1) circle [x radius=0.4, y radius=0.2];
  \draw [dashed,red] (3.52,1.09) to [out=110, in=180] (3.6,1.3) to [out=0, in=70] (3.68,1.09);
  \draw [fill] (3.35,1.06) circle [radius=.3pt];
  \draw [fill] (3.45,0.91) circle [radius=.3pt];
  \draw [fill] (3.8,0.93) circle [radius=.3pt];
  \draw [fill] (3.85,1.05) circle [radius=.3pt];
  \draw [fill,red] (3.52,1.09) circle [radius=.3pt];
  \draw [fill,red] (3.68,1.09) circle [radius=.3pt];
  \node at (4.2,1) {$\rightarrow$};
  \node at (2.5,1) {Case III:};
 \shade[shading=ball, ball color=light-gray] (4.8,1) circle [x radius=0.4, y radius=0.2];
 \shade[shading=ball, ball color=light-gray] (5.35,1) circle [x radius=0.15, y radius=0.13];
  \draw [dashed,red] (5.4,1.055) to [out=20, in=90] (5.65,1.005) to [out=270, in=340] (5.4,0.955);
  \draw [fill] (4.55,1.06) circle [radius=.3pt];
  \draw [fill] (4.65,0.91) circle [radius=.3pt];
  \draw [fill] (5,0.93) circle [radius=.3pt];
  \draw [fill] (4.88,1.12) circle [radius=.3pt];
  \draw [fill,red] (5.4,1.055) circle [radius=.3pt];
  \draw [fill,red] (5.4,0.955) circle [radius=.3pt];
  \draw [fill] (5.2,1) circle [radius=.3pt];
 \shade[shading=ball, ball color=light-gray] (3.6,1.7) circle [x radius=0.4, y radius=0.2];
  \draw [dashed,red] (3.52,1.79) to [out=110, in=180] (3.6,2) to [out=0, in=70] (3.68,1.79);
  \draw [fill] (3.35,1.76) circle [radius=.3pt];
  \draw [fill] (3.45,1.61) circle [radius=.3pt];
  \draw [fill] (3.8,1.63) circle [radius=.3pt];
  \draw [fill] (3.85,1.75) circle [radius=.3pt];
  \draw [fill,red] (3.52,1.79) circle [radius=.3pt];
  \draw [fill,red] (3.68,1.79) circle [radius=.3pt];
  \node at (4.2,1.7) {$\rightarrow$};
  \node at (2.5,1.7) {Case II:};
 \shade[shading=ball, ball color=light-gray] (4.75,1.7) circle [x radius=0.3, y radius=0.15];
 \shade[shading=ball, ball color=light-gray] (5.35,1.7) circle [x radius=0.3, y radius=0.15];
  \draw [dashed,red] (4.82,1.79) to [out=60, in=180] (5.07,1.95) to [out=0, in=120] (5.32,1.79);
  \draw [fill] (4.6,1.75) circle [radius=.3pt];
  \draw [fill] (4.87,1.63) circle [radius=.3pt];
  \draw [fill] (5.3,1.62) circle [radius=.3pt];
  \draw [fill] (5.52,1.69) circle [radius=.3pt];
  \draw [fill,red] (4.82,1.79) circle [radius=.3pt];
  \draw [fill,red] (5.32,1.79) circle [radius=.3pt];
  \draw [fill] (5.05,1.7) circle [radius=.3pt];
\end{tikzpicture}\end{center}
\caption{Different possible worldsheet degenerations.}
\label{fig:degen}
\end{figure}

\subsection{Factorisation II - integrands} \label{sec:fact_int}
Whether we actually have a pole or not in the factorisation limit depends on the weight of the integrand $\cI$ in $\varepsilon$ as $\varepsilon\rightarrow 0$.
In this section, we consider  the integrands  for the $n$-gons, Yang-Mills, gravity and the bi-adjoint scalar in more detail. In particular, we will find that all these integrands behave as
\begin{equation}
 \cI^{(1)}=\varepsilon^{-2(|\cI|-1)}\cI_I^{(0)}\cI_{\bar I}^{(1)}\,,
\end{equation}
in case I and 
\begin{equation} \label{eq:fact_intcase2}
 \cI^{(1)}=\frac{1}{\ell^2}\varepsilon^{-2(|\cI|-1)}\cI_I^{(0)}\cI_{\bar I}^{(0)}\,,
\end{equation}
in case II, where $\cI_I$ ($\cI_{\bar I}$) depends only on the on-shell momenta $k_I$ ($k_{\bar I}$) and $\ell_I=\ell+\eta$, $\eta^2=-\ell^2$. With the measure contributing a factor of $\varepsilon^{2(|\cI|-1)} \,\bar{\delta}(\tilde{s}_I+\varepsilon\mathcal{F})\,\frac{\d\varepsilon}{\varepsilon}$, the overall amplitude scales as $\varepsilon^{-1}$, and we can perform the integral against the $\delta$-function explicitly, leading to a pole in $\tilde{s}_I$. Therefore, the full amplitude factorises on the expected poles, with residues given by the corresponding subamplitudes. Moreover, as evident from above,  for case II this factorisation channel corresponds to a Q-cut;
\begin{equation}
\begin{split}
 \cM^{(1)}&=\int \cI^{(1)} d\mu=\frac{1}{\ell^2}\int  \frac{\d\varepsilon}{\varepsilon} \,\bar{\delta}(\tilde{s}_I+\varepsilon\mathcal{F})\,\cI_I^{(0)}\cI_{\bar I}^{(0)}\,d\mu_I\, d\mu_{\bar I}\\
 &=\cM_I^{(0)}(\dots,\ell_I,\ell_I+k_I)\frac{1}{\ell^2}\frac{1}{\tilde{s}_I}\cM_{\bar I}^{(0)}(-\ell_I,-\ell_I-k_I,\dots)\,.
 \end{split}
\end{equation}
The analysis of the integrand can be further simplified by focussing on the `left' and `right' contributions $\cI^{L,R}$ to the integrand individually, where $\cI=\cI^L \cI^R$. From the discussion of the preceding sections, we identify $\cI^L=\cI^R=\sum_r \pf(M_{\text{NS}}^r)$ for pure gravity in $d$ dimensions, $\cI^L=\sum_r \pf(M_{\text{NS}}^r)$ and $\cI^R=\cI^{cPT}$ for pure Yang-Mills, and $\cI^L=\cI^R=\cI^{cPT}$ for the bi-adjoint scalar theory. Therefore, it is sufficient to prove that these building blocks have weight $-(|I|-1)$ in $\varepsilon$,
\begin{equation}
 {\cI^{(1)}}^{L,R}=\varepsilon^{-(|\cI|-1)}{\cI^{(0)}}^{L,R}_{I}\,{\cI^{(0)}}^{L,R}_{\bar I}\,.
\end{equation}

\subsubsection{The \texorpdfstring{$n$}{n}-gon integrand}
Let us first consider the $n$-gon integrand
\begin{equation}
 \hat{\cI}^{n-gon}=\frac{1}{\sigma_{\ell^+\,\ell^-}^4}\prod_{i=1}^n \left(\frac{\sigma_{\ell^+\,\ell^-}}{\sigma_{i\,\ell^+}\sigma_{i\,\ell^-}}\right)^2\,.
\end{equation}
It is straightforward to see that case I cannot contribute since the integrand scales as $\varepsilon^0$, and thus the amplitude behaves as $\varepsilon^{2|I|-3}=\cO(\varepsilon)$. Therefore only case II contributes; and the integrand factorises as
\begin{equation}
 \hat{\cI}^{n-gon}=\varepsilon^{-2(|\cI|-1)}\frac{1}{\sigma_{I\,\ell^-}^4}\prod_{i\in I}\frac{1}{x_i^2}\,\prod_{p\in\bar I}\left(\frac{\sigma_{I\,\ell^-}}{\sigma_{p\,I}\sigma_{p\,\ell^-}}\right)^2=\varepsilon^{-2(|\cI|-1)}\,\cI_{I}^{n-gon}\cI_{\bar I}^{n-gon}\,.
\end{equation}
Note that we have used explicitly the chosen gauge fixing $x_{\bar I}=\infty$ for the second equality. In particular, since $\cI^{n-gon}_{I, \bar I}$ do not depend on $\ell$, this gives the correct residues for the respective Q-cuts of the $n$-gon. 

\subsubsection{The Parke-Taylor factor}
Consider next the Parke-Taylor-like integrands
\begin{align}
 \cI^{cPT}(\sigma_{\ell^+},\,\rho,\,\sigma_{\ell^-})= \sum_{\rho\in S_n}\cI^{PT}(\sigma_{\ell^+},\,\rho,\,\sigma_{\ell^-}) =\sum_{\rho\in S_n}\frac{\sigma_{\ell^+\,\ell^-}}{\sigma_{\ell^+\,\rho(1)}\sigma_{\rho(1)\,\rho(2)}\dots\sigma_{\rho(n)\,\ell^-}}\,.
\end{align}
If the set $I$ is not consecutive in any of the orderings of the Parke-Taylor factors in the cyclic sum above, the amplitude scales as $\cO(\varepsilon)$ and thus vanishes. Therefore the only non-vanishing contributions come from terms where all $\sigma_i$, $i\in I$ are consecutive with respect to the ordering defined by the Parke-Taylor factors.\\

In case II, with $\sigma_{\ell^+}\in I$, the only term contributing is $\cI^{PT}(\sigma_{\ell^+},\,I,\,\bar I,\,\sigma_{\ell^-})$,
and we find the correct scaling behaviour to reproduce the pole,
\begin{equation} \cI^{cPT}(\sigma_{\ell^+},\,\rho,\,\sigma_{\ell^-})=\varepsilon^{-(|I|-1)}\,\cI^{PT}_I(x_{\ell^+},\,I,\,x_{\bar I})\,\cI^{PT}_{\bar I}(\sigma_{I},\,\bar I,\,\sigma_{\ell^-})\,.
\end{equation}
In particular, the integrands are again independent of the loop momentum $\ell$, and are straightforwardly identified as the tree-level Parke-Taylor factors, $\cI^{PT}=\cI^{PT,(0)}$.
Note furthermore that the reduction of the integrands from a sum over cyclic Parke-Taylor factors to single terms can be understood directly in terms of diagrams, as only a single diagram will contribute to a given pole. However, another nice interpretation can be given for the bi-adjoint scalar theory discussed in \cite{He:2015yua}: Here, the cyclic sum is understood as a tool to remove unwanted tadpole contributions to the amplitude. The factorising Riemann surface however separates the insertions of the loop momenta, and thereby automatically removes these tadpole diagrams.\\

In case I, the same argument as above can be used to deduce that the only terms contributing on the factorised Riemann sphere are those where all $i\in I$ appear in a consecutive ordering, and we find
\begin{equation}
 \cI^{cPT}(\sigma_{\ell^+},\,\rho,\,\sigma_{\ell^-})=\varepsilon^{-(|I|-1)}\cI^{PT}_I(I\cup\{x_{\bar I}\})\, \cI^{cPT}_{\bar I}(\sigma_{\ell^+},\,\rho(\bar I\cup\{\sigma_I\}),\,\sigma_{\ell^-})\,,
\end{equation}
which again leads to the expected poles and residues, with $\cI^{cPT}_{\bar I}=\cI^{(1)}_{\bar I}$.

\subsubsection{Non-supersymmetric theories}
\label{sec:non-supersymm-theor-1}
In both the case of the $n$-gon and the Parke-Taylor factors, the integrand was independent of $\ell$, and thus Q-cuts were easily identified. For non-supersymmetric theories, with Pfaffians in the integrands, this identification becomes more involved. We will focus first on the NS sector,\footnote{In the case of Yang-Mills, this is identical to the pure case, see \cref{eq:pure-YM}.}
\begin{equation}
 \hat{\cI}^{\text{NS}}= \frac1{\sigma_{\ell^+\ell^-}}\sum_r\pf'(M_{\text{NS}}^r)\,,
\end{equation}
where $\pf'(M_{\text{NS}})\equiv \frac{1}{\sigma_{\ell^+\,\ell^-}}\,\pf({M_{\text{NS}}}_{(\ell^+,\ell^-)})$, and the matrix $M_{\text{NS}}$ was defined in \cref{eq:defMNS_details}. As above, we have used the subscript $(ij)$ to denote that both the rows and the columns $i$ and $j$ have been removed from the matrix. Consider first again the case II where the Riemann sphere degenerates as $\sigma_i \rightarrow \sigma_I+\varepsilon x_i +\cO(\varepsilon^2)$ with $\sigma_{\ell^+}\in I$. Then the entries in $M_{\text{NS}}$ behave to leading order in $\varepsilon$ as
\begin{equation}
 \frac{1}{\sigma_{ij}}=\begin{cases} \varepsilon^{-1}\frac{1}{x_{ij}} & i,j\in I\\
                                     \frac{1}{\sigma_{Ij}} & i\in I,\, j\in \bar I\\
                                     \frac{1}{\sigma_{ij}} & i,j\in \bar I\\
                       \end{cases} ,\qquad P^\mu(\sigma_i)=\begin{cases} \varepsilon^{-1}P^\mu_I(x_i)&i\in I\\ P^\mu_{\bar I}(\sigma_i)&i\in\bar I \end{cases}\,.
\end{equation}
Using antisymmetry of the Pfaffian, we can rearrange the rows and columns such that $M_{\text{NS}}^r$ takes the following form:
\begin{equation}
\begingroup \renewcommand*{\arraystretch}{1.4}
 M_{\text{NS}}^r=\begin{pmatrix}\varepsilon^{-1} {M^{(0)}_I}_{(\bar I\bar I')} & N\\ -N^T & {M^{(0)}_{\bar I}}_{(II')}\end{pmatrix}\,,
 \endgroup
\end{equation}
where $M_{I}^{(0)}$ is the tree-level matrix, depending only on higher-dimensional on-shell deformations of the loop momentum $\ell_I=\ell+\eta$ with polarisation $\epsilon^r$, and the momenta $k_i$, $i\in I$. In particular, the diagonal entries $C_{ii}=\epsilon_i\cdot P(\sigma_i)$ respect this decomposition due to the one-form $P_\mu$ factorising appropriately. The matrices $N$ are defined (to leading order in $\varepsilon)$ by $N_{ij}=\mu_i\cdot \nu_i$, with 
\begin{align*}
 \mu=\left(\ell_I,k_i,\epsilon^r,\epsilon_i\right)\,,\qquad \nu=\left(\frac{-\ell_I}{\sigma_{I\ell^-}},\frac{k_j}{\sigma_{Ij}},\frac{\epsilon^r}{\sigma_{I\ell^-}},\frac{\epsilon_j}{\sigma_{Ij}}\right)\,,
\end{align*}
for $i\in I$, $j\in \bar I$, where $\ell_I=\ell+\eta$, with $\eta^2=-\ell^2$. Note in particular that this ensures that $N_{\ell^+\ell^-}=0$. To identify the scaling of the integrand $\cI_{\text{NS}}$, we have to consider the Pfaffian of the reduced matrix ${M_{\text{NS}}^r}_{(\ell^+\ell^-)}$;
\begin{equation}
\begingroup \renewcommand*{\arraystretch}{1.5}
 {M_{\text{NS}}^r}_{(\ell^+\ell^-)}=\begin{pmatrix}\varepsilon^{-1} {M^{(0)}_I}_{(\ell^+\bar I\bar I')} & N_{[\ell^+\ell^-]}\\ -N^T_{[\ell^+\ell^-]} & {M^{(0)}_{\bar I}}_{(\ell^- II')}\end{pmatrix}\,,
 \endgroup
\end{equation}
where in $N$ only the row (column) associated to $\ell^+$ ($\ell^-$) has been removed.
Note in particular that the matrices ${M^{(0)}_I}_{(\ell^+\bar I\bar I')}$ and ${M^{(0)}_{\bar I}}_{(\ell^- II')}$ have odd dimensions, so the scaling in $\varepsilon$ is non-trivial. To identify the leading behaviour of the bosonic integrand in $\varepsilon$, we will use the following lemma:
\begin{lemma}[Factorisation Lemma \cite{Dolan:2013isa}] \label{lemma:fact}
 Let $M_I$ and $M_{\bar I}$ be antisymmetric matrices of dimensions $m_I\times m_I$ and $m_{\bar I}\times m_{\bar I}$ respectively; and $N=\mu_i\nu_j$, with $d$-dimensional vectors $\mu_i=\mu_i^\mu$, $\nu_j=\nu_j^\mu$ for $i\in I$, $j\in \bar I$. Then the leading behaviour of the Pfaffian of
 \begin{equation}
  M=\begin{pmatrix} \varepsilon^{-1}M_I & N \\ -N^T & M_{\bar I} \end{pmatrix}\,,
 \end{equation}
as $\varepsilon\rightarrow 0$ is given by
\begin{itemize}
 \item $m_I +m_{\bar I}$ is odd: $\pf(M)=0$.
 \item $m_I +m_{\bar I}$ is even, $m_I$ and $m_{\bar I}$ are even: $\pf(M)\sim\varepsilon^{-\frac{m_I}{2}}\pf(M_I)\pf(M_{\bar I})$
 \item $m_I +m_{\bar I}$ is even, $m_I$ and $m_{\bar I}$ are odd: 
 \begin{equation}
  \pf(M)\sim \varepsilon^{-\frac{m_I-1}{2}}\sum_s\pf(\widetilde{M}_I^s)\,\pf(\widetilde{M}_{\bar I}^s)\,,
 \end{equation}
where $s$ runs over a basis $\epsilon^s$, and 
\begin{equation}
 \widetilde{M}_I^s=\begin{pmatrix} M_I & (\mu\cdot\epsilon^s)^T\\ -\mu\cdot\epsilon^s & 0 \end{pmatrix}\,, \qquad \widetilde{M}_{\bar I}^s=\begin{pmatrix} 0 & \nu\cdot\epsilon^s\\ -(\nu\cdot\epsilon^s)^T & M_{\bar I} \end{pmatrix}\,.
\end{equation}
\end{itemize}
\end{lemma}
The interested reader is referred to \cite{Dolan:2013isa} for the proof of this lemma relying on basic properties of the Pfaffian. Applying \cref{lemma:fact} to the integrand $\cI_{\text{NS}}$, we can identify ${\widetilde{M}^{(0)}_I}$\hspace{-4pt}$_{(\ell^+\bar I\bar I')}$ with ${M^{(0)}_I}_{(\ell^+\bar I)}^{rs}$ by identifying\footnote{and similarly for ${\widetilde{M}^{(0)}_{\bar I}}$\hspace{-4pt}$_{(\ell^- I I')}$ and ${M^{(0)}_{\bar I}}_{(\ell^-I)}^{rs}$.} the additional `$s$' row and column with the ones associated to the interchanged particle $\bar I'$. To leading order in $\varepsilon$, the integrand therefore becomes
\begin{equation}
 \hat{\cI}^{\text{NS}}=
 \varepsilon^{-(|I|-1)} \frac{1}{\sigma_{I\ell^-}}\sum_{r,s} \pf\left({M^{(0)}_I}_{(\ell^+\bar I)}^{rs}\right)\pf\left({M^{(0)}_{\bar I}}_{(\ell^- I)}^{rs}\right)\,.
\end{equation}
Recalling furthermore the gauge fixing choice $x_{\bar I}=\infty$ in degenerating the worldsheet, this can be identified with a product of reduced Pfaffians,
\begin{equation}
 \hat{\cI}^{\text{NS}}=
 \varepsilon^{-(|I|-1)} \sum_{r,s} \pf'\left({M^{(0)}_I}^{rs}\right)\pf'\left({M^{(0)}_{\bar I}}^{rs}\right)\,.
\end{equation}
As seen from the discussion above, this provides both the correct weight in the degeneration parameter $\varepsilon$ and the correct residues for the Q-cut factorisation.\\

The discussion for case I proceeds along similar lines: For convenience, we choose to remove rows and columns associated to one particle on each side of the degeneration from $M_{\text{NS}}$. Following through the same steps as for case II, the integrand then factorises as 
\begin{equation}
 \hat{\cI}^{\text{NS}}=
 \varepsilon^{-(|I|-1)} \sum_{r,s} \pf'\left({M^{(0)}_I}^{s}\right)\pf'\left({M^{(1)}_{\bar I}}^{rs}\right)\,.
\end{equation}
This correctly reproduces the poles and residues for the bubbling of a Riemann sphere: as a partial sum of the external momenta goes null, the residue is a product of a one-loop amplitude and a tree-level amplitude.\\

\paragraph{Factorisation for Pure Yang Mills and gravity amplitudes.} At this point it is easy to see how this analysis extends to pure Yang-Mills and gravity. Note first of all that for Yang-Mills, the NS and the pure sector are identical, see \cref{eq:pure-YM-d},
\begin{equation}
  \mathcal{I}_{\mathrm{pure-YM}}=
  ( \pf(M_{3})\big|_{\sqrt{q}}+ (d-2)\,\pf(M_{3})\big|_{q^{0}})\,\mathcal{I}^{cPT}\,.
\end{equation}
For pure gravity, \cref{eq:pure-gravity-d}, we have
\begin{equation}
  \mathcal{I}_{\mathrm{pure-grav}}=( \pf(M_{3})\big|_{\sqrt{q}}+ (d-2)\,\pf(M_{3})\big|_{q^{0}})^{2}-\alpha\,(\pf(M_{3})\big|_{q^{0}})^{2}\,,
\end{equation}
where $\alpha=\frac{1}{2}(d-2)(d-3)+1$ is given by the degrees of freedom of the B-field and the dilaton. This factorises,
\begin{equation}
\begin{aligned}
  &\mathcal{I}_{\mathrm{pure-grav}}=\cI_L\,\cI_R\,,\qquad\text{with} &&\cI_L=\pf(M_{3})\big|_{\sqrt{q}}+ \alpha_d\,\pf(M_{3})\big|_{q^{0}}\,,\\
  & && \cI_R=\pf(M_{3})\big|_{\sqrt{q}}+ \tilde{\alpha}_d\,\pf(M_{3})\big|_{q^{0}}\,,
\end{aligned}
\end{equation}
where $\alpha_d=d-2+\sqrt{\alpha}$ and $\tilde{\alpha}_d=d-2-\sqrt{\alpha}$. An analogous calculation to the one in \cref{sec:MNS} then straightforwardly leads to
\begin{equation}
 \hat{\cI}_L=\sum_r \pf'(M_{\alpha_d}^r)\,, \qquad \hat{\cI}_R=\sum_r\pf'(M_{\tilde{\alpha}_d}^r)\,,
\end{equation}
where $M_{\alpha_d}^r$ and $M_{\tilde{\alpha}_d}^r$ have been defined as $M_{\text{NS}}^r$, but with the element $B_{\ell^+\ell^-}^{\text{NS}}=\frac{d-2}{\sigma_{\ell^+\ell^-}}$ replaced by $B_{\ell^+\ell^-}^{\alpha_d}=\frac{\alpha_d}{\sigma_{\ell^+\ell^-}}$ and $B_{\ell^+\ell^-}^{\tilde{\alpha}_d}=\frac{\tilde{\alpha}_d}{\sigma_{\ell^+\ell^-}}$ respectively. Then the discussion given above for the NS sector generalises straightforwardly, and the factorisation lemma, in conjunction with the same identification of the matrices, yields again
\begin{equation}
 \hat{\cI}^{\text{pure}}=
 \varepsilon^{-(|I|-1)} \sum_{r,s} \pf'\left({M^{(0)}_I}^{rs}\right)\pf'\left({M^{(0)}_{\bar I}}^{rs}\right)\,.
\end{equation}
Note in particular that since only the matrix $N$ is affected by the change $d-2\rightarrow \alpha_d$, the residues are unchanged, and thus still correspond to the expected tree-level amplitudes for pure Yang-Mills and gravity. Again, case I proceeds in close analogy to the NS sector discussion above.

\subsection{UV behaviour of the one-loop amplitudes} \label{sec:powercounting}
Consider now the UV behaviour of the one-loop amplitudes; $\ell\rightarrow\lambda^{-1}\ell$, with $\lambda\rightarrow 0$. In this set-up, the scattering equations only yield solutions if the two insertion points of the loop momentum coincide, $\sigma_{\ell^-}\rightarrow\sigma_{\ell^+}$. The factorisation of the scattering equations and the measure will be closely related to \cref{sec:fact_SE}, so we will restrict the discussion here to highlight only the differences due to the factor of $\lambda^{-1}$. As above, we will blow up the concentration point into a bubbled-off Riemann sphere,
\begin{equation}\label{eq:degen_UV}
 \sigma_{\ell^-}=\sigma_{\ell^+}+\varepsilon +\varepsilon^2 y_{\ell}+\cO(\varepsilon^3)\,,
\end{equation}
where we have used the M\"obius invariance on the sphere to fix $x_{\ell^-}=1$. 
We thus find the scattering equations
\begin{subequations}
  \begin{align}
    &0=\text{Res}_{\ell^+}P^2=\lambda^{-1} \sum_i \frac{\ell\cdot k_i}{\sigma_{\ell^+i}} \,,\\
    &0=\text{Res}_{\ell^-}P^2=-\lambda^{-1} \sum_i \frac{\ell\cdot k_i}{\sigma_{\ell^+i}}+\varepsilon \lambda^{-1}\sum_i \frac{\ell\cdot k_i}{\sigma_{\ell^+i}^2}+\cO(\varepsilon^2)  \label{eq:SE_UV}\,,\\
    &0=k_i\cdot P(\sigma_i)=-\varepsilon \lambda^{-1} \frac{\ell\cdot k_i}{\sigma_{\ell^+i}^2} +\sum_j \frac{k_i\cdot k_j}{\sigma_{ij}} +\cO(\varepsilon^2)\,.
  \end{align}
\end{subequations}
On the support of the  scattering equations at $\sigma_{\ell^{+}}$ and $\sigma_i$ for $i\neq i_1,i_2,i_3$, \cref{eq:SE_UV} simplifies to
\begin{align}
 0=\text{Res}_{\ell^-}P^2=\varepsilon \lambda^{-1}\sum_{i=i_1,i_2,i_3} \frac{\ell\cdot k_i}{\sigma_{\ell^+i}^2}+\sum_{\substack{i,j\\i\neq i_1,i_2,i_3}}\frac{k_i\cdot k_j}{\sigma_{ij}}+\cO(\varepsilon^2)
 \equiv \varepsilon \lambda^{-1} \mathcal{F}_1-\mathcal{F}_2\,,
\end{align}
where the explicit form of $\mathcal{F}_{1,2}$ will be irrelevant for the following discussion. Including the factor of $\ell^{-2}$, the measure therefore factorises (to leading order) as
\begin{equation}
\begin{split}
 d\mu&\equiv\frac{\lambda^2}{\ell^2}\bar{\delta}(\text{Res}_{\ell^+}P^2)\bar{\delta}(\text{Res}_{\ell^-}P^2)\prod_{i\neq i_1,i_2,i_3}\bar{\delta}(k_i\cdot P(\sigma_i))d\sigma_id\sigma_{\ell^+}d\sigma_{\ell^-}\\
 &=\lambda^4 \,\bar{\delta}\left(\varepsilon -\lambda \frac{\mathcal{F}_2}{\mathcal{F}_1}\right) \,d\varepsilon\, d\tilde{\mu}\,,
 \end{split}
\end{equation}
where $d\tilde{\mu}$ is independent of $\lambda$ and $\varepsilon$. The remaining delta-function thus fixes the worldsheet degeneration $\varepsilon$ to be proportional to the UV scaling $\lambda$ of the loop momentum $\ell$.\\

Again, this factorisation behaviour of the measure is universal for all theories, and only the specific form of the integrand will dictate the UV scaling of the theory. Denoting the weight of $\mathcal{I}_{L,R}$ in $\varepsilon$ by $N_{L,R}$, the scattering equation fixing $\varepsilon$ implies that the one-loop amplitudes scale as
\begin{equation}
 \cM\rightarrow\lambda^{4+N_L+N_R}\cM\,.
\end{equation}
Let us now consider the different supersymmetric and non-supersymmetric theories discussed above.\\

\subsubsection{The \texorpdfstring{$n$}{n}-gon} The integrand of the $n$-gon,
\begin{equation}
 \hat{\cI}^{n-gon}=\frac{1}{\sigma_{\ell^+\ell^-}^4}\prod_{i=1}^n \left(\frac{\sigma_{\ell^+\ell^-}}{\sigma_{i\ell^+}\sigma_{i\ell^-}}\right)^2\,,
\end{equation}
manifestly has weight $\varepsilon^{2n-4}$ under the worldsheet degeneration \ref{eq:degen_UV}. The leading behaviour of the amplitudes is thus given by $\lambda^{4+N}=\lambda^4$ for $\lambda\rightarrow 0$, and therefore the $n$-gons scale as $\ell^{-2n}$ in the UV limit.\\

\subsubsection{The Parke-Taylor factor}
The Parke-Taylor integrand \cref{eq:PTloop-def} contributing in Yang-Mills and the bi-adjoint scalar theory is given by\footnote{Note that we have chosen to include a factor of $\sigma_{\ell^+\ell^-}$ symmetrically in each integrand $\cI^{L,R}$.}
\begin{equation}
 \hat{\cI}^{PT}=\frac{1}{\sigma_{\ell^+\ell^-}}\sum_{i=1}^n \frac{1}{\sigma_{\ell^+\,i}\sigma_{i+1\, i}\sigma_{i+2\, i+1}\ldots \sigma_{i+n\,\ell^-}}\, .
\end{equation}
While naively this has weight $-1$ in $\varepsilon$, the leading order
cancels due to the photon decoupling identity --- a special case of
the KK relations ---
\begin{equation}
 0=\sum_{i=1}^n \frac{1}{\sigma_{\ell\,i}\sigma_{i+1\, i}\sigma_{i+2\, i+1}\ldots \sigma_{i+n\,\ell}}\, ,
\end{equation}
and the integrand thus scales as $\varepsilon^0=1$. In particular, this allows us to identify immediately the UV behaviour of the bi-adjoint scalar theory as $\ell^{-4}$. This result can be given an intuitive interpretation in terms of Feynman diagrams; the UV behaviour of the theory is determined by the diagrams involving bubbles, which scale as $\ell^{-4}$.

\subsubsection{Supersymmetric theories}
For supersymmetric theories, the UV behaviour is governed by the scaling of the integrand \cref{eq:IL0-sugra}
\begin{equation}
 \hat{\cI}_0^{L,R}=\frac{1}{\sigma_{\ell^+\ell^-}^2}\cI_0^{L,R}\,,\qquad \cI_0^{L,R}=\pf(M_{3}) \big|_{\sqrt{q}}+
  8\left(\pf(M_{3}) \big|_{q^0}-\pf(M_{2} )\big|_{q^0}\right)\,,
\end{equation}
under the worldsheet degeneration described above. Note first that the Szeg\H{o} kernels become (see \cref{eq:szegolimitsl2c})
\begin{align*}
 S_2(\sigma_{ij}) &= \frac{1}{2\sigma_{i\,j}} \left(\sqrt{\frac{\sigma_{i\,\ell^+}\,\sigma_{j\,\ell^-}}{\sigma_{j\,\ell^+}\,\sigma_{i\,\ell^-}}}+ \sqrt{\frac{\sigma_{j\,\ell^+}\,\sigma_{i\,\ell^-}}{\sigma_{i\,\ell^+}\,\sigma_{j\,\ell^-}}} \right)  \sqrt{d\sigma_ i} \sqrt{d\sigma_ j} &\to\sum_{m=0}^\infty  \varepsilon^m S_2^{(m)}\sqrt{d\sigma_ i} \sqrt{d\sigma_ j}\\
S_3(\sigma_{ij}) &=\frac{1}{\sigma_{i\,j}} \left(1 +\sqrt{q} \;\frac{(\sigma_{i\,j}\,\sigma_{\ell^+\,\ell^-})^2}{\sigma_{i\,\ell^+}\,\sigma_{i\,\ell^-}\,\sigma_{j\,\ell^+}\,\sigma_{j\,\ell^-}}\right) \sqrt{d\sigma_ i} \sqrt{d\sigma_ j} &\to \sum_{m=0}^\infty \varepsilon^m S_3^{(m)}\sqrt{d\sigma_ i} \sqrt{d\sigma_ j}\,,
\end{align*}
where
\begin{subequations}
  \begin{align}
    S_2^{(0)}&= \frac{1}{\sigma_{ij}} &  S_3^{(0)}&= \frac{1}{\sigma_{ij}}\\
    S_2^{(1)}&= 0 &  S_3^{(1)}&= 0\\
    S_2^{(2)}&= \frac{1}{8}\,\frac{\sigma_{ij}}{\sigma_{i\ell^+}^2\sigma_{j\ell^+}^2} &  S_3^{(2)}&= \frac{\sigma_{ij}}{\sigma_{i\ell^+}^2\sigma_{j\ell^+}^2}\\
    S_2^{(3)}&= \frac{1}{8}\left(\frac{\sigma_{ij}(\sigma_{i\ell^+}+\sigma_{j\ell^+})}{\sigma_{i\ell^+}^3\sigma_{j\ell^+}^3}+\frac{2\,y_\ell \,\sigma_{ij}}{\sigma_{i\ell^+}^2\sigma_{j\ell^+}^2}\right) &  S_3^{(3)}&=\frac{\sigma_{ij}(\sigma_{i\ell^+}+\sigma_{j\ell^+})}{\sigma_{i\ell^+}^3\sigma_{j\ell^+}^3}+\frac{2\,y_\ell \,\sigma_{ij}}{\sigma_{i\ell^+}^2\sigma_{j\ell^+}^2}\,.
  \end{align}
\end{subequations}
Expanding the integrand $\cI_0^{L,R}$ in powers of $\varepsilon$, $\pf(M_{3}) \big|_{q^0}$ and $\pf(M_{2}) \big|_{q^0}$ potentially contribute at order $\varepsilon^0$, whereas $\pf(M_{3}) \big|_{\sqrt{q}}$ can only contribute to $\varepsilon^2$. However, the leading contribution cancels among $\pf(M_{3}) \big|_{q^0}$ and $\pf(M_{2}) \big|_{q^0}$, as well as all higher order contribution (starting at order $\varepsilon^1$) coming from the diagonal entries of $C$. The weight in $\varepsilon$ is thus governed by the higher order behaviour of the Szeg\H{o} kernels. Moreover, due to the factor of $1/8$ between $S_2^{(2,3)}$ and $S_3^{(2,3)}$, the terms of order $\cO(\varepsilon^2)$ and $\cO(\varepsilon^3)$ originating from $\pf(M_{2}) \big|_{q^0}$ cancel\footnote{A bit more care is needed at $\cO(\varepsilon^3)$: While there are no contributions to $\cO(\varepsilon^2)$ from products of terms of order $\varepsilon$, these have to be taken into account at order $\cO(\varepsilon^3)$. However, the same reasoning as above guarantees their cancellation: The only possible origin for terms of order $\varepsilon$ are the diagonal entries of $C$, which coincide for $M_2$ and $M_3$. The cancellations to second order thus carry forwards to ensure that there will be no contributions from products of lower order terms up to $\cO(\varepsilon^3)$.} against the contributions from $\pf(M_{3}) \big|_{\sqrt{q}}$. \\

A short investigation confirms that there are no further cancellations, and thus $\hat{\cI}_0^{L,R}$ has weight 2 in $\varepsilon$. In particular, using $N=4+N_L+N_R$, this implies that our one-loop supergravity amplitudes scale as $\ell^{-8}$ in the UV limit, and super Yang-Mills as $\ell^{-6}$ (using $N_R=0$ for the Parke-Taylor integrand derived above). Naively, this seems to be a lower UV behaviour for super Yang-Mills than expected from the Feynman diagram expansion, for which the UV limit is given by the contribution from the boxes. However, a more detailed investigation of the Q-cut representation of the integrand  demonstrates that this is indeed the expected scaling, and that only the (higher) symmetry properties of the gravity integrand guarantee the same scaling for both the conventional Feynman diagram expression and the sum over Q-cuts. 

\subsubsection{Non-supersymmetric theories}
In the supersymmetric case discussed above, cancellations between the NS and the R sector ensured the correct scaling of the integrand. However, these cancellations are absent in the purely bosonic case,
\begin{equation}
 \hat{\cI}^{\text{NS}}= \frac{1}{\sigma_{\ell^+\ell^-}}\sum_r\pf({M_{\text{NS}}^r}_{(\ell^+\ell^-)})\,,
\end{equation}
so naively the integrand seems to scale as $\varepsilon^{-2}$. However, the leading contribution is given by the Pfaffian of the full tree-level matrix $M^{(0)}_{\text{NS}}=M^{(0)}$, which vanishes on the support of the scattering equations. More explicitly, let us expand the reduced matrix ${M_{\text{NS}}^r}_{(\ell^+\ell^-)}$ in $\varepsilon$;
\begin{equation}
 \pf({M_{\text{NS}}}_{(\ell^+\ell^-)})=\pf(M^{(0)})+\varepsilon \,\pf(M_{\text{NS}}^{(1)}) +\cO(\varepsilon)\,.
\end{equation}
The vanishing of the leading term can then be seen from the existence of two vecors, 
\begin{equation}\label{eq:v_0}
 v_0 = (\sigma_1,\dots,\sigma_n,0\dots,0) \quad \text{and } \quad \tilde{v}_0 = (1,\dots,1,0\dots,0)\,,
\end{equation}
in the kernel of $M_{\text{NS}}^{(0)}=M^{(0)}$. This argument can in fact be extended to subleading order: expanding both the matrix $M_{\text{NS}}$ and a potential vecor in the kernel to subleading order, 
\begin{equation}
 {M_{\text{NS}}}_{(\ell^+\ell^-)}=M^{(0)} +\varepsilon M^{(1)}_{\text{NS}}\,,\qquad v=v_0 +\varepsilon \,v_1\,,
\end{equation}
we note that the condition for $\cI^{\text{NS}}$ to scale as $\cO(\varepsilon^0)$ is that $ {M_{\text{NS}}}_{(\ell^+\ell^-)}$ has co-rank two to order $\varepsilon$, and thus two vectors spanning its kernel. Finding one of these is enough, as it guarantees the existence of the second. But this on the other hand is equivalent to 
\begin{equation}
 \left(M^{(0)} +\varepsilon M^{(1)}_{\text{NS}}\right)\left(v_0 +\varepsilon\, v_1\right)=\cO(\varepsilon^2)\,,
\end{equation}
vanishing to order $\cO(\varepsilon^2)$. Expanding this out, we obtain the conditions
\begin{align}\label{eq:cond_UV}
 M^{(0)}v_0=0\,,\qquad M^{(0)}\,v_1=-M^{(1)}_{\text{NS}}\, v_0\,.
\end{align}
As commented above, the first condition is satisfied with $v_0$ given in \cref{eq:v_0}. Note furthermore that the second condition cannot be straightforwardly inverted, since det($M^{(0)})=0$. The constraint for a solution to exist is thus the vanishing of both sides of the equation under a contraction with a vector in the kernel of the matrix. Since the only contribution to $M^{(1)}_{\text{NS}}$ to order $\cO(\varepsilon)$ comes from the diagonal entries $C_{ii}$, we get
\begin{equation}
 M^{(1)}_{\text{NS}}\, v_0= (0,\dots,0,\lambda^{-1}\frac{\epsilon_i\cdot\ell}{\sigma_{i\ell^+}},\dots)\,.
\end{equation}
This vanishes trivially when contracted with $v_0$ and $\tilde{v}_0$, and thus there exists a solution $v_1$ to \cref{eq:cond_UV}. The contributions of order $\cO(\varepsilon^{-1})$ to $\cI^{\text{NS}}$ therefore vanish, and the integrand is of order one. In particular, this implies that both pure Yang-Mills and pure gravity one-loop amplitudes scale as $\ell^{-4}$ in the UV limit, which is the expected behaviour from both the Q-cut and the Feynman diagram expansion. Moreover, the analysis above is equally applicable to the NS and the pure theories, similarly to the discussion given in \cref{sec:fact_int}.\\

To conclude the proof of \cref{lemma:powercounting}, let us summarise these results for the UV scaling of our one-loop amplitudes:
\begin{center}
\begin{tabular}{l|c}
 theory & scaling $\lambda^N\sim\ell^{-N}$\\ \hline 
 $n$-gon & $N=2n$\\
 supergravity & $N=8$\\
 super Yang-Mills & $N=6$\\
 pure gravity & $N=4$\\
 pure Yang-Mills & $N=4$\\
 bi-adjoint scalar & $N=4$
\end{tabular}
\end{center}
Note that the scaling of the non-supersymmetric theories (pure gravity and Yang-Mills, as well as the bi-adjoint scalar) corresponds to Feynman diagrams involving bubbles, whereas the higher scaling of the supersymmetric theories ensures that only boxes contribute. As observed above, Yang-Mills exhibits a lower scaling than expected from the Feynman diagram expansion, but which coincides with the expected scaling in the Q-cut representation. \\

Let us comment briefly on the closely related discussions regarding the contribution of the singular solutions $\sigma_{\ell^-}=\sigma_{\ell^+} +\varepsilon +\cO(\varepsilon^2)$. The same arguments as above, without the rescaled loop momenta, ensure that the measure scales to leading order as
\begin{equation}
 d\mu=\bar{\delta}\left(\varepsilon - \frac{\mathcal{F}_2}{\mathcal{F}_1}\right) \,d\varepsilon\, d\tilde{\mu}\,,
\end{equation}
where the measure $d\tilde{\mu}$ is again independent of $\varepsilon$. Then the same powercounting argument in the degeneration parameter $\varepsilon$ gives the following scaling for the different theories:
\begin{center}
\begin{tabular}{l|c}
 theory & weight $N$ in $\varepsilon$\\ \hline 
 $n$-gon & $N=2n-4$\\
 supergravity & $N=4$\\
 super Yang-Mills & $N=2$\\
 pure gravity & $N=0$\\
 pure Yang-Mills & $N=0$\\
 bi-adjoint scalar & $N=0$
\end{tabular}
\end{center}
The contribution from the singular solutions $\sigma_{\ell^-}=\sigma_{\ell^+} +\varepsilon +\cO(\varepsilon^2)$ to the $n$-gon and the supersymmetric theories thus vanishes, whereas they can clearly be seen to contribute for the bi-adjoint scalar theory and Yang-Mills and gravity in the absence of supersymmetry. Moreover, since the integrands scale as $\varepsilon^0$, the contributions from the singular solutions are clearly finite. However, they evidently do not contribute to the Q-cuts \cite{Baadsgaard:2015twa}, and thus do not contribute to the integrated amplitudes and can be discarded. This complements the discussion given in \cref{sec:general-form-one}.

\section{All-loop integrands}\label{sec:Loops_all-loop}
The ACS proposals have natural extensions to Riemann surfaces $\Sigma^g$ of arbitrary genus $g$ for $g$-loop amplitudes \cite{Adamo:2013tsa,Ohmori:2015sha, Adamo:2015hoa}.  We can again attempt to use residue theorems to localise on a preferred boundary component of the moduli space. Here we choose a basis of $g$ $a$-cycles to contract in $g$ non-separating degenerations, to obtain Riemann spheres  $\Sigma^g_0$ with $g$ nodes, i.e  pairs of double points $(\sigma_r,\sigma_{r'})$, $r=1,\ldots,g$. (We still expect separating degenerations to be suppressed by the remaining scattering equations for generic momenta.) This fixes $g$ of the moduli, and the remaining $2g-3$ moduli are now associated with the $2g$ new marked points modulo M\"obius transformations.  On nodal curves, 1-forms are allowed to have simple poles at the nodes so that the nodal Riemann sphere $\Sigma^g_0$ is endowed with a basis of $g$ global holomorphic 1-forms 
\begin{equation}
 \omega_r=\frac{(\sigma_r-\sigma_{r'})d\sigma}{(\sigma-\sigma_r)(\sigma-\sigma_{r'})}\, .
\end{equation}

For $P$ with poles at $n$ further marked points $\sigma_i$ and residues $k_i$, we have 
\begin{equation}
P=\sum_{r=1}^g \ell_r\omega_r+\sum_i k_i \frac{d\sigma}{\sigma-\sigma_i} \, ,
\end{equation}
where $\ell_r\in \R^d$ are the zero modes in $P$ representing the loop momenta. Setting
\begin{equation}
S(\sigma):=P^2-\sum_{r=1}^g \ell_r^2 \omega_r^2+\sum_{r,s}^g a_{rs}(\ell_r^2+\ell_s^2)\omega_r\omega_s\, ,
\end{equation}
 a quadratic differential with simple poles at all the marked points including $\sigma_r,\sigma_{r'}$,  the multiloop off-shell scattering 
 equations are
\begin{equation}
\mathrm{Res}_{\sigma_i} S=0\, , \quad i=1,\ldots ,n+2g \,,
\end{equation}
where $i$ now ranges over all the marked points. The coefficients $a_{rs}$ in the definition of $S(\sigma)$ are fixed uniquely by requiring the amplitude to factorise correctly.  We have as before three relations between the scattering equations arising from the vanishing of the sum of the residues of $\sum\sigma_\alpha \sigma_\beta S=0$.  Thus  if we impose $n+2g-3$ of them, the remaining ones must also be satisfied, so that $S$ is holomorphic and, being of negative weight, vanishes.

\begin{figure}[ht]
\begin{center}
  \begin{tikzpicture} [scale=2.5]
  \draw (0,2.6) to [out=90, in=180] (0.15,2.75) to [out=0, in=180] (0.35,2.67) to [out=0, in=180] (0.8,2.78) to [out=0,in=90] (1.1,2.49) to [out=270,in=0] (0.8,2.2) to [out=180,in=0] (0.3,2.5) to [out=180,in=0] (0.15,2.45) to [out=180, in=270] (0,2.6);
  \draw (0.09,2.61) arc [radius=0.17, start angle=240, end angle=300];
  \draw (0.23,2.6) arc [radius=0.17, start angle=70, end angle=110];
  \draw (0.65,2.65) arc [radius=0.17, start angle=240, end angle=300];
  \draw (0.79,2.64) arc [radius=0.17, start angle=70, end angle=110];
  \draw (0.75,2.36) arc [radius=0.17, start angle=240, end angle=300];
  \draw (0.89,2.35) arc [radius=0.17, start angle=70, end angle=110];
 \draw [fill] (0.11,2.53) circle [radius=.3pt];
 \draw [fill] (0.42,2.61) circle [radius=.3pt];
 \draw [fill] (0.75,2.47) circle [radius=.3pt];
 \draw [fill] (0.94,2.3) circle [radius=.3pt];
 \draw [fill] (0.93,2.68) circle [radius=.3pt];
  \node at (1.35,2.5) {$\rightarrow$};
  \draw (2.0,2.5) circle [x radius=0.4, y radius=0.3];
  \draw [fill,red3] (1.73,2.67) circle [radius=.3pt];
  \draw [fill,red3] (1.69,2.62) circle [radius=.3pt];
  \draw [dashed, red3] (1.73,2.67) to [out=125, in=45] (1.57,2.82) to [out=225, in=155] (1.69,2.62);
  \draw [fill] (1.85,2.32) circle [radius=.3pt];
  \draw [fill] (2.0,2.5) circle [radius=.3pt];
  \draw [fill] (1.93,2.72) circle [radius=.3pt];
  \draw [fill,red3] (2.08,2.75) circle [radius=.3pt];
  \draw [fill,red3] (2.14,2.73) circle [radius=.3pt];
  \draw [dashed,red3] (2.08,2.75) to [out=100,in=160] (2.16,2.96) to [out=340,in=70](2.14,2.73);
  \draw [fill,red3] (2.24,2.33) circle [radius=.3pt]; 
  \draw [fill,red3] (2.19,2.29) circle [radius=.3pt]; 
  \draw [dashed,red3] (2.19,2.29) to [out=300,in=225] (2.35,2.15) to [out=45,in=330] (2.24,2.33);
  \draw [fill] (2.27,2.55) circle [radius=.3pt];
  \draw [fill] (1.67,2.45) circle [radius=.3pt];
 \end{tikzpicture}
 \end{center}
 \caption{Residue theorem at higher genus.}
\end{figure}
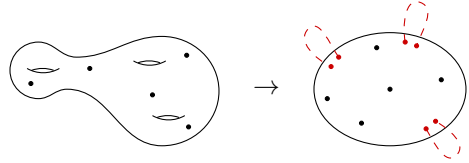

This leads to the following proposal for the all-loop supergravity integrand
\begin{equation}
\hat{\cM}^{(g)}_{SG}=\int_{(\CP^1)^{n+2g}} \frac{\cI^L_0\cI^R_0}{\mathrm{Vol}\, G} \prod_{r=1}^{g} \frac{1}{\ell_r^2} \prod_{i=1}^{n+2g} \bar{\delta}(\mathrm{Res}_{\sigma_i} S(\sigma_i))\, ,
\end{equation}
where $G=\SL(2,\C)^2$ is the residual gauge symmetry of the ambitwistor string. It is fixed in standard Faddeev-Popov fashion by fixing three points to $(0,1,\infty)$ and removing their corresponding delta functions. In this formula, the integrand factors $\cI^R$ and $\cI^L$ depend on the marked points, momenta and polarisation data, and take values in 1-forms in each integration variable. They are most simply defined to be  the sum over spin structures of the worldsheet correlator \cite{Adamo:2013tsa,Adamo:2015hoa} of $n$ type-II supergravity vertex operators on a genus $g$ Riemann surface $\Sigma^g$, then taken to the $g$-fold nodal limit $\Sigma_0^g$.  There is of course much work to be done to make such correlators explicit, but they strongly resemble those that arise in conventional string theory, which have been the object of intensive research. This is done in our context at four points and two loops in \cite{Adamo:2015hoa}.

Similar conjectures can be made for planar super Yang-Mills and for the analogue of bi-adjoint scalar amplitudes at all loops.  We should respectively replace one or both of the $\cI^R$ and $\cI^L$ by  the sum of all Parke-Taylor factors that are compatible with some given ordering of the external particles, but which also  run through all the loops, generalizing the one-loop case. Indeed, one can conceive of conjectures based on the ingredients of \cite{Cachazo:2014xea,Casali:2015vta} for the theories described there.

\section{Discussion}\label{sec:discussion}
\paragraph{Summary.} In giving the theory that underlies the CHY formulae for tree amplitudes, ambitwistor strings gave a route to conjectures for the extension of those formulae to loop amplitudes.  Being chiral string theories, ambitwistor strings potentially have more anomalies than conventional strings, but nevertheless the version appropriate to type II supergravity led to consistent proposals for amplitude formulae at one and two loops 
\cite{Adamo:2013tsa,Casali:2014hfa,Adamo:2015hoa}.  However, the other main ambitwistor string models would seem to have problems on the torus, either due to anomalies, or because the full ambitwistor string theories have unphysical modes associated with their gravity sectors that would propagate in the loops and corrupt for example a  pure Yang-Mills loop amplitude.  Furthermore, once on the torus, it is a moot point as to how much can be done with the formulae, requiring as they do, the full machinery of theta functions.  Issues such as the Schottky problem will make higher genus formulae difficult to write down explicitly.

In this chapter, we have seen that the conjectures of \cite{Adamo:2013tsa,Casali:2014hfa}, with the adjustment to the scattering equations as described in \cref{sec6:general}, are equivalent to much simpler conjectures on the Riemann sphere.  These formulae are now of the same complexity as the CHY tree-level scattering formulae on the Riemann sphere with the addition of two marked points, corresponding to loop momenta insertions.  It is therefore possible to apply methods that have been developed at tree-level on the Riemann sphere here also at one-loop to both extend and prove the conjectures.

In particular, we were able to give strong evidence for 1-loop integrands in super Yang-Mills theory and supergravity, and prove the conjectures for integrands in non-supersymmetric Yang-Mills, gravity and bi-adjoint scalar theory. All these amplitudes share as a fundamental building block the {\it off-shell scattering equations}, and can be written in manifestly SL$(2,\C)$ invariant form (c.f. \cref{1-loopgeneral}),
\begin{equation}
 \cM^{(1)}= -\int d^d \ell \, \frac{1}{\ell^2} \int 
\frac{d\sigma_{\ell^+} d\sigma_{\ell^-} d^n\sigma}{\text{vol}\,G} \;\; \hat{\cI}
\;\bar\delta(\text{Res}_{\sigma_{\ell^+}} S) \bar\delta(\text{Res}_{\sigma_{\ell^-}} S) \prod_i{} \;\bar\delta(\mathrm{Res}_{\sigma_i}S)\,, 
\end{equation}
In \cref{table6:summary}, we have listed the integrands for the most relevant theories described here.

\begin{table}[ht]
\begin{center}
 \begin{tabular}{l!{\color{light-gray-table}\vrule}c!{\color{light-gray-table}\vrule}c}\hline
  Theory & Integrand $\cI^L$ & Integrand $\cI^R$\\ \hline
  supergravity & $\hat{\cI}_{0}$ & $\hat{\cI}_{0}$\\
  super Yang-Mills & $\hat{\cI}^{cPT}$ & $\hat{\cI}_{0}$\\ \hline
  NS-NS gravity & $\hat{\cI}_{d}^{\text{NS}}$ & $\hat{\cI}_{d}^{\text{NS}}$\\
  Einstein gravity & $\hat{\cI}_{\alpha_d}^{\text{NS}}$ & $\hat{\cI}_{\tilde{\alpha}_d}^{\text{NS}}$\\
  Yang-Mills & $\hat{\cI}^{cPT}$ & $\hat{\cI}_{d}^{\text{NS}}$\\
  Bi-adjoint scalar &  $\hat{\cI}^{cPT}$ &  $\hat{\cI}^{cPT}$\\ \hline
 \end{tabular}
 \caption{Integrands for theories at one loop. \label{table6:summary}}
 \end{center}
\end{table}

\noindent
Here, $\cI^{cPT}$ denotes a Parke-Taylor factor running through the loop, and the NS integrand and its supersymmetric counterpart are given by
\begin{subequations}
\begin{align}
&\hat{\cI}_{0}=\hat{\cI}_{d}^{\text{NS}}-\frac{c_d}{\sigma_{\ell^+\ell^-}^2}\pf(M_2)\,,\\
&\hat{\cI}_{d}^{\text{NS}}=\frac{1}{\sigma_{\ell^+\ell^-}}\sum_r\pf'(M_{d}^r)\,,\\
& \hat{\cI}^{cPT}=\sum_{\rho\in S_n}\mathcal{C}_{n+2}(\sigma_{\ell^+},\rho,\sigma_{\ell^-})\,.
\end{align}
\end{subequations}

\paragraph{Discussion and Outlook.} Fixing the marked points associated to the loop momenta, $\sigma_{\ell^\pm}$, gives rise to $(n-1)!-2(n-2)!$ solutions to the scattering equations for an $n$ particle amplitude at one-loop.   This counting was more clearly understood in \cite{He:2015yua}:  the $(n-1)!$ is the number of solutions that one obtains for $n+2$ points on the sphere with arbitrary null momenta at $n$ points, and off-shell momenta at the remaining two points (all summing to zero).  If one takes the forward limit in which  the two off-shell momenta become equal and opposite, one finds that there are two classes of  $(n-2)!$ degenerate solutions, in which the  two loop insertion points  come together (or alternatively all the other points come together); the two classes are distinguished by the rate at which the points come together as the forward limit is taken.   In the forward limit which corresponds to the one-loop amplitude, the most degenerate  class no longer applies but, in general, we should consider the other, leaving $(n-1)-(n-2)!$ solutions.  For amplitudes in supersymmetric Yang-Mills and gravity, these degenerate solutions give a vanishing contribution to the loop integrand. However, they do contribute in the case of the bi-adjoint scalar theory, as shown in \cite{He:2015yua}, and they also contribute in the cases of non-supersymmetric Yang-Mills and gravity presented in this chapter. 
 
As seen in \cref{sec:factorization}, the degenerate solutions do not contribute to the Q-cuts.  So,  to arrive at a loop integrand that computes the correct amplitude under dimensional regularisation, we  simply discard them in our proposed formulae also.  Having discarded these terms, our formulae will not necessarily give the integrand itself as a sum of Feynman diagrams.  In the bi-adjoint scalar theory for example, there will be terms that look like tree amplitudes with bubbles on each external leg that vanish under dimensional regularisation.  These terms are correctly computed if the degenerate solutions are included as shown by \cite{He:2015yua}.  It would be interesting to see if this persists for all our formulae as we have seen that they make sense on the degenerate solutions. 

It should also be possible to prove our one-loop formulae for supersymmetric theories via factorisation.  The  gap in our argument is that we do not have a good closed-form formula for the Ramond sector contributions at tree level, as would be required to prove factorisation.  Our  representation of the Ramond sector in the loop as the Pfaffian of $M_2$ should provide some hint as to how to do this, and further insights could possibly come from the pure spinor formalism.

Ideally, there should be no need to solve the scattering equations explicitly. The main result of this chapter, which relies on the use of a residue theorem to localise the modular parameter, was inspired by \cite{Dolan:2013isa}, where the tree-level CHY integrals were computed by successive applications of residue theorems, rather than by solving the scattering equations. The way forward is to use the map between integrals over the moduli space of the Riemann sphere and rational functions of the kinematic invariants, which is implicit in the scattering equations. Recently, there has been intense work on making this map more practical \cite{Cachazo:2015nwa,Baadsgaard:2015voa,Baadsgaard:2015ifa,Baadsgaard:2015hia,Kalousios:2013eca,Kalousios:2015fya,Cardona:2015eba,Cardona:2015ouc,Sogaard:2015dba,Huang:2015yka,Lam:2014tga,Lam:2015sqb,Dolan:2015iln}. We expect that this will make the use of our formulae much more efficient.\\

It was argued in \cref{sec:Loops_all-loop} that the scheme explored in detail at one-loop has a natural extension to all loops.  Similarly, the Q-cut formalism of \cite{Baadsgaard:2015twa} also has a natural extension to all loops.  It will be interesting to see  whether the factorisation strategy presented in this chapter can be extended to give a correspondence with the Q-cut formalism at higher loop order.  Obtaining better control of higher loop Pfaffians will be crucial for  using these ideas to understand gauge theories and gravity. A formulation as correlators on the Riemann sphere, as suggested by our introduction of $M_{\text{NS}}$, may play a key role.  

Proving this all-loop proposal would be of utmost interest, since it would effectively solve the theories in question, albeit perturbatively, by reformulating an $n$ particle $g$ loop amplitude as an object of the same complexity as an $n+2g$ particle tree amplitude. However, new features emerge already at two loops, and an extension to three loops becomes highly non-trivial. In particular, the two-loop scattering equations are manifestly off-shell, with coefficients $a_{rs}=1$. A resolution to these difficulties would be of high impact, with the potential to resolve the long-standing pressing question regarding the UV-behaviour of maximal supergravity.

\chapter{Conclusion}
I have demonstrated in this thesis both the wide-ranging impact of ambitwistor strings on the study of tree-level amplitudes and its new insights at loop-level. With the scattering equations as the backbone, there exists now a wide variety of formulae for scattering amplitudes derived from underlying ambitwistor string theories. The great number of formulae is due to the flexibility of the ambitwistor string approach, allowing both for different representations (see \cref{chapter4,chapter5} and \cite{Geyer:2014lca,Adamo:2014yya,Adamo:2015fwa,Geyer:2014fka,Berkovits:2013xba,Gomez:2013wza,Adamo:2015hoa}) and a variety of different theories (see \cref{chapter3} and \cite{Cachazo:2014nsa,Cachazo:2014xea,Ohmori:2015sha}). In four dimensions, the ambitwistor string not only leads to the simplest known expression of amplitudes for any degree of supersymmetry, but also makes contact with twistor strings for Yang-Mills and gravity, framing the ambitwistor string in arbitrary dimensions as the natural generalisation of the twistorial ideas of $d=4$.  Moreover, I have demonstrated the strength of the target space geometric structure in studying the asymptotic symmetries of the S-matrix and their relation to the low-energy behaviour of the theory, both in arbitrary dimensions and the twistorial representation in four-dimensional space-time.

At loop level, I have demonstrated that further simplifications occur due to the localisation on the scattering equations. In particular, a contour integration argument in the fundamental domain reduces the computationally challenging ambitwistor higher-genus expressions to simpler formulae on nodal Riemann spheres. These inherit both the flexibility and simplicity of tree-level amplitudes. I have shown this explicitly by constructing integrands for supersymmetric and non-supersymmetric  Einstein gravity and Yang-Mills theory, the former giving compelling evidence for the validity of the ambitwistor string on the elliptic curve, while the latter was proven systematically using factorisation properties. The proposal of the all-loop integrand concluding this chapter gives a small insight into where this research might be headed - if proven, it would extend this widely applicable framework to reduce higher-loop integrals to formulae on nodal Riemann spheres, of a complexity comparable to tree-amplitudes with two more particles for each loop.
\\

This thesis represents a small step towards a new understanding of quantum field theories. In this last paragraph, I will aim to set this into a wider context by giving a short overview over some recent spectacular advances originating in ambitwistor string theory, 
with a focus on long-term goals of research in this field.

One of the most intriguing features of the ambitwistor worldsheet models is their relation to string theory. Despite the similarity in structure, string theories depend - via the string length $\sqrt{\alpha'}$ - on an additional parameter, and the ambitwistor worldsheet models are more readily understood as holomorphic complexification of worldline formulations. In particular, they describe perturbative general relativity, and thus fail to make sense at high energies, whereas string theory as a theory of quantum gravity is well-defined. Another piece of this puzzle are the scattering equations - originally discovered in string theory, they feature most prominently in the ambitwistor models. 
Many recent insights into the mathematical structure of both string theory and ambitwistor strings provide an angle on this puzzle. Most noteworthy in this context is the work on the contours used to localize on the scattering equations \cite{Ohmori:2015sha,Witten:2010cx}, and recent progress on the twistorial origin of the pure spinor string \cite{Berkovits:2014aia,Berkovits:2015yra}. Moreover, ambitwistor strings have been shown \cite{Adamo:2014wea} to encode the full non-linear geometry of supergravity in the ambitwistor string current algebra, providing a key step towards understanding non-perturbative aspects from an ambitwistor string approach.
Research on the relation between ambitwistor models and string theory, as well as non-perturbative aspects, has the scope to not only transform our understanding of quantum field theories, but could also provide new insights into the geometric properties of string theory and its relation to physics. 

Ambitwistor strings are therefore a very active area of research, and I am looking forward to the insights, surprises and discoveries of the coming years.
\begin{acknowledgements}
First and foremost, I would like to thank my supervisor Lionel Mason, for everything I have learnt over the last few years and the joy of working with him, his mentorship, guidance and continuous encouragement and support. For slowly guiding me from theoretical physics into mathematical physics by talking about bundles and cohomology classes from day one. For dealing with my emails in the middle of the night without showing too much concern the next morning, and for providing a point of sanity in all our discussions with his reliable coffee breaks. Most of all though, for his endless enthusiasm, and sharing the ups and downs of research with me.\smallskip

I would also like to thank my collaborators, Ricardo Monteiro, Piotr Tourkine, Arthur Lipstein, Eduardo Casali and Kai Roehrig -- without whom this thesis would not be what it is -- for our long discussions sorting out both ideas and technical details, and for their patience with my over-excitement as well as my concerns. I have learnt a lot from and with them, and their insightful ideas, enthusiasm  and practical advice were an essential part of my time in Oxford.
\smallskip

Freddy Cachazo  introduced me to this area of research, and his enthusiasm and intuition are unparalleled. I am also thankful for spending four months of my DPhil working with him, many interesting discussions, his encouragement, and, maybe most widely applicable, learning to always `give it a shot'.
\smallskip

David Skinner, Tim Adamo and Tomasz Lukowski have always been there to listen and give advice, and I am grateful for many insightful comments, discussions and constructive criticisms. \smallskip

Finally, I would like to thank my parents and family, for their unconditional support and tireless proofreading.\smallskip

My DPhil in Oxford has been supported by an EPSRC DTA award and the Mathematical Prizes Fund.

 
 


\end{acknowledgements}


\appendix
\chapter{New Models}
\section{Correlators for \texorpdfstring{$S_{YM,\Psi}$}{S-YM,Psi}}\label{sec:ecomb-proof}
In this appendix, we provide the  proof of  theorem \ref{ym-correlator}. While several versions for Einstein-Yang-Mills theory are realised in \cref{chapter3}, we will demonstrate and prove the mechanism in the simplest setting containing all necessary ingredients, and comment on adaptations and restrictions afterwards. Concretely we use the action $S_{YM,\Psi}$describing a single free fermion $\rho^a $ and a generic level zero current $j^a$. The fields have the same OPEs as above, that is $j^a$ form a current algebra and $\rho^a$ are in the adjoint presentation of the $j$-algebra, i.e.
\begin{equation}\label{eqn:appendixA-OPEs}
\rho^a(\sigma) \rho^b(0) \sim \frac{1}{\sigma} \delta^{ab}  ~, \qquad j^a(\sigma) j^b(0) \sim \frac{1}{\sigma} f^{abc} j^c  ~, \qquad j^a(\sigma) \rho^b(0) \sim \frac{1}{\sigma} f^{abc} \rho^c  ~.
\end{equation}

The strategy for proving the theorem is as follows: both the tree-level CHY amplitude $\cM_{\text{CHY}}(g,h)$ and the world-sheet correlator $\cM(g,h)$ derived from the ambitwistor string are sums of simple terms. The sum in $\cM_{\text{CHY}}(g,h)$ is over trace sectors as well as a choice of gluon labels, while the sum in $\cM(g,h)$ is simply the Wick expansion of the expectation value, schematically expressed by
\begin{equation}
\mathcal{A} = \sum_{x \in X} \mathcal{A}(x) \qquad \qquad \text{and} \qquad \qquad A = \sum_{y\in Y} A(y)\,.
\end{equation}
Demonstrating that $x \in X \Rightarrow x \in Y$ and $y \in Y \Rightarrow y \in X$, as well as uniqueness of each element leads to $X=Y$. Along the way we will see that $\cM(x) = \cM_{\text{CHY}}(x)$, hence establishing $\cM = \cM_{\text{CHY}}$.

To clarify the structure of the discussion we firstly only insert integrated vertex operators on the world-sheet -- which corresponds to considering the full Pfaffian in the CHY formula -- keeping in mind that to get a non-vanishing result we need to go over to the reduced Pfaffian. That step will be taken at the end.

Let us examine the correlation function of two types of operators,
\begin{equation}\label{eqn:appendixA-one}
\mathcal{O}^{gl} = k \cdot \Psi ~ t \cdot \rho + t \cdot j \qquad \text{and} \qquad  \mathcal{O}^{gr} = k \cdot \Psi ~ \epsilon \cdot \Psi + \epsilon \cdot P  ~,
\end{equation}
for (one half of) the gluon and graviton integrated vertex operators respectively. The claim is that the ambitwistor string string worldsheet correlator
\begin{equation}\label{eqn:appendixA-two}
\cM( g , h) := \left\langle ~ \prod_{a \in g} \mathcal{O}^{gl}_a  ~ ~ \prod_{b \in h}  \mathcal{O}^{gr}_b  ~ \right\rangle
\end{equation}
where $g$ and $h$ are the sets containing the gluon and graviton labels respectively, is equal to (one part of the CHY representation of) the tree-level amplitude
\begin{equation}\label{eqn:appendixA-three}
\cM_{\text{CHY}} = \sum_{\text{trace sectors}} \mathcal{C}_1 \cdots \mathcal{C}_m ~ \pf \Pi ~,
\end{equation}
where the sum ranges over all possible trace sectors, including a sum over the number of traces $m = 1, \cdots , [ |g| /2 ]$. The matrix $\Pi$, defined in \cite{Cachazo:2014xea}, of course depends on this trace sector.

The main step in going between the representations \cref{eqn:appendixA-two} and \cref{eqn:appendixA-three} is the identity \cref{eqn:Comb-Trace-relation}, which we repeat here for the readers convenience
\begin{equation}\label{eqn:appendixA-four}
\sigma_{ab} ~ \mathcal{C} (T) = \mathcal{K}( b,a | T) ~,
\end{equation}
where the `comb structure' $\mathcal{K}$ was defined in the main text. Using the anti-symmetry and multi-linearity of the Pfaffian, \cref{eqn:appendixA-three} can be recast as
\begin{equation}\label{eqn:appendixA-five}
\cM_{\text{CHY}} = \sum_{\substack{\text{trace}\\ \text{sectors}}}  \sum_{\substack{a_1 < b_1 \in T_1 \\ \cdots \\ a_m < b_m \in T_m }}  \mathcal{K}(a_1 ,b_1 | T_1) \cdots  \mathcal{K}(a_m ,b_m | T_m)  ~~ \pf M( h , \{a_i \} , \{b_i\} | h ) ~.
\end{equation}
This is the representation of the amplitude which the world-sheet correlator \cref{eqn:appendixA-two} will land us on.

Let us now evaluate the correlator $\cM(g,h)$. We will see that it gives rise to a multiple sum over terms, which turn out to be the same that \cref{eqn:appendixA-five} sums over. The first step is to expand the product of all vertex operators $\mathcal{O}^{gl}_a$ into a sum over $2^{|g|}$ terms corresponding to factors of $k\Psi \rho$ or a $j$ for each gluon. Labelling the set of gluons with $k\Psi \rho$ insertions by $e$, the path integral over the $\Psi $ field can be performed for each term individually, leading to $\pf M( h , e | h )$. Since $\Psi$ is fermionic, the path integral vanishes unless $|e|=2m$ is even. The correlator $\cM$ is now a sum over ways of partitioning $g$ into $e$ and $g-e$, with the condition that $|e|$ be even, and each term is given by\footnote{From now onwards we omit the colour structure and abbreviate $t_a \cdot \rho (\sigma_a) = \rho_a$ and $t_a \cdot j (\sigma_a) = j_a$.}
\begin{equation}\label{eqn:appendixA-seven}
\left\langle ~ \prod _{a \in e} \rho_a  ~ \prod_{a \in g-e} j_a ~ \right\rangle ~~ \pf M( h , e | h ) ~.
\end{equation}
The remaining worldsheet correlator evidently gives rise to the product of $\mathcal{K}$s and the remaining sum over partitions. Care is needed due to the fermionic nature of the $\rho$ insertions, the resulting factors of $(-1)$ are absorbed into a reordering of the rows and columns of $M$:
\begin{equation}
\sum ~  \mathcal{K}(a_1 , b_1 | T_1) \cdots \mathcal{K}(a_m , b_m | T_m) ~~ \pf M( h , \{ a_1 ,b_1, \cdots , a_m,b_m\} | h) ~,
\end{equation}
which is precisely the summand appearing in the full space-time amplitude. We repeat that Wick expansion ensures both that every possible configuration is summed over exactly once. 


We have shown that the expressions $\cM$ and $\cM_{\text{CHY}}$ are sums over the same simple terms, involving $\mathcal{K}$s and the corresponding $\pf M$. To clarify the differences, the sum in the correlator $\cM$ ranges over different ways of choosing $m$ pairs  from $g$ and different ways of forming $m$ unordered sets $T_i$ from the labels left over, as well as the sum over $m$ - indicated above by the set $X$. On the other hand, the sum in $\cM_{\text{CHY}}$ ranges over ways of splitting the labels $g$ into $m$ unordered subsets $T_i$ and picking a pair from each subset, as well as the sum over $m$ - corresponding above to $Y$. The set of these choices is $Y$.  Indeed, these sums are actually identical: each term in $\cM$ has a counterpart in $\cM_{\text{CHY}}$ and vice versa. Moreover, Wick's theorem and the construction of terms in the CHY representation guarantees that each element is unique, and thus
\begin{equation}
 \cM=\cM_{\text{CHY}}\,.
\end{equation}

\subsection{The reduced Pfaffian}

 The Pfaffian we discussed so far actually vanishes for physical systems, i.e. when momentum conservation, gauge invariance and the scattering equations hold. Hence it is replaced by the reduced Pfaffian  $\pfr  \Pi$ defined in either of the following equivalent ways
 \begin{equation}
\label{eq:redpf}
 \pfr \Pi := \pf \Pi_{i,j^\prime} = \frac{(-)^a}{\sigma_a} \pf \Pi_{a,i}  = - \frac{(-)^a}{\sigma_a} \pf \Pi_{a,j^\prime} = \frac{(-)^{a+b}}{\sigma_{ab}} \pf \Pi_{a,b}
 \end{equation}
where $a,b$ label gravitons, with the restriction to not remove any row/column of the matrix $B$, and the $i, j^\prime$ label traces. In the ambitwistor string, this corresponds to the observation that BRST invariance ensures invariance of the amplitude under the choice of fixed vertex operators.  Hence, if there are at least two gravitons and arbitrarily many gluons, the amplitude must be equal to the CHY formula. The validity of \cref{eq:redpf} can also be shown explicitly in the case of fixed vertex operators for two gluons or one gluon and one graviton.

\paragraph{Two gluons fixed.}
Denote the labels of the fixed gluon operators as $c,d$. With the reduced Pfaffian defined as
\begin{equation}
\pfr \Pi = \pf \Pi_{i,j^\prime}
\end{equation}
there are two cases, $j^\prime = i$ or $j^\prime \neq i$. In the first case the trace $T_i$ is totally removed from the Pfaffian and we can write
\begin{equation}
\cdots \, \mathcal{C}_i \, \cdots ~ \pf \Pi_{i,i^\prime} = \frac{1}{(d \, c)} \, \cdots  \,  \mathcal{K}(c,d | T_i) ~ \pf \Pi_{i,i^\prime}  ~,
\end{equation}
with the gluons $c,d$ being members of the trace $T_i$. The factor $\frac{1}{(d \, c)}$ fits into the interpretation of \cite{Mason:2013sva} as ghost field correlator. Note that there is no sum over choices of pairs in $T_i$, instead the comb $\mathcal{K}$ appears with fixed start/end points, corresponding to the insertion of fixed vertex operators for the gluons $c,d$.

In the second case ($j^\prime \neq i$), name the traces such that $c \in T_1$ and $d \in T_2$. Now each term in the expansion of the worldsheet correlator will look like (omitting all irrelevant factors)
\begin{equation}
\begin{aligned}
  \frac{1}{\sigma_{cd}} \sum_{\substack{a\in T_1 \\ b \in T_2}} & \mathcal{K}(c,a | T_1)      \mathcal{K}(d,b | T_2) ~ \pf ( a , b , \cdots  )   =  \mathcal{C}(T_1) \, \mathcal{C}(T_2) \, \sum_{\substack{a\in T_1 \\ b \in T_2}}  \frac{\sigma_{ac} \sigma_{bd} }{\sigma_{cd} } ~ \pf ( a , b , \cdots  ) \\
  &    =  \mathcal{C}(T_1) \, \mathcal{C}(T_2) \, \sum_{a \in T_1 }  \frac{\sigma_{ac} }{\sigma_{cd} } ~ \pf ( a , \sum _{b \in T_2 } \sigma_{bd} \,  b , \cdots  )    \\ 
  &  =  \mathcal{C}(T_1) \, \mathcal{C}(T_2) \, \sum_{a \in T_1 }  \frac{\sigma_{ac} }{\sigma_{cd} } ~ \pf ( a , - \sum _{b \in T_1 } \sigma_{bd} \,  b , \cdots  )  \\  
& =  \mathcal{C}(T_1) \, \mathcal{C}(T_2) \, \sum_{ a< b \in T_1 }  \sigma_{ba}  ~ \pf ( a , b , \cdots  )    \equiv  \mathcal{C}(T_1) \, \mathcal{C}(T_2) ~ \pf \Pi_{2,2^\prime} ~.
\end{aligned}
\end{equation}
Note that we had to use the scattering equations and the antisymmetry of the Pfaffian to arrive at the final result.

\paragraph{One gluon, one graviton fixed}
The computation for fixing one gluon and one graviton vertex operator is largely analogous to the previous one. Moreover, BRST invariance guarantees that the final result will be as desired. Let us nevertheless demonstrate the necessary manipulations. Denote the fixed gluon by $c$, with $c \in T_1$, and the fixed graviton by $m$
\begin{equation}
\begin{aligned}
 \frac{1}{\sigma_{mc}} \sum_{a\in T_1} \mathcal{K}(a,c |T_1) &~ \pf ( a , \cdots , \check{m} , \cdots )  = \mathcal{C}(T_1) \sum_{a\in T_1 } \frac{\sigma_{ca}}{\sigma_{mc}} ~ \pf ( a , \cdots , \check{m} , \cdots ) \\
& = \mathcal{C}(T_1) \frac{1}{\sigma_{mc}}   ~ \pf (\sum_{a\in T_1 }  \sigma_{ca} \, a , \cdots , \check{m} , \cdots ) \\
&  = \mathcal{C}(T_1) \frac{1}{\sigma_{mc}}   ~ \pf ( - \sigma_{cm} \, m , \cdots , \check{m} , \cdots )   \\
& = \mathcal{C}(T_1)   ~ \pf (   m , \cdots , \check{m} , \cdots )   \equiv \mathcal{C}(T_1)   ~ \pf \Pi _{1,1^\prime} ~.
\end{aligned}
\end{equation}
Again we had to make use of the scattering equations.

\subsection{Adaption and restriction}
As mentioned in the text, it seems not to be possible to find a level zero current via descent from $\rho $ in a consistent way. Hence, the main text contains an adaption of the system discussed above, using two fermions $\rho^a , \tilde{\rho}^a$, conjugate to each other. Via the descent, $\rho^a$ gives rise to $j^a$ while $\tilde{\rho}^a$ gives rise to $\tilde{j}^a$. The OPEs between the currents and the fields are
\begin{equation}
\begin{aligned}
& \rho^a(z) j^b(0) \sim \frac{1}{z} f^{abc} \rho^c ~, ~~ \tilde{\rho}^a(z) j^b(0) \sim \frac{1}{z} f^{abc} \tilde{\rho}^c ~,~~  \\
&\rho^a(z) \tilde{j}^b(0) \sim \frac{1}{z} f^{abc} \tilde{\rho}^c ~, ~~ \tilde{\rho}^a(z) \tilde{j}^c(0) \sim 0 ~.
\end{aligned}
\end{equation}
We shall now examine the correlators of this system.

First, note that by taking the fixed vertex operators to be $(\rho + \tilde{\rho} )$, the discussion above would carry over verbatim. There is a crucial difference however: the current appearing in the associated integrated vertex operator does not lead to $\mathcal{K}(a , b | T)$, but instead gives
\begin{equation}
 \left\langle ~ (\rho_{a} + \tilde{\rho}_a ) \,(\rho_{b} + \tilde{\rho}_b )  ~ \prod_{c \in T} ( j_{c} + \tilde{j}_c) ~\right\rangle  = ~ |T| \,  ~ \mathcal{K}(a , b | T) ~.
\end{equation}
So each contribution from a different trace sector will come with a different prefactor $\prod_i^m |T_i|$, spoiling the relative coefficient between partial amplitudes. As the prefactor depends on the given partition of particles into traces, it cannot be removed by a field rescaling. The origin of this factor can be understood by simply counting the ways in which a full comb can be generated. Observe that each contraction must have exactly one insertion of $\tilde{j}$ or $\tilde{\rho}$  independent of the length $n$ of the chain, while there are $n-1$ insertions of $j$ or $\rho$. Summing over the possible positions of the tilded operator in the chain gives rise to the over-counting by $|T|$. Note that each contraction contributes exactly the same analytical and colour structure.

Having understood the (non--trivial) origin of the factor $|T|$, the remainder of the discussion, showing how to remove it, follows trivially. Denote $v$ the vertex operator containing $\rho$ and $j$ and  $\tilde{v} $ the one containing $\tilde{\rho}$ and $\tilde{f}$, either integrated or fixed. It is now clear that choosing to insert $\tilde{v}$ at $m$ of the gluon punctures and $v$ at the others will give rise (following the general discussion above) to the complete color ordered partial amplitude with $m$ traces
\begin{equation}
\mathcal{C}(T_1) \cdots \mathcal{C}(T_m) ~ \pfr \Pi ~,
\end{equation}
which concludes the discussion.
\chapter{Details of the correlators with soft limits}

Integrated graviton vertex operators implement symplectic diffeomorphisms of $T^*\scri$ in the worldsheet ambitwistor string theory. We have seen explicitly how these vertex operators can be expanded in powers of the soft momentum, and have identified the leading and subleading terms as generators of supertranslations and superrotations on $\scri$. Analogous results hold for Yang-Mills theory. In this appendix, we deduce the associated Ward identities, both in the $d$ dimensional model and the four dimensional twistorial model, from the worldsheet CFT of the ambitwistor string.  In particular, we compute correlators with insertions of supertranslation/superrotation generators and obtain the leading/subleading terms in the soft theorems for Yang-Mills and gravity. 

\section{Yang-Mills soft limits in \texorpdfstring{$d$}{d}-dimensional model} \label{ymd}

\subsection*{Leading terms}

Let $\epsilon$ and $s$ be the polarization and momentum
of a soft gluon. If we expand the vertex operator in $s$, the leading
term corresponds to the generator of a singular gauge transformation that only
depends on $p$.  This is the gauge analogue of a  supertranslation and
we denote it by $\mathcal{V}_s^{ym,0}$:
\[
\mathcal{V}_s^{ym,0}=\frac{1}{2 \pi i}\oint d\sigma_{s}\frac{\epsilon\cdot
  p(\sigma_{s})}{s\cdot p\left(\sigma_{s}\right)}j(\sigma_{s})\,,
\]
where $j(\sigma_{s})$ is the worldsheet current algebra contracted
with an
element of the corresponding Lie algebra. Since we are dealing with
color-stripped amplitudes, we will leave out generators of the Lie
algebra and simply take the single trace term when we take the
correlation function. 

Consider the correlator of a soft gluon with $n$ other gluons. This is given by
\begin{equation*}
\left\langle \mathcal{V}_{1}...\mathcal{V}_{n} \mathcal{V}_s^{ym,0} \right\rangle =\frac{1}{2 \pi i}\sum_{j=1}^{n}\left\langle \mathcal{V}_{1}...\mathcal{V}_{n}\oint_{\left|\sigma_{s}-\sigma_{j}\right|<\epsilon}d\sigma_{s}\frac{\left(\sigma_{n}-\sigma_{1}\right)}{\left(\sigma_{s}-\sigma_{1}\right)\left(\sigma_{n}-\sigma_{s}\right)}\frac{\epsilon\cdot p(\sigma_{s})}{s\cdot p\left(\sigma_{s}\right)}\right\rangle\,, \label{eq:ym}
\end{equation*}
where $\epsilon \rightarrow 0$. We have used \eqref{p} and the single trace term in the current correlator to obtain a
Parke-Taylor denominator from which we have extracted the soft term.
As the soft gluon vertex operator approaches
one of the other vertex operators, we have 
\begin{equation}
\lim_{\sigma_{s}\rightarrow\sigma_{j}}\frac{\epsilon\cdot p(\sigma_{s})}{s\cdot p\left(\sigma_{s}\right)}=\frac{\epsilon\cdot k_{j}}{s\cdot k_{j}}\,.
\end{equation}
Plugging this into equation \eqref{eq:ym} and performing the contour
integral finally gives the leading order contribution to the soft limit \cref{eq4:Ward-YM},
\begin{equation}
\left\langle \mathcal{V}_{1}...\mathcal{V}_{n} \mathcal{V}_s^{ym,0}\right\rangle =\left(\frac{\epsilon\cdot k_{1}}{s\cdot k_{1}}-\frac{\epsilon\cdot k_{n}}{s\cdot k_{n}}\right)\left\langle \mathcal{V}_{1}...\mathcal{V}_{n}\right\rangle .
\end{equation}

\subsection*{Subleading terms}

Expanding the vertex operator further in $s$, the gauge analogue of
the superrotation generator corresponds the terms linear in $s$ 
\[
\mathcal{V}_s^{ym,1}=Q_R^{orbit}+Q_R^{spin}
\]
where 
\begin{subequations}
\begin{align}
Q_{R}^{orbit}&=\frac{1}{2\pi i}\oint d\sigma_{s}\frac{iq(\sigma_{s})\cdot s\,\epsilon\cdot p(\sigma_{s})}{s\cdot p(\sigma_{s})}j(\sigma_{s})\,,\\
Q_{R}^{spin}&=\frac{1}{2\pi i}\oint
d\sigma_{s}\frac{\epsilon\cdot\Psi(\sigma_{s})s\cdot\Psi(\sigma_{s})}{s\cdot
  p(\sigma_{s})}j(\sigma_{s})\, .
\end{align}
\end{subequations}
Let's compute the correlator
of $\mathcal{V}_s^{ym,1}$ with $n$ other vertex operators. If we focus only on
the delta functions in the other vertex operators, we can neglect
$Q_{R}^{spin}$, since the delta functions do not depend on fermionic
fields. Hence, we only need the following OPE: 
\begin{equation}
is\cdot v(\sigma_{s})\bar{\delta}\left(k_{j}\cdot P(\sigma_{j})\right)=\frac{k_{j}\cdot s}{\sigma_{s}-\sigma_{j}}\bar{\delta}^{(1)}\left(k_{j}\cdot P(\sigma_{j})\right)+...\label{eq:ope2}
\end{equation}
where $\bar{\delta}^{(n)}(x)=\left(\frac{\partial}{\partial x}\right)^{n}\left(x^{-1}\right)$, which follows from \eqref{eq:ope}. Focusing on the delta functions
of the vertex operators and using the above OPE, one easily finds
that 
\[
\left\langle \mathcal{V}_{1}...\mathcal{V}_{n}Q_{R}^{orbit}\right\rangle =\frac{1}{2\pi i}\int d^{2n}\sigma\oint d\sigma_{s}\frac{\epsilon\cdot P\left(\sigma_{s}\right)}{s\cdot P(\sigma_{s})}\frac{\sigma_{n1}}{\sigma_{s1}\sigma_{ns}}\sum_{j=1}^{n}\frac{s\cdot k_{j}}{\sigma_{sj}}\bar{\delta}_{j}^{(1)}\Pi_{a=1,a\neq j}^{n}\bar{\delta}_{a}I_{n}\,,
\]
where $\sigma_{ij}=\sigma_{i}-\sigma_{j}$, $\bar{\delta}_{j}=\bar{\delta}\left(k_{j}\cdot P(\sigma_{j})\right)$,
and $I_{n}$ indicates that the remainder of the integrand does not depend
on $\sigma_{s}$. Note that this integral is precisely equation (19)
of \cite{Schwab:2014xua}. Following the calculations of that paper, we can easily see that this will indeed correspond to the subleading soft limit terms $S^{(1)}$, with the derivatives taken to act exclusively on the scattering equations obtained from the momentum eigenstates in the vertex operators.\\

To obtain the full subleading soft factors, we will have to include the contributions from the correlation function $\left\langle \mathcal{V}_{1}...\mathcal{V}_{n}Q_{R}^{orbit}\right\rangle$, as well as additional contributions from $Q_R^{spin}$. In particular, we find that
\[
\begin{split}
 &\left\langle \mathcal{V}_{1}...\mathcal{V}_{n}Q_{R}^{orbit}\right\rangle \\
 &\qquad=\frac{1}{2\pi i}\int \frac{d^{n}\sigma}{\vol\, \SL(2,\C)}\oint d\sigma_{s}\frac{\epsilon\cdot P\left(\sigma_{s}\right)}{s\cdot P(\sigma_{s})}\frac{\sigma_{n1}}{\sigma_{s1}\sigma_{ns}}\\
 &\qquad\qquad\left(\sum_{j=1}^{n}\frac{s\cdot k_{j}}{\sigma_{sj}}\bar{\delta}_{j}^{(1)}\Pi_{a=1,a\neq j}^{n}\bar{\delta}_{a}\frac{\text{Pf}(M^{(n)})}{\prod_{b}\sigma_{b,b+1}}+\sum_{a=1}^n \frac{s\cdot \epsilon_a}{\sigma_{sa}}\Pi_b'\bar{\delta}_b\frac{\text{Pf}({M^{(n)}}^{a,a+n}_{a,a+n})}{\prod_{b}\sigma_{b,b+1}}\right)\,,
\end{split}
\]
where we denote the CHY matrix obtained from $n$ vertex operator insertions by $M^{(n)}$. Note especially that this does not contain any data of the soft gluon. As mentioned above, additional contributions to the orbital part of the subleading soft limit (in addition to the spin contribution) will originate from the correlation function involving $Q_R^{spin}$;
\[
 \begin{split}
  &\left\langle \mathcal{V}_{1}...\mathcal{V}_{n}Q_{R}^{spin}\right\rangle \\
 &\qquad=\frac{1}{2\pi i}\int \frac{d^{n}\sigma}{\vol \,\SL(2,\C)}\frac{1}{\prod_{b}\sigma_{b,b+1}}\oint d\sigma_{s}\frac1{s\cdot P(\sigma_{s})}\frac{\sigma_{n1}}{\sigma_{s1}\sigma_{ns}}\\
 &\qquad\qquad\sum_{a,b}(-1)^{a+b}\Bigg(\frac{\epsilon\cdot k_a}{\sigma_{sa}}\frac{s\cdot\epsilon_b}{\sigma_{sb}}\text{Pf}({M^{(n)}}^{a,b+n}_{a,b+n})
 -\frac{s\cdot k_a}{\sigma_{sa}}\frac{\epsilon\cdot \epsilon_b}{\sigma_{sb}}\text{Pf}({M^{(n)}}^{a,b+n}_{a,b+n})\\
 &\qquad\qquad+\frac{\epsilon\cdot k_a}{\sigma_{sa}}\frac{s\cdot k_b}{\sigma_{sb}}\text{Pf}({M^{(n)}}^{a,b}_{a,b})-\frac{\epsilon\cdot \epsilon_a}{\sigma_{sa}}\frac{s\cdot\epsilon_b}{\sigma_{sb}}\text{Pf}({M^{(n)}}^{a+n,b+n}_{a+n,b+n})\Bigg)\,.
 \end{split}
\]
A closer look at the structure and origin of these terms already indicates how to match them to the contributions to the subleading soft limits found in \cite{Schwab:2014xua}. Recall from the original ambitwistor string \cite{Mason:2013sva} that in the correlation functions, the fermionic fields $\Psi$ give rise to the Pfaffians, with the diagonal terms $C_{aa}$ coming from the contributions $\epsilon\cdot P(\sigma)$. An insertion of $Q_{R}^{orbit}$ will therefore contribute the subleading soft limits, where the derivative is taken to act on the scattering equations, as well as an additional term due to the appearance of the soft gluon in the diagonal terms of the matrix $C$. The charge $Q_{R}^{spin}$, on the other hand, will give the remaining contributions of the soft particle in the Pfaffian, as well as the spin contribution $J_{spin,a}^{\mu\nu}=\epsilon_{a}^{[\mu} k_{a}^{\nu]}$, stemming from the double contractions where both soft gluon $\Psi_s$ fields are contracted to the fields $\Psi_a$ of {\it one} external gluon $a$. Combining these terms and following the manipulations described in \cite{Schwab:2014xua}, one then finds the subleading soft limit \cref{eq4:Ward-YM-sub},
\begin{equation}
  \left\langle \mathcal{V}_{1}\dots\mathcal{V}_{n}\mathcal{V}_s^{ym,1}\right\rangle= \left(\frac{\epsilon_{\mu}s_{\nu}J_{1}^{\mu\nu}}{s\cdot k_{1}}-\frac{\epsilon_{\mu}s_{\nu}J_{n}^{\mu\nu}}{s\cdot k_{n}}\right)\left\langle \mathcal{V}_{1}...\mathcal{V}_{n}\right\rangle \,,
\end{equation}
where $J_a^{\mu\nu}=J_{orb,a}^{\mu\nu}+J_{spin,a}^{\mu\nu}$, with $J_{orb,a}^{\mu\nu}=k_{a}^{[\mu}\frac{\partial}{\partial k_{a,\nu]}}$ and $J_{spin,a}^{\mu\nu}=\epsilon_{a}^{[\mu} k_{a}^{\nu]}$.

\section{Gravity soft limits in \texorpdfstring{$d$}{d}-dimensional model} \label{gravd}

\subsection*{Leading terms}
For a soft graviton $s$, we are interested in computing the Ward identitiy associated to the leading order term $\cV_{s}^0$ in the soft expansion of the vertex operator. As we have seen above, this corresponds to a supertranslation on $\scri$. With one insertion of $\cV_{s}^0$, the correlator becomes
\begin{equation*}
\left\langle \mathcal{V}_{1}...\mathcal{V}_{n} \cV_{s}^0\right\rangle =\frac{1}{2\pi i}\sum_{j=1}^{n}\left\langle \mathcal{V}_{1}...\mathcal{V}_{n}\oint_{\left|\sigma_{s}-\sigma_{j}\right|<\epsilon}d\sigma_{s}\frac{(\epsilon\cdot p(\sigma_{s}))^2}{s\cdot p\left(\sigma_{s}\right)}\right\rangle\,, \label{eq:gr}
\end{equation*}
where $\epsilon \rightarrow 0$ and $p(\sigma)$ is given in equation \eqref{p}. When the soft graviton vertex operator approaches one of
the other vertex operators, from \eqref{p} we have 
\[
\lim_{\sigma_{s}\rightarrow\sigma_{j}}\frac{(\epsilon\cdot p(\sigma_{s}))^2}{s\cdot p\left(\sigma_{s}\right)}=\frac{(\epsilon\cdot k_{j})^2}{s\cdot k_{j}\left(\sigma_{s}-\sigma_{j}\right)}.
\]
Plugging this into equation \eqref{eq:gr} and performing the contour
integral yields the Weinberg soft graviton theorem,
\begin{equation}
\left\langle \mathcal{V}_{1}...\mathcal{V}_{n} \cV_{s}^0\right\rangle =\left(\sum_{j=1}^{n}\frac{(\epsilon\cdot k_j)^2}{s\cdot k_{j}}\right)\left\langle \mathcal{V}_{1}...\mathcal{V}_{n}\right\rangle. 
\end{equation}

\subsection*{Subleading terms}
Expanding the soft graviton vertex operator further in $s$, we obtain a term $\cV^1_s$ linear in $s$ which corresponds to the generator of a supertranslation on $\scri$. Note that $\cV^1_s$ is made out of $r_{\mu\nu}J^{\mu\nu}$ which breaks up into an orbital part $q^{[\mu}p^{\nu]}$ and spin part $\Psi_r^\mu\Psi_r^\nu$:
\[
\mathcal{V}_s^{1}=Q_R^{orbit}+Q_R^{spin},
\]
where the orbital and spin contributions are given by
\begin{subequations}
\begin{align}
Q_{R}^{orbit}&=\frac{1}{2\pi i}\oint d\sigma_{s}\frac{i \epsilon\cdot p(\sigma_s)^{[\mu}\epsilon^{\nu]} q(\sigma_{s})_\mu p(\sigma_{s})_\nu}{s\cdot p(\sigma_{s})}\,,\\
Q_{R}^{spin}&=\frac{1}{2\pi i}\oint d\sigma_{s}\frac{\epsilon\cdot p(\sigma_{s})s\cdot\Psi_{1}(\sigma_{s})\epsilon\cdot\Psi_{1}(\sigma_{s})+(1\leftrightarrow2)}{s\cdot p(\sigma_{s})}\,.
\end{align}
\end{subequations}
The correlation functions involving these vertex operators are computed using the OPE \eqref{eq:ope}. 
Related calculations have been performed in detail in \cite{Schwab:2014xua} to compute subleading soft
limits. There, the authors focus on the soft limits of the
delta functions in the CHY formulae, which contributes to
the orbital part of the subleading soft limit. The remainder of the orbital part and the spin part of the
subleading soft limit then comes from analysing the soft limits of
the Pfaffians. Similarly, when we compute the correlation functions
of $\mathcal{V}_s^{1}$ with other vertex operators, we will first focus on the contractions involving the delta
functions of the other vertex operators. This will allow us to make contact with the calculations in \cite{Schwab:2014xua} to demonstrate that $Q_{R}^{orbit}$ indeed generates the correct contributions to the orbital part
of the subleading soft limit. One can then show that $Q_{R}^{spin}$ generates the spin part of the subleading soft limit, as well as the missing contributions to the orbital part. $\mathcal{V}_s^{1}=Q_R^{orbit}+Q_R^{spin}$ will therefore generate the full subleading soft gluon or graviton contribution as discussed in \cite{Cachazo:2014fwa}.

To compute the correlator of $Q_{R}$ with $n$ other vertex operators, we will focus first only on the delta functions in the other vertex operators,
and neglect $Q_{R}^{spin}$. Furthermore, using equation \ref{eq:ope2},
one finds that
\[
\left\langle \mathcal{V}_{1}...\mathcal{V}_{n}Q_{R}^{orbit}\right\rangle =\frac{1}{2\pi i}\int d^{2n}\sigma\oint d\sigma_{s}\frac{\epsilon_{1}\cdot P(\sigma_{s})\epsilon_{2}\cdot P(\sigma_{s})}{s\cdot P(\sigma_{s})}\sum_{j=1}^{n}\frac{s\cdot k_{j}}{\sigma_{sj}}\bar{\delta}_{j}^{(1)}\Pi_{a=1,a\neq j}^{n}\bar{\delta}_{a}I_{n}\,,
\]
where we use the notation defined in the previous subsection. Note that this integral is precisely
equation 23 of \cite{Schwab:2014xua}. Again, the remaining correlation function,
\[
\left\langle \mathcal{V}_{1}...\mathcal{V}_{n}Q_{R}^{spin}\right\rangle\,,
\]
can be calculated along similar lines as in Yang-Mills, described in appendix \ref{ymd}. Following the manipulations outlined in \cite{Schwab:2014xua}, we find indeed the subleading soft graviton limit derived in \cite{Cachazo:2014fwa}
\begin{equation}
\left\langle \mathcal{V}_{1}...\mathcal{V}_{n}\mathcal{V}_s^{1}\right\rangle =\sum_{a=1}^{n}\frac{\epsilon_{\mu\nu}k_{a}^{\mu}s_{\lambda}J_{a}^{\lambda\nu}}{s\cdot k_{a}}\left\langle \mathcal{V}_{1}...\mathcal{V}_{n}\right\rangle \,,
\end{equation}
where $\epsilon^{\mu\nu}=\epsilon_{1}^{(\mu}\epsilon_{2}^{\nu)}$
and $J_{a}^{\mu\nu}$ was defined in appendix \ref{ymd}.

\section{Yang-Mills soft limits in the twistorial model} \label{ym4}

\subsection*{Leading terms}

The action of the worldsheet model for the ambitwistor string is based on the symplectic potential of $\A$, and the singular parts of the OPE of operators in the ambitwistor string is thus given by the Poisson structure on $\A=T^*\scri$. In calculating the soft limits in the twistorial model, the following OPE's of fields in the ambitwistor string will be useful:
\begin{equation}
\lambda_{\alpha}(z)\tilde{\mu}^{\beta}(w)=\frac{\delta_{\alpha}^{\beta}}{z-w}+...,\,\,\,\tilde{\lambda}_{\dot{\alpha}}(z)\mu^{\dot{\beta}}(w)=\frac{\delta_{\dot{\alpha}}^{\dot{\beta}}}{z-w}+...\label{eq:sope}
\end{equation}

Expanding an integrated gluon vertex operator in the soft momentum,
the leading term is given by

\[
\mathcal{V}^{ym,0}_s=\frac{1}{2\pi i}\oint d\sigma_{s}\frac{\left\langle \xi\lambda\left(\sigma_{s}\right)\right\rangle }{\left\langle \xi\,s\right\rangle \left\langle s\,\lambda\left(\sigma_{s}\right)\right\rangle }j\left(\sigma_{s}\right)\,,
\]
where $l=\lambda_{s}\tilde{\lambda}_{s}$ is the soft momentum,
$\xi_{\alpha}$ is a reference spinor, and 
\begin{equation}
\lambda(\sigma)=\sum_{i=1}^{k}\frac{s_{i}\lambda_{i}}{\sigma-\sigma_{i}}\,.\label{eq:lambdas}
\end{equation}
Let us compute the correlator of $\mathcal{V}^{ym,0}_s$ with $k$ negative helicity
vertex operator $\wt{\mathcal{V}}$ and $n-k$ positive helicity
vertex operators $\mathcal{V}$:
\begin{equation}
\left\langle \wt{\mathcal{V}}_{1}...\wt{\mathcal{V}}_{k}\mathcal{V}_{k+1}...\mathcal{V}_{n}\mathcal{V}^{ym,0}_s\right\rangle =\frac{1}{2\pi i}\frac{1}{\left\langle \xi\,s\right\rangle }\left\langle \wt{\mathcal{V}}_{1}...\wt{\mathcal{V}}_{k}\mathcal{V}_{k+1}...\mathcal{V}_{n}\oint d\sigma_{s}\frac{\sigma_{n1}}{\sigma_{ns}\sigma_{s1}}\frac{\left\langle \xi\,\lambda\left(\sigma_{s}\right)\right\rangle }{\left\langle s\,\lambda\left(\sigma_{s}\right)\right\rangle }\right\rangle \,.\label{eq:ymqt}
\end{equation}
Note from equation \ref{eq:lambdas} that
\[
\lim_{\sigma_{s}\rightarrow\sigma_{1}}\frac{\left\langle \xi\,\lambda\left(\sigma_{s}\right)\right\rangle }{\left\langle s\,\lambda\left(\sigma_{s}\right)\right\rangle }=\frac{\left\langle \xi\,1\right\rangle }{\left\langle s\,1\right\rangle }\,.
\]
Furthermore, on the support of the delta functions in $\mathcal{V}_{n}$,
we have 
\[
\lim_{\sigma_{s}\rightarrow\sigma_{n}}\frac{\left\langle \xi\,\lambda\left(\sigma_{s}\right)\right\rangle }{\left\langle s\,\lambda\left(\sigma_{s}\right)\right\rangle }=\frac{\left\langle \xi\,n\right\rangle }{\left\langle s\,n\right\rangle }\,.
\]
Hence, when we evaluate the contour integral in \eqref{eq:ymqt},
the residues at $\sigma_{s}=\sigma_{1}$ and $\sigma_{s}=\sigma_{n}$
give the soft graviton contribution to leading order
\[
\left\langle \wt{\mathcal{V}}_{1}...\wt{\mathcal{V}}_{k}\mathcal{V}_{k+1}...\mathcal{V}_{n}\mathcal{V}^{ym,0}_s\right\rangle =\frac{\left\langle 1\,n\right\rangle }{\left\langle s\,1\right\rangle \left\langle s\,n\right\rangle }\left\langle \wt{\mathcal{V}}_{1}...\wt{\mathcal{V}}_{k}\mathcal{V}_{k+1}...\mathcal{V}_{n}\right\rangle \,,
\]
where we have used the Schouten identity.

\subsection*{Subleading terms}

Expanding the gluon vertex operator further to first order in the soft momentum
gives
\[
\mathcal{V}^{ym,1}_s=\frac{1}{2\pi i}\oint d\sigma_{s}\frac{\left[\mu\left(\sigma_{s}\right)s\right]}{\left\langle s\,\lambda(\sigma_{s})\right\rangle }J\left(\sigma_{s}\right)\,.
\]
Note that there is subtlety in defining this operator, since the equations
of motion for the $\tilde{\lambda}$ field imply that $\mu=0$. On
the other hand, $\mu$ will have non-zero contractions with the $\tilde{\lambda}$
fields which appear in the delta functions of other vertex operators,
so correlation functions of $\mathcal{V}^{ym,1}_s$ will be non-vanishing. In particular,
from \eqref{eq:sope}, we see that 
\[
\left[\mu\left(\sigma_{s}\right)s\right]\bar{\delta}^{2}\left(\tilde{\lambda}_{i}-t_{i}\tilde{\lambda}\left(\sigma_{i}\right)\right)=\frac{1}{\sigma_{s}-\sigma_{i}}\tilde{\lambda}_{s}\cdot\frac{\partial}{\partial\tilde{\lambda}\left(\sigma_{i}\right)}\bar{\delta}^{2}\left(\tilde{\lambda}_{i}-t_{i}\tilde{\lambda}\left(\sigma_{i}\right)\right)+...
\]
where we omitted non-singular terms. 
The subleading contribution to the soft gluon will arise from the correlator of $\mathcal{V}^{ym,1}_s$ with $k$ negative helicity
vertex operator $\wt{\mathcal{V}}$ and $n-k$ positive helicity
vertex operators $\mathcal{V}$:
\begin{equation*}
\left\langle \tilde{\mathcal{V}}_{1}...\tilde{\mathcal{V}}_{k}\mathcal{V}_{k+1}...\mathcal{V}_{n}\mathcal{V}^{ym,1}_s\right\rangle =\frac{1}{2\pi i}\int d^{2n}\sigma\oint d\sigma_{s}\frac{1}{\left\langle s\,\lambda(\sigma_{s})\right\rangle }\frac{\sigma_{n1}}{\sigma_{ns}\sigma_{s1}}\sum_{i=1}^{k}\frac{1}{\sigma_{si}}\tilde{\lambda}_{s}\cdot\frac{\partial}{\partial\tilde{\lambda}\left(\sigma_{i}\right)}I_{n}\,,\label{eq:ymr}
\end{equation*}
where $I_{n}$ indicates that the rest of the integrand does not depend
on $\sigma_{s}$. Noting that
\[
\lim_{\sigma_{s}\rightarrow\sigma_{1}}\frac{1}{\left\langle s\,\lambda(\sigma_{s})\right\rangle }=\frac{\sigma_{s1}}{s_{1}\left\langle s\,1\right\rangle }\,,
\]
the residue at $\sigma_{s}=\sigma_{1}$ gives to
\[
\frac{1}{\left\langle s\,1\right\rangle }\tilde{\lambda}_{s}\cdot\frac{\partial}{\partial\tilde{\lambda}_{1}}\left\langle \wt{\mathcal{V}}_{1}...\wt{\mathcal{V}}_{k}\mathcal{V}_{k+1}...\mathcal{V}_{n}\right\rangle\, .
\]
Furthermore, the residue at $\sigma_{s}=\sigma_{n}$ corresponds to
\[
\int d^{2n}\sigma\frac{1}{\left\langle \lambda(\sigma_{n})\,s\right\rangle }\sum_{i=1}^{k}\frac{1}{\sigma_{ni}}\tilde{\lambda}_{s}\cdot\frac{\partial}{\partial\tilde{\lambda}\left(\sigma_{i}\right)}I_{n}=\frac{1}{\left\langle n\,s\right\rangle }\tilde{\lambda}_{s}\cdot\frac{\partial}{\partial\tilde{\lambda}_{n}}\left\langle \wt{\mathcal{V}}_{1}...\wt{\mathcal{V}}_{k}\mathcal{V}_{k+1}...\mathcal{V}_{n}\right\rangle \,,
\]
where we noted that on the suppport of the delta functions in $\mathcal{V}_{n}$, we have
$\left\langle \lambda(\sigma_{n})s\right\rangle =\left\langle n\,s\right\rangle /s_{n}$.
Hence, we find that the correlator in equation \ref{eq:ymr} reduces
to the subleading soft gluon contribution from \cite{Casali:2014xpa}
\[
\left\langle \tilde{\mathcal{V}}_{1}...\tilde{\mathcal{V}}_{k}\mathcal{V}_{k+1}...\mathcal{V}_{n}\mathcal{V}^{ym,1}_s\right\rangle =\left(\frac{1}{\left\langle s\,1\right\rangle }\tilde{\lambda}_{s}\cdot\frac{\partial}{\partial\tilde{\lambda}_{1}}+\frac{1}{\left\langle n\,s\right\rangle }\tilde{\lambda}_{s}\cdot\frac{\partial}{\partial\tilde{\lambda}_{n}}\right)\left\langle \tilde{\mathcal{V}}_{1}...\tilde{\mathcal{V}}_{k}\mathcal{V}_{k+1}...\mathcal{V}_{n}\right\rangle .
\]

\section{Gravity soft limits in the twistorial model} \label{grav4}

As we have seen above, the terms in the soft limit expansion of the integrated vertex operators for gravity correspond to generators for the symmetries of $\scri$. In particular we find generators of translations $\cV_s^0$ at leading order, and generators of superrotations $\cV_s^1$ at subleading order. By imposing the constraints \eqref{constraints}, the equations for the generators \eqref{4dgravV} can be simplified drastically. Moreover, in contraction with the vertex operators introduced above, the pieces $\rho_\alpha\frac{\p}{\p\lambda_\alpha}+\rho_A\frac{\p}{\p\chi_A}$ can be ignored, as there remains always at least one $\tilde\rho$ in one of the vertex operators, which causes the path integral to vanish. Keeping this in mind, the symmetry generators due to a soft graviton are given by
\begin{subequations}
\begin{align}
\cV_s^0&=\frac1{2\pi i}\oint    \frac{\la \xi \, \lambda(\sigma_s)\ra^2 [\tilde\lambda(\sigma_s) \, \tilde \lambda_s]}{\la \xi \,  \lambda_s\ra^2 \la \lambda_s\, \lambda (\sigma_s)\ra} \,, \\
 \cV_s^1&= \frac{1}{2\pi}\oint  \left( \frac{\la \xi \, \lambda(\sigma_s)\ra [\tilde\lambda(\sigma_s) \, \tilde \lambda_s]  [\mu(\sigma_s) \tilde\lambda_s]}{\la \xi \,  \lambda_s\ra \la \lambda_s\, \lambda (\sigma_s)\ra}+\frac{\la \xi\,\lambda(\sigma_s)\ra[\rho\,\tilde\lambda_s][\tilde\rho\,\tilde\lambda_s]}{\la\xi\,\lambda_s\ra\la\lambda_s\,\lambda(\sigma_s)\ra}\right) \,,\\
 \cV_s^2&= \frac1{2\pi i}\oint  \left(  \frac1{2}\frac{ [\tilde\lambda(\sigma_s) \, \tilde \lambda_s]  [\mu(\sigma_s) \tilde\lambda_s]^2}{ \la \lambda_s\, \lambda (\sigma_s)\ra}+\frac{[\rho\,\tilde\lambda_s][\tilde\rho\,\tilde\lambda_s][\mu(\sigma_s)\,\tilde\lambda_s]}{\la\lambda_s\,\lambda(\sigma_s)\ra}\right) \, .
\end{align}
\end{subequations}

\subsection*{Leading terms}

In particular, we can investigate the Ward identity of the first order contribution of an integrated vertex operator in the soft limit, which we have identified with a charge associated to superrotations. For a soft graviton $s$, the superrotation generator is then given by
\begin{equation}
\cV_s^0=\frac1{2\pi i}\oint \rd \sigma_s\frac{\la \xi \, \lambda(\sigma_s)\ra^2 [s\,\tilde\lambda(\sigma_s)]}{\la \xi \,  s\ra^2 \la s\, \lambda (\sigma_s)\ra}.
\end{equation}
We are interested in the Ward identity of
\begin{equation*}
 \left\langle\cV_1\dots \cV_n\cV_s^0\right\rangle\,,
\end{equation*}
for momentum eigenstates, where the equations of motion determine $\lambda(\sigma)$ and $\tilde\lambda(\sigma)$ to be
\begin{equation*}
 \lambda(\sigma)=\sum_{i=1}^k \frac{s_i\lambda_i}{\sigma-\sigma_i}\,,\qquad \tilde\lambda(\sigma)=\sum_{p=k+1}^n \frac{s_p\tilde\lambda_p}{\sigma-\sigma_p}\,.
\end{equation*}
Recall that, from the form of $\lambda(\sigma_s)$ and on the support of the delta-functions, which will eventually be interpreted as the scattering equations, the limit 
\begin{equation*}
 \lim_{\sigma_s\rightarrow\sigma_a}\frac{\la \xi\,\lambda(\sigma_s)\ra}{\la s\,\lambda(\sigma_s)\ra}=\frac{\la\xi\,a\ra}{\la s\,a\ra}\,,\qquad a\in\{1,\dots,n\}\,.
\end{equation*}
Using the residue theorem and the support of the remaining scattering equations, the soft graviton Ward identity then takes the form
\begin{equation}
 \left\langle\cV_1\dots \cV_n \cV_s^0\right\rangle=\sum_{a=1}^n\frac{[as]\la \xi\,a\ra^2}{\la a\,s\ra\la \xi\, s\ra^2}\left\langle\cV_1\dots \cV_n\right\rangle\,,
\end{equation}
which can be identified straightforwardly as the soft graviton contribution. The soft graviton term thus arises from a specific charge generating supertranslation, which can be manifestly identified with the leading order expansion of an insertion of a soft graviton.\\

\subsection*{Subleading terms}

Expanding the integrated graviton vertex operator to first order in the soft momentum $s$ defines a superrotation,
\begin{equation}
\cV_s^1=\frac{1}{2\pi}\oint  \left( \frac{\la \xi \, \lambda(\sigma_s)\ra [\tilde\lambda(\sigma_s) \, \tilde \lambda_s]  [\mu(\sigma_s) \tilde\lambda_s]}{\la \xi \,  \lambda_s\ra \la \lambda_s\, \lambda (\sigma_s)\ra}+\frac{\la \xi\,\lambda(\sigma_s)\ra[\rho\,\tilde\lambda_s][\tilde\rho\,\tilde\lambda_s]}{\la\xi\,\lambda_s\ra\la\lambda_s\,\lambda(\sigma_s)\ra}\right) \,.
\end{equation}
Again, we can investigate the `Ward identity' associated to this superrotation, where we insert $\cV_s^1$ in a correlation function of graviton vertex operators,
\begin{equation}
 \left\langle \tilde{\mathcal{V}}_{1}...\tilde{\mathcal{V}}_{k}\mathcal{V}_{k+1}...\mathcal{V}_{n}\cV_s^1\right\rangle\,.
\end{equation}
Using
\[
\left[\mu\left(\sigma_{s}\right)s\right]\bar{\delta}^{2}\left(\tilde{\lambda}_{i}-s_{i}\tilde{\lambda}\left(\sigma_{i}\right)\right)=\frac{1}{\sigma_{s}-\sigma_{i}}\tilde{\lambda}_{s}\cdot\frac{\partial}{\partial\tilde{\lambda}\left(\sigma_{i}\right)}\bar{\delta}^{2}\left(\tilde{\lambda}_{i}-s_{i}\tilde{\lambda}\left(\sigma_{i}\right)\right)+...
\]
we can calculate the correlation functions easily,
\begin{align*}
 &\left\langle \tilde{\mathcal{V}}_{1}...\tilde{\mathcal{V}}_{k}\mathcal{V}_{k+1}...\mathcal{V}_{n}\cV_s^1\right\rangle\\\
 &\qquad=\frac{1}{2\pi}\left\langle \oint \rd \sigma_s\sum_{i=1}^k \frac{\la \xi\,\lambda(\sigma_s)\ra [\tilde\lambda(\sigma_s)\,s]}{\la \xi\,s\ra \la \lambda(\sigma_s)\,s\ra}\frac1{\sigma_s-\sigma_i}\tilde\lambda_s\cdot\frac{\p}{\p\tilde\lambda(\sigma_i)} \tilde{\mathcal{V}}_{1}...\tilde{\mathcal{V}}_{k}\mathcal{V}_{k+1}...\mathcal{V}_{n}\right\rangle\\
 &\qquad\qquad + \frac{1}{2\pi}\Bigg\langle \oint \rd \sigma_s \sum_{p=k+1}^n \frac{\la \xi\,\lambda(\sigma_s)\ra }{\la \xi\,s\ra \la \lambda(\sigma_s)\,s\ra} I_R \;\tilde{\mathcal{V}}_{1}...\tilde{\mathcal{V}}_{k}\mathcal{V}_{k+1}...\widehat{\mathcal{V}}_{p}...\mathcal{V}_{n} \Bigg\rangle\,,
\end{align*}
where the notation $\widehat{\mathcal{V}}_{p}$ indicates that the integrand of the corresponding vertex operator is omitted from the correlation function, still leaving the integration over the variable $\rd s_p/s_p^3$ and the scattering equation $\bd^2(\lambda_p-s_p\lambda(\sigma_p))$, and where
\begin{align*}
 I_R=[\tilde\lambda(\sigma_s)\,s]\frac{[s\,p]}{\sigma_s-\sigma_p} +(-1)^p\left([\tilde\rho(\sigma_p)\,p][\rho(\sigma_s)\,s]+[\tilde\rho(\sigma_s)\,s][\rho(\sigma_p)\,p]\right)\frac{[s\,p]}{\sigma_s-\sigma_p}.
\end{align*}
Trivially, the derivative $\tilde\lambda_s\cdot\frac{\p}{\p\tilde\lambda(\sigma_i)}$ can be taken to act on all vertex operators, as the only occurence of $\tilde\lambda(\sigma_i)$ is in the scattering equations. Note furthermore that the terms in $I_R$ can be obtained alternatively by acting with $\frac{[s\,p]}{\sigma_s-\sigma_p}\tilde\lambda_s\cdot\frac{\p}{\p\tilde\lambda_p}$ on the vertex operators in the correlation function, with the first term arising from the diagonal elements of $\HH_{pp}$, and the remaining terms from the off-diagonal contributions $\HH_{pq}$. We can thus rewrite the correlation function as 
\begin{align*}
  &\left\langle \tilde{\mathcal{V}}_{1}...\tilde{\mathcal{V}}_{k}\mathcal{V}_{k+1}...\mathcal{V}_{n}\cV_s^1\right\rangle\\\
 &\qquad=\frac{1}{2\pi}\left\langle \oint \rd \sigma_s\sum_{i=1}^k \frac{\la \xi\,\lambda(\sigma_s)\ra [\tilde\lambda(\sigma_s)\,s]}{\la \xi\,s\ra \la \lambda(\sigma_s)\,s\ra}\frac1{\sigma_s-\sigma_i}\tilde\lambda_s\cdot\frac{\p}{\p\tilde\lambda(\sigma_i)} \tilde{\mathcal{V}}_{1}...\tilde{\mathcal{V}}_{k}\mathcal{V}_{k+1}...\mathcal{V}_{n}\right\rangle\\
 &\qquad\qquad + \frac{1}{2\pi}\Bigg\langle \oint \rd \sigma_s \sum_{p=k+1}^n \frac{\la \xi\,\lambda(\sigma_s)\ra }{\la \xi\,s\ra \la \lambda(\sigma_s)\,s\ra} \frac{[s\,p]}{\sigma_s-\sigma_p}\tilde\lambda_s\cdot\frac{\p}{\p\tilde\lambda_p} \;\tilde{\mathcal{V}}_{1}...\tilde{\mathcal{V}}_{k}\mathcal{V}_{k+1}...\widehat{\mathcal{V}}_{p}...\mathcal{V}_{n} \Bigg\rangle\,.
\end{align*}
Now the integral can be calculated straightforwardly, using the explicit expressions for $\lambda(\sigma)$ and $\tilde{\lambda}(\sigma)$, as well as the support of the delta-functions of the vertex operators. Thus the Ward identity for the superrotation charge obtained from the soft expansion of the graviton vertex operator gives the subleading terms of the soft limit,
\begin{equation}
 \left\langle \tilde{\mathcal{V}}_{1}...\tilde{\mathcal{V}}_{k}\mathcal{V}_{k+1}...\mathcal{V}_{n}\cV_s^1\right\rangle=\sum_{a=1}^n\frac{[a\,s]\la\xi\,a\ra}{\la a\,s\ra\la\xi\,s\ra}\tilde\lambda_s\cdot\frac{\p}{\p \tilde\lambda_a}\left\langle \tilde{\mathcal{V}}_{1}...\tilde{\mathcal{V}}_{k}\mathcal{V}_{k+1}...\mathcal{V}_{n}\right\rangle\,.
\end{equation}

\subsubsection*{Sub-subleading terms}
Although there is no symmetry principle to protect the subsubleading terms in the soft expansion of the graviton vertex operators, we can still calculate the corresponding tree-level soft limits. In particular, the Ward identity for the diffeomorphism on ambitwistor space induced by the soft vertex operator to subsubleading order is given by
\begin{align*}
 &\left\langle \tilde{\mathcal{V}}_{1}...\tilde{\mathcal{V}}_{k}\mathcal{V}_{k+1}...\mathcal{V}_{n}\cV_s^{2}\right\rangle\\
 &\qquad= \frac1{2\pi i}\left\langle \oint \rd \sigma_s \frac{[\rho\,s][\tilde\rho\,s]}{\la s\,\lambda(\sigma_s)\ra}\;I_1\;\tilde{\mathcal{V}}_{1}...\tilde{\mathcal{V}}_{k}\mathcal{V}_{k+1}...\mathcal{V}_{n}\right\rangle\\
 &\qquad \qquad + \frac1{2\pi i}\left\langle \oint \rd \sigma_s \frac1{2}\frac{[\tilde\lambda(\sigma_s)\,s]}{\la s\,\lambda(\sigma_s)\ra}\;I_2\;\tilde{\mathcal{V}}_{1}...\tilde{\mathcal{V}}_{k}\mathcal{V}_{k+1}...\mathcal{V}_{n}\right\rangle\,,
\end{align*}
where we have chosen to abbreviate the integrands by
\begin{align*}
 I_1&=\sum_{i=1}^k \frac1{\sigma_{si}}\tilde\lambda_s\cdot\frac{\p}{\p\tilde\lambda(\sigma_i)}+\sum_{p=k+1}^n \frac{[s\,p]}{\sigma_s-\sigma_p}\widehat\cV_p\,,\\
 I_2&=\sum_{i,j=1}^k  \frac1{\sigma_{si}}\tilde\lambda_s\cdot\frac{\p}{\p\tilde\lambda(\sigma_i)} \frac1{\sigma_{sj}}\tilde\lambda_s\cdot\frac{\p}{\p\tilde\lambda(\sigma_j)} + \sum_{p,q=k+1; p\neq q}^n \frac{[s\,p]}{\sigma_{sp}}\frac{[s\,q]}{\sigma_{sq}}\widehat\cV_p\widehat\cV_q\\
 &\quad+ \sum_{p=k+1}^n\sum_{i=1}^k \frac{[s\,p]}{\sigma_{sp}}\frac1{\sigma_{si}}\tilde\lambda_s\cdot\frac{\p}{\p\tilde\lambda(\sigma_i)}\widehat\cV_p\,.
\end{align*}
Again, we have indicated by $\widehat\cV_p$ that the corresponding integrand of the vertex opertor is omitted from the correlation function. Calculating the residues and comparing the results to the derivatives obtained by acting with $\tilde{\lambda}_s\cdot\frac{\p}{\p\tilde\lambda_p}$ for $p\in\{k+1,\dots,n\}$, all unwanted residues cancel and the only contributions are coming from 
\begin{equation}
\begin{split}
 &\left\langle \tilde{\mathcal{V}}_{1}...\tilde{\mathcal{V}}_{k}\mathcal{V}_{k+1}...\mathcal{V}_{n}\cV_s^{2}\right\rangle\\
 &\qquad=\frac1{2}\sum_{a=1}^n \frac{[a\,s]}{\la a\,s\ra}\tilde{\lambda}_s\cdot\frac{\p}{\p\tilde\lambda_a}\tilde{\lambda}_s\cdot\frac{\p}{\p\tilde\lambda_a}\left\langle \tilde{\mathcal{V}}_{1}...\tilde{\mathcal{V}}_{k}\mathcal{V}_{k+1}...\mathcal{V}_{n}\right\rangle\,,
 \end{split}
\end{equation}
which is the subsubleading soft graviton contribution discovered in \cite{Cachazo:2014fwa}.
\chapter{Review of scattering amplitudes in four dimensions}\label{sec5:review}
In this appendix, we will discuss briefly twistor theory and its impact on scattering amplitudes in $\cN=4$ super Yang-Mills and $\cN=8$ supergravity in four dimensions, see \cite{Elvang:2013cua} for a review and references below for more details.

\section{Spinor-helicity formalism and twistor space}\label{sec5:rev-twistor}
\paragraph{The spinor-helicity formalism.}
In four dimensions, the spinor-helicity formalism was used very successfully in the study of scattering amplitudes in $\cN=4$ super Yang-Mills and $\cN=8$ supergravity. It relies on the isomorphism between the restricted Lorentz group SO$(4,\C)$ on complexified space-time and special linear group PSL$(2,\C)$, established by $p_{\alpha\dot{\alpha}}= \sigma^\mu_{\alpha\dot{\alpha}}p_\mu$, where $\sigma_\mu$ denote the Pauli matrices. 
For massless particles, det$(p_{\alpha\dot{\alpha}})=p^\mu p_\mu=0$, and thus the isomorphism relates null vectors on space-time to hermitian matrices of rank one, which may always be decomposed into an outer product of two complex two-dimensional Weyl spinors, 
\begin{equation}
 p_{\alpha\dot{\alpha}} =\lambda_\alpha \tilde{\lambda}_{\dot{\alpha}}\,.
\end{equation}
The isomorphism therefore extends to the double cover of the Lorentz group $(\text{SL}(2,\C)\times\text{SL}(2\C))/\mathbb{Z}_2$. The tangent bundle of complexified Minkowski space $TM\cong S\otimes\tilde{S}$ can therefore be seen as a tensor product\footnote{with the isomorphism constructed from the map $\partial_\mu \stackrel{\sigma}{\longleftrightarrow} \p_{\alpha\dot{\alpha}}$} of the (self-dual and anti self-dual) spin bundles $S$ and $\tilde{S}$, with the two copies of SL$(2,\mathbb{C})$ acting independently on the spin bundles. Since the complexified Lorentz group is locally isomorphic to $\text{SL}(2,\C)\times\text{SL}(2\C)$, all finite-dimensional irreducible representations of the spinor algebra can be classified by a pair of integers or half-integers $(p, q)$. Spinors transforming in $(1/2,0)$ ($(0,1/2)$) are referred to as `negative (positive) chirality' spinors.\footnote{Note that vectors on space-time lie within the $(1/2,1/2)$ representation, hence they are indeed associated with a pair of spinor indices $p^{\alpha\dot{\alpha}}$, consistent with the decomposition for a null vector given above.} Moreover, the spin spaces $S$ and $\tilde{S}$ are equipped with symplectic forms $\epsilon_{\alpha\beta}$ and $\epsilon_{\dot{\alpha}\dot{\beta}}$, which can be used to raise and lower the spinor indices of $\lambda_\alpha\in S$ and $\tilde{\lambda}_{\dot{\alpha}} \in\tilde{S}$. These can be used to construct the SL$(2,\C)$-invariant inner products between spinors of each chirality,
\begin{align}
 \langle1,2\rangle\equiv\langle\lambda_1,\lambda_2\rangle=\epsilon^{\alpha\beta}\lambda_{1,\alpha}\lambda_{2,\beta}\,\quad\text{and}\quad
 [1,2]\equiv[\tilde{\lambda}_1,\tilde{\lambda}_2]=\epsilon^{\dot{\alpha}\dot{\beta}}\tilde{\lambda}_{1,\dot{\alpha}}\tilde{\lambda}_{2,\dot{\beta}}\,.
\end{align}

\paragraph{Twistor space.}
The spinor formalism introduced above provides the fundamental language in which twistor theory is formulated. Twistors space is a complex manifold $P\T$ endowed with holomorphic and algebraic structures encoding space-time points and fields. Physical data on Minkowski space $M$ is represented by complex geometry in twistor space; a correspondence established via the Penrose transform, the Ward correspondence and the non-linear graviton. We will only be able to give a very brief introduction here, for more detail see \cite{Huggett,WardWells} and \cite{Wolf} for lecture notes highlighting the implications for scattering amplitudes.\\

Recall from the previous section that the tangent bundle of complexified Minkowski space is isomorphic to the tensor product of two complex spin bundles, $TM\cong S\otimes\tilde{S}$. Consider now the `correspondence space' give by the projective dual spin bundle, $P\mathbb{\tilde{S}}\cong \mathbb{CP}^1\times\mathbb{C}^4$ with coordinates $(\lambda_{\alpha},x^{\alpha\dot{\alpha}})$. Twistor space is constructed from $P\mathbb{\tilde{S}}$ as the quotient space of the foliation induced by the vector fields $V_{\dot{\alpha}}=\lambda^\alpha \p_{\alpha\dot{\alpha}}$. By construction, the $P\T$,  $P\mathbb{\tilde{S}}$ and $M$ satisfy the double fibration
\begin{equation}\label{eq:doublefibration}
 P\T\stackrel{\pi_{P\T}}{\longleftarrow}P\mathbb{\tilde{S}}\stackrel{\pi_{M}}{\longrightarrow}M\,,
\end{equation}
where $\pi_{M}$ is the trivial projection and $\pi_{P\T}$ is defined by
\begin{align}
 \pi_{P\T}:&(\lambda_{\alpha},x^{\alpha\dot{\alpha}})\mapsto (\lambda_{\alpha},\mu^{\dot{\alpha}})=(\lambda_{\alpha},x^{\alpha\dot{\alpha}}\lambda_{\alpha})\\
 \pi_{M}:&(\lambda_{\alpha},x^{\alpha\dot{\alpha}})\mapsto x^{\alpha\dot{\alpha}}.
\end{align}
The correspondence space is therefore the space linking Minkowski space $M$ with twistor space $P\T$. The relation 
\begin{equation}\label{eq:incidencerelation}
 \mu^{\dot{\alpha}}=ix^{\alpha\dot{\alpha}}\lambda_\alpha,
\end{equation}
defining the projection $\pi_{P\T}$ from $P\mathbb{\tilde{S}}$ to twistor space is known as the {\it incidence relation}, and connects the space-time variables $x_{\alpha\dot{\alpha}}$ to twistor space and allows us to identift $P\T\cong \CP^3$.

The double fibration (\ref{eq:doublefibration}) establishes a geometric correspondence between twistor space and space-time: by considering the pullback and subsequent push-forward, a point in space-time $x\in M$ corresponds geometrically to a line in twistor space. Conversely, a point $Z\in P\T$ is dual to a line in spacetime\footnote{or more accurately, an $\alpha$ plane, defined as a self-dual null 2-plane.} via the incidence relation (\ref{eq:incidencerelation}). This two-way correspondence highlights the fact that two points in space-time are null separated if and only if the dual lines in twistor space intersect. Twistor space is therefore by construction the complexification of the space of null rays. Specifying the light cones on spacetime, which determines the conformal structure, is therefore equivalent to specifying the complex structure on twistor space.  Moreover, it can be seen\footnote{by an extension of the argument given above for the Lorentz group in the spinor formalism} that the conformal group $\text{SO}(4,\C)\cong\text{SL}(4\C)$ acts linearly on the homogeneous coordinates of twistor space, and thus twistors transform in the fundamental representation of the conformal group.\\

There is a natural extension of the twistor space introduced above to a supersymmetric manifold, $P\T\cong \mathbb{CP}^{3|\mathcal{N}}$, with homogeneous coordinated given by
\begin{equation}
 Z=Z_I=(\lambda_\alpha,\mu^{\dot{\alpha}},\eta^A)\in\mathbb{T}=\mathbb{C}^2\times\mathbb{C}^2\times\mathbb{C}^{0|\mathcal{N}}\quad Z\sim tZ,\:t\in\mathbb{C}^*
\end{equation}
The conformal algebra on twistor space becomes enhanced to a superconformal algebra $\mathfrak{sl}(4,\mathcal{N})$, and the incidence relation, incorporating supersymmetry, gets extended to
\begin{equation}
 \mu^{\dot{\alpha}}=ix^{\alpha\dot{\alpha}}\lambda_\alpha\quad\quad \eta^A=\theta^{A\alpha}\lambda_\alpha.
\end{equation}

\paragraph{The Penrose transform.}
The correspondence between physical data on spacetime and complex geometry on twistor space is founded on three theorems: the Penrose transform \cite{Penrose:1969ae}, relating zero rest-mass fields on $M$ to cohomology classes in twistor space, the Ward correspondence \cite{Ward:1977ta} between Yang-Mills instantons on spacetime and holomorphic vector bundles on $P\T$, and the non-linear graviton construction \cite{Penrose:1976js,Ward:1980am}, realising self-dual 4-manifolds as integrable complex structures on twistor space. In the context of this thesis, we will only review the Penrose transform.\\

Zero-rest mass fields of helicity $\pm h$ on spacetime are spinor fields $\phi_{\alpha_1\dots \alpha_n}$, $\phi_{\dot{\alpha}_1\dots \dot{\alpha}_n}$, $\Phi$, with $n=2|h|$ symmetric spinor indices, satisfying the partial differential equations
\begin{align} \label{eq:zrm}
 \p^{\alpha_1 \dot{\alpha}}\phi_{\alpha_1\dots \alpha_n}(x)&=0\,,\qquad\quad\p^{\alpha \dot{\alpha}_1}\phi_{\dot{\alpha}_1\dots \dot{\alpha}_n}(x)=0\,,\qquad\quad \Box \Phi=\p^{\alpha \dot{\alpha}}\p_{\alpha \dot{\alpha}}\Phi=0\,.
\end{align}
The Penrose transform relates these massless fields to the first cohomology class on twistor space:

\begin{thm} \label{thm:Penrose}{\bf (Penrose transform \cite{Penrose:1969ae})} Let $P\T'$ be an open subset of $P\T$, and denote by $\mathbb{M}'$ the corresponding open subset of $\mathbb{M}$, $\mathbb{M}'= \pi_{\mathbb{M}}\circ \pi_{P\T}^{-1}(P\T')$. Then the first cohomology group $H^1(P\T',\mathcal{O}(2h-2))$ is isomorphic to the set of massless fields on $\mathbb{M}'$;
\begin{equation}
 H^1(P\T',\mathcal{O}(2h-2))\cong\{\text{On-shell zero rest mass fields on $\mathbb{M}'$ of helicity $h$}\}.
\end{equation}
\end{thm}

Here, $\mathcal{O}(k)$ denotes the sheaf of holomorphic functions homogeneous of degree $k$. Therefore, gluon wave functions are represented on $P\T$ by holomorphic functions of homogeneity $0$ and $-4$, while gravitons correspond to functions of homogeneity $2$ and $-6$.

\paragraph{Gravity on twistor space.}
As mentioned above, twistors transform in the fundamental representation of the space-time superconformal group SL$(4|\mathcal{N})$, and thus the symmetries of $\mathcal{N}=4$ super Yang-Mills are perfectly encoded in the twistor space geometry. General relativity however is not conformally invariant, so we need to introduce a structure on twistor space to break superconformal invariance. This is achieved by the skew-symmetric infinity twistor $I^{IJ}$ \cite{Mason:2008jy, Cachazo:2012kg, Skinner:2013xp}, determining a metric on spacetime \cite{Penrose:1976jq}.\footnote{To be a bit more specific, a simple skew-symmetric bi-twistor $X^{ab}=Z_1^{[a}Z_2^{b]}$ parametrizes a line in twistor space $P\T$. As a consequence of the incidence relation, this corresponds to a point $x$ in asymptotically flat space-time ($\Lambda=0$). The metric determined by the infinity twistor is then given by
\begin{equation}
 ds^2=\frac{\epsilon_{abcd}dX^{ab}dX^{cd}}{(I_{ef}X^{ef})^2},
\end{equation}
and infinity $\mathscr{I}$ is determined by the surface $I^{ab}X_{ab}=0$. The name `infinity twistor' is derived from the relation $I_{ab}I^{ab}=0$ in flat space-time, which defines itself a point corresponding to $\iota^0$ in conformally compactified space-time.}
Likewise, the fermionic components of the infinity twistor define a metric on the R-symmetry group, which corresponds to a gauge choice of the R-symmetry \cite{Wolf:2007tx}. Choosing $I$ to be block-diagonal in its bosonic and fermionic entries, the bosonic part is required to satisfy
\begin{equation}
 I_{ab}=\frac{1}{2}\epsilon_{abcd}I^{ab}, \qquad I_{ab}I^{bc}=\Lambda\delta_a^c\,.
\end{equation}
In terms of the spinor decomposition of a twistor $Z=(\lambda_\alpha, \mu^{\dot{\alpha}}, \eta_A)$, we have $I_{AB}=\sqrt{\Lambda}\delta_{AB}$, and the bosonic components are given by
\begin{align}
 I^{ab}=\begin{pmatrix}\Lambda \epsilon_{\alpha\beta}&0\\0&\epsilon^{\dot{\alpha}\dot{\beta}}\end{pmatrix}\,, \qquad I_{ab}=\epsilon_{abcd}I^{ab}\begin{pmatrix}\epsilon^{\alpha\beta}&0\\0&\Lambda\epsilon_{\dot{\alpha}\dot{\beta}}\end{pmatrix}\,.
\end{align}
Note that with this choice, the infinity twistor $I$ becomes degenerate of rank two in flat space-time ($\Lambda=0$), whereas it has rank four otherwise. Geometrically, the infinity twistor $I^{IJ}$ and its inverse $I_{IJ}$ define a holomorphic Poisson structure and $\{\:,\:\}$ and a contact structure $\tau$ on twistor space by
\begin{equation}
 \{h_1,h_2\}= I^{IJ}\frac{\p h_1}{\p Z_1^I}\frac{\p h_2}{\p Z_2^J} \qquad \tau=I_{IJ}Z^I\text{d}Z^J\,.
\end{equation}
In a flat Minkowski space-time, the degenerate infinity twistor for $\Lambda=0$ simplifies to
\begin{equation}
 I^{ab}\frac{\p}{\p Z_1^a}\frac{\p}{\p Z_2^b}=\left[\frac{\p}{\p \mu_1}, \frac{\p}{\p \mu_2}\right], \qquad I_{ab}Z_1^aZ_2^b=\langle1\,,2\rangle\,.
\end{equation}

\section{Scattering amplitudes}\label{sec5:review_gravity}
The ideas developed above were immensely successful in the study of scattering amplitudes in $\cN=4$ super Yang-Mills and $\cN=8$ supergravity.  Some key concepts are briefly reviewed here.

\paragraph{$\cN=4$ super Yang-Mills.}
The first step towards a new approach to scattering amplitudes in $\mathcal{N}=4 $ super Yang-Mills theory was the impressive calculation of Parke and Taylor\footnote{proven in \cite{Berends:1987me}, see also \cite{Nair:1988bq} for the first twistor string-like expression inspiring \cite{Witten:2003nn}.} \cite{ParkeTaylor:1986}  demonstrating that the Feynman diagram expression for tree-level colour-ordered scattering of two gluons of negative helicity and $n-2$ of positive helicity exhibits a remarkable simplicity;
\begin{equation}
 \mathcal{A}_{\text{MHV}}(1,2,\dots,n)=\frac{\delta^{4|8}\left(\sum_{i=1}^n\lambda_i\tilde{\lambda}_i\right)}{\prod_{i=1}^n\langle i\,i+1\rangle}\,.
\end{equation}
This provided the base for rapid progress in a multitude of directions, both computational and conceptual. Most importantly, it inspired Witten's twistor string \cite{Witten:2003nn}, leading to a representation of the $n$ particle N$^{d-1}$MHV amplitude\footnote{where we denote the amplitude involving exactly $k$ particles of negative helicity by N${k-2}$MHV.} as an integral over the space of rational curves of degree $d$ in twistor space \cite{Roiban:2004yf,Witten:2003nn}
\begin{equation}\label{eq3a:RSVW}
 \mathcal{A}_{n,d}=\int\frac{\prod_{a=0}^d Z_a}{\text{vol}(\text{GL}(2;\mathbb{C}))}\frac{1}{(12)(23)\dots(n1)}\prod_{i=1}^n a_i(Z(\sigma_i))(\sigma_i d\sigma_i)\,.
\end{equation}
Here, the integral is taken over the moduli space of holomorphic maps $Z:\Sigma\rightarrow P\T$ of degree $d$ from the Riemann sphere $\Sigma$ to supertwistor space, with $Z^I(\sigma)= \sum_{a=0}^dZ_a^I(\sigma^{\underline{1}})^a(\sigma^{\underline{2}})^{d-a}$ for homogeneous coordinates $\sigma^{\underline{\alpha}}=(\sigma^{\underline{1}},\sigma^{\underline{2}})$ on the Riemann sphere. Moreover, $a_i$ are twistor representatives of gluon wave functions, see \cref{thm:Penrose}.

\paragraph{$\cN=8$ supergravity.}
Hodges proved in \cite{Hodges:2012ym, Hodges:2011wm} that MHV tree amplitudes in $\mathcal{N}=8$ supergravity are given by the strikingly simple and elegant generalized determinant formula
\begin{equation}\label{eq:Hodges}
 \mathcal{M}_{\text{MHV}}= \frac{1}{(\langle i\,j\rangle \langle j\,k\rangle\langle k\,i\rangle)^2}\big|\wt{\phi}\big|^{ijk}_{ijk}\:\delta^{4|16} \left(\sum_{i=1}^n\lambda_i\tilde{\lambda}_i\right) \equiv \text{det}'(\wt{\phi})\:\delta^{4|16}\left(\sum_{i=1}^n\lambda_i\tilde{\lambda}_i\right).
\end{equation}
where the `Hodges' matrix $\Phi$ is determined by its entries (with $i\neq j$)
\begin{equation}\label{eq:HodgesMatrix}
 \wt{\phi}_{ij}=\frac{[ij]}{\braket{ij}}\,, \qquad
 \wt{\phi}_{ii}=-\sum_{j\neq i} \phi_{i}\frac{\langle j\,\xi_1\rangle\langle j\,\xi_2\rangle}{\langle i\,\xi_1\rangle\langle i\,\xi_2\rangle} \,.
\end{equation}
Momentum conservation $\sum_i\lambda_i\tilde{\lambda}_i=0$ ensures that the $\wt{\phi}_{ii}$ are well-defined and hence independent of the reference spinors $\xi_{1,2}$. Furthermore, the $n\times n$ symmetric matrix $\wt{\phi}$ has co-rank three due to $\sum_{j}\phi_{ij}\lambda_{j, \alpha}\lambda_{j,\beta}=0$, and thus $\text{det}'(\wt{\phi})$ is independent of the rows and columns removed. This relation further ensures that the gravitational MHV amplitude is fully permutations symmetric. \\

Shortly after Hodges' acclaimed representation of gravity MHV amplitudes, Cachazo and Skinner \cite{Cachazo:2012kg,Cachazo:2012pz} suggested a related formula describing all classical amplitudes of $\mathcal{N}=8$ supergravity. In analogy to Witten's representation for tree amplitudes in $\mathcal{N}=4$ super Yang-Mills, the $n$ particle N$^{d-1}$MHV gravity amplitude is written as an integral over the space of rational maps $Z:\Sigma \rightarrow P\T$ to supertwistor space;
\begin{equation}\label{eq:CS}
 \mathcal{M}_{n,d}=\int\frac{\prod_{a=0}^d D^{4|8}Z_a}{\text{vol}(\text{GL}(2;\mathbb{C}))}\text{det}'(\wt{\Phi})\text{det}'(\Phi)\prod_{i=1}^n h_i(Z(\sigma_i))(\sigma_i d\sigma_i).
\end{equation}
The integrand, involving two generalized determinants, is reminiscent of Hodges' formula for MHV amplitudes (\ref{eq:Hodges}) and the Cachazo-Skinner (CS) formula makes both the $\mathcal{N}=8$ supersymmetry and permutation invariance of gravity amplitudes manifest. Furthermore, the matrices $\wt{\Phi}$ and $\Phi$ explicitly depend on the  infinity twistor, and thereby break the conformal symmetry. Moreover, the $h_i$ are Penrose representatives of graviton wave functions, see \cref{thm:Penrose}. In \cite{Skinner:2013xp}, Skinner demonstrated that this representation of the scattering amplitude is reproduced by a twistor string theory, constructed in analogy to \cite{Witten:2003nn}, but on a split worldsheet supermanifold.\\

Let us examine more carefully the ingredients entering in the gravity amplitudes (\ref{eq:CS}): the $n\times n$ matrix $\wt{\Phi}$ is defined in close analogy to Hodges' matrix $\wt{\phi}$;
\begin{equation}
 \wt{\Phi}_{ij}=\frac{1}{(ij)}\left[\frac{\p}{\p\mu_i},\frac{\p}{\p\mu_j}\right] \,,\qquad
 \wt{\Phi}_{ii}=-\sum_{j\neq i}\wt{\Phi}_{ij} \prod_{r=0}^d\frac{(w_r j)}{(w_r i)}\,.
\end{equation}
Note however that $\wt{\Phi}$ has co-rank $d+2$, therefore the generalized determinant det$'(\wt{\Phi})$ is defined to be any $(n-d-2)\times(n-d-2)$ minor of $\wt{\Phi}$ divided by the Vandermonde determinant of the worldsheet coordinates corresponding to the removed rows and columns,
\begin{equation}
 \text{det}'(\wt{\Phi})=\frac{|\wt{\Phi}_{(n-d-2)}|^{r_1\dots r_{d+2}}_{c_1\dots c_{d+2}}}{|\sigma_{r_1}\dots\sigma_{r_{d+2}}|\:|\sigma_{c_1}\dots\sigma_{c_{d+2}}|}.
\end{equation}

The other crucial ingredient of (\ref{eq:CS}) is the dual $n\times n$ matrix $\Phi$ of rank $d$, which balances the counting of angular and square bracket factors exactly to what is required for gravity amplitudes from the BCFW recursion relation \cite{Cachazo:2012kg, Cachazo:2012pz, Skinner:2013xp}. $\Phi$ is defined to be an $n\times n$ matrix of rank $d$ with entries 
\begin{align}
 \Phi_{ij}=\frac{\langle\lambda(\sigma_i),\lambda(\sigma_j)\rangle}{(ij)}\,,\qquad
 \Phi_{ii}=-\sum_{j\neq i}\Phi_{ij} \prod_{r=0}^{n-d-2}\frac{(u_r j)}{(u_r i)}\prod_{k\neq i,j}\frac{(ki)}{(kj)}\,.
\end{align}
Again, det$'(\Phi)$ denotes any $d\times d$ minor of $\Phi$, divided by the corresponding Vandermonde determinants for the rows and columns remaining in the minor
\begin{equation}
 \text{det}'(\Phi)=\frac{|\Phi_{(d)}|^{r_1\dots r_{d}}_{c_1\dots c_{d}}}{|\sigma_{r_1}\dots\sigma_{r_{d}}|\:|\sigma_{c_1}\dots\sigma_{c_{d}}|}\,.
\end{equation}

The full formula was proven to satisfy the BCFW recursion relation \cite{Cachazo:2012pz}. As it furthermore reproduces the 3-point seed amplitudes and exhibits the correct behaviour at infinity, this is the correct formulation for all tree-level scattering amplitudes in $\mathcal{N}=8$ supergravity. In \cite{Skinner:2013xp}, it has been shown to arise from a twistor string theory, which is closely related to the one discussed in \cref{sec5:gravity}.

\chapter{Loops}
\section{Moduli space of Riemann surfaces} \label{sec6:moduli_space}
Since it will feature frequently, let us briefly discuss the moduli space $\mathscr{M}_{g,n}$ of Riemann surfaces.\footnote{For a review in the context of superstring theory, see for example \cite{Witten:2012bh,Dijkgraaf:1997ip}} In general, there are three complementary ways of describing Riemann surfaces:
\begin{itemize}
 \item As complex curves, hence as topological surfaces with a complex structure\footnote{A linear map satisfying $J^2=-1$ and the integrability condition $\nabla J=0$; and that thus splits the complexified tangent space in holomorphic and antiholomorphic vectors.} $J$, or equivalently a $\dbar$ operator.
 \item As algebraic curves, and thus as solutions to homogeneous polynomial equations $f(x,y,z)=0$ in $\CP^2$.
 \item As surfaces with (a conformal class of) Riemann metrics. This directly makes contact with the first approach since a metric relates to the complex structure via $J_a^b=\sqrt{g}\epsilon_{ac}g^{cb}$, where $g$ is the determinant of the metric. Since the complex structure $J$ is invariant under Weyl rescalings $g_{ab}\rightarrow \rho(x)g_{ab}$, the complex structure is equivalent to the conformal class of the metric.
\end{itemize}
The moduli space of a Riemann surface is defined as the parameter space of all Riemann surfaces, and has dimension
\begin{equation}
 \text{dim}_\C\mathscr{M}_g=\begin{cases}
                     0 & g=0\,,\\ 1 & g=1\,,\\ 3g-3 & g\geqslant 2\,.
                    \end{cases}
\end{equation}
In worldsheet models, we will generally be interested in the moduli space of a Riemann surface with marked points. For $n$ points the moduli space has the dimension
\begin{equation}
 \text{dim}_\C\mathscr{M}_{g,n}=\begin{cases}
                     n-3 & g=0\,,\\ n & g=1\,,\\ n+3g-3 & g\geqslant 2\,.
                    \end{cases}
\end{equation}
Note that, at genus zero, the moduli space is not given by $(\CP^1)^n$ without diagonal elements ($\sigma_{ij}\neq 0$), as one could guessed naively, since $\CP^1$ has a non-trivial group of automorphisms Aut$(\CP^1)=PGL(2,\C)$. These act as M\"{o}bius transformations on the coordinates, and the quotient can be eliminated explicitly by fixing three points - as seen in discussion of scattering amplitudes at genus zero in \cref{chapter2}.\\

The moduli space $\mathscr{M}_{g,n}$ described above is not compact, since it does not contain points on boundary corresponding to degenerations of the Riemann surface. However it allows for the so-called  Deligne-Mumford compactification $\overline{\mathscr{M}}_{g,n}$ \cite{DeligneMumford,Witten:2012bh} that adjoins singular Riemann surfaces at the boundary of the moduli space. To be more specific, this compactification adds singular surfaces where a non-trivial homology cycle is pinched, corresponding to a surface of one genus less and two more marked points, as well as contributions from pinching dividing cycles, separating the Riemann surface into two components with one additional node each. In addition to these non-separating and separating degenerations, the Deligne-Mumford compactification is constructed such that marked points never collide, $\sigma_{ij}\neq 0$. The configuration $\sigma_i- \sigma_j=\varepsilon$ is instead understood, more in the spirit of conformal field theory, as a separating degeneration, with one component containing the marked points $\sigma_i$ and $\sigma_j$ at a finite separation.\footnote{The two descriptions are, of course, equivalent up to coordinate transformations $\sigma\rightarrow \sigma/\varepsilon$, causing the surface to develop a long, thin neck of length $\log \varepsilon$, which in turn is equivalent to the factorisation described in the main text, see e.g. \cite{Polchinski:1998}.}\\

Since the focus of \cref{chapter6} is on the elliptic curve, we will describe briefly the genus one case in more detail. Using language of complex curves, we can represent the elliptic curve as $\Sigma_\tau =\C/\Lambda_{\tau}$, where we quotient by the lattice
\begin{equation}
 \Lambda_{\tau}=\mathbb{Z}\oplus\tau\mathbb{Z}\,,
\end{equation}
and thus identify $z\sim z+1\sim z+\tau$ for Im$(\tau)>0$.\footnote{We have chosen to set one generator of the torus to one since the only invariant is the quotient of the two generators.} Moreover, SL$(2,\mathbb{Z})$ transformations relate different choices of basis generating the same elliptic curve, so we have to further identify by the action of PSL$(2,\mathbb{Z})$. This is acting on the the modulus $\tau$ as
\begin{equation}
 \tau\rightarrow \frac{a\tau +b}{c\tau+d}\,,\quad \text{for }\begin{pmatrix} a&b\\ c&d\end{pmatrix}\in\text{PSL}(2,\mathbb{Z})\,.
\end{equation}
PSL$(2,\mathbb{Z})$ is generated by $T:\tau\rightarrow\tau+1$ and $S:\tau\rightarrow -\frac{1}{\tau}$. The moduli space is thus obtianed as a quotient of the upper half plane by the modular group, which yields the fundamental domain $\mathcal{F}=\{\tau|\,|\tau|\geqslant 1\text{ and }|\Re(\tau)|\leqslant\frac{1}{2}\}$, see \cref{fig6:fund-dom0}.

\begin{figure}[ht]
\begin{center}
 \begin{tikzpicture} [scale=2]
  \filldraw [fill=light-grayII, draw=white] (0.5,0.866) arc [radius=1, start angle=60, end angle= 120] -- (-0.5,2.7) -- (0.5,2.7) -- (0.5,0.866);
  \draw [black] (1.1,0) -- (-1.1,0);
  \draw [black] (0,0) -- (0,1);
  \draw [black,dashed] (0,1) -- (0,1.6);
  \draw [black,dashed] (0,2) -- (0,2.5);
  \node at (0,1.8) {$\mathcal{F}$};
  \node at (0.5,-0.15) {$\frac{1}{2}$};
  \node at (-0.5,-0.15) {-$\frac{1}{2}$};
  \node at (1.195,2.63) {$\tau$};
  \draw (1.1,2.7) -- (1.1,2.55) -- (1.25,2.55);
  \draw [gray] (0,0) arc [radius=1, start angle=0, end angle= 90];
  \draw [black] (1,0) arc [radius=1, start angle=0, end angle= 180];
  \draw [gray] (1,1) arc [radius=1, start angle=90, end angle= 180];
  \draw [gray] (0.5,0) -- (0.5,0.866);
  \draw [gray] (-0.5,0) -- (-0.5,0.866);
  \draw (0.5,0.866) -- (0.5,2);
  \draw (-0.5,0.866) -- (-0.5,2);
  \draw [dashed] (0.5,2) -- (0.5,2.5);
  \draw [dashed] (-0.5,2) -- (-0.5,2.5);
  \draw (0.5,2.5) -- (0.5,2.7);
  \draw (-0.5,2.5) -- (-0.5,2.7);
  \draw [gray] (0,0) arc [radius=0.333, start angle=0, end angle= 180];
  \draw [gray] (-0.334,0) arc [radius=0.333, start angle=0, end angle= 180];
  \draw [gray] (0.666,0) arc [radius=0.333, start angle=0, end angle= 180];
  \draw [gray] (1,0) arc [radius=0.333, start angle=0, end angle= 180];
  \draw [gray] (0,0) arc [radius=0.196, start angle=0, end angle= 180];
  \draw [gray] (-0.608,0) arc [radius=0.196, start angle=0, end angle= 180];
  \draw [gray] (0.392,0) arc [radius=0.196, start angle=0, end angle= 180];
  \draw [gray] (1,0) arc [radius=0.196, start angle=0, end angle= 180];
  \draw [gray] (0,0) arc [radius=0.142, start angle=0, end angle= 180];
  \draw [gray] (-0.716,0) arc [radius=0.142, start angle=0, end angle= 180];
  \draw [gray] (0.284,0) arc [radius=0.142, start angle=0, end angle= 180];
  \draw [gray] (1,0) arc [radius=0.142, start angle=0, end angle= 180];
  \draw [gray] (0,0) arc [radius=0.108, start angle=0, end angle= 180];
  \draw [gray] (-0.784,0) arc [radius=0.108, start angle=0, end angle= 180];
  \draw [gray] (0.216,0) arc [radius=0.108, start angle=0, end angle= 180];
  \draw [gray] (1,0) arc [radius=0.108, start angle=0, end angle= 180];
  \draw [gray] (0.5,0) arc [radius=0.122, start angle=0, end angle= 180];
  \draw [gray] (-0.5,0) arc [radius=0.122, start angle=0, end angle= 180];
  \draw [gray] (0.744,0) arc [radius=0.122, start angle=0, end angle= 180];
  \draw [gray] (-0.256,0) arc [radius=0.122, start angle=0, end angle= 180];
  \draw [gray] (0.5,0) arc [radius=0.060, start angle=0, end angle= 180];
  \draw [gray] (-0.5,0) arc [radius=0.060, start angle=0, end angle= 180];
  \draw [gray] (0.620,0) arc [radius=0.060, start angle=0, end angle= 180];
  \draw [gray] (-0.380,0) arc [radius=0.060, start angle=0, end angle= 180];
  \draw [gray] (0.666,0) arc [radius=0.039, start angle=0, end angle= 180];
  \draw [gray] (0.412,0) arc [radius=0.039, start angle=0, end angle= 180];
  \draw [gray] (-0.588,0) arc [radius=0.039, start angle=0, end angle= 180];
  \draw [gray] (-0.334,0) arc [radius=0.039, start angle=0, end angle= 180];
  \draw [gray] (0.334,0) arc [radius=0.062, start angle=0, end angle= 180];
  \draw [gray] (0.790,0) arc [radius=0.062, start angle=0, end angle= 180];
  \draw [gray] (-0.666,0) arc [radius=0.062, start angle=0, end angle= 180];
  \draw [gray] (-0.210,0) arc [radius=0.062, start angle=0, end angle= 180];
 \end{tikzpicture}
 
\caption{The fundamental domain.}
\label{fig6:fund-dom0}
\end{center}
\end{figure}
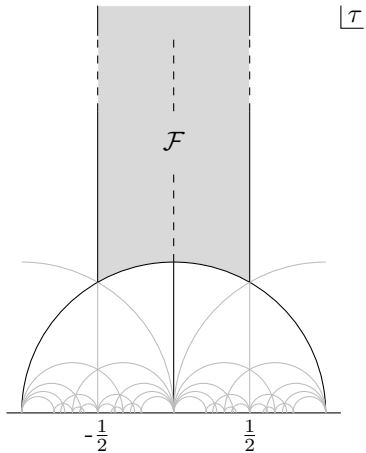  

\section{Checks}
\subsection{Solutions to the four-point 1-loop scattering equations}\label{4pt-soln}

In this appendix, we briefly discuss the solutions to the one-loop scattering equations for $n=4$. There are two solutions to \eqref{SE2n}, in agreement with the counting $(n-1)!-2(n-2)!$. Fixing $\sigma_1=1$, along with the choices $\sigma_{\ell^+}=0$ and $\sigma_{\ell^-}=\infty$ understood in \eqref{SE2n}, these two solutions are given by
\begin{align}
\sigma_2 &= \frac{\ell\cdot k_2 \,(W+1) (k_2\cdot k_3+ \ell\cdot k_4 W)}{W (\ell\cdot k_4 (W
   (\ell+k_1)\cdot k_2+\ell\cdot k_2)+ k_2\cdot k_3\,\ell\cdot k_2)} \nonumber \\
\sigma_3 &= -\frac{(W+1) (\ell\cdot k_4 W-\ell\cdot k_2) (W k_1\cdot k_2+\ell\cdot k_2) (k_2\cdot k_3+\ell\cdot k_4 W)}{W (W
   (k_1\cdot k_2-\ell\cdot k_4)-k_2\cdot k_3+\ell\cdot k_2) (\ell\cdot k_4 (W (\ell+k_1)\cdot k_2+\ell\cdot k_2)+ k_2\cdot k_3\,\ell\cdot k_2)} \nonumber \\
\sigma_4 &= \frac{\ell\cdot k_4 \,(W+1) (W k_1\cdot k_2+\ell\cdot k_2)}{\ell\cdot k_4 (W
   (\ell+k_1)\cdot k_2+\ell\cdot k_2)+ k_2\cdot k_3\,\ell\cdot k_2}  
\end{align}
where $W$ can take the two values
\begin{equation}
W^{(\pm)} = \frac{k_1\cdot k_2\,\ell\cdot k_2-k_2\cdot k_4\,\ell\cdot k_3+k_2\cdot k_3\,\ell\cdot k_4-2 \ell\cdot k_2\,\ell\cdot k_4 \pm \sqrt{4\prod_{i=1}^4 \ell\cdot k_i + \det U}}{2 \ell\cdot k_4
   (k_1\cdot k_2+\ell\cdot (k_1+k_2))},
\end{equation}
\begin{equation}
U = \left( \begin{matrix}
k_i\cdot k_j \;\;& \ell\cdot k_i\\
\ell\cdot k_j  \;\;&  0
\end{matrix} \right) \quad i,j=1,2,3.
\end{equation}
The expressions for $\sigma_ i$ solve $f_2=f_4=0$ for any $W$. The expression for $W$ is then determined by solving $f_3=0$, which takes a quadratic form.

In the case of more general theories, as discussed in section \ref{sec:general-form-one}, there are two additional solutions contributing at four points, in agreement with $(n-2)!$, so that the total number of solutions is $(n-1)!-(n-2)!$. The `regular' solutions are the ones described above, but we should now express them in a different SL$(2,\mathbb{C})$ gauge, where we don't fix both $\sigma_{\ell^+}$ and $\sigma_{\ell^-}$. Let us use coordinates $\sigma'$ such that $(\sigma'_1,\sigma'_2,\sigma'_3)=(0,1,\infty)$. Then we obtain the two `regular' solutions from the expressions above by the change of coordinates
$$
\sigma' = \frac{\sigma_{23}}{\sigma_{21}} \, \frac{\sigma-\sigma_1}{\sigma-\sigma_3}.
$$
For the `singular' solutions, we have $\sigma'_{\ell^+}=\sigma'_{\ell^-}$. The remaining $\sigma'_i$ must satisfy the tree-level scattering equations, so that in our choice $\sigma'_4=-k_1\cdot k_4 / k_1\cdot k_2$. The two solutions for $\sigma'_\ell$ are then determined by
$$
\ell\cdot k_3\,{\sigma'_\ell}^2 + \left( \ell\cdot (k_1+k_4) + \sigma'_4\, \ell\cdot (k_1+k_2) \right) \,\sigma'_\ell  - \sigma'_4\, \ell\cdot k_1 =0.
$$

\subsection{Checks for supersymmetric theories}
\label{sec:checks}

The conjectures for supersymmetric theories can be verified explicitly at low multiplicity. In \cref{sec:factorization}, we will provide further evidence for
these conjectures at any multiplicity, based on the factorisation
properties of the formulae. In this section, however, we will simply
provide details of the numerical checks. These were performed in four dimensions, where
there exist simple known expressions for $N=4$ super Yang-Mills theory
and $N=8$ supergravity. These expressions should match our $D=10$
formulae when we restrict the external data to four dimensions, as we
argue in \cref{sec:dimens-reduct}. We make use of the spinor-helicity
formalism; see \cref{sec5:rev-twistor} and e.g.~\cite{Elvang:2013cua} for a review. The
polarisation vectors for positive and negative helicities are
represented as
\begin{equation}
\epsilon_i^{(+)}= \frac{|\xi\rangle [i|}{\langle i \xi\rangle} \ , \qquad 
\epsilon_i^{(-)}= \frac{|i\rangle [\xi|}{[ i \xi]} \ ,
\label{treenum5}
\end{equation}
where $\xi=|\xi\rangle [\xi|$ is a reference vector. The four-point checks were performed analytically, using the solutions to the scattering equations presented in \cref{4pt-soln}, whereas the five-point checks were performed numerically.

For the theories at hand, due to supersymmetry, the only external helicity configurations which lead to a non-vanishing amplitude have at least two particles of each helicity. We verified that our formulae for both super Yang-Mills theory and supergravity vanish if that condition is not satisfied.

For $n=4$, non-vanishing amplitudes must have two particles of each helicity type. Let us label the negative-helicity particles as $r$ and $s$. The loop integrands for these super Yang-Mills and supergravity amplitudes are well known \cite{Green:1982sw,Bern:2010tq}. After the application of our shift procedure, they are given by
\begin{align}
\hat{\cM}^{(1)}_{4,\,\text{SYM}}=  \frac1{\ell^2} \sum_{\rho\in \text{cyc}(1234)}
 \frac{N_4}{{\prod_{i=1}^{3} }\left(2\ell\cdot\sum_{j=1}^i k_{\rho(i)} +\Big(\sum_{j=1}^i
k_{\rho(i)}\Big)^2\right) }
\end{align}
and
\begin{align}
\hat{\cM}^{(1)}_{4,\,\text{SG}}=  \frac1{\ell^2} \sum_{\rho\in S_4}
 \frac{N_4^2}{{\prod_{i=1}^{3} }\left(2\ell\cdot\sum_{j=1}^i k_{\rho(i)} +\Big(\sum_{j=1}^i
k_{\rho(i)}\Big)^2\right) },
\end{align}
where we sum over cyclic permutations for gauge theory and over all permutations for gravity. The numerator 
\begin{equation}
N_4= \langle rs \rangle^4 \,\frac{[12][34]}{\langle 12 \rangle\langle 34 \rangle}
\end{equation}
is given by a permutation-invariant kinematic function, times the factor $\langle rs \rangle^4$ involving the negative-helicity particles. The fact that this numerator appears squared in gravity with respect to gauge theory is the simplest one-loop example of the BCJ double copy.\footnote{In the supergravity case, we could have distinguished the choice of $r,s$ in $\epsilon_i^\mu$ and $\tilde r,\tilde s$ in $\tilde \epsilon_i^\mu$, with the obvious outcome of substituting $\langle rs \rangle^8$ by $\langle rs \rangle^4\langle \tilde r\tilde s \rangle^4$ in the final result.} We verified that these expressions match our formulae. The amplitude for supergravity follows from the $n$-gon conjecture \eqref{formulangon}. This is due to the fact that, at four points, the quantities $\mathcal{I}^L$ and $\mathcal{I}^R$ are constant \cite{Casali:2014hfa}, as discussed above, each coinciding with the numerator $N_4$.

For $n=5$, we will consider the case of two negative-helicity particles (for two positive helicities, we should simply exchange the chirality of the spinors in the formulae). The complete integrands involve both pentagon and box integrals. In their shifted form, they are given by
\begin{align}
\hat{\cM}^{(1)}_{5,\,\text{SYM}}=\frac{1}{\ell^2} \sum_{\rho\in \text{cyc}(12345)}
\frac{N_{\rho(1)\rho(2)\rho(3)\rho(4)\rho(5)}}{\prod_{i=1}^4 \big(2\ell\cdot\sum_{j=1}^i k_{\rho(i)} +(\sum_{j=1}^i k_{\rho(i)})^2\big)} \,,
\label{5ptsym}
\end{align}
and
\begin{align}
\hat{\cM}^{(1)}_{5,\,\text{SG}}=\frac{1}{\ell^2}& \sum_{\rho\in S_5}
\frac{N_{\rho(1)\rho(2)\rho(3)\rho(4)\rho(5)}}{\prod_{i=1}^4 \big(2\ell\cdot\sum_{j=1}^i k_{\rho(i)} +(\sum_{j=1}^i k_{\rho(i)})^2\big)}\,,
\label{5ptsugra}
\end{align}
where
\begin{equation}
N= N^{\text{pent}}_{\rho(1)\rho(2)\rho(3)\rho(4)\rho(5)} \,+ 
\frac{1}{2} \sum_{i=1}^4 N^{\text{box}}_{[\rho(i),\rho(i+1)]} \,\frac{2\ell\cdot\sum_{j=1}^i k_{\rho(i)} +(\sum_{j=1}^i k_{\rho(i)})^2}{k_{\rho(i)}\cdot k_{\rho(i+1)}} \,.
\end{equation}
A valid choice for the pentagon and box numerators was presented in \cite{Carrasco:2011mn},
\begin{equation}
N^{\text{pent}}_{12345} = \langle rs \rangle^4 \,
\frac{[12][23][34][45][51]}{[12]\langle 23\rangle [34]\langle 41\rangle
-\langle 12\rangle [23] \langle 34\rangle [41]}\,,
\end{equation}
and
\begin{equation}
N^{\text{box}}_{[1,2]} = N^{\text{box}}_{[1,2]345} = N^{\text{pent}}_{12345}-N^{\text{pent}}_{21345}\,.
\end{equation}
The numerator $N^{\text{box}}_{[1,2]}$ corresponds to a box with one massive corner, $K=k_1+k_2$, and is independent of the ordering of 3,4,5. We verified that our expressions match these integrands. There are other choices of numerators leading to the same integrands, such as the choice in \cite{He:2015wgf}, which extends to MHV amplitudes of any multiplicity, and arises as the dimensional reduction of the superstring-derived numerators of \cite{Mafra:2014gja}. In that case, the pentagon numerators depend on the loop momentum, but \eqref{5ptsym} and \eqref{5ptsugra} take the same form, because the relevant shifts are of the type
\begin{equation}
N^{\text{pent}}_{12345}(\ell-k_1)=N^{\text{pent}}_{23451}(\ell).
\end{equation}
Here, we define the loop momentum as flowing between the first and last leg of the numerator, and this behaviour with respect to shifts follows from cyclic symmetry. It is therefore trivial to translate between the shifted representation of the integrand and the standard one.

\subsection{Checks on all-plus amplitudes}
\label{sec:gso-projection}

In \cref{sec:non-supersymm-theor}, we have presented proposals for the integrands of four-dimensional $n$-particle amplitudes in non-supersymmetric gauge theory and gravity. In the gravity case, we distinguished between the cases of pure gravity and the NS-NS sector of supergravity, containing a graviton, a dilaton and a B-field. While we focussed on four dimensions for the sake of being explicit, it is clear that analogous constructions can be made of theories with different matter couplings in various dimensions, including different degrees of supersymmetry if we also introduce fermions.

We checked our conjectures against known expressions for the simplest class of non-supersymmetric four-dimensional amplitudes. These are the amplitudes for which all external legs have the same helicity, which we will choose to be positive. The supersymmetric Ward identities \cite{Grisaru:1979re} lead to the following relations for these non-supersymmetric amplitudes:
\begin{subequations}\label{allplusscalars}
\begin{align}
\cM^{(1)}_{\text{pure-YM}}(\text{all-plus}) =&\, 2\, \cM^{(1,\,\text{scalar})}_{\text{pure-YM}}(\text{all-plus}) \\
\cM^{(1)}_{\text{NS-NS-grav}}(\text{all-plus}) =& \,2\, \cM^{(1)}_{\text{pure-grav}}(\text{all-plus}) =4\, \cM^{(1,\,\text{scalar})}_{\text{pure-grav}}(\text{all-plus}) \ .
\end{align}
\end{subequations}
The superscript on the right-hand side indicates an amplitude where
only one real minimally-coupled scalar is running in the loop. For
gauge theory and for pure gravity, the two helicity states running in
the loop are effectively equivalent to two real scalars, hence the
factor of two, whereas for NS-NS gravity there are two extra states
(dilaton and axion), leading to four real scalars. We checked at four and five points that
\begin{equation}
\pf(M_{3})\big|_{\sqrt{q}}(\text{all-plus})=0 \ .
\label{derallpluszero}
\end{equation}
From this simple fact, it is easy to see that our conjectured expressions satisfy the relations \eqref{allplusscalars}. We believe this to hold for any multiplicity. These statements also apply to amplitudes with one helicity distinct from all others (say one minus, rest plus), which also satisfy the relations \eqref{allplusscalars}.

We have explicitly checked that our conjectures for pure gauge theory and gravity match the (shifted) integrands for all-plus amplitudes in the case of $n=4$. For concreteness, we will write down the integrands explicitly. The Feynman rules for the all-plus amplitudes take a particularly simple form in light-cone gauge, because such amplitudes correspond to the self-dual sector of the theory \cite{Chalmers:1996rq,Cangemi:1996rx}. The rules for the vertices and external factors in all-plus amplitudes in gauge theory can be taken to be \cite{Boels:2013bi}
\begin{equation}
(i^+,j^+,k^-)=X_{i,j} \, f^{a_i a_j a_k}, \qquad e_i^{(+)}=\frac{1}{\langle \xi i\rangle^2},
\label{extstates}
\end{equation}
whereas in gravity they are
\begin{equation}
(i^{++},j^{++},k^{--})=X_{i,j}^2, \qquad e_i^{(++)}=\frac{1}{\langle \xi i\rangle^4} \ .
\end{equation}
We are again making use of the spinor helicity formalism, and taking $\xi=|\xi\rangle [\xi|$ to be a reference vector. Gauge invariance implies that the amplitudes are independent of the choice of $\xi$. The object $X_{i,j}$ is defined with respect to the spinors $|\hat{i}]=K_i |\xi\rangle$, which can be defined for any (generically off-shell) momentum $K_i$ using the reference spinor $|\xi\rangle$,
\begin{equation}
X_{i,j}=-[\hat{i}\hat{j}] = - X_{j,i}\ .
\end{equation}
The direct ``square'' relation between the rules in gauge theory and in gravity makes the BCJ double copy manifest for these amplitudes \cite{Monteiro:2011pc}.

Using the diagrammatic rules above, we can write the (shifted) integrand for the single-trace contribution to gluon scattering as
\begin{align}
\hat{\cM}^{(1,\,\text{scalar})}_{\text{pure-YM}}(1^+2^+3^+4^+) =  \frac1{ \ell^2 \prod_{i=1}^n \langle \xi i\rangle^2}  \sum_{\rho\in \text{cyc}(1234)} I_{\rho(1)\rho(2)\rho(3)\rho(4)} \, ,
\end{align}
with
\begin{equation}
 I_{\rho(1)\rho(2)\rho(3)\rho(4)}^{\text{YM}}=I^{\text{box-YM}}_{\rho(1)\rho(2)\rho(3)\rho(4)} +I^{\text{tri-YM}}_{[\rho(1),\rho(2)]\rho(3)\rho(4)} + \frac1{2}\,I^{\text{bub-YM}}_{[\rho(1),\rho(2)][\rho(3),\rho(4)]}\,,
\end{equation}
and 
\begin{align}
I^{\text{box-YM}}_{1234} =&  \frac{X_{\ell,1}X_{\ell+1,2}X_{\ell-4,3}X_{\ell,4}}{(2\ell\cdot k_1)(2\ell\cdot (k_1+k_2)+2k_1\cdot k_2)(-2\ell\cdot k_4)} 
\ , \nonumber \\
I^{\text{tri-YM}}_{[1,2]34} =& \frac{X_{1,2}}{2k_1\cdot k_2} \Bigg(
 \frac{X_{\ell,1+2}X_{\ell,3}X_{\ell+3,4}}{(2\ell\cdot k_3)(2\ell\cdot (k_3+k_4)+2k_3\cdot k_4)}
\nonumber \\
& \quad + \frac{X_{\ell,1+2}X_{\ell-4,3}X_{\ell,4}}{(-2\ell\cdot (k_3+k_4)+2k_3\cdot k_4)(-2\ell\cdot k_4)}   + \frac{X_{\ell+4,1+2}X_{\ell,3}X_{\ell,4}}{(-2\ell\cdot k_3)(2\ell\cdot k_4)}
 \Bigg) \ ,\nonumber \\
I^{\text{bub-YM}}_{[1,2][3,4]} =& \frac{X_{1,2}X_{3,4}X_{\ell,1+2}X_{\ell,3+4}}{(2k_1\cdot k_2)^2} 
\Bigg( \frac1{2\ell\cdot (k_1+k_2)+2k_1\cdot k_2} + \frac1{-2\ell\cdot (k_1+k_2)+2k_1\cdot k_2}  \Bigg)\nonumber\ .
\end{align}
Notice that there is no contribution from bubbles in the external legs. As discussed in \cite{He:2015yua} in the context of the bi-adjoint scalar theory, such contribution must be proportional to the tree-level amplitude, which vanishes for the all-plus helicity sector. We should mention that the `singular' solutions of the scattering equations, for which $\sigma_{\ell^+}=\sigma_{\ell^-}$, give a directly vanishing contribution to the all-plus loop integrand.

For the scattering of gravitons, we have
\begin{align}
\hat{\cM}^{(1,\,\text{scalar})}_{\text{pure-grav}}(1^+2^+3^+4^+)=  \frac1{ \ell^2 \prod_{i=1}^n \langle \xi i\rangle^4} \sum_{\rho\in S_4}
I_{\rho(1)\rho(2)\rho(3)\rho(4)}^{\text{grav}} \ ,
\end{align}
where 
\begin{equation}
 I_{\rho(1)\rho(2)\rho(3)\rho(4)}^{\text{grav}}=I^{\text{box-grav}}_{\rho(1)\rho(2)\rho(3)\rho(4)} +\frac{1}{2}I^{\text{tri-grav}}_{[\rho(1),\rho(2)]\rho(3)\rho(4)} + \frac1{4}\,I^{\text{bub-grav}}_{[\rho(1),\rho(2)][\rho(3),\rho(4)]}\,,
\end{equation}
and $I^{\text{box-grav}}$, $I^{\text{tri-grav}}$ and $I^{\text{bub-grav}}$ are respectively obtained from $I^{\text{box-YM}}$, $I^{\text{tri-YM}}$ and $I^{\text{bub-YM}}$ via the substitution $X_{.,.}\to (X_{.,.})^2$ in the numerators.

\section{Motivation from ambitwistor heterotic models}
\label{sec:motiv-form-ambitw}

The single trace sector of the heterotic ambitwistor model was used to
derive the CHY formulae for gluon amplitudes in
\cite{Mason:2013sva}, see also \cref{chapter2}. It was however noted that generically these
amplitudes contain unphysical gravitational degrees of
freedom, leading in particular to multi-trace interactions, absent
from Yang-Mills theories. At one loop, the presence of these would-be
gravitational interactions leads to a double pole at the boundary of
the moduli space $dq/q^{2}$, coming from the bosonic sector of the
theory. In string theory, the level matching prevents these tachyonic
modes from propagating, and heterotic models were used to write down a
set of rules to compute gluon amplitudes in the 90's
\cite{Bern:1991aq}. Here this double pole simply renders the theory
ill-defined.

One could hope that a subsector of the theory may be well defined at
one-loop, just like at tree-level, but this is not the case.
Even by restricting to the single-trace sector, these additional states do
not automatically decouple, so we have to be more careful when attempting to extract a portion of
the heterotic amplitude. Let us start by writing it out;

\begin{equation}
  \label{eq:het-one-loop}
  \mathcal{M}^{\mathrm{n-gluon}\,(!)}_{\mathrm{1-loop}}= \int
  \frac{dq}q \prod_{i}d z_{i} 
 \frac12 \sum_{\pmb\alpha}
  Z^{\mathrm{Het}}_{\pmb{\alpha}}
  \mathrm{Pf}{(M^{\mathrm{col}}_{\pmb{\beta}})}\times 
  \frac12\sum_{\pmb{\beta}}
  Z_{\pmb\beta}\mathrm{Pf}({M}_{\pmb\beta} )\,,
\end{equation}
where the symbol $(!)$  emphasizes that this is not a well
defined amplitude in a well defined theory. Nevertheless, we shall try to
extract physically menaingful parts of it below.
The matrix $M$ is the ``kinematical'' one of
\cref{eq:Malpha}, while the partition functions $Z_{\pmb\alpha}$
were defined in \cref{eq6:pt-funs}
The new ingredient here is the colour part which contains a partition
function for the $32$ Majorana-Weyl fermions that realise the current
algebra and the colour Pfaffian coming from Wick contractions between
them. The partition functions are given by
\begin{equation}
  \label{eq:het-partfun}
 Z^{\mathrm{Het}}_{\pmb\alpha}=\frac{\theta_{\alpha}(0|\tau)^{16}}{\eta^{16}(\tau)}\frac{1}{\eta(\tau)^{8}}\,,
\end{equation}
The ``colour'' Pfaffian is built by applying Wick's theorem
to the gauge currents
\begin{equation}
  \label{eq:gauge-currents}
  J^{a}=T^{a}_{ij}\psi^{i}\psi^{j}, \quad i,j=1,\ldots,16\,,
\end{equation}
using the fermion propagator
\begin{equation}
  \langle \psi^{i}(z)\psi^{j}(0)\rangle_{\pmb\alpha} = S_{\pmb\alpha}(z|\tau)\,,
\end{equation}
in the spin structure $\pmb\alpha$. In all, we obtain the Pfaffian of a
matrix with elements $T^{a}_{ij}\,S_{\pmb\alpha}(z_{i})-z_{j}|\tau)$.
  However, contrary to the case of kinematics where the compact Pfaffian structure is useful despite obscuring the supersymmetry of the amplitude, the colour structure is most conveniently expressed by the colour ordering decomposition. To mnaifest this and decouple as many gravity states as possible,
  we shall from now on restrict our attention to one particular single trace term
  in this Pfaffian, $\mathrm{Tr}(T^{1}T^{2}T^{3}T^{4})$, such that
\begin{equation}
  \label{eq:pftosingletrace}
  \mathrm{Pf}_{\alpha}\to \prod_{i=1}^{4}
  S_{\alpha}(z_{i}-z_{i+1}|\tau) \mathrm{Tr}(T^{1}T^{2}T^{3}T^{4})\,,\quad z_{5}\equiv z_{1}\,.
\end{equation}

To apply the integration by parts procedure and obtain the integrands on the nodal Riemann
sphere, we need the following $q$-expansions;
\begin{equation}
  \label{eq:het-partfun-qexp}
  \begin{aligned}
    \frac{\theta_{2}(0|\tau)^{16}}{\eta^{24}(\tau)} &= 2^{16}
    q+O(q^{2}),\\
    \frac{\theta_{3}(0|\tau)^{16}}{\eta^{24}(\tau)} &= \frac 1q +\frac{32}{\sqrt{q}}+504+  O(q),\\
    \frac{\theta_{4}(0|\tau)^{16}}{\eta^{24}(\tau)} &= \frac 1q -\frac{32}{\sqrt{q}}+504+O(q),\\
  \end{aligned}
\end{equation}
Since the Pfaffians constructed from Szeg\H{o} kernels are holomorphic in
$q$, the spin structure $\pmb\alpha=2$ does not contribute to the colour decomposition after applying the residue theorem. Using the $q$-expansions of
the Szeg\H{o} kernels given in \cref{eq:szegolimitsl2c} in their torus parametrization,
\begin{subequations}\label{eq:szego-exp}
\begin{align}
  &S_{3} = \frac{\pi }{\sin (\pi  z)}+4 \pi  \left(-\sqrt{q}+q\right)
  \sin (\pi  z)\,,\\
&S_{4}=\frac{\pi }{\sin (\pi  z)}+4 \pi  \left(\sqrt{q}+q\right) \sin (\pi  z)\,,
\end{align}
\end{subequations}
we have
$Z_{3}^{\mathrm{Het}}\mathrm{Pf}(M^{\mathrm{col}})_{3}\big|_{q^{-1/2}}=
-Z_{4}^{\mathrm{Het}}\mathrm{Pf}(M^{\mathrm{col}})_{4}\big|_{q^{-1/2}}$.

Thus, the integrand of \eqref{eq:het-one-loop} has a double pole in $q$
with coefficient $2\mathrm{Pf}_{3}\big|_{q^{0}}$ corresponding to the unphysical gravity degrees of freedom,
and a single pole in $q$, given by the leading order contribution. 

The terms contributing at leading order are
\begin{subequations}
\begin{align}
  \frac{32}{\sqrt{q}}&\to 32 \frac{\sin(\pi
                      z_{12})} {\sin(\pi z_{23}) \sin(\pi
                      z_{34}) \sin(\pi z_{41})} + \mathrm{
                      cyclic} \,,\\
  \frac 1q & \to  \frac{\sin(\pi
             z_{12})} {\sin(\pi z_{23}) \sin(\pi
             z_{34}) \sin(\pi z_{41})}  +\frac{\sin^{2}}{\sin^{2}} + \mathrm{
             perms}\,,\\
  504 & \to \frac{504}{\sin(\pi
                      z_{12})\sin(\pi z_{23}) \sin(\pi
                      z_{34}) \sin(\pi z_{41})}\,,
\end{align}
\end{subequations}
where $\frac{\sin^{2}}{\sin^{2}}$ indicate terms of the form
$\frac{\sin(\pi z_{12}) \sin(\pi z_{34})}{\sin(\pi z_{13}) \sin(\pi
  z_{24}) } $.
After the usual change of variables, the $"\sin/\sin^{3}"$ terms give
rise to the Parke-Taylor part of the Yang-Mills integrands given in
\eqref{eq:PTloop-def} (including the reversed ones);
\begin{equation}
  \label{eq:sin-sincube}
  \frac{1}{\sigma_{1}\sigma_{2}\sigma_{3}\sigma_{4}}\frac{\sin(\pi
    z_{12})} {\sin(\pi z_{23}) \sin(\pi
    z_{34}) \sin(\pi z_{41})} =
  \frac{\sigma_{41}}{\sigma_{1}\sigma_{12} \sigma_{23}\sigma_{34}\sigma_{4}}\,,
\end{equation}
where an additional factor of
$(\sigma_{1}\sigma_{2}\sigma_{3}\sigma_{4})^{-1}$ has been included from
the measure $\prod d\sigma_{i}/\sigma_{i}^{2}$. Note also that the
counting for these terms produces a numerical factor of $(32-1)=31$,
which, after suitable counting of the powers of 2, builds up 
\begin{equation}
  496 = 2^{4}\times 31
\end{equation}
which is the dimension of the adjoint\footnote{Note that the gluons are in the adjoint representation of the gauge
group in these models.} of $SO(32)$. The fact that loops
in gauge theories come with a factor of $N$ at leading order is
well-known.

However, we find additional terms. We haven't been able to determine
their origin with precision, but we suspect that they could originate
from bi-adjoint scalars running in the loop, if they are not simple
artefacts of the inconsistency of the model.

\section{The NS part of the integrand}\label{sec:MNS}
This appendix provides the proof for the equivalence of the two expressions for the NS sector of the integrand;
\begin{equation}
 \cI^{\text{NS}}\equiv(d-2)\pf(M_3)\,\big|_{q^0}+\pf(M_3)\,\big|_{\sqrt{q}} =\sum_r\pf'(M_{\text{NS}}^r)\,,
\end{equation}
with $\pf'(M_{\text{NS}})= \frac{1}{\sigma_{\ell^+\,\ell^-}}\,\pf({M_{\text{NS}}}_{(\ell^+,\ell^-)})$ and $d$ denoting the dimension of the space-time. Recall the definition of the matrix $M_{\text{NS}}$ given in \cref{sec:analys-indiv-gso};
\begin{equation}
 M_{\text{NS}}^r=\begin{pmatrix}A & -C^T\\ C & B \end{pmatrix}
\end{equation}
where the elements of $M_{\text{NS}}^r$ were defined by
\begin{align*}
 &A_{\ell^+\ell^+}=A_{\ell^-\ell^-}=0 && A_{\ell^+\ell^-}=0 && A_{\ell^+ i}=\frac{\ell\cdot k_i}{\sigma_{\ell^+i}} && A_{\ell^- i}=-\frac{\ell\cdot k_i}{\sigma_{\ell^-i}}\\
 & && && A_{ij}=\frac{k_i\cdot k_j}{\sigma_{ij}} && A_{ii}=0\\
 &B_{\ell^+\ell^+}=B_{\ell^-\ell^-}=0 && B_{\ell^+\ell^-}=\frac{d-2}{\sigma_{\ell^+\ell^-}} && B_{\ell^+ i} =\frac{\epsilon^r\cdot\epsilon_i}{\sigma_{\ell^+i}} && B_{\ell^- i}=\frac{\epsilon^r\cdot \epsilon_i}{\sigma_{\ell^-i}}\\
 & && && B_{ij}=\frac{\epsilon_i\cdot \epsilon_j}{\sigma_{ij}} && B_{ii}=0\\
 &C_{\ell^+\ell^+}=-\epsilon^r\cdot P(\sigma_{\ell^+}) && C_{\ell^+\ell^-}=-\frac{\epsilon^r\cdot\ell}{\sigma_{\ell^+\ell^-}} && C_{\ell^+ i}=\frac{\epsilon^r\cdot k_i}{\sigma_{\ell^+i}} && C_{i\ell^+}=\frac{\epsilon_i\cdot \ell}{\sigma_{i\ell^+}}\\
 &C_{\ell^-\ell^-}=-\epsilon^r\cdot P(\sigma_{\ell^-}) && C_{\ell^-\ell^+}=\frac{\epsilon^r\cdot\ell}{\sigma_{\ell^-\ell^+}} && C_{\ell^- i}=-\frac{\epsilon^r\cdot k_i}{\sigma_{\ell^-i}} && C_{i\ell^-}=-\frac{\epsilon_i\cdot \ell}{\sigma_{i\ell^-}}\\
 & && && C_{ij}=\frac{\epsilon_i\cdot k_j}{\sigma_{ij}} && C_{ii}=-\epsilon_i\cdot P(\sigma_i)\,.
\end{align*}
Using the recursive definition of the Pfaffian,
\begin{equation}
 \pf(M)=\sum_{j\neq i}(-1)^{i+j+1+\theta(i-j)}m_{ij}\,\pf(M_{(ij)})\,,
\end{equation}
we can expand the reduced\footnote{For convenience, we have chosen to remove the rows and columns associated to $\ell^+$ and $\ell^-$ in the reduced matrix.}  Pfaffian $\pf'(M_{\text{NS}}^r)= \frac{1}{\sigma_{\ell^+\,\ell^-}}\,\pf({M_{\text{NS}}^r}_{(\ell^+,\ell^-)})$, in the remaining rows associated to $\ell^+$ and $\ell^-$.
\begin{align*}
 \pf'(M_{\text{NS}}^r)=&\frac{1}{\sigma_{\ell^+\ell^-}} B_{\ell^+\ell^-} \,\pf(M_{(\ell^+\ell^-{\ell^+}'{\ell^-}')}) \\
 & + \frac{1}{\sigma_{\ell^+\ell^-}} \hspace{-7pt}\sum_{i,j\neq {\ell^+}',{\ell^-}'} \hspace{-12pt}(-1)^{1+i+j+\theta(n+1-i)+\theta(n+2-j)+\theta(i-j)} m_{{\ell^+}'i} m_{{\ell^-}'j}\, \pf(M_{(\ell^+\ell^-{\ell^+}'{\ell^-}'ij)})\,.
\end{align*}
To briefly comment on the notation used for $M_{\text{NS}}$, rows and columns in $\{1,\dots,n+2\}$ are denoted by indices $i$, whereas we use the conventional $i'=n+2+i$ for rows in $\{n+3,\dots,2(n+2)\}$. Now note that after summing over $r$, the entries and prefactors simplify,
\begin{align*}
 (-1)^{1+i+j+\theta(n+1-i)+\theta(n+2-j)+\theta(i-j)}&=\begin{cases} (-1)^{1+i+j+\theta(i-j)} & i,j\in I\,,\\
								    (-1)^{1+i+j+\theta(i-j)} & i,j\in I'\,,\\
								    (-1)^{i+j} & i\in I',\, j\in I\,,
                                                      \end{cases}\\
\sum_r m_{{\ell^+}'i} m_{{\ell^-}'j}&=\begin{cases} \frac{k_i\cdot k_j}{\sigma_{\ell^+i}\sigma_{\ell^-j}} & i,j\in I\,,\\
				       \frac{\epsilon_i\cdot \epsilon_j}{\sigma_{\ell^+i}\sigma_{\ell^-j}} & i,j\in I'\,,\\
				       \frac{\epsilon_i\cdot k_j}{\sigma_{\ell^+i}\sigma_{\ell^-j}} & i\in I',\, j\in I\,,
                         \end{cases}
\end{align*}
where $I=\{1,\dots,n\}$, $I'=\{1',\dots,n'\}$. Moreover, using $M_{(\ell^+\ell^-{\ell^+}'{\ell^-}')}=M_3$ and choosing for simplicity $\sigma_{\ell^+}=0$, $\sigma_{\ell^-}=\infty$, the result further simplifies to
\begin{equation}
 \begin{split}
 \sum_r \pf'(M_\text{NS}^r)=(d-2)\,\pf(M_3)& + \sum_{i,j\in I}(-1)^{1+i+j+\theta(i-j)} \frac{k_i\cdot k_j}{\sigma_i} \,\pf({M_3}^{ij}_{ij})\\
  &+ \sum_{i,j\in I'}(-1)^{1+i+j+\theta(i-j)} \frac{\epsilon_i\cdot \epsilon_j}{\sigma_i} \,\pf({M_3}^{i'j'}_{i'j'})\\
  &+\sum_{i\in I',j\in I}(-1)^{i+j} \left(\frac{\epsilon_i\cdot k_j}{\sigma_i}-\frac{\epsilon_i\cdot k_j}{\sigma_j}\right) \,\pf({M_3}^{i'j}_{i'j})\,.
 \end{split}
\end{equation}
To see that, as claimed above,
\begin{equation}
  \sum_r \pf'(M_\text{NS}^r)=(d-2)\, \pf(M_3) \big|_{q^0}+ \pf(M_3) \,\big|_{\sqrt{q}}\,,
\end{equation}
expand the Pfaffian on the RHS to order $\sqrt{q}$, using
\begin{equation}
 \pf(M)=\sum_{\alpha\in\Pi}\text{sgn}(\alpha) m_{i_1j_1}\dots m_{i_nj_n}=\frac{1}{2}\sum_{i_k,j_k}(-1)^{1+i_k+j_k+\theta(i_k-j_k)}m_{i_kj_k}\pf(M_{i_kj_k}^{i_kj_k})\,.
\end{equation}
This leads to
\begin{align*}
 \pf(M_3) \,\big|_{q^{1/2}}&=\frac{1}{2}\sum_{i_k,j_k}(-1)^{1+i_k+j_k+\theta(i_k-j_k)}m_{i_kj_k} \,\big|_{\sqrt{q}}\,\,\pf({M_3}_{i_kj_k}^{i_kj_k}) \,\big|_{q^{0}}\\
 &= \sum_{i,j\in I} (-1)^{1+i+j+\theta(i-j)} \frac{k_i\cdot k_j}{\sigma_i} \,\pf({M_3}^{ij}_{ij})\\
  &\qquad+ \sum_{i,j\in I'}(-1)^{1+i+j+\theta(i-j)} \frac{\epsilon_i\cdot \epsilon_j}{\sigma_i} \,\pf({M_3}^{i'j'}_{i'j'})\\
  &\qquad +\sum_{i\in I',j\in I}(-1)^{i+j} \left(\frac{\epsilon_i\cdot k_j}{\sigma_i}-\frac{\epsilon_i\cdot k_j}{\sigma_j}\right) \,\pf({M_3}^{i'j}_{i'j})\,.
\end{align*}
The diagonal terms $C_{ii}=\epsilon_i\cdot P(\sigma_i)$ do not contribute since $P(\sigma)$ exclusively contains terms of the form $k_i \tilde{S}_1\sim k_i \frac{\theta'_1}{\theta_1}$, which only contribute at higher order in $q$. This concludes the proof of \cref{eq:PfNs-def}.

\section{Dimensional reduction}
\label{sec:dimens-reduct}

In this section we discuss considerations that are general and
somewhat out of the scope of the main article.  The point is to
discuss how one can dimensionally reduce the ambitwistor string to $d$
dimensions. At the core of these considerations is the work of
ACS~\cite{Adamo:2014wea} where the ambitwistor string was be
formulated in generic (on-shell) curved spaces; toroidal
compactifications are just a subcategory of the latter spaces.

In the usual string compactified on a circle of radius $R$, wrapping
modes or worldsheet instantons are solutions that obey the periodicity
conditions $X=X+2\pi m R$. Their classical values are given by
$X_{class}=2\pi R(n \xi_{1}+m \xi_{2})$ where $\xi_{1,2}$ are the
worldsheet coordinates, such that $z=\xi_{1}-i\xi_{2}$. This cannot be
made holomorphic, except if $n=m=0$, therefore none of these can
contribute in the ambitwistor string, which is quintessentially holomorphic.
This may not exclude the possibility of having other type of more
exotic instantons, as mentioned in the final section of
\cite{Adamo:2014wea}, but we shall proceed here under the assumption that
none of these are generated.
In total, the Kaluza-Klein reduction of the amplitude
above~\eqref{e:graveven} is simply obtained by replacing the
$10$-dimensional loop momentum integral by a $d$-dimensional one and a
$(10-d)$-dimensional discrete sum
\begin{align}
\label{eq:KKreduction}
  \int d^{10} \ell \to \frac{1}{R^{10-d}} \sum_{\pmb{n}\in \mathbb{Z}^{{10-d}}}
  \int d^{d}\ell_{(d)}
\end{align}
where $\pmb{n}$ is an integer valued $(10-d)$-dim vector. A $10-d$ torus
with $10-d$ different radii $R_{1},\ldots,R_{10-d}$ is dealt with at
the cost of minor obvious modifications of the previous expression.

The loop momentum square is then given by
\begin{align}
\label{eq:lsqKK}
\ell^{2}= \ell_{(d)}^{2}+\frac{ {\pmb{n}}^{2}}{R^{2}}
\end{align}

In this way, the transformation rule $\ell\to\tau\ell$ of the loop
momentum after a modular transformation is generalized to the compact
dimensions by demanding $R\to R/\tau$ and the integral is still
modular invariant, in the sense of~\cite{Adamo:2013tsa}.

Ultimately, we take the radius $R$ of the torus to zero in order to
decouple the KK states. In this limit, $\ell$ simply becomes
$\ell_{d}$  wherever it appears, therefore this process
is achieved by, loosely speaking, restricting the loop momentum integral
by hand.

In conclusion, standard compactification techniques of string theory
on tori and orbifolds thereof apply straightforwardly.

\singlespacing
\addcontentsline{toc}{chapter}{References}
\renewcommand{\bibname}{References}
\bibliography{twistorbib} 
\bibliographystyle{utphys}

\end{document}